\documentclass{article}
\usepackage{graphicx} %
\usepackage{amsthm,amsfonts,amsmath,amssymb}
\usepackage{thmtools}
\usepackage{dsfont}
\usepackage{bbm,bm}
\usepackage{float}
\usepackage{xcolor}
\usepackage{graphicx} 
\usepackage{tikz}
\usepackage{ragged2e}
\usepackage{xspace}
\usepackage{mdframed}
\usepackage{tabularx}
\usepackage{booktabs}
\usepackage[
a4paper,
top=1in,
bottom=1in,
left=1in,
right=1in]{geometry}
\usepackage{comment}
\usepackage{svg}

\usepackage[T1]{fontenc}

\let\epsilon\varepsilon

\usepackage{algorithm}
\usepackage{algpseudocode}
\usepackage{thmtools}

\definecolor{toc}{RGB}{13,55,174}
\definecolor{myPurple}{RGB}{175,0,124}

\usepackage[hidelinks]{hyperref}	
\hypersetup{
    colorlinks=true,
    citecolor=myPurple,
    linkcolor=toc,
    urlcolor=toc
}

\usepackage[shortlabels]{enumitem}		%
\usepackage{fontawesome5}   %
\usepackage{adjustbox}
\allowdisplaybreaks			%
\usepackage{subcaption}     %
\usepackage{mathtools}
\usepackage{relsize}
\usepackage{nicefrac}

\usepackage[symbol*]{footmisc}

\usepackage{thm-restate}
\newtheorem{theorem}{Theorem}[section]
\newtheorem{lemma}[theorem]{Lemma}
\newtheorem{corollary}[theorem]{Corollary}
\newtheorem{claim}[theorem]{Claim}

\newtheorem{property}[theorem]{Property}
\newtheorem{remark}[theorem]{Remark}
\newtheorem{definition}[theorem]{Definition}

\newtheorem{observation}[theorem]{Observation}

\newcommand{\poly}{\textnormal{poly}}

\newcommand{\OPT}{\mathsf{OPT}}
\newcommand{\opt}{\mathsf{opt}}

\newcommand{\cost}{\textnormal{cost}}

\newcommand{\reals}{\mathbb R}

\newcommand{\eps}{\epsilon}

\renewcommand{\AA}{\mathbb{A}}
\newcommand{\BB}{\mathbb{B}}

\newcommand{\FF}{\mathbb{F}}

\renewcommand{\SS}{\mathbb{S}}

\newcommand{\ind}{\mathds{1}}

\def\<{\langle}
\def\>{\rangle}

\providecommand*{\Paren}[1]{\left(#1\right)}

\newcommand{\Set}[1]{\left\{#1\right\}}

\newcommand{\specialevent}{\textsc{freezing}}
\newcommand{\Dgood}{D^\star_{\mathrm{good}}}
\newcommand{\klast}{k^{\mathrm{last}}}

\newcommand{\kfirst}{k^{\mathrm{first}}}
\newcommand{\dsum}{d_{\mathrm{sum}}}

\newcommand{\completesol}{\textrm{CompleteSolution}\xspace}

\newcommand{\completesequence}{\textrm{CompleteSequence}\xspace}

\newcommand{\clcost}{\text{cost}}

\newcommand{\clients}{D}

\newcommand{\facilities}{\FF}
\newcommand{\sopt}{\mathsf{opt}}

\newcommand{\Maxdist}{M}

\newcommand{\Sf}{\SS_{\mathrm{free}}}
\newcommand{\Sr}{\SS_{\mathrm{reg}}}

\newcommand{\ones}[1]{\text{Ones}(#1)}
\newcommand{\hg}{\hat{g}}

\newcommand{\Mrho}{2-(\gamma-1)^2}

\newcommand{\Mcoeff}{\gamma}

\newcommand{\Mcoeffpone}{1+\gamma}
\newcommand{\Mcoeffmone}{\varepsilon}

\newcommand{\Mtone}{\tau}
\newcommand{\MRtone}{\bar\tau}
\newcommand{\Mttwo}{1-\theta}
\newcommand{\MRttwo}{\theta}
\newcommand{\Mfrac}{\frac{\varepsilon}{\tau}}

\title{A Better-than-$2$ Approximation for k-Median\\ Part I: The Power of General Bids}
\author{
Rares-Darius Buhai\thanks{EPFL, Lausanne and supported by the Swiss State Secretariat for Education, Research and Innovation (SERI) under contract number MB22.00054.}
\and
Vincent Cohen-Addad\thanks{Google Research, New York NY.}
\and
Henry Fleischmann\thanks{Computer Science Department, Carnegie Mellon
  University, Pittsburgh PA. Work supported by a National Science Foundation Graduate
Research Fellowship Program under Grant No(s).\ DGE2140739 and DGE2631988.}
\and
Anupam Gupta\thanks{Computer Science Department, Courant Institute of
  Mathematical Sciences, New York University, New York, and Google Research. Work supported in part by NSF awards CCF-2422926 and CCF-2608359.}
\and 
Antoni Jub\'es Monforte\thanks{EPFL, Lausanne}
\and 
Madhusudhan Reddy Pittu\thanks{Computer Science Department, Courant Institute of
  Mathematical Sciences, New York University, New York, and  H.\ Milton
  Stewart School of Industrial and Systems Engineering, Georgia Tech,
  Atlanta GA.}
\and 
Miltiadis Stouras\footnotemark[1]
\and 
Ola Svensson\footnotemark[1]
\and 
Lukas Vogl\footnotemark[1]
}
\date{}

\begin{document}
\pagenumbering{gobble}

\maketitle

\newcommand{\pos}[1]{\left[#1\right]^+} %
\newcommand{\Ball}[2]{B(#1,#2)}         %
\newcommand{\dist}{d}                   %

\begin{abstract}
The Greedy algorithm of Jain, Mahdian, Markakis, Saberi, and Vazirani was the best Lagrangian-multiplier-preserving (LMP) approximation algorithm for facility location, and thus the engine behind the best $k$-median guarantees, for over two decades. Its $2$-LMP guarantee was only recently broken by a small amount ($(2-\eta)$-LMP where $\eta \approx 2.25 \cdot 10^{-7}$) through local-search post-processing. That analysis relies on the structure of worst-case instances, and it appears difficult to obtain substantial further improvements along those lines.

We take a different approach and change the algorithm itself: we introduce \emph{general bidding functions}, replacing the classical linear bid with a richer family of rules and prove a polynomial-time $1.95$-LMP approximation, a large improvement over the long-standing factor of~$2$.

We further show that these generalized bidding rules can be made smooth (log-adaptive), building on the recent framework developed for the classical Greedy algorithm by Cohen-Addad, Grandoni, Lee, Schwiegelshohn, and Svensson. This yields a polynomial-time $1.95$-approximation that opens $k+O_\varepsilon(\log n)$ centers. Combined with a $(1.89+\varepsilon)$-approximation algorithm for stable instances in our companion work \cite{stableinstances}, we achieve a better-than-2 approximation for the $k$-median problem in polynomial-time.
\end{abstract}

\newpage

\tableofcontents

\newpage

\pagenumbering{arabic}

\section{Introduction}
\label{sec:introduction}

The metric $k$-median problem is the following classical clustering problem: Given a set $\clients$ of $n$ clients, a set $\facilities$ of facilities, metric distances $\dist(\cdot, \cdot)$ on $\clients \cup \facilities$, and an integer parameter $k\ge 1$, the goal is to choose a set $\SS \subseteq \facilities$ of $k$ open facilities minimizing
\[
    \clcost(\SS):=\sum_{j\in \clients} d(j,\SS)= \sum_{j\in \clients} \min_{i\in \SS} \dist(i,j)\,,
\]
where we use $d(j,\SS)$ as a shorthand for $\min_{i\in \SS} \dist(i,j)$. We further denote by $\sopt$ the cost of an optimal solution. 
Metric $k$-median is well-known to be hard to approximate within a factor better than $1+\nicefrac{2}{e}$~\cite{GuK99,JMS02,JainMMSV03}, and achieving matching algorithmic guarantees has been the subject of intense study over the last two decades. 

Most work has focused on two algorithmic paradigms. On the one hand, the basic local-search algorithm of Arya, Garg, Khandekar, Meyerson, Munagala, and Pandit (swap $O(1)$ centers while improving the solution) gives a $(3+\eps)$-approximation, and for a long time this was the best known polynomial-time guarantee \cite{AryaGKMMP04}. 
On the other hand, the influential primal-dual/Lagrangian framework initiated by Jain and Vazirani~\cite{JaV01} and sharpened by Jain, Mahdian, Markakis, Saberi and Vazirani~\cite{JMS02,JainMMSV03} reduces $k$-median to the task of devising Lagrangian-multiplier-preserving (LMP) approximation algorithms for uncapacitated facility location. 

Here, one first considers facility location with a uniform opening cost $f$, where a solution $\SS\subseteq \facilities$ pays $\clcost(\SS)+f|\SS|$.
Since a higher opening cost $f$ leads to fewer open facilities and a lower cost leads to more, one can binary search on $f$ to find two consecutive solutions, one opening fewer than $k$ facilities and one opening more, whose convex combination opens exactly $k$ facilities in a fractional sense.
If the facility location algorithm used is an $\alpha_{\mathrm{LMP}}$-LMP approximation, this already gives a solution of cost at most $\alpha_{\mathrm{LMP}}\cdot\opt$ that opens $k$ centers \emph{in expectation}.
A second step, bi-point rounding, is then needed to obtain a feasible $k$-median solution with at most $k$ open centers, and this incurs an additional multiplicative factor of $2$, as shown by~\cite{JaV01}.
The overall approximation factor is therefore $2\alpha_{\mathrm{LMP}}$. In this language, the original Jain, Vazirani framework yields a factor $6$ algorithm, while the Greedy algorithm of~\cite{JMS02,JainMMSV03} improves the LMP stage to $2$ and hence gives a factor $4$ algorithm \cite{JaV01,JMS02,JainMMSV03}. 

Li and Svensson~\cite{LiS16} improved the bi-point rounding step by allowing the algorithm to open $k+O(1)$ centers. They also showed that any $\alpha$-approximation opening $k+c$ centers can be transformed into an $(\alpha+\eps)$-approximation opening exactly $k$ centers in time $n^{O(c/\eps)}$. This sparked a sequence of refinements of bi-point rounding algorithms opening $k+ O(1)$ centers~\cite{ByrkaPRST17,GowdaPST23}; combined with the small but important improvement to a $(2-\eta)$-LMP approximation, where $\eta = 2.25 \cdot 10^{-7}$ by Cohen-Addad, Grandoni, Lee, and Schwiegelshohn \cite{Cohen-AddadLS23}, this line eventually led to a $2.613$-approximation \cite{Cohen-AddadLS23,GowdaPST23}.

Recently, Cohen-Addad, Grandoni, Lee, Schwiegelshohn, and Svensson~\cite{MR4928455-Cohen25} achieved a significantly improved $(2+\eps)$-approximation for metric $k$-median (for every $\eps >0$). 
Their result has two ingredients. First, they give a ``smooth'' (so called log-adaptive) implementation of the Greedy algorithm that opens at most $k+O(\log n/\eps^2)$ centers while keeping the approximation ratio at $(2+\eps)$.
Second, they complement this with a polynomial-time $(2+\eps)$-approximation for so-called $(\zeta/\log n)$-stable instances, for every $\eps,\zeta>0$, which is exactly the ingredient needed to turn the quasi-polynomial algorithm into a polynomial-time one. In a separate development, Byrka, Guo, Hu, Li, Wan and Wang \cite{byrka2026kclusteringiterativerandomizedrounding} came up with a novel approach based on an iterative rounding procedure for the standard LP relaxation to reproduce the $(2+\varepsilon)$-approximation ratio.
Thus, by now the gap between what we know how to achieve with $k$ centers in expectation and with at most $k$ centers is essentially closed at factor~$2$. This raises the question whether the factor~$2$ is a natural barrier or if it is possible to close the gap between $2$ and $1+\nicefrac{2}{e}$? We answer the first part of the question and show that there is a $1.95$-approximation algorithm for $k$-median.

\subsection{Our Contribution}
Our main contribution is an improved approximation algorithm for $k$-median with an approximation ratio below $2$:
\begin{theorem}
  \label{thm:kmed-approx}
  There is a randomized polynomial-time $1.95$-approximation algorithm for $k$-median. 
\end{theorem}
To explain how we obtain this result, we first outline the two-pronged
framework from \cite{MR4928455-Cohen25} giving the $(2+\eps)$-approximation;
this framework was first used by \cite{LiS16}. Given a metric space
$(D \cup \FF, d)$, define $\sopt_\ell$ to be the optimal cost of
$\ell$-median on this instance. Naturally, fewer centers result in a
larger cost, and hence $\sopt_k \leq \sopt_{k-c}$ for any $c \geq 0$.
The framework is the following:
\begin{enumerate}[label=(\roman*)]
\item \label{outline:1} Use a greedy algorithm (or some other
  approach) to give a $(2+\eps)$-pseudo-approximation for $k$-median:
  i.e., instead of returning a solution $\SS$ with $k$ centers having
  cost $(2+\eps)\cdot \sopt_k$, return a pseudo-solution---a set
  $\SS_1$ with $k$ centers, which has a cost at most
  $(2+\eps) \cdot \sopt_{k - c}$, where $c = O((\log n)/\poly(\eps))$. Equivalently, we are interested in an approximation algorithm that finds a solution $\SS_2$ with $k + c$ centers and cost at most $(2 + \eps) \sopt_k$. Indeed, invoking such an algorithm with $k^\prime = k-c$ will return a pseudo-approximation.

\item \label{outline:2} Next, observe that if
  $\sopt_{k - c} \leq (1+\eps) \, \sopt_k$, then the solution $\SS_1$
  has cost $(2+O(\eps)) \, \sopt_k$, and we have the desired
  approximation. If not, then reducing the number of centers from $k$
  to $k-c$ causes the optimal value to increase by at least a factor
  of $(1+\eps)$. In other words, there exists $k' \in [k-c, k]$ such
  that (a)~$\sopt_{k'} \leq (1+\eps)\sopt_k$, and moreover,
  (b)~$\sopt_{k'-1} \geq (1+\Omega(\eps/c))\sopt_{k'}$. I.e., the
  optimal cost increases by at least a $(1+\beta)$-factor when one
  reduces the number of available centers by $1$---from $k'$ to
  $k'-1$. (Here $\beta = \Omega(\eps/c)$.)  Such an instance of
  $k'$-median is called a \emph{$\beta$-stable} instance, and stable
  instances can allow for better-than-worst-case
  approximations given their structural properties~\cite{ABS10, Cohen-AddadS17}.
  
\item \label{outline:3} The final ingredient is a
  $(2+\eps)$-approximation algorithm for $\beta$-stable instances
  running in time $\exp((\beta \cdot \poly(\eps))^{-1})\,
  \poly(n)$. Using this, one gets a $k'$-median solution $\SS_2$
  (which for $k' \leq k$ is also a $k$-median solution) with cost
  $(2+\eps)\,\sopt_{k'} \leq (2+O(\eps))\,\sopt_{k}$ in time
  $n^{\poly(1/\eps)}$.
\end{enumerate}
Note that step~\ref{outline:2} is completely generic; however, both
the other steps~\ref{outline:1} and \ref{outline:3} lose a factor of
$(2+\eps)$. To do better, we need to improve both these steps.
In this paper we improve upon step \ref{outline:1} and show that we can find a $(1.9499+\eps)$-approximation that opens $k + c$ centers.
\begin{theorem}[See Theorem~\ref{thm:k-median} for full version]
\label{thm:k-median-informal}
For any $\epsilon > 0$, there exists a polynomial-time algorithm for $k$-median that outputs a solution with at most $k + O(\log n / \eps^2)$ centers and cost at most $(1.9499+\eps) \cdot \sopt$.
\end{theorem}
In a companion paper \cite{stableinstances} we address \ref{outline:3} and prove the following theorem for $\beta$-stable instances.

\begin{theorem}[Stable Instances \cite{stableinstances}]
    \label{thm:stable-result}
  Let $(D \cup \FF, d)$ be a $\beta$-stable instance of $k$-median and
  $\eps \in (0, \nicefrac{1}{200})$. There exists an algorithm running in
  $\exp\{O(\beta^{-1} \poly(1/\eps))\} n^{\poly(1/\eps)}$ time which with
  probability $1- \eps$ outputs $\SS \subseteq \FF$ with $|\SS| = k$
  such that
  \[ \cost(\SS) \leq (1.89 +
    O(\eps))\sopt_k. \]
\end{theorem}

Similar to previous work, we use the LMP-framework by \cite{JaV01}. First, we give an improved LMP approximation algorithm for the uncapacitated facility location problem. Here, we take a different route than \cite{Cohen-AddadLS23} who augment the Greedy algorithm with local-search post-processing to arrive at their $(2-\eta)$-LMP approximation. Instead, we revisit the greedy algorithm itself and modify the bidding rule. 
In our unified notation, the classical choice is the linear bid $\beta(\alpha,d)=[\alpha-d]^+$. We generalize this to broader bidding functions $\beta(\alpha,d)$ and in particular, bidding functions of the form $\beta(\alpha,d)=[\alpha-\gamma\,d]^+$, parametrized by $\gamma > 1$.  
A similar approach has been applied to improve the approximation ratio for the $k$-means problem that asks to minimize the sum of the squared distances (as opposed to the sum of distances) \cite{CharikarCGGLW25,CharikarCGGLW26}. They employ a dual fitting analysis while we, on the other hand, study how the choice of $\beta(\alpha,d)$ changes the factor-revealing program. Indeed, their kind of dual fitting analysis doesn't work for $k$-median; there exists an example (Figure~\ref{fig:alpha-hard-example}) where bounding the algorithm's cost by the sum of dual variables $\sum_{j \in C} \alpha_j$ results in an LMP ratio approaching $2$ (both for our variant and their variant of the greedy algorithm).

Intuitively, our bids allow clients to discount distant facilities more aggressively, which already breaks the tight examples behind the classical factor-$2$ analysis. This added flexibility comes, however, at a price: the classical Greedy analysis is governed by a single factor-revealing LP, whereas general bids lead to genuinely non-linear programs, and even piecewise-linear bids lead to a family of LPs rather than one LP. This viewpoint appears surprisingly powerful even for very simple families of bids. We show that Algorithm~\ref{alg:dual_generalized} with bidding function $\beta(\alpha,d)=[\alpha-1.224\,d]^+$ is a $1.9499$-LMP approximation algorithm.

\begin{theorem}[See Theorem~\ref{thm:1.95-lmp} for full version]
\label{thm:1.95-lmp-informal}
There exists a polynomial-time $1.9499$-LMP approximation algorithm for uncapacitated facility location.
\end{theorem}
\begin{remark}
\label{rem:AI-LMP}
    Building on the results described above, GPT-6 Astra managed to improve the LMP approximation ratio to 1.844. The bidding function used to achieve this ratio is
    \begin{equation}\label{eq:bid}
  \beta(\alpha,d)=\frac1{10^8}\sum_{s=1}^{13}W_s[\alpha-\gamma_s d]^+,
\end{equation}
where the multipliers $\gamma_s$ and weights $W_s$ are
specified in Table~\ref{tab:bid}.
Furthermore, GPT-6 Astra managed to prove that for any admissible bidding function, as in Definition~\ref{def:admissible-bidding}, the supremum of the factor revealing program used in our analysis over all $n \geq 1$ has value at least $1.8$.
Therefore, our approach cannot immediately yield an approximation ratio below $1.8$.
Lean proofs for these claims can be found at \url{https://github.com/epfl-thl2/project-k-median}.
\end{remark}

\begin{table}[H]
\centering
\renewcommand{\arraystretch}{1.1}
\setlength{\tabcolsep}{12pt}
\begin{tabular}{rrr@{\qquad}rrr}
\toprule
$s$ & $\gamma_s$ & $W_s$ & $s$ & $\gamma_s$ & $W_s$\\
\midrule
1 & $137/125$ & $66{,}678{,}856$ & 8 & $9/5$ & $659{,}515$\\
2 & $29/20$ & $616{,}823$ & 9 & $37/20$ & $2{,}091{,}320$\\
3 & $3/2$ & $12{,}834{,}612$ & 10 & $19/10$ & $2{,}245{,}119$\\
4 & $29/18$ & $4{,}845{,}692$ & 11 & $39/20$ & $2{,}936{,}362$\\
5 & $33/20$ & $1{,}922{,}894$ & 12 & $2$ & $766{,}770$\\
6 & $17/10$ & $70{,}695$ & 13 & $41/20$ & $2{,}724{,}139$\\
7 & $7/4$ & $1{,}607{,}203$ &  &  & \\
\bottomrule
\end{tabular}
\caption{The bidding function. Each coefficient of $[x-\gamma_s y]^+$ is
$W_s/10^8$.}
\label{tab:bid}
\end{table}

Our second contribution is to make these generalized bids \emph{smooth} (more specifically, low in terms of adaptivity). We adapt and generalize the recent log-adaptive framework of \cite{MR4928455-Cohen25} to our broader class of bidding functions. A key conceptual point is that we do \emph{not} redo the (in our case, complex) approximation analysis of the log-adaptive algorithm from scratch. Instead, we prove a black-box transfer statement: if the factor-revealing mathematical program with a bidding function $\beta = \pos{a - \gamma d}$ has value at most $\rho$, then the corresponding smooth version opens at most $k+O(\log n/\eps^2)$ centers while losing only an additive $\eps$ in the approximation guarantee. In other words, the smooth analysis inherits the bound from the original factor-revealing program almost for free. Combined with Theorem~\ref{thm:1.95-lmp-informal}, this yields Theorem~\ref{thm:k-median-informal}.

\subsection{Further Related Work}

\paragraph{Facility location and dual fitting.}
The connection between $k$-median and facility location dates back to the primal-dual schema of Jain and Vazirani~\cite{JaV01} and the dual-fitting framework of Jain, Mahdian, Markakis, Saberi, and Vazirani~\cite{JMS02,JainMMSV03}. These ideas have been highly influential for facility location itself, leading to improved approximation algorithms culminating in the $1.488$-approximation of Li~\cite{Li13}; see also~\cite{CharikarGTS99,ChudakS03, ByrkaA10}. The factor-revealing LP technique introduced in~\cite{JainMMSV03} is a key tool underlying both the classical Greedy analysis and our generalization to broader bidding functions.

\paragraph{Local search.}
The $(3+\eps)$-approximation of Arya et al.~\cite{AryaGKMMP04} via multi-swap local search remained the best polynomial-time guarantee for metric $k$-median for over a decade. Cohen-Addad, Gupta, Hu, Oh, and Saulpic~\cite{Cohen-AddadGHOS22} introduced a non-oblivious local-search approach that beats the classical factor~$3$, though it does not yet improve the best overall guarantee for the problem.

\paragraph{Bi-point rounding.}
As discussed in the introduction, converting an LMP solution into an exact-$k$ solution via bi-point rounding has been a major bottleneck. The line of work by Li and Svensson~\cite{LiS16}, Byrka et al.~\cite{ByrkaPRST17}, and Gowda et al.~\cite{GowdaPST23} progressively refined this step, eventually reaching a $2.613$-approximation in combination with the LMP improvement of~\cite{Cohen-AddadLS23}. It is known that no LMP-based bi-point rounding approach can breach the $(1+2/e)\sqrt{\phi} \approx 2.208$ barrier~\cite{GowdaPST23}, which further motivates the smooth/log-adaptive approach that avoids bi-point rounding altogether.

\paragraph{Euclidean $k$-Median.}
A different approach to the rounding loss was used by Ahmadian, Norouzi-Fard, Svensson, and Ward~\cite{AhmadianNSW20} and further developed by Cohen-Addad, Esfandiari, Mirrokni, and Narayanan~\cite{Cohen-AddadEMN22} and Grandoni, Ostrovsky, Rabani, Schulman and Venkat~\cite{GrandoniORSV22}, who improved guarantees for $k$-means and Euclidean $k$-median by showing how to run the simpler primal-dual algorithm~\cite{JaV01} in a ``smooth'' way so as to respect the hard cardinality constraint without paying the bi-point rounding loss.

\subsection{Outline}

The remainder of the paper is organized as follows. We start with some additional definitions in Section~\ref{sec:prelim}. We then, in Section~\ref{sec:greedy_algorithm_general_bids}, present the Greedy algorithm with general bidding functions and use simple examples to illustrate why modified bids can avoid the classical factor-$2$ lower-bound constructions. In Section~\ref{sec:approx_guarantee}, we formulate the corresponding factor-revealing mathematical program (for general bidding functions), and give an outline of the proof of our $1.9499$-LMP guarantee for the bid $\beta(\alpha,d)=[\alpha-1.224\,d]^+$.

In Section~\ref{sec:logad} we show how to transfer any upper bound on the factor-revealing program to a log-adaptive version of the algorithm, and we state the resulting polynomial-time $k$-median guarantee that opens at most $k+O(\log n)$ centers (which then follows from adapting the recipe of~\cite{MR4928455-Cohen25}). Finally, in Section~\ref{sec:195-dual-fitting} we give the full, more technically involved, dual-fitting proof of the $1.9499$ bound. 

\subsection*{AI Acknowledgment and Concurrent Work}
The mathematical proofs presented in this work (apart from Remark \ref{rem:AI-LMP}) were developed solely by the authors without the use of AI. 
GPT-5.6 Sol and GPT-6 Astra were used for proof-checking, to assist the writing process and to obtain feedback on the presentation. In particular, AI has been used to reformulate paragraphs in Sections~\ref{sec:introduction},~\ref{sec:prelim}, and in the introductory parts of the other sections. The lean certificates mentioned in Remark~\ref{rem:AI-LMP} and the certificate for Theorem~\ref{thm:1.95-frp} were obtained by GPT-6 Astra.  

In concurrent work announced on October 6, 2026, OpenAI presented
two different better-than-two approximations for general 
$k$-median instances~\cite{OAI:The-Approximation-Threshold-for-Metric-k-Median-September-24-2026, OAI:Single-Exponential-Recovery-and-Bounded-Price-Strictness-for-Metric-k-Median-September-24-2026}. The first extends the iterative rounding-based $(2+\varepsilon)$-approximation of~\cite{byrka2026kclusteringiterativerandomizedrounding} to achieve an optimal $(1 + 2/e + \eps)$-approximation. The other instead builds on the stable instance/pseudo-approximation framework of~\cite{MR4928455-Cohen25} giving a $(2 -\sigma)$-approximation for some absolute constant $\sigma > 0$, thus sharing many similarities with this work and its companion paper. 
Finally, through personal communication, we learned that Jaroslaw Byrka, Yuhao Guo, Yang Hu, Shi Li, Yongho Shin, Phanu Vajanopath, Chengzhang Wan and Zaixuan Wang have independently achieved a $2-\sigma$-approximation algorithm for a small but constant $\sigma >0$. Their approach further develops the iterative randomized rounding framework of~\cite{byrka2026kclusteringiterativerandomizedrounding}, providing an alternative route to a better-than-$2$ approximation.

\section{Preliminaries}
\label{sec:prelim}
We are given a metric $d(\cdot,\cdot)$ on the union of a set of clients $\clients$ and a set of facilities $\facilities$.
For a set of facilities (also referred to as centers) $\SS \subseteq \facilities$ and a client $j \in \clients$, we write
\[
    \dist(j,\SS) := \min_{i \in \SS} \dist(i,j)\,,
\]
with the convention $\dist(j,\emptyset)=+\infty$. For a facility $i \in \facilities$ and radius $r \ge 0$, we let
\[
    \Ball{i}{r} := \{j \in \clients : \dist(i,j) \le r\}
\]
be the ball around $i$ of radius $r$.
We also use the shorthand $[x]^+ := \max\{x,0\}$ for $x \in \mathbb{R}$.

In general, we will denote sets of clients by capital letters such as $A, B$, and sets of facilities by blackboard-bold capital letters such as $\AA, \BB$.

\paragraph{Facility location and $k$-median.}
In the \emph{metric uncapacitated facility location} problem (UFL) with uniform facility opening cost $f \ge 0$, the goal is to choose a set $\SS \subseteq \facilities$ minimizing
\[
    \cost_f(\SS) := \sum_{j \in \clients} \dist(j,\SS) + f\cdot |\SS|\,.
\]
We denote by $\OPT_f$ the optimum value of this facility-location instance.
We omit the subscript when it is clear from context.

In the \emph{metric $k$-median} problem, we are given in addition an integer $k \ge 1$, and the goal is to choose a set $\SS \subseteq \facilities$ of at most $k$ facilities minimizing the connection cost
\[
    \clcost(\SS) := \sum_{j \in \clients} \dist(j,\SS)\,.
\]
We denote by $\opt_k$ the optimum value of the $k$-median instance. Whenever the value of $k$ is clear from context, we simply write $\opt$.

The standard Lagrangian viewpoint relates the two problems by regarding the cardinality constraint $|\SS| \le k$ as a penalty term $f \cdot |\SS|$. A facility-location algorithm therefore yields a family of solutions parameterized by the opening cost $f$, and one tries to recover a near-optimal $k$-median solution from this family.

\begin{definition}[LMP approximation ratio]
    \label{def:lmp}
    A facility-location algorithm $\mathcal{A}$ is a \emph{$\rho$-Lagrangian-multiplier-preserving approximation} (or \emph{$\rho$-LMP approximation}) if for every instance with uniform opening cost $f$, the algorithm outputs a set $\SS \subseteq \facilities$ satisfying
    \[
        \sum_{j \in \clients} \dist(j,\SS) \le \rho\bigl(\OPT_f - f\cdot |\SS|\bigr)\,.
    \]
\end{definition}
\paragraph{}
Intuitively, a $\rho$-LMP guarantee says that after fully paying the facility opening costs the algorithm is still a $\rho$-approximation in the connectivity cost.

\vspace{1em}
By a standard reduction as in Lemma 2.1 in \cite{MR4928455-Cohen25}, we may assume that all distances are polynomially bounded and bounded away from zero; concretely, we assume
\begin{equation}
\label{eq:dist-bound}
    1 \le \dist(i,j) \le \poly(|\clients|)
    \qquad \forall i \in \facilities,\ j \in \clients\,.
\end{equation}
We also denote $\Maxdist = \max_{i \in \facilities, j \in \clients} d(i,j)$.

\section{Greedy Algorithm with General Bids}
\label{sec:greedy_algorithm_general_bids}

In this section we revisit the Greedy algorithm of \cite{JMS02, JainMMSV03} and present its generalized version with more flexible bidding functions. The classical Greedy algorithm maintains a set $A$ of active clients, a set $S$ of opened facilities, and for each client $j$ a total initial budget $\alpha_j$ that grows uniformly while $j$ remains active. As time increases, one of two events can occur. Either some unopened facility receives enough total bids to pay for its opening cost, in which case the facility is opened, or some active client reaches an already opened facility and becomes inactive. In the standard Greedy algorithm, an active client $j$ first reserves $d(i,j)$ units of its budget to pay for connecting to a candidate facility $i$, and then contributes all of its remaining budget, namely $[\alpha_j-d(i,j)]^+$, towards opening $i$. Likewise, an inactive client that is currently connected at cost $d(j,\SS)$ contributes the full improvement $[d(j,\SS)-d(i,j)]^+$ that opening $i$ would offer.

It is therefore natural to ask what happens if clients are not forced to bid their entire remaining budget to a candidate facility, but are instead allowed to contribute only a fraction of it.
Intuitively, one would like clients to support nearby facilities more aggressively and discount their bids to facilities that are farther away.

Our goal is to understand the power of the bidding function in the Greedy algorithm. For a client with value $\alpha$, we denote its bid towards a facility at distance $d$ by $\beta(\alpha,d)$, and for an inactive client we use the same function with $\alpha=d(j,\SS)$, namely $\beta(d(j,\SS),d)$. We give the pseudocode of the Generalized Greedy algorithm in Algorithm~\ref{alg:dual_generalized}. Of course, the bidding function $\beta(\alpha,d)$ must
satisfy some natural properties, e.g. the bids should be non-negative and should not surpass
the remaining budget. We outline these properties in Section~\ref{sec:natural-properties-bidding} and from this point onwards, whenever we refer to the Generalized Greedy algorithm (Algorithm~\ref{alg:dual_generalized}), we assume that its bidding function satisfies these properties.

\begin{algorithm}[H]
\caption{Greedy Algorithm with a Generalized Bidding Strategy}
\label{alg:dual_generalized}
\begin{algorithmic}[h]
\State \textbf{Initialization:} Set $A=D$ (set of active clients), $\SS=\emptyset$ (set of opened facilities), $\alpha=0$ (global time).
\State For each client $j \in D$, let $\alpha_j=\alpha$ while $j \in A$, and freeze $\alpha_j$ when $j$ first becomes inactive.
\State For any client $j$ and facility $i$, define the bid:
\[
\text{bid}(j, i) = 
\begin{cases} 
\beta(\alpha_j, d(i,j)) & \text{if } j \in A \\
\beta(d(j,\SS), d(i,j)) & \text{if } j \notin A\,.
\end{cases}
\]
\While{$A \ne \emptyset$}
    \State Increase $\alpha$ uniformly until one of the following events occurs:
    \State \textbf{Event 1:} An unopened facility $i \in \FF$ is fully paid for: $\sum_{j \in D} \text{bid}(j, i) = f$
    \State \indent Open facility $i$: $\SS \gets \SS \cup \{i\}$
    \State \indent Mark all clients $j \in A$ such that $\alpha_j \geq \dist(j,\SS)$ as served (i.e., remove them from $A$).

    \State \textbf{Event 2:} An active client $j \in A$ reaches an open facility $i \in \SS$: $\alpha_j = d(i,j)$
    \State \indent Mark $j$ as served (i.e., remove $j$ from $A$).
    
    \Statex \textit{(If multiple events occur simultaneously, process them in an arbitrary order)}
\EndWhile
\end{algorithmic}
\end{algorithm}

\subsection{Motivation for General Bidding Functions}

In this section we recall the classic factor-$2$ lower bound for the standard Greedy algorithm
of \cite{JMS02, JainMMSV03} and show how a slightly different bidding function can help us break the barrier of $2$. The examples we present do not have uniform opening costs for the facilities, i.e. every facility has a different opening cost $f_i$ and it opens when
the total bids towards it are equal to $f_i$. This is done purely for presentational 
simplicity as similar examples can be obtained for uniform opening costs. 

In these examples, the quantity that ``reveals'' the LMP factor of the algorithm is the following. Let $i$ be an optimal facility with opening cost $f$, and $d_1, \dots, d_n$ 
be the distances between the facility and the clients connected to it in the optimal solution. Furthermore, let $\cost_j$ 
be the total cost attributed to client $j$: its final connection cost plus all its bids towards the facility opening costs. We are interested in bounding the following ratio:
\[\frac{\sum_{j = 1}^n \cost_j - f}{\sum_{j = 1}^n d_j}\,.\]
\noindent The above quantity is a crucial ingredient in the analysis of the algorithm and is essentially equal to the final LMP approximation. In fact, when $\cost_j=\alpha_j$ this is equal to the objective of the Factor Revealing program developed
in \cite{JMS02, JainMMSV03} to analyze the standard Greedy algorithm.

\begin{figure}[h]
\centering
\begin{tikzpicture}[x=1cm,y=1cm,every node/.style={font=\small}]
    \def\boxsize{0.21}
    \def\dotr{0.065}
    \def\solodotr{0.075}
    \def\seglen{2.35}
    \def\basey{0.105}
    \def\clienty{0.74}
    \def\labely{0.95}
    \def\fx{0.105}
    \def\sx{\fx+\seglen}
    \def\fpx{\fx+2*\seglen}

    \filldraw[black] ({\fx-\boxsize/2},{\basey-\boxsize/2}) rectangle ({\fx+\boxsize/2},{\basey+\boxsize/2});
    \filldraw[black] ({\fpx-\boxsize/2},{\basey-\boxsize/2}) rectangle ({\fpx+\boxsize/2},{\basey+\boxsize/2});

    \filldraw[black] ({\fx-0.21},{\clienty}) circle (\dotr);
    \filldraw[black] ({\fx-0.07},{\clienty}) circle (\dotr);
    \filldraw[black] ({\fx+0.07},{\clienty}) circle (\dotr);
    \filldraw[black] ({\fx+0.21},{\clienty}) circle (\dotr);
    \node[anchor=south] at ({\fx},{\labely}) {$n-1$ clients};

    \draw[dotted, line width = 0.73pt] ({\fx-0.21},{\clienty-\dotr}) -- ({\fx},{\basey});
    \draw[dotted, line width = 0.73pt] ({\fx-0.07},{\clienty-\dotr}) -- ({\fx},{\basey});
    \draw[dotted, line width = 0.73pt] ({\fx+0.07},{\clienty-\dotr}) -- ({\fx},{\basey});
    \draw[dotted, line width = 0.73pt] ({\fx+0.21},{\clienty-\dotr}) -- ({\fx},{\basey});
    \node[anchor=east] at ({\fx-0.32},{0.5*(\clienty+\basey)}) {$0$};

    \node[anchor=north] at ({\fx},{\basey-\boxsize/2-0.08}) {$f =  2(n-1)$};

    \filldraw[black] ({\sx},{\basey}) circle (\solodotr);
    \node[anchor=south] at ({\sx},{\basey+0.24}) {$n$-th client};

    \node[anchor=north] at ({\fpx},{\basey-\boxsize/2-0.08}) {$f' = 1-\frac1n$};

    \draw[dashed] ({\fx+\boxsize/2},{\basey}) -- node[above] {$1$} ({\sx-\solodotr-0.04},{\basey});
    \draw[dashed] ({\sx+\solodotr+0.04},{\basey}) -- node[above] {$1$} ({\fpx-\boxsize/2},{\basey});
\end{tikzpicture}
\caption{A $2$ lower bound for standard Greedy with bidding function: $[\alpha-d]^+$.}
\label{fig:greedy-hard-example}
\end{figure}
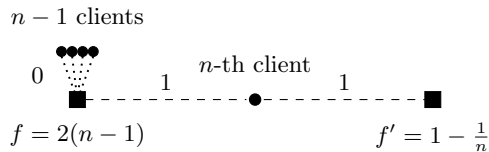

In the classical tight example, depicted in  Figure~\ref{fig:greedy-hard-example}, 
we have two facilities at distance $2$ from each other with opening costs $f$ and $f'$.
The instance consists of $n-1$ clients that are co-located with the facility $f$ (i.e. at distance $0$) and one client that is at distance $1$ from both facilities.
On a high level, in this example the standard Greedy algorithm might open the facility $f'$
using only the bid of the separated client, which then forces the $n-1$ clients
to raise to a higher level in order to pay for the opening of $f$.

Indeed, when the active clients reach value $\alpha = 2-\frac1n$, the separated client bids exactly $1-\frac1n$ to both facilities, while each of the $n-1$ co-located clients bids $2-\frac1n$ to $f$ and $0$ to $f'$. Therefore, both facilities receive bids that are exactly equal to their opening costs. If Greedy chooses to open $f'$ first, then only the $n$-th client becomes inactive. The other $n-1$ clients are still at distance $2$ from the open facility, so they remain active (as their budget cannot yet pay for their connection cost) and keep increasing their values until $\alpha=2$, when they can either open facility $f$ by themselves or simply connect to $f'$.
In the resulting execution, the final client values are
\[
    \alpha_n = 2-\frac1n 
    \qquad\text{and}\qquad
    \alpha_j = 2 \quad \text{for every } j \leq n-1\,.
\]

\noindent At the same time, the optimal solution opens only $f$, paying opening cost $2(n-1)$ and connection cost $1$ for the $n$-th client. As a result, the factor of the algorithm for this example is

$$
\frac{\sum_{j = 1}^n \alpha_j - f}{\sum_{j=1}^n d_j} = \frac{2(n-1) + 2-\frac1n - 2(n-1)}{1} = 2-\frac1n\,.
$$

\noindent This same example can be generalized to any similar algorithm whose clients end up paying $\cost_j = \alpha_j$, as long as $f=2(n-1)$ and changing appropriately the value of $f'$. 

Consider now a slightly different bidding rule, namely $\beta(\alpha,d)=[\alpha-(1+\eps)d]^+$, for some $0<\eps< 1$.
The main difference is that the facility $f'$ can no longer be opened at time $\alpha = 2-\frac1n$, because the isolated client does not contribute its full remaining budget but contributes only
\[
    \left[2-\frac1n-(1+\eps)\right]^+
    =
    \left[1-\eps-\frac1n\right]^+
    <
    1-\frac1n\,.
\]
On the other hand, the facility $f$ can open as soon as $(n-1)\alpha + [\alpha-(1+\eps)]^+ = 2(n-1)$, that is, at time $\alpha = 2-\frac{1-\eps}{n}$. At this time, the facility $f'$ still does not receive enough bids, as the isolated client's contribution is
\[
    [\alpha-(1+\eps)]^+ = \frac{(n-1)(1-\eps)}{n} < 1-\frac1n\,.
\]
\noindent As a result, Greedy opens $f$ first and terminates as all the clients connect to $f$. 
Therefore, the final client values for this execution are 
\[
    \cost_j \le \alpha_j = 2-\frac{1-\eps}{n}, \quad \text{for every } j \leq n\,.
\]
The final factor for this example is thus
$$
\frac{\sum_{j = 1}^n \cost_j - f}{\sum_{j = 1}^n d_j} \le  n\left(2-\frac{1-\eps}{n}\right)-2(n-1)=1+\eps\,.
$$

\noindent Hence, even this simple modification already improves the first example from $2-\frac1n$ to $1+\eps$. In fact, for any $\gamma \in (1,2)$, one can follow the same analysis and get that the bidding function $[\alpha-\gamma d]^+$ gives an upper bound of $\gamma$ on this example, which is strictly better than standard Greedy for sufficiently large $n$.

Of course, the previous example (Figure~\ref{fig:greedy-hard-example}) is tailored
to the standard Greedy algorithm, or equivalently to $\gamma = 1$, as the ``bad event'' (i.e. 
that $f'$ can be opened before $f$) only happens for $\gamma = 1$. Following the same rationale,
this example can be generalized so that $f'$ and $f$ become tight exactly at the same time for any $\gamma \in (1,1.224]$. We present this generalized example in Figure~\ref{fig:greedy-hard-example-gamma}.
However, even in this case the resulting ratio is $2-(1-\gamma)^2$ which is strictly better
than $2$ for any $\gamma \in (1,1.224]$. In fact, we later prove that the aforementioned ratio
is exactly tight for $\gamma \in (1,1.224]$; this proves that Algorithm~\ref{alg:dual_generalized} with the bidding function $\beta_{1.224}(\alpha,d) = [\alpha - 1.224 d]^+$ is a $1.9499$-LMP approximation.

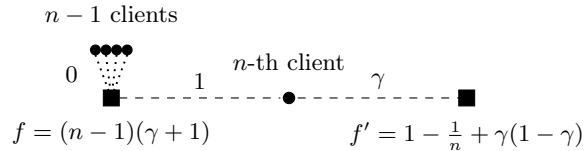
\begin{figure}[h]
\centering
\begin{tikzpicture}[x=1cm,y=1cm,every node/.style={font=\small}]
    \def\boxsize{0.21}
    \def\dotr{0.065}
    \def\solodotr{0.075}
    \def\seglen{2.35}
    \def\basey{0.105}
    \def\clienty{0.74}
    \def\labely{0.95}
    \def\fx{0.105}
    \def\sx{\fx+\seglen}
    \def\fpx{\fx+2*\seglen}

    \filldraw[black] ({\fx-\boxsize/2},{\basey-\boxsize/2}) rectangle ({\fx+\boxsize/2},{\basey+\boxsize/2});
    \filldraw[black] ({\fpx-\boxsize/2},{\basey-\boxsize/2}) rectangle ({\fpx+\boxsize/2},{\basey+\boxsize/2});

    \filldraw[black] ({\fx-0.21},{\clienty}) circle (\dotr);
    \filldraw[black] ({\fx-0.07},{\clienty}) circle (\dotr);
    \filldraw[black] ({\fx+0.07},{\clienty}) circle (\dotr);
    \filldraw[black] ({\fx+0.21},{\clienty}) circle (\dotr);
    \node[anchor=south] at ({\fx},{\labely}) {$n-1$ clients};

    \draw[dotted, line width = 0.73pt] ({\fx-0.21},{\clienty-\dotr}) -- ({\fx},{\basey});
    \draw[dotted, line width = 0.73pt] ({\fx-0.07},{\clienty-\dotr}) -- ({\fx},{\basey});
    \draw[dotted, line width = 0.73pt] ({\fx+0.07},{\clienty-\dotr}) -- ({\fx},{\basey});
    \draw[dotted, line width = 0.73pt] ({\fx+0.21},{\clienty-\dotr}) -- ({\fx},{\basey});
    \node[anchor=east] at ({\fx-0.32},{0.5*(\clienty+\basey)}) {$0$};

    \node[anchor=north] at ({\fx},{\basey-\boxsize/2-0.08}) {$f = (n-1)(\gamma+1)$};

    \filldraw[black] ({\sx},{\basey}) circle (\solodotr);
    \node[anchor=south] at ({\sx},{\basey+0.24}) {$n$-th client};

    \node[anchor=north] at ({\fpx},{\basey-\boxsize/2-0.08}) {$f' = 1-\frac1n+\gamma(1-\gamma)$};

    \draw[dashed] ({\fx+\boxsize/2},{\basey}) -- node[above] {$1$} ({\sx-\solodotr-0.04},{\basey});
    \draw[dashed] ({\sx+\solodotr+0.04},{\basey}) -- node[above] {$\gamma$} ({\fpx-\boxsize/2},{\basey});
\end{tikzpicture}
\caption{A $2-(\gamma-1)^2$ lower bound for Greedy with the bidding rule $\beta_{\gamma}(\alpha,d)=[\alpha-\gamma d]^+$.}
\label{fig:greedy-hard-example-gamma}
\end{figure}

\vspace{1em}

We show that the example of Figure~\ref{fig:greedy-hard-example-gamma} gives indeed a lower bound of $2-(1-\gamma)^2-\frac1n$ for the bidding functions $\beta_\gamma(\alpha,d) = [\alpha - \gamma d]^+$ with $\gamma \in (1,1.224]$.

\noindent Let us define $\tau = 1+\gamma-\frac1n.$ At time $\alpha = \tau$, the co-located
clients do not yet bid towards $f'$ because $\tau - \gamma(1+\gamma) = 1-\gamma^2 -\frac1n < 0$.
At the same time, the isolated client bids $[\tau-\gamma^2]^+ = 1+\gamma-\gamma^2-\frac1n$ towards facility $f'$ and $[\tau-\gamma]^+ = 1-\frac1n$ towards facility $f$. Notice that at this time, both facilities can be paid for since
$$
[\tau-\gamma^2]^+ = f' \quad\text{and}\quad (n-1)\tau + [\tau-\gamma]^+ = f\,.
$$
\noindent Suppose Greedy breaks the tie by opening $f'$ first. Then the $n$-th client becomes inactive, while the other $n-1$ clients remain active because they are still at distance $1+\gamma > \tau$ from the opened facility. Moreover, the isolated client, now inactive, does not contribute to facility $f$ anymore, because its current distance to $f'$ is $\gamma$
and its bid under the bidding function is $[\gamma - 1\cdot\gamma]^+ = 0$. 
Therefore only the $n-1$ co-located clients continue to bid for $f$ and they open $f$ by time
\[
    \frac{f}{n-1} = 1+\gamma\,.
\]
At time $\alpha = 1+\gamma$, the co-located clients can either open $f$ or connect to $f'$.
Suppose that they connect to $f'$; notice that opening $f$ would result in a smaller final cost
because the isolated client would get its connection distance dropped to $1$ without extra bidding.
At the end of the execution, the $n$-th client has paid $\tau-\gamma^2$ in its first bid and
$\gamma$ for its connection cost. The $n-1$ co-located clients together pay exactly the opening cost of $f$, and therefore the factor of this example is
$$
\frac{\sum_{j = 1}^n \cost_j - f}{\sum_{j = 1}^n d_j} = \tau - \gamma^2 + \gamma = 2-(1-\gamma)^2 -\frac1n\,.
$$

Here, it is important that we use $\cost_j$ instead of the upper bound $\cost_j \le \alpha_j$. In particular, there is an example where the ratio 
$$
\frac{\sum_{j = 1}^n \alpha_j - f}{\sum_{j = 1}^n d_j},
$$
is approaching $2$; it is the same example as in Figure~\ref{fig:greedy-hard-example} with $f^\prime$ adjusted to $f^\prime=(2-\gamma)(1-\frac1n)$.

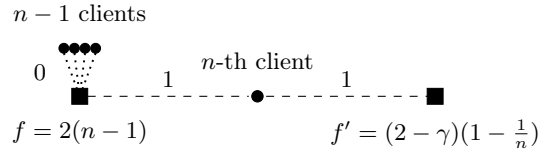
\begin{figure}[h]
\centering
\begin{tikzpicture}[x=1cm,y=1cm,every node/.style={font=\small}]
    \def\boxsize{0.21}
    \def\dotr{0.065}
    \def\solodotr{0.075}
    \def\seglen{2.35}
    \def\basey{0.105}
    \def\clienty{0.74}
    \def\labely{0.95}
    \def\fx{0.105}
    \def\sx{\fx+\seglen}
    \def\fpx{\fx+2*\seglen}

    \filldraw[black] ({\fx-\boxsize/2},{\basey-\boxsize/2}) rectangle ({\fx+\boxsize/2},{\basey+\boxsize/2});
    \filldraw[black] ({\fpx-\boxsize/2},{\basey-\boxsize/2}) rectangle ({\fpx+\boxsize/2},{\basey+\boxsize/2});

    \filldraw[black] ({\fx-0.21},{\clienty}) circle (\dotr);
    \filldraw[black] ({\fx-0.07},{\clienty}) circle (\dotr);
    \filldraw[black] ({\fx+0.07},{\clienty}) circle (\dotr);
    \filldraw[black] ({\fx+0.21},{\clienty}) circle (\dotr);
    \node[anchor=south] at ({\fx},{\labely}) {$n-1$ clients};

    \draw[dotted, line width = 0.73pt] ({\fx-0.21},{\clienty-\dotr}) -- ({\fx},{\basey});
    \draw[dotted, line width = 0.73pt] ({\fx-0.07},{\clienty-\dotr}) -- ({\fx},{\basey});
    \draw[dotted, line width = 0.73pt] ({\fx+0.07},{\clienty-\dotr}) -- ({\fx},{\basey});
    \draw[dotted, line width = 0.73pt] ({\fx+0.21},{\clienty-\dotr}) -- ({\fx},{\basey});
    \node[anchor=east] at ({\fx-0.32},{0.5*(\clienty+\basey)}) {$0$};

    \node[anchor=north] at ({\fx},{\basey-\boxsize/2-0.08}) {$f = 2(n-1)$};

    \filldraw[black] ({\sx},{\basey}) circle (\solodotr);
    \node[anchor=south] at ({\sx},{\basey+0.24}) {$n$-th client};

    \node[anchor=north] at ({\fpx},{\basey-\boxsize/2-0.08}) {$f' = (2-\gamma)(1-\frac1n)$};

    \draw[dashed] ({\fx+\boxsize/2},{\basey}) -- node[above] {$1$} ({\sx-\solodotr-0.04},{\basey});
    \draw[dashed] ({\sx+\solodotr+0.04},{\basey}) -- node[above] {$1$} ({\fpx-\boxsize/2},{\basey});
\end{tikzpicture}
\caption{A $2$ lower bound when replacing $\cost_j$ with $\alpha_j$ for bidding functions: $[\alpha-\gamma d]^+$.}
\label{fig:alpha-hard-example}
\end{figure}
Indeed, consider an execution of the greedy algorithm with $\beta(\alpha,d) = [\alpha-\gamma d]^+$ for the example in Figure~\ref{fig:alpha-hard-example}. When $\alpha =2 - \frac{2-\gamma}{n}$ both facilities are tight:
$$
[\alpha-\gamma]^+ = f' \quad\text{and}\quad (n-1)\alpha + [\alpha-\gamma]^+ = f\,.
$$
Again, we assume that $f^\prime$ is opened first. The bid from client $n$ towards facility $f$ drops to $\beta(1,1) = 0$. Thus the remaining active clients $j \leq n - 1$ increase their budget to $\alpha_j = 2$ at which point they can either open $f$ or connect to facility $f^\prime$.
In the latter case the ratio becomes,
$$
\frac{\sum_{j = 1}^n \alpha_j - f}{\sum_{j = 1}^n d_j} = \alpha_n = 2 - \frac{2-\gamma}{n}.
$$
However, for client $n$ it holds that 
$$
\cost_n = \beta(\alpha,1) + 1 = 2 -\frac{2-\gamma}{n} - \gamma + 1 \le 2 - (\gamma - 1).
$$
For all other $j \le n-1$, however, the cost remains unchanged, $\cost_j = \alpha_j$.
Thus, if we consider $\cost_j$ instead, one achieves a ratio of
$$
\frac{\sum_{j = 1}^n \cost_j - f}{\sum_{j = 1}^n d_j} = \cost_n \le 2 - (\gamma - 1).
$$

\subsection{Properties of Bidding Functions}\label{sec:natural-properties-bidding}

We now state a number of natural structural properties that the bidding functions must satisfy. 
These are precisely the properties used later in the analysis, and they capture the basic 
intuition that bids should behave smoothly, monotonically, and consistently under rescaling.

\begin{definition}[Admissible bidding functions]
\label{def:admissible-bidding}
A bidding function $\beta : \reals_{\ge 0}^2 \to \reals_{\ge 0}$ must satisfy the following properties.
\begin{enumerate}
    \item \emph{Non-negative:} $\beta(\alpha,d)\ge 0$ for every $\alpha,d \ge 0$.
    \item \emph{Budget-feasible:} $\beta(\alpha,d) \leq [\alpha-d]^+$ for every $\alpha,d \ge 0$.
    \item \emph{Continuous:} $\beta$ is continuous in both arguments.
    \item \emph{Monotonic:} $\beta$ is non-decreasing in its first argument and non-increasing in its second argument.
    \item \emph{Scale equivariant:} for every $\alpha>0$ and $d\ge 0$, it holds that 
    $$ \beta(\alpha,d)=\alpha \cdot \beta(1,d/\alpha)$$
    \item \emph{Superadditive:} whenever $x_1 \ge x_2 \ge x_3 \ge 0$,
    $$
        \beta(x_1,x_2) + \beta(x_2,x_3) \le \beta(x_1,x_3).
    $$
\end{enumerate}
\end{definition}

The first four properties are immediate from the intended interpretation of a bid. A client should never make a negative contribution or contribute more than its remaining budget after paying for its distance. Bids should also not change discontinuously under an infinitesimal perturbation of either the client's value or the distance to the facility. Additionally, the monotonicity guarantees that a client with larger value should not bid less while a farther facility should also not receive a larger bid.

Scale equivariance expresses that the bidding rule should be invariant under a uniform rescaling of the metric and of the opening costs. This is natural both algorithmically and analytically: if all distances and opening costs are multiplied by the same factor, then the algorithm should simply scale its bids by the same factor and preserve the same approximation guarantee.
Finally, superadditivity formalizes the idea that reconnecting through an intermediate facility should not cost more than connecting directly to the final facility. If a client with current value $x_1$ first freezes at distance $x_2$ and later improves its connection to distance $x_3$, then the total amount it bids along this two-step process, which is $\beta(x_1, x_2) + \beta(x_2, x_3)$, should be at most the amount it would bid if it connected directly from $x_1$ to $x_3$, i.e. $\beta(x_1,x_3)$.

\paragraph{Cost of the algorithm.} Having formally defined Algorithm~\ref{alg:dual_generalized}
and the properties of the bidding function, we can now calculate the total cost in an execution of the algorithm. Fix an execution of Algorithm~\ref{alg:dual_generalized}, and let $i_0,i_1,\ldots,i_t$ be the successive facilities that strictly improve the connection of a client $j$, ordered by time of appearance. Thus
\[
    d(i_0,j) > d(i_1,j) > \cdots > d(i_t,j)\,,
\]
where $i_0$ is the facility to which $j$ first connects, and $i_t$ is the facility serving $j$ at termination. If $j$ becomes inactive when its value is $\alpha_j^\star$, its total cost is
\[
    \cost_j(\SS^\star)
    =
    \beta(\alpha_j^\star,d(i_0,j))
    +
    \sum_{\ell=0}^{t-1}\beta(d(i_\ell,j),d(i_{\ell+1},j))
    +
    d(i_t,j)\,.
\]
The first term corresponds to the amount paid when the client first becomes inactive, the intermediate terms account for later improvements of its connection as new facilities open, and the final term is the remaining connection cost to the facility that serves the client in the final solution. The total cost of the algorithm is then
\[
    \cost(\SS^\star) = \sum_{j\in D}\cost_j(\SS^\star)\,.
\]

For the standard Greedy, which corresponds to the bid function $\beta(\alpha,d)=[\alpha-d]^+$, the summation of subsequent bids in a client's cost telescopes exactly. Indeed, every term is of the form ``old distance minus new distance'', and therefore
\[
    \cost_j(\SS^\star)
    =
    (\alpha_j^\star-d(i_0,j))
    +
    \sum_{\ell=0}^{t-1}\bigl(d(i_\ell,j)-d(i_{\ell+1},j)\bigr)
    +
    d(i_t,j)
    =
    \alpha_j^\star\,.
\]
This is the reason why the classical Greedy analysis is especially clean: the total cost of the algorithm is simply
\[
    \cost(\SS^\star)=\sum_{j\in D}\alpha_j^\star\,.
\]
For general bidding functions, this exact telescoping identity no longer holds, and this introduces one of the main difficulties in the analysis: even for $\beta_\gamma(\alpha,d)=[\alpha-\gamma d]^+$, the approximation ratio of the algorithm can no longer be upper bounded by the value of the same linear program, but instead must be upper bounded by the maximum of a family of $2^{\Theta(n^2)}$ linear programs.

\section{Approximation Guarantees}
\label{sec:approx_guarantee}
In this section we develop the approximation analysis for Algorithm~\ref{alg:dual_generalized}. In the classical Greedy analysis of \cite{JMS02,JainMMSV03}, the guarantee is captured by a factor-revealing linear program. For a general bidding function $\beta$, the same high-level approach still applies, but the resulting optimization problem is no longer linear in general. We therefore begin by formulating a factor-revealing mathematical program whose optimum controls the worst-case LMP ratio of the algorithm.

This framework will be used as the starting point for a rigorous proof for the one parameter family of bidding functions $\beta(\alpha,d)=[\alpha-\gamma d]^+$, yielding our provable $1.9499$-LMP guarantee. We first state the general factor-revealing program and its implication for the approximation ratio.

\begin{definition}[Factor-revealing Program]
\label{def:frp}
\text{Fix $n\ge 1$.}
\begin{align}
\max \quad & \sum_{j=1}^n \Paren{\sum_{k=j}^{n} \beta(r_{j,k}, r_{j,k+1}) + r_{j,n+1}} - f
\label{eq:fr-obj}
\\
\text{s.t.}\quad 
& \sum_{i=1}^{n} d_i \;=\; 1
\qquad && && \text{(normalization)}
\label{eq:fr-norm}
\\
& r_{j,k+1} \;\le\; r_{j,k}
\qquad && \forall\, 1\le j \leq k \leq n && \text{(monotonicity of $r$)}
\label{eq:fr-r-mono}
\\
& r_{k,k} \;\le\; r_{j,k} + d_k + d_j
\qquad && \forall\, 1\le j<k\le n && \text{(triangle constraint)}
\label{eq:fr-triangle}
\\
& r_{j, \ell} \leq r_{k, \ell} + d_k + d_j && \forall\, 1 \leq j \neq k < \ell \leq n + 1 && \text{(metric inequalities of $r$)}
\label{eq:fr-r-metric}
\\
& \sum_{j=1}^{k-1} \beta(r_{j,k},d_j)
  \;+\;
  \sum_{j=k}^{n} \beta(r_{k,k},d_j)
  \;\le\; f
\qquad && \forall\, 1\le k\le n && \text{(facility constraint)}
\label{eq:fr-facility}
\\
& d_j, f, r_{j,k} \;\ge\; 0
\qquad && \forall\, j,k \text{ in range.} && \text{(nonnegativity)}
\label{eq:fr-nonneg}
\end{align}
\end{definition}

As in the classical analysis of \cite{JMS02,JainMMSV03}, the analysis can be reduced to a worst-case local configuration around a single facility of an optimal solution. Fix such a facility $i$ and index the clients $D^\prime$ assigned to it. 
To motivate the definition of the factor-revealing program in Definition~\ref{def:frp}, consider an execution of Algorithm \ref{alg:dual_generalized} where we order the clients in $D^\prime$ from $1$ to $n=|D^\prime|$ in increasing order of their final $\alpha$-values, $\alpha^\star_1 \le \alpha^\star_2 \le \dots \le \alpha^\star_n$. Equivalently, one can order the clients according to the order of events in which they become inactive.
For the sake of the presentation we assume that no two clients become inactive in the same event.

We set $d_j = d(i,j)$ to be the distance of client $j \in D^\prime$ to the fixed facility $i$. Note that the program is scale-invariant and we can normalize the distances by imposing $\eqref{eq:fr-norm}$. 
For two clients $j < k$, set $r_{j,k}$ as the connection distance to the closest open facility of $j$ at the time of the algorithm just before $k$ becomes inactive. Moreover, $r_{j,j} = \alpha^\star_j$ will be equal to the final $\alpha$-value of a client $j$. 
The first term of the objective 
\begin{equation*}
    \sum_{j=1}^n \Paren{\sum_{k=j}^{n} \beta(r_{j,k}, r_{j,k+1}) + r_{j,n+1}}
\end{equation*}
is an upper bound to the cost of Algorithm \ref{alg:dual_generalized} for clients in $D^\prime$.
For each client $j$, the sum 
\[
     \beta(r_{j,j}, r_{j,j+1}) + \sum_{k=j+1}^{n} \beta(r_{j,k}, r_{j,k+1})+ r_{j,n+1}
\]
upper bounds
\[
    \cost_j(\SS^\star)
    =
    \beta(\alpha_j^\star,d(i_0,j))
    +
    \sum_{\ell=0}^{t-1}\beta(d(i_\ell,j),d(i_{\ell+1},j))
    +
    d(i_t,j)\,.
\]
To verify this claim, we repeatedly use the 
\emph{superadditivity} of the bidding function as per Definition~\ref{def:admissible-bidding} to go from the exact cost to the upper bound. 
This property implies that, for any two times $p<q$, we can bound $\sum_{\ell=p}^{q-1}\beta(d(i_\ell,j),d(i_{\ell+1},j))$ by $\beta(d(i_p,j),d(i_{q},j))$, and thus $\beta(r_{j,k},r_{j,k+1})$ upper bounds all contributions done by the clients for the corresponding distance changes between $r_{j,k}$ and $r_{j,k+1}$.
Moreover, $r_{j,n+1}$ is equal to the connection cost of that client at the end of the algorithm, $d(i_t,j)$.

Next, we will argue that $(r, d, f)$ is indeed feasible for the factor revealing program (Definition \ref{def:frp}).  The monotonicity constraint $\eqref{eq:fr-r-mono}$ follows from the monotonicity of $d(j,\SS)$. For the triangle constraint \eqref{eq:fr-triangle} note that the algorithm never increases the $\alpha$-value of an active client $j$ past $d(j,\SS)$. Now, just before $k$ becomes inactive, let $i_0 \in \SS$ be the facility client $j$ is connected to, so $d(j,i_0)=r_{j,k}$. Then,
\begin{equation*}
    r_{k,k} = \alpha^\star_k \leq d(k,\SS) \le d(k,i_0) \le d(k,i) + d(j,i) + d(j,i_0) = r_{j,k} + d_j + d_k\,.
\end{equation*}
Similarly, for the metric inequality constraint \eqref{eq:fr-r-metric}, let $i_1$ and $i_2$ be the facilities $j$ and $k$ are connected to when $\ell > j,k$ becomes inactive, respectively. Then
\begin{equation*}
    r_{j,\ell} = d(j, i_1) \le d(j,i_2) \le d(j,i) + d(k,i) + d(k,i_2) = r_{k,\ell} + d_k + d_j\,.
\end{equation*}
Next, we consider the facility constraint \eqref{eq:fr-facility} and the time in the algorithm when client $k$ becomes inactive. 
Again, by the simplifying assumption, no other client becomes inactive together with $k$. Then all clients $j < k$ are already inactive and their bids are equal to $ \beta(d(j,\SS), d(j,i)) = \beta(r_{j,k}, d_j)$. Similarly, all clients $j \ge k$ are still active and bid with the current $\alpha$-value which is equal to $\alpha^\star_k$ as $k$ becomes inactive at that time. Thus, all clients $j \ge k$ bid $\beta(\alpha^\star_k, d(i,j)) = \beta(r_{k,k}, d_j)$. The inequality holds because the sum of bids towards a facility never exceeds the opening cost.

An upper bound on the optimum of the factor-revealing program immediately yields an LMP-approximation guarantee for Algorithm~\ref{alg:dual_generalized}. The following theorem states this reduction for our generalized bidding setting. 

\begin{restatable}{theorem}{theoremfacilitylocation}
\label{theorem:facility-location}
    Let $\rho \geq 1$ be an upper bound on the optimal value of the factor-revealing program in Definition~\ref{def:frp} over all $n \geq 1$.
    Then, Algorithm \ref{alg:dual_generalized} run with facility cost $\hat{f} = \rho \cdot f$ is a $\rho$-LMP approximation algorithm for UFL instances with facility cost $f$.
\end{restatable}
The proof of Theorem \ref{theorem:facility-location} is similar to the proof of Lemma 3 in \cite{JMS02} and relies on the observations made above. The proof is left for the Appendix~\ref{sec:proof-lmp-approx-ratio-greedy}. 

Given Theorem~\ref{theorem:facility-location}, to upper bound the LMP-approximation ratio of Algorithm \ref{alg:dual_generalized} it is enough to upper bound the optimum of the factor-revealing program in Definition~\ref{def:frp} over all $n \geq 1$. For a general bidding function, obtaining such an upper bound appears difficult analytically, and is non-trivial even computationally.

We now turn to the analysis of the LMP-approximation ratio for a simple generalization of the linear bidding function in \cite{JMS02,JainMMSV03} which changes the slope of $d$:
\begin{equation*}    
    \beta_{\Mcoeff}(\alpha, d) = [\alpha - \Mcoeff \cdot d]^+\,.
\end{equation*}
For a small value of $\gamma$, the algorithms created by this family of bidding functions can be well analyzed, as the following theorem shows.
\begin{restatable}{theorem}{thmoneninefivelmp}\label{thm:1.95-lmp}
    Algorithm~\ref{alg:dual_generalized} run with facility cost $\hat{f} = 1.9499 \cdot f$ and the bidding function $\beta(\alpha,d) = [\alpha - 1.224\cdot d]^{+}$ is a $1.9499$-LMP approximation algorithm for UFL instances with facility cost $f$.
\end{restatable}

To prove this result, Theorem~\ref{theorem:facility-location}  needs an upper bound on the factor-revealing program. This is what the following technical theorem describes.
\begin{restatable}{theorem}{thmupperboundfrp}\label{thm:1.95-frp}
    For values of $\gamma\in [1,1.224]$, the factor-revealing program from Definition~\ref{def:frp} with $\beta_\gamma(\alpha,d)=[\alpha-\gamma\cdot d]^+$ is upper-bounded by $\Mrho$ for all $n \geq 1$.
\end{restatable}
Note that this constant makes the program tight, given the lower bound from Section~\ref{sec:greedy_algorithm_general_bids}.
Let us outline the proof of Theorem~\ref{thm:1.95-frp}; the full proof is deferred to Section~\ref{sec:195-dual-fitting}. Moreover, a lean certificate for Theorem~\ref{thm:1.95-frp} can be found at \url{https://github.com/epfl-thl2/project-k-median}.
First, we substitute the bidding function into the factor revealing objective function from Definition~\ref{def:frp}, obtaining:
\begin{equation*}
    \sum_{j=1}^{n} \left(\sum_{i=j}^{n} \beta(r_{j,i}, r_{j,i+1})  + r_{j,n+1} \right) - f
    = \sum_{j=1}^{n} \left(\sum_{i=j}^{n} [r_{j,i} - \Mcoeff \cdot r_{j,i+1}]^+  + r_{j,n+1} \right) - f\,.
\end{equation*}

The main difference between this objective function and that of the original program of \cite{JMS02,JainMMSV03} is that the latter can be reduced to a single linear program, while ours cannot. 
This makes it difficult to directly analyze the program through its dual, as it cannot be easily worked with.
Thus, we reduce our program to a family of \emph{exponentially many} linear programs, to then upper bound each one in a similar way as they will all share some structure which we can exploit.

First, the constraints can be easily linearized, just as in the original analysis in \cite{JMS02,JainMMSV03}. The only nonlinear terms in the constraints are $\beta(r_{j,k},d_j)$ or $\beta(r_{k,k},d_j)$. Introducing an auxiliary variable using its epigraph does the trick.

The problem is in the objective, as we mentioned earlier. To linearize it, we introduce new binary variables $b_{i,j} \in \{0, 1\}$ indicating whether each bid $\beta(r_{j,i}, r_{j,i+1})$ in the objective is positive; hence, $\beta(r_{j,i}, r_{j,i+1})=b_{i,j} \cdot [r_{j,i} - \Mcoeff \cdot r_{j,i+1}] $.
Hiding the constraints, this transforms the program into:
\begin{equation*}
    \max_{\substack{\\ \forall b_{i,j}\in\{0,1\} \\ \forall i,j:\  1\leq j\leq i \leq n}} \max_{r,f,d} \quad 
            \sum_{j=1}^{n} \left(\sum_{i=j}^{n} b_{i,j} \cdot (r_{j,i} -\Mcoeff \cdot r_{j,i+1})  + r_{j,n+1}\right) - f\,.
\end{equation*}
\noindent This program will now be interpreted as a maximization of an LP program value over all possible choices of the binary variables $b_{i,j}$.
As we are interested in finding an upper bound to the whole program, we define and work with the dual of each inner linear program. After eliminating auxiliary dual variables, we obtain the following family of dual programs, which we denote all together as the dual of the whole program. 
\begin{restatable}[Factor-revealing dual]{definition}{deffrdual}
\label{def:frdual}
    Fix $n\geq 1$.
    \begin{align*}
    \smash{\max_{\substack{\forall b_{i,j}\in\{0,1\} \\ \forall i,j: 1\leq j\leq i \leq n}}}\qquad 
    & \min \quad \rho \\
    &\begin{aligned}
        &\text{\emph{subject to:}} 
        &&  \sum_{i=1}^{n} \hat{g}_i \leq 1 && \text{($f$)} \\
        & \forall \ 1 \le j, i \le n: 
        && \hat{g}_i \geq g_{i,j} && \text{($\gamma_{i,j}$ and $\beta_{i,j}$)}\\
        & \forall \ 1\le j \le i \le n: 
        && \sum_{q=j+1}^n g_{j,q} + \sum_{p = j}^{i} \left(g_{p,j} + t_{p,j} - s_{p,j}\right) \geq b_{i,j} - (\Mcoeff-1) \sum_{p=j}^{i-1} b_{p,j}
        && \text{($r_{j,i}$)}\\
        & \forall \ 1\leq i \leq n: 
        &&\sum_{p=1}^i \sum_{j=1}^p s_{p,j} - \sum_{p=1}^i \sum_{j=1}^p t_{p,j} \ge 0 && \text{($m_i$)}\\
        & \forall \ 1\leq j\leq n:
        && \rho \geq \ \Mcoeff\cdot \sum_{i=1}^n g_{i,j} + \sum_{i=j}^n s_{i,j} + \sum_{i=j}^n t_{i,j} && \text{($d_j$)} \\
        & \forall \ 1 \le i,j \le n: && \hat{g}_{i}, \ g_{i,j}, \ t_{i,j}, \ s_{i,j} \geq 0 ,\ s_{i,i}=0
    \end{aligned}
    \end{align*}
\end{restatable}
Appendix~\ref{sec:appendix-deriving-dual} gives the derivation.
\begin{restatable}{lemma}{lemmamathematicalprogram}\label{lemma:mathematical-program-1.99}
    The optimum value of the factor-revealing dual in Definition~\ref{def:frdual} upper bounds the optimum value of the factor-revealing program in Definition~\ref{def:frp} for the bidding function $\beta_{\Mcoeff}(\alpha, d) = [\alpha - \Mcoeff \cdot d]^+$.
\end{restatable}

The remaining goal is to find a feasible solution of value at most $\rho$ for each inner linear program. There appears to be no single dual solution that can work for all choices of $b_{i,j}\in\{0,1\}$, since the constraints depend on these values. To work around this issue, we need to design an adaptive dual solution that depends on the values of $b_{i,j}$.

Our construction adapts specifically to the \emph{first nonzero bidding time} $f_j$, i.e. the first time $i$ that has $b_{i,j}=1$, while ignoring the rest of values $b_{i,j}$ after that time instance. This simplifies the dual as we can upper bound the right hand side of the $(r_{j,i})$ constraint as: $0$ for $i<f_j$, $1$ for $i=f_j$ and $2-\gamma$ for $i>f_j$; reducing the complexity of finding a solution a lot.

Given these values of $f_j$, a central obstacle is finding a correct profile to the values of $\hg_j$ that appear in the dual. We create a group of indices $D$ that all have $\hg_j=0$ for all $j\in D$, and all indices out of this group having a different kind of behavior with $\hg_j>0$. The group $D$ is defined through a threshold $t_3^\star$, such that all $j$ with $f_j>t_3^\star$ and some with $f_j=t_3^\star$ are in $D$. We group the values of $t_3^\star$ and the memberships in $D$ of elements with $f_j=t_3^\star$ together in some advice variables $\Gamma$. Then, we run our procedure with the generated $D$ and if at the end we indeed have $\sum_{i=1}^n \hg_i = 1$ then we prove that we get a correct solution. This is summarized in the following lemma.

\begin{lemma}\label{lem:informal-general-case-procedure}
    For an assignment of the $b_{i,j}\in\{0,1\}$ values, correct advice variables $\Gamma$ and a value of $\Mcoeff\in(1,1.224]$, there exists a procedure that generates a $\Mrho$ feasible solution for the corresponding inner program of the factor-revealing dual from Definition~\ref{def:frdual}.
\end{lemma}

The second part of the proof relies on proving that there do exist correct advice variables $\Gamma$ such that through them, the final check $\sum_1^n \hg_i = 1$ follows. This follows from a continuity argument. We make the memberships in $D$ for the clients with $f_j=t_3^\star$ fractional, so that we can continuously interpolate between all the different possible advice variables. For example, suppose that two clients $p<q$ have $f_p=f_q=t_3^\star<n$, and assume that both of them are in $D$. Increasing by one the value of $t_3^\star$ changes their assignments discretely. Instead, we can first decrease the fractional membership $\mu_q^\star$ continuously from $1$ to $0$,
and then decrease $\mu_p^\star$ from $1$ to $0$. This shows how we can continuously interpolate between all choices.

We prove that the resulting value of $\sum_1^n \hg_i$ also varies continuously, with two endpoints that satisfy that the quantity is at least $1$ on one and at most $1$ on the other. Then, applying the intermediate value theorem alone does the trick. This is summarized by the following lemma.

\begin{lemma}\label{lem:informal-good-tau-existence}
    For any assignment of the $b_{i,j}\in\{0,1\}$ values with $\Mcoeff\in(1,1.224]$, there always exist correct advice variables $\Gamma$ with respect to the procedure from Lemma~\ref{lem:informal-general-case-procedure}.
\end{lemma}

These are simplified versions of the lemmas that we use in the full proof; see Section~\ref{sec:195-dual-fitting} for the full statements and proofs. We can now combine these lemmas to achieve the final result:

\begin{proof}[Proof sketch of Theorem~\ref{thm:1.95-lmp}]
    The combination of Lemma \ref{lem:informal-general-case-procedure} and Lemma \ref{lem:informal-good-tau-existence} implies that there always exists a solution of cost at most $\Mrho$ to each of the dual inner programs. This gives an upper bound of $\Mrho$ on the general dual program. Using this together with Lemma~\ref{lemma:mathematical-program-1.99} that states that this dual program is an upper bound on the factor revealing, it proves the theorem, given that $2-(1.224-1)^2\leq 1.9499$.
\end{proof}

\section{Log-Adaptive Algorithm}
\label{sec:logad}

In this section we present a \textit{smooth} version (Algorithm \ref{alg:log-adaptive-greedy}) of our greedy algorithm with general bids (Algorithm \ref{alg:dual_generalized}), which we call the log-adaptive algorithm.
Thus, we follow the recipe of~\cite{MR4928455-Cohen25} for obtaining an approximation algorithm for $k$-median that opens $k + O(\log n)$ centers: they first develop a log-adaptive version of the Greedy algorithm, and then combine it with a merging procedure that walks between solutions to find one that opens $k + O(\log n)$ centers.
We also follow these two steps, although our strategy for analyzing the log-adaptive algorithm is substantially different from theirs.
We then state our guarantees for $k$-median in Subsection~\ref{subsec:k-median}; the merging procedure is deferred to Section~\ref{sec:merging}.

The analysis of~\cite{MR4928455-Cohen25} derives the guarantees of their log-adaptive algorithm from scratch, by redoing the analysis of~\cite{JMS02,JainMMSV03} in the log-adaptive context.
In contrast, our goal is to connect the approximation ratio of the log-adaptive algorithm to the factor-revealing program in Definition~\ref{def:frp}.
Thus, we show that the log-adaptive algorithm induces a feasible solution to the factor-revealing program, and transfer the upper bound on the value of that program to an upper bound on the cost of our algorithm.

To make the Greedy algorithm smoother, \cite{MR4928455-Cohen25} increase the $\alpha$-values in geometrically spaced steps, i.e., the time parameter $\theta$ takes values in $\{1,(1+\eps^2),(1+\eps^2)^2,(1+\eps^2)^3, \dots\}$.
With this strategy, however, it is possible for clients to significantly overbid for facilities when $\theta$ jumps to $(1+\eps^2)\theta$.
To fix this issue, \cite{MR4928455-Cohen25} scale down the metric $d$ by a $1-\delta$ factor, which leads to a multiplicative $(1-O(\delta))$ loss in the LMP-approximation ratio.
This still only resolves the issue for clients that are far from facilities, so they choose to allow clients that are close to facilities to raise their $\alpha$-value in-between steps.
We also apply these fixes, but in addition we will need stronger structural properties in our analysis, and thus our algorithm will alternate between two stages of opening facilities which use \emph{different} scalings of the metric.
This gives us a clusterability property that we prove in Lemma~\ref{lem:structure}.

Below, we define two notions of openable facilities, which we call 1-openable (Definition \ref{def:openable-mod-1}) and 2-openable (Definition \ref{def:openable-mod-2}), used in Stage 1 and Stage 2 of the algorithm, respectively.
Similarly to Algorithm \ref{alg:dual_generalized}, we maintain a set of active clients $A$ and a set of open facilities $\SS$, initialized to $A=D$ and $\SS = \emptyset$, respectively.
Moreover, the state of the algorithm can be identified with the triple $(\SS,A,\theta)$ where $\theta$ is the time parameter.
\begin{definition}[1-openable facility]
\label{def:openable-mod-1}
Fix a state $(\SS,A,\theta)$ where $\SS\subseteq \FF$ is the set of opened facilities, $A\subseteq D$ are the active clients, $I:=D\setminus A$ are the inactive clients, and
    $\alpha_j=\theta$ for all $j\in A$.
A facility $i\in \FF\setminus \SS$ is 1-\emph{openable} w.r.t.\ $(\SS,A,\theta)$ if it 
satisfies:
\[
\sum_{j\in A}\beta(\theta, (1-\delta)\dist(i,j))
\;+\;
\sum_{j\in I}\beta((1-\delta)\dist(j,\SS), (1-\delta)\dist(i,j))
\;\ge\; f\,.
\]
\end{definition}

\begin{definition}[2-openable facility]
    \label{def:openable-mod-2}
    Fix a state $(\SS,A,\theta)$ where $\SS\subseteq \FF$ is the set of opened facilities,
    $A\subseteq D$ are the active clients, $I:=D\setminus A$ are the inactive clients, and
    $\alpha_j=\theta$ for all $j\in A$.
    A facility $i\in \FF\setminus \SS$ is 2-\emph{openable} w.r.t.\ $(\SS,A,\theta)$ if there exist bids $(\tau_j)_{j\in A}$
    satisfying:
    \begin{enumerate}
        \item[\textbf{(i)}] (\textbf{Only nearby clients may increase})
              \[
                  \tau_j=\alpha_j \qquad \forall j\in A\setminus \Ball{i}{\epsilon\theta}\,.
              \]
        \item[\textbf{(ii)}] (\textbf{Nearby clients increase only slightly})
              \[
                  \alpha_j \le \tau_j \leq \min\Bigl\{(1-\delta)\dist(j,\SS),\, (1+\epsilon^2) \theta \Bigr\}
                  \qquad \forall j\in A\cap \Ball{i}{\epsilon\theta}\,.
              \]
        \item[\textbf{(iii)}] (\textbf{Facility $i$ is paid for in the $(1-\delta/2)$-scaled metric})
              \[
                  \sum_{j\in A}\beta(\tau_j, (1- \delta/2)\dist(i,j))
                  \;+\;
                  \sum_{j\in I}\beta((1-\delta)\dist(j,\SS), (1- \delta/2)\dist(i,j))
                  \;\ge\; f\,.
              \]
        \item[\textbf{(iv)}] (\textbf{No-overbidding after updating to $\tau$})
              \[
                  \forall i_0\in \FF:\qquad
                  \sum_{j\in A}\beta(\tau_j,\dist(i_0,j))
                  \;+\;
                  \sum_{j\in I}\beta((1-\delta)  \dist(j,\SS), \dist(i_0,j))
                  \;\le\;  f\,.
              \]     
        \item[\textbf{(v)}] (\textbf{Triangle inequality})
              \[
                  \tau_j \le \tau_k + \dist(j,k)
                  \qquad \forall j,k\in A\,.
              \]
    \end{enumerate}
\end{definition}
Here, 1-openable facilities are facilities that are paid for with the current value of $\theta$, when the existing active clients contribute enough in the scaled metric.
Thus, Stage~1 captures facilities that can be opened immediately. 
On the other hand, 2-openable facilities are facilities where we have to raise the bids of clients that are close to the facility in order to pay for the opening cost: clients in $B(i,\epsilon \theta)$ may increase their bid from $\theta$ to intermediate values $\tau_j$.
The no-overbidding constraint ensures that we do not overbid towards another facility when raising the $\alpha$-values to $\tau$.
Furthermore, the triangle constraints are necessary because the factor-revealing program requires triangle inequalities on the $\alpha$-values, see constraint \eqref{eq:fr-triangle}. %

The log-adaptive algorithm (Algorithm \ref{alg:log-adaptive-greedy}) has a three-stage structure within each phase. Stage~1 greedily opens 1-openable facilities. Stage~2 opens 2-openable facilities, potentially raising the $\alpha$-values of some clients. Then, once there are no more 2-openable facilities, we proceed to Stage~3, increasing $\theta$ to $(1+\eps^2)\theta$.

\begin{algorithm}[H]
    \caption{Log-Adaptive Algorithm}
    \label{alg:log-adaptive-greedy}
    \textbf{Initialization:} Set $\SS = \emptyset$, $A = D$, $\theta = 1$, and
    $\alpha_j = \theta$ for every $j \in A = D$.

    \textbf{Repeat the following while} $A \neq \emptyset$:
              \begin{enumerate}
                  \item While there is an unopened facility $i$ that is 1-openable (chosen arbitrarily if there are many):
                        \begin{quote}
                            Open $i$ (i.e., add it to $\SS$) and mark all
                            clients $j \in A$ such that $\alpha_j \geq (1-\delta)\dist(j,\SS)$ as served
                            (i.e., remove them from $A$).
                        \end{quote}
                  \item While there is an unopened facility $i$ that is 2-openable (chosen arbitrarily if there are many):
                        \begin{quote}
                            Calculate $(\tau_j)_{j\in A}$ so that $i$ is openable with bids $(\tau_j)_{j\in A}$.
                            Open $i$ (i.e., add it to $\SS$), let $\alpha_j \leftarrow \tau_j$ for every $j \in A$
                            (note that only the values of the clients in $A \cap B(i,\eps\theta)$ change), and mark all
                            clients $j \in A$ such that $\alpha_j \geq (1-\delta)\dist(j,\SS)$ as served
                            (i.e., remove them from $A$).
                        \end{quote}
                  \item Perform the following step and proceed to the next phase:
                        \begin{quote}
                            For every remaining $j \in A$, let
                            \[
                                \alpha_j \leftarrow \min\bigl((1+\eps^2)\theta,\,(1-\delta)\dist(j,\SS)\bigr)\,,
                            \]
                            remove $j$ from $A$ when $\alpha_j = (1-\delta)\dist(j,\SS)$, and update
                            $\theta \leftarrow (1+\eps^2)\theta$.
                        \end{quote}
              \end{enumerate}
    \textbf{Output $\SS$.}
\end{algorithm}

\begin{remark}
In Stage 2, if a client $j \in A$ has $\tau_j > \theta$, then it is $\epsilon \theta$-close to the facility that is opened at that step and it immediately becomes inactive.
Therefore, it is valid to assume in Definition~\ref{def:openable-mod-2} that $\alpha_j = \theta$ for all $j \in A$.
\end{remark}
In this section, we will prove the following main theorem that states that the log-adaptive algorithm (Algorithm \ref{alg:log-adaptive-greedy}) is indeed an LMP-approximation algorithm with factor $(1+2\delta)\rho$, where $\rho$ is an upper bound on the factor revealing program. In other words, the connection cost of the algorithm is bounded by $(\OPT - |\SS^\star| \cdot f)$ times the approximation factor $(1+2\delta)\rho$.
\begin{theorem}
\label{theorem:alg-lmp}
    Let $\gamma \geq 1$ and let $\beta(\alpha,d) = [\alpha - \gamma \cdot d]^+$.\footnote{For simplicity, we specialized the results of this section to the family of bidding functions that we study, and similarly in Appendix~\ref{sec:merging} and Appendix~\ref{sec:logad-robust}.
    However, the results are adaptable to any bidding function that satisfies Definition~\ref{def:admissible-bidding} and, in addition, satisfies the following properties: Let $\hat\beta(x) = \beta(1,x)$ for all $0 \leq x \leq 1$. We require that $\hat\beta(0)>0$. In addition, we require that (1) $\hat\beta$ is $L_1$-Lipschitz for some $L_1 >0$, and (2) $\hat\beta(x) - \hat\beta(y) \geq L_2(y-x)$ for some $L_2 > 0$ and all $0 \leq x \leq y \leq 1$ for which $\hat\beta(y) > 0$.} Let $\rho \geq 1.01 \gamma$ be an upper bound on the optimal value of the factor-revealing program in Definition~\ref{def:frp} over all $n \geq 1$.
    For any $\delta > 0$ small enough compared to $1/\gamma$ as well as any $\eps > 0$ small enough compared to $1/\gamma$ and $\delta$,
    Algorithm~\ref{alg:log-adaptive-greedy} run with facility cost $\hat{f} = \rho \cdot f$ is a $(1+2\delta)\rho$-LMP approximation for UFL instances with facility cost $f$.
\end{theorem}
This theorem is analogous to Theorem~\ref{theorem:facility-location}, with an additional loss of $2\delta$ because of the log-adaptivity.
We will prove the theorem in the following sections. The proof is concluded in Section~\ref{subsec:lmp-proof}.
We remark that, if $f=0$, then Stage 1 opens all facilities and the resulting
solution is optimal; therefore, for the rest of this section we assume $f>0$.

\subsection{Preliminaries}
Throughout this section we use two small absolute constants $\delta,\eps > 0$.
We write
\[
    1/\gamma \gg \delta \gg \eps > 0\,,
\]
where ``$\gg$'' means that each quantity is chosen sufficiently small as a function of the preceding one.

We define an \emph{atomic step} in the algorithm as the execution of one iteration in Stage 1, the execution of one iteration in Stage 2, or the execution of Stage 3.
When we say something holds at ``every step'' in the algorithm, we mean before and after every atomic step.

For a client $j \in D$, we denote by $\alpha_j^\star$ its final $\alpha_j$ value.
For the purposes of the analysis, we order the clients from $1$ to $n$ by the order of the atomic steps in which they become inactive. 
In case two clients become inactive in the same atomic step, we order them non-decreasingly by their $\alpha^\star$ value.
We denote by $\SS^{(k)}$ the set of open facilities at the beginning of the atomic step in which client $k$ becomes inactive.
We denote by $\SS^{(n+1)} = \SS^\star$ the set of open facilities at the end of the algorithm.
For $1 \le j \leq k \le n$, we set 
\[
    r_{jk} = 
    \begin{cases}
        \alpha_j^\star &
        \begin{array}[t]{@{}l@{}}
            \text{if $j$ is active at the beginning of the atomic step}\\
            \text{in which $k$ becomes inactive}\,,
        \end{array}
        \\[2mm]
       (1-\delta)~d(j, \SS^{(k)}) & \text{otherwise}\,.
    \end{cases}
\]
In other words, if $j \le k$ become inactive in the same atomic step, then $ r_{jk} = \alpha^\star_j$ and otherwise $ r_{jk} =  (1-\delta)~d(j, \SS^{(k)})$.
We also set $r_{j,n+1} = (1-\delta)d(j, \SS^\star)$. 

\begin{observation}
\label{obs:r-noninc}
We have $r_{jk} \geq r_{jk+1}$ for all $1 \leq j \leq k \leq n$.
\end{observation}
\begin{proof}
    We have trivially that $\alpha_j^\star \leq \alpha_j^\star$ and $(1-\delta) d(j, \SS^{(k)}) \geq (1-\delta) d(j, \SS^{(k')})$ for $k < k'$ as $\SS^{(k)}$ is monotone.
    Moreover, $\alpha_j^\star \geq (1-\delta)d(j, \SS)$ for the set $\SS$ of open facilities after the atomic step when $j$ becomes inactive, so the claim follows.
\end{proof}

Our analysis will also use the notion of candidate facilities which are facilities that are paid for as in the definition of 2-openable. 
\begin{definition}
We say that a facility $i$ is a candidate at some point of the algorithm if it satisfies
\begin{align*}
    \sum_{j\in A}\beta((1+\eps^2) \theta, (1- \delta/2)\dist(i,j))
    \;+\;
    \sum_{j\in I}\beta((1-\delta)\dist(j,\SS), (1-\delta/2)\dist(i,j))
    \;\ge\; \hat f\,.
\end{align*}
\end{definition}
In other words, a candidate has the possibility to be paid for by increasing all the active clients to the maximum of that stage. In particular, all facilities that are opened in Stage 2 will be candidate facilities at that point of the algorithm. 

\subsection{Algorithm Properties}
First, we prove that at every step of the algorithm the clients do not overbid toward any facility.
\begin{lemma}[No overbidding]
    \label{lem:overbidding}
    At any step of the algorithm, for every facility $i \in F$, 
    \[
        \sum_{j \in A} \beta(\theta, d(i,j)) + \sum_{j \in I} \beta((1-\delta)d(j, \SS), d(i,j)) \leq \hat f\,.
    \]
\end{lemma}
\begin{proof}
The inequality holds at the beginning of the algorithm because $\alpha_j = 1 \le d(i,j)$ for every client $j$ by assumption \eqref{eq:dist-bound}.
We will show that invariant holds after every atomic step of the algorithm.
The atomic step might happen in each of the three stages.

\paragraph{Stage 1 and Stage 2}
Opening up a facility $i_0$ decreases the left hand side since $d(j,\SS)$ is monotonically decreasing with $\SS$.
The only other change to the left hand side that can happen is that a client $j$ becomes inactive.
In that case, however, either $\theta = \alpha_j^\star \ge (1-\delta)d(j,\SS\cup \{i_0\})$ or $j \in B(i_0,\epsilon \theta)$ and $\epsilon \theta \ge d(j,i_0) \ge  (1-\delta)d(j,\SS\cup \{i_0\})$.
Thus, the left hand side decreases by the monotonicity of $\beta$.

\paragraph{Stage 3}
Let $A, I$ be the sets of active and inactive clients after raising all $\alpha_j$ to $\min\bigl((1+\eps^2)\theta,\,(1-\delta)\dist(j,\SS)\bigr)$ for all $j$ that are active at the beginning of this step.
Assume for contradiction that after raising $\theta$ to $(1+\eps^2)\theta$ the invariant is violated for some facility, and let $\FF_0$ be the set of all such facilities, such that for all $i \in \FF_0$
\[
    \sum_{j \in A} \beta((1+\epsilon^2)\theta, d(i,j)) + \sum_{j \in I} \beta((1-\delta)d(j, \SS), d(i,j)) > \hat f\,.
\]
Note that none of the facilities $i \in \FF_0$ are open because otherwise there would be no client with a positive bid to $i$ (for clients $j \in A$ we would have $(1+\eps^2)\theta < (1-\delta)d(j,\SS) \leq (1-\delta)d(j, i)$, and for clients $j \in I$ we would have $(1-\delta)d(j,\SS) \leq (1-\delta)d(j,i) \leq d(j, i)$).

Also let $A' \supseteq A$ and $I' \subseteq I$ be the active and inactive sets at the beginning of Stage 3, respectively.
For each facility $i_0 \in \FF_0$ compute $\theta \le \tau^{(i_0)} \le (1+\eps^2)\theta$ such that for $\alpha^\prime_{i_0,j} = \min\bigl(\tau^{(i_0)},\,(1-\delta)\dist(j,\SS)\bigr)$ we have
\[
    \sum_{j \in A'} \beta(\alpha^\prime_{i_0,j}, d(i_0,j)) + \sum_{j \in I'} \beta((1-\delta)d(j, \SS), d(i_0,j)) = \hat f\,.
\]
Let us prove that such $\tau^{(i_0)}$ exists: We know by the invariant before this atomic step that
\[\sum_{j \in A'} \beta(\theta, d(i_0,j)) + \sum_{j \in I'} \beta((1-\delta)d(j, \SS), d(i_0,j)) \leq \hat f,\]
so it also holds for $\alpha^\prime_{i_0,j}$ with the choice $\tau^{(i_0)} = \theta$.
We also show that the invariant is violated for $\alpha^\prime_{i_0,j}$ with the choice $\tau^{(i_0)} = (1+\eps^2)\theta$, and then the existence of $\tau^{(i_0)}$ follows by continuity.
Recall that
\[\sum_{j \in A} \beta((1+\epsilon^2)\theta, d(i_0,j)) + \sum_{j \in I} \beta((1-\delta)d(j, \SS), d(i_0,j)) > \hat f.\]
Then, for $j \in A$ we have $(1+\eps^2)\theta < (1-\delta)d(j,\SS)$ and so $\alpha^\prime_{i_0,j} = (1+\eps^2)\theta$, and for $j \in A' \cap I$ we have $(1+\eps^2)\theta \geq (1-\delta) d(j,\SS)$ and so $\alpha^\prime_{i_0,j} = (1-\delta)d(j,\SS)$, so also
\[\sum_{j \in A'} \beta(\alpha^\prime_{i_0,j}, d(i_0,j)) + \sum_{j \in I'} \beta((1-\delta)d(j, \SS), d(i_0,j)) > \hat f.\]
Therefore $\tau^{(i_0)}$ exists by continuity.

Now select the facility $i \in \FF_0$ that minimizes $\tau^{(i)}$. This implies that $\alpha^\prime_{i,j} \leq \alpha^\prime_{i_0,j}$ for all $i_0 \in \FF_0$ and $j \in D$, so
\begin{equation}
    \label{eq:minimal-facility}
    \sum_{j \in A'} \beta(\alpha^\prime_{i,j}, d(i_0,j)) + \sum_{j \in I'} \beta((1-\delta)d(j, \SS), d(i_0,j)) \le \hat f\,.
\end{equation}
From now on, let $\alpha_j^\prime = \alpha_{i,j}^\prime$.
We will show that $i$ was 2-openable at the end of Stage 2 with bids $\tau_j = \min\bigl(\tau^{(i)},\,(1-\delta)\dist(j,\SS), (1+\epsilon) \theta - d(j,i)\bigr)$ for all $j \in A' \cap B(i,\epsilon \theta)$ and $\tau_j = \theta$ for all $j \in A'\setminus B(i,\epsilon \theta)$.
Let us check the conditions one by one,
\begin{enumerate}
        \item[\textbf{(i)} - \textbf{(ii)}] Clearly holds by the definition of $\tau$.
        \item[\textbf{(iii)}]
        We have that 
        \begin{align*}
            \hat f & = \sum_{j \in A'} \beta(\alpha_j^\prime, d(i,j)) + \sum_{j \in I'} \beta((1-\delta)d(j, \SS), d(i,j)) \\ 
            & \le \sum_{j \in A'} \beta(\tau_j, (1- \delta/2)\dist(i,j)) + \sum_{j \in I'} \beta((1-\delta)d(j, \SS), (1- \delta/2)\dist(i,j))\,.
        \end{align*}
        Let us verify that the inequality holds. First, consider $j \in B(i,\epsilon \theta)$. In the case where $\tau_j = \alpha^\prime_j$ we have by the monotonicity of $\beta$ that $\beta(\alpha_j^\prime, d(i,j)) \le \beta(\tau_j,  (1-\delta/2) d(i,j))$. In the other case when $\tau_j = (1+\epsilon) \theta - d(i,j)$ observe that
        \begin{align*}
            \beta(\alpha_j^\prime, d(i,j)) & \le \pos{(1+\epsilon^2)\theta - \gamma  d(i,j)} \\ & \le \pos{(1+\epsilon)\theta - d(i,j) - \gamma (1-\delta/2) d(i,j)}
            \\ & = \beta(\tau_j,  (1-\delta/2) d(i,j))
        \end{align*}
        because using that $\gamma \geq \gamma (1-\delta/2) + \eps$ for $\delta \gg \eps$ we have
    \begin{align*}
        (1+\epsilon) \theta - d(i, j)-(\gamma (1-\delta/2)) d(i, j)
        &= (1+\epsilon) \theta - (1-\eps)d(i, j)-(\gamma (1-\delta/2) + \eps) d(i, j)\\
        &\geq (1+\eps^2) \theta - (\gamma (1-\delta/2) + \eps)d(i,j)\\
        &\geq (1+\eps^2)\theta - \gamma d(i, j)\\
        &\geq \alpha_j' - \gamma d(i, j)\,.
    \end{align*}
        For $j \in A' \setminus B(i,\epsilon \theta)$, using that $\eps^2 \theta \le \epsilon \gamma \delta /2~\theta \le \gamma \delta/2~ d(i,j)$,
        \begin{align*}
            \beta(\alpha_j^\prime, d(i,j)) = \pos{\alpha_j^\prime - \gamma  d(i,j)} \le \pos{\theta - \gamma (1-\delta/2) d(i,j)} = \beta(\tau_j,  (1-\delta/2) d(i,j))\,.
        \end{align*}
        \item[\textbf{(iv)}] Assume that the no-overbidding inequality is violated for some facility $i_0 \in \FF_0$,
        \begin{align*}
            \hat f & < \sum_{j\in A'}\beta(\tau_j,\dist(i_0,j))
                  \;+\;
                  \sum_{j\in I'}\beta((1-\delta)  \dist(j,\SS), \dist(i_0,j))
                  \\ & \le \sum_{j\in A'} \beta(\alpha^\prime_j,\dist(i_0,j))
                  \;+\;
                  \sum_{j\in I'}\beta((1-\delta)  \dist(j,\SS), \dist(i_0,j))\,.
        \end{align*}
        The last inequality holds by the monotonicity of $\beta$.
        We arrive at a contradiction to the minimal choice of $i$. In particular, a contradiction to (\ref{eq:minimal-facility}). 
        
        On the other hand, for $i_0 \in \FF \setminus \FF_0$, the claim follows because we have 
        \[
\sum_{j \in A} \beta((1+\epsilon^2)\theta, d(i_0,j)) + \sum_{j \in I} \beta((1-\delta)d(j, \SS), d(i_0,j)) \leq \hat f
    \]
        and we have $(1+\epsilon^2)\theta \geq \tau_j$ and for $j \in A' \cap I$ we have $(1-\delta)d(j, \SS) \geq \tau_j$ so also
        \[
\sum_{j \in A'} \beta(\tau_j, d(i_0,j)) + \sum_{j \in I'} \beta((1-\delta)d(j, \SS), d(i_0,j)) \leq \hat f\,.
    \]
        \item[\textbf{(v)}] We have to show that
              \[
                  \tau_j \le \tau_k + \dist(j,k)
                  \qquad \forall j,k\in A'
              \]
              which is clearly true if $\tau_j = \theta$ or $\tau_k = \tau^{(i)}$. Moreover, if $\tau_k = (1+\epsilon)\theta - d(i,k)$ then
              \[
                \tau_j \leq (1+\epsilon) \theta - d(i, j) \leq (1+\epsilon) \theta - (d(i, k) - d(j, k)) = \tau_k + d(j,k)\,.
            \]
            From now on we assume that $j \in B(i,\epsilon \theta)$. 
            Consider the two remaining cases where either $\tau_k =\theta$ and $k \in A' \setminus B(i,\epsilon \theta)$ or $\tau_k = (1-\delta)d(k,\SS) = (1-\delta)d(k,i_0)$ for some open facility $i_0$. In the first case,
              \[
                \tau_j \le (1+\epsilon) \theta - d(j,i) \le (1+\epsilon) \theta - (d(k,i) - d(j,k)) \le \theta + d(j,k) = \tau_k + d(j,k)\,.
              \]
              In the second case, 
              \[
                \tau_j \le (1-\delta)d(j,\SS) \le (1-\delta)d(j,i_0) \le (1-\delta) (d(i_0,k) + d(j,k)) \le \tau_k + d(j,k)\,.
              \]
    \end{enumerate}
\end{proof}

Next, we will show that without 1-openable facilities, which is the case after finishing Stage 1, all facilities that are candidates in Stage 2 are either far apart or close together. This structural property will be crucial to prove feasibility for the factor revealing program.
\begin{lemma}
\label{lem:structure}
    Suppose two facilities $i_1$ and $i_2$ are candidates at the same step during Stage 2.
    Then either
    $d(i_1,i_2)\le \delta \theta$ or $d(i_1,i_2)\ge \theta / \gamma$.
\end{lemma}
\begin{proof}
If $i_1 = i_2$ then the claim is immediate so we assume $i_1 \neq i_2$.
Let $N_1 = A \cap B(i_1, \epsilon \theta)$ and $N_2 = A \cap B(i_2, \epsilon \theta)$.
Let us assume without loss of generality that $|N_1| \le |N_2|$.
Let us also assume for the sake of contradiction that $\delta \theta < d(i_1, i_2) < \theta/\gamma$, which implies that $N_1$ and $N_2$ are disjoint as $\delta \gg \eps$.

Let 
\[P = \sum_{j\in A}\beta((1+\eps^2) \theta, (1- \delta/2)\dist(i_1,j)) +
\sum_{j\in I}\beta((1-\delta)\dist(j,\SS), (1-\delta/2)\dist(i_1,j))\]
and 
\[Q = \sum_{j\in A}\beta(\theta, (1-\delta)\dist(i_1,j)) +\sum_{j\in I}\beta((1-\delta)\dist(j,\SS), (1-\delta)\dist(i_1,j))\,.\]
Because $i_1$ is a candidate, $P \geq \hat f$.
We also claim that $Q < \hat f$.
If $i_1 \not\in \SS$, we claim that this is true because otherwise facility $i_1$ would have been 1-openable in Stage 1.
Let $\SS' \subseteq \SS$ be the set of open facilities at the beginning of Stage 2, and let $A', I'$ be the set of active and inactive clients at the beginning of Stage 2.
Note that $A' \supseteq A$ and $I' \subseteq I$.
For $j \in I'$ we use $(1-\delta) d(j,\SS) \leq (1-\delta) d(j, \SS')$, and for $j \in A' \cap I$ we use $(1-\delta) d(j,\SS) \leq \theta$.
The last inequality is immediate if $\alpha_j^\star = \theta$, and if $\alpha_j^\star > \theta$ then $j$ must have been $\eps \theta$-close to some facility in $\SS$ and hence $d(j,\SS) \leq \eps\theta$.
Then
\begin{align*}
Q &= \sum_{j\in A}\beta(\theta, (1- \delta)\dist(i_1,j)) +
\sum_{j\in I}\beta((1-\delta)\dist(j,\SS), (1-\delta)\dist(i_1,j))\\
&\leq \sum_{j\in A'}\beta(\theta, (1- \delta)\dist(i_1,j)) +
\sum_{j\in I'}\beta((1-\delta)\dist(j,\SS'), (1-\delta)\dist(i_1,j))\,,
\end{align*}
so the last term would have been at least $\hat f$ if $Q \geq \hat f$, and then facility $i_1$ would have been 1-openable in Stage 1.
On the other hand, if $i_1 \in \SS$, the claim is true simply because $\theta < (1-\delta) d(i_1, j)$ for $j \in A$ and $d(j, \SS) \leq d(i_1, j)$ for $j \in I$, so the left-hand side would be zero.
Hence $Q < \hat f$.

We will prove that $Q \geq P$, which is impossible because $Q < \hat f \leq P$, and thus we reach the desired contradiction.
Note that the inactive term of $P$ is trivially upper bounded by the inactive term of $Q$, so it suffices to prove that this is also true for their active terms.
We consider three cases:
\begin{itemize}
    \item \textbf{Case $j \in N_1$.} 
    Then, because $d(i_1, j) \leq \eps \theta$, we have that $\beta((1+\eps^2)\theta, (1-\delta/2)d(i_1,j))$ and $\beta(\theta,(1-\delta)d(i_1,j))$ are both positive and
    \begin{align*}
    &[\theta - \gamma(1-\delta)d(i_1,j)]^+ - [(1+\eps^2)\theta - \gamma(1-\delta/2)d(i_1,j)]^+\\
    =& -\eps^2 \theta + (\gamma \delta / 2) d(i_1, j)
    \geq -\eps^2 \theta\,.
    \end{align*}
    \item \textbf{Case $j \in N_2$.}
    Then $\delta \theta - \eps \theta \leq d(i_1, j) \leq \theta / \gamma + \eps \theta$, and we have again that $\beta((1+\eps^2)\theta, (1-\delta/2)d(i_1,j))$ and $\beta(\theta,(1-\delta)d(i_1,j))$ are both positive and
    \begin{align*}
    &[\theta - \gamma(1-\delta)d(i_1,j)]^+ - [(1+\eps^2)\theta - \gamma(1-\delta/2)d(i_1,j)]^+\\
    =& -\eps^2 \theta + (\gamma \delta / 2) d(i_1, j)
    \geq -\eps^2 \theta + (\gamma \delta / 2)(\delta-\eps)\theta
    \geq \eps^2 \theta\,,
    \end{align*}
    where we used $\delta \gg \eps$.
    \item \textbf{Case $j \in A \setminus (N_1 \cup N_2)$.}
    Then $d(i_1, j) \geq \eps \theta$.
    First, we claim that $\beta((1+\eps^2)\theta, (1-\delta/2)d(i_1,j))$ is positive only if $\beta(\theta,(1-\delta)d(i_1,j))$ is positive:
    Suppose
    \[(1+\eps^2)\theta - \gamma(1-\delta/2)d(i_1, j) \geq 0\,.\]
    Then
    \begin{align*}
    \theta - \gamma (1-\delta) d(i_1, j)
    \geq \frac{\gamma(1-\delta/2)}{1+\eps^2}d(i_1,j)-\gamma(1-\delta)d(i_1,j) \geq 0\,,
    \end{align*}
    where we used $\delta \gg \eps$.
    So it suffices to consider the case when both are positive.
    In that case,
    \begin{align*}
    &[\theta - \gamma(1-\delta)d(i_1,j)]^+ - [(1+\eps^2)\theta - \gamma(1-\delta/2)d(i_1,j)]^+\\
    =& -\eps^2 \theta + (\gamma \delta / 2) d(i_1, j)
    \geq -\eps^2 \theta + (\gamma \delta / 2) \eps \theta
    \geq 0\,,
    \end{align*}
    where we used again $\delta \gg \eps$.
\end{itemize}
Summing up these contributions over all $j \in A$, and using that $|N_1| \leq |N_2|$, we get that $Q \geq P$, finishing the proof.
\end{proof}

Finally, we prove a structural that is rather technical but is essential in the proof that our solution is feasible for the factor-revealing program:
For any fixed facility $i^\star$, if one considers the subset of clients that are not too far from $i^\star$ when they become inactive, then if any of these clients $j$ is active in Stage $2$ with $\tau_j > \theta$, then all other active clients in this subset also become inactive in the same atomic step.

\begin{lemma}
    \label{lem:freezing}
    Assume $\rho \geq 1.01 \gamma$.
    Consider a facility $i^\star$ and a subset of clients $D^\star \subseteq [n]$ such that for all $j \in D^\star$ it holds that $\alpha_j^\star \geq \rho \cdot d(j, i^\star)$.
    Suppose that a client $j \in D^\star$ becomes inactive in Stage 2 and has $\alpha_j^\star > \theta$ at that step.
    Suppose also that facility $i^\star$ was a candidate at that step.
    Then, for the facility $i_0$ that was opened at that step, we have $d(i^\star, i_0) \leq \delta \theta$.
    Furthermore, all other clients in $D^\star$ that were active before this atomic step become inactive in this atomic step.
\end{lemma}
\begin{proof}
    Both $i^\star$ and $i_0$ are candidates, so by Lemma~\ref{lem:structure} we have that either $d(i^\star,i_0) \leq \delta \theta$ or $d(i^\star,i_0) \geq \theta/\gamma$.
    Note that, because $\alpha_j^\star > \theta$, we have $d(j, i_0) \leq \epsilon \theta$.
    In Stage 2 we have $\alpha^\star_j \leq (1+\eps^2) \theta$, so combining with $\alpha_j^\star \geq \rho \cdot d(j, i^\star)$ we get that $(1+\eps^2) \theta \geq \rho \cdot d(j, i^\star)$.

    If $d(i^\star, i_0) \geq \theta/\gamma$, we have that $d(j, i^\star) \geq d(i^\star, i_0) - d(j, i_0) \geq \theta/\gamma - \epsilon\theta \geq 0.999 \cdot \theta/\gamma$ (using $1/\gamma \gg \epsilon$), which violates the assumption that $d(j, i^\star) \leq (1+\eps^2) \theta/\rho$.

    Hence, it must be the case that $d(i^\star, i_0) \leq \delta \theta$.
    Suppose some $k \in D^\star$ that is active before this atomic step does not become inactive at this atomic step. 
    This implies that $\tau_k < (1-\delta) d(k, i_0)$, but we must also have $\tau_k = \theta$ because otherwise $k$ would have been $\epsilon \theta$-close to $i_0$ and it would have become inactive at this atomic step, so $\theta < (1-\delta) d(k, i_0)$.
    However, after $i_0$ is opened, this client will become inactive at the latest when its $\theta$ value increases to at least $(1-\delta) d(k, i_0)$.
    Then, because $\theta$ increases in multiples of $1+\eps^2$, we have that $\alpha^\star_k \leq (1+\eps^2)(1-\delta)d(k, i_0)$.
    On the other hand,
    \[\alpha^\star_k \geq \rho \cdot d(k, i^\star) \geq \rho \cdot (d(k, i_0) - d(i^\star, i_0)) \geq \rho \cdot d(k, i_0) - \rho \delta \theta \geq \rho \Paren{1 - \delta(1-\delta)} d(k, i_0)\,.\]
    Putting the two together, we get that
    \[\rho \Paren{1 - \delta(1-\delta)} d(k, i_0) \leq (1+\eps^2)(1-\delta)d(k, i_0)\,,\]
    which is a contradiction for $\rho \geq 1.01$ for $\delta \gg \eps > 0$ small enough.
\end{proof}

\subsection{Algorithm Cost}

We now bound the cost of the algorithm in terms of the variables $r$.

\begin{lemma}
\label{lemma:alg-cost}
    The cost (with respect to opening cost $\hat{f}$) of a solution output by our algorithm is at most
    \[\sum_{j=1}^n \Paren{\sum_{i=j}^{n} \beta(r_{j,i}, r_{j,i+1}) + r_{j,n+1}} + \delta \cdot \sum_{j=1}^n d(j,\SS^\star)\,.\]
\end{lemma}
\begin{proof}
    The $j$-th client has connection cost $d(j, \SS^\star)$.
    We also need to calculate the contribution of $j$ to the opening of each facility in $\SS^\star$.
    It costs $\hat{f}$ to open a facility, and we know that at the atomic step when it was opened we had
    \[
\sum_{j\in A}\beta(\theta, (1-\delta)\dist(i,j))
\;+\;
\sum_{j\in I}\beta((1-\delta)\dist(j,\SS), (1-\delta)\dist(i,j))
\;\ge\; \hat f
\]
    if it was opened in Stage 1 (by Definition~\ref{def:openable-mod-1})
    or 
  \[
      \sum_{j\in A}\beta(\tau_j, (1- \delta/2)\dist(i,j))
      \;+\;
      \sum_{j\in I}\beta((1-\delta)\dist(j,\SS), (1- \delta/2)\dist(i,j))
      \;\ge\; \hat f.
  \]
  if it was opened in Stage 2 (by Definition~\ref{def:openable-mod-2}).
  Therefore, it will suffice to bound these left-hand sides.
    
    Also let $\SS_0 \subsetneq \SS_1 \subsetneq \ldots \subsetneq \SS_m = \SS^\star$ with $m = |\SS^\star|$ be the sets of opened facilities during the execution of the algorithm.
    Let $s(j)$ be the index $k$ such that $\SS_k = \SS^{(j)}$ and $s(n+1) = m$.
    For $1 \le j \leq n$ and $0 \leq k \le m - 1$ such that $s(j) \leq k$, we set 
    \[
    R_{jk} =
    \begin{cases}
        \alpha^\star_j &
        \begin{array}[t]{@{}l@{}}
            \text{if $j$ is active at the beginning of the atomic step}\\
            \text{in which facility $\SS_{k+1}\setminus \SS_k$ is opened}\,,
        \end{array}
        \\[2mm]
        (1-\delta)d(j,\SS_k) & \text{otherwise}\,.
    \end{cases}
    \]
    We also define $R_{j,m}= (1-\delta)d(j, \SS^\star)$ for all $j$.
    We note that $r_{j,k} = R_{j, s(k)}$ for all valid $j,k$: The only potential difficulty is if $k$ becomes inactive in Stage 3, because then it is possible for $j$ to be active w.r.t. $\SS^{(k)}$ at the beginning of the atomic step in which $k$ becomes inactive but for it to be inactive w.r.t. $\SS^{(k)}$ at the beginning of the atomic step in which the next facility is opened.
    However, in this one special case, we have in fact $\alpha^\star_j = (1-\delta)d(j,  \SS^{(k)})$, so the difference in definitions does not matter and we still have $r_{j,k} = R_{j, s(k)}$.

    \paragraph{Stage $1$}
    Suppose that $i$ is opened in Stage $1$ with respect to a set of open facilities $\SS$ and with active and inactive sets $A$ and $I$, respectively.
    Then, if $j \in A$ then the contribution is $\beta(\theta, (1-\delta)d(i,j))$ and if $j \in I$ then the contribution is $\beta((1-\delta)d(j,\SS), (1-\delta)d(i, j))$.
    
    We note that if $i$ is opened before the atomic step in which $j$ becomes inactive, we have that the contribution is zero, because $j$ is active after the opening of $i$ and hence $\theta < (1-\delta)d(j,\SS \cup \{i\})$.

    We now bound the contribution of client $j$ to the opening of the facility $\SS_{k+1} \setminus \SS_{k}$, for $k \geq s(j)$, by $\beta(R_{j,k}, R_{j,k+1})$.
    This follows by the definition of $R$.
    Specifically, if $j \in A$, then $R_{j,s(j)} = \alpha_j^\star = \theta$ because $j$ must become inactive in this atomic step, matching the above.
    If $j \in I$ but $k = s(j)$, which can happen if $j$ becomes inactive in Stage 3, we have $R_{j,s(j)} = \alpha_j^\star = (1-\delta)d(j, \SS)$, again matching the above.
    For $k > s(j)$ we have by definition $R_{j, k} = (1-\delta)d(j, \SS)$, again matching the above.
    Finally, for all $k \geq s(j)$, $R_{j,k+1} = (1-\delta) d(j, \SS_{k+1})$, where $d(j, \SS_{k+1})$ is a lower bound on the distance from $j$ to the facility, so together the above suffice to upper bound the contribution by $\beta(R_{j,k}, R_{j,k+1})$.

    \paragraph{Stage $2$.}
    Suppose that $i$ is opened in Stage $2$ with respect to a set of open facilities $\SS$ and with active and inactive sets $A$ and $I$, respectively.
    If $j \in A$ then the contribution is $\beta(\tau_j, (1-\delta/2)d(j,i))$ and if $j \in I$ then the contribution is $\beta((1-\delta)d(j, \SS), (1-\delta/2) d(j, i))$.

    Suppose $i$ is opened before the atomic step in which $j$ becomes inactive.
    If $\tau_j > \theta$ then $d(j, i) \leq \epsilon \theta$ and $j$ in fact becomes inactive in this step, so we must have $\tau_j = \theta$.
    But then, because $j$ is active after the opening of $i$, $\theta < (1-\delta) d(j, \SS \cup\{i\}) \leq (1-\delta) d(j, i)$, which implies that $\beta(\tau_j, (1-\delta/2)d(j,i)) = 0$.

    We now bound the contribution of client $j$ to the opening of the facility $\SS_{k+1} \setminus \SS_{k}$, for $k \geq s(j)$, by $\beta(R_{j,k}, R_{j,k+1})$.
    This follows by the definition of $R$.
    Specifically, if $j \in A$, then $R_{j,s(j)} = \alpha_j^\star = \tau_j$ because $j$ must become inactive in this atomic step, matching the above.
    If $j \in I$ but $k = s(j)$, which can happen if $j$ becomes inactive in Stage 3, we have $R_{j,s(j)} = \alpha_j^\star = (1-\delta)d(j, \SS)$, again matching the above.
    For $k > s(j)$ we have by definition $R_{j, k} = (1-\delta)d(j, \SS)$, matching the above.
    Finally, for all $k \geq s(j)$, $R_{j,k+1} = (1-\delta) d(j, \SS_{k+1})$, where $(1-\delta)d(j, \SS_{k+1})$ is a lower bound on $(1-\delta/2)$ times the distance from $j$ to the facility, so together the above suffice to upper bound the contribution by $\beta(R_{j,k}, R_{j,k+1})$.
    
    \paragraph{Putting things together.}
    The above gives that the total opening costs are upper bounded by
    \begin{align*}
    \sum_{j=1}^n \sum_{k=s(j)}^{m-1} \beta(R_{j,k}, R_{j, k+1})
    &= \sum_{j=1}^n \Paren{\sum_{i=j}^n \sum_{k=s(i)}^{s(i+1)-1} \beta(R_{j,k}, R_{j, k+1})}\,.
    \end{align*}
    Then, because $r_{j,i} = R_{j,s(i)}$ for all valid $j,i$, and because $R_{j,k} \geq R_{j,k+1}$ for all valid $j,k$ (by the same reason as in Observation~\ref{obs:r-noninc}), we have by the superadditivity of $\beta$ that furthermore
    \[\sum_{j=1}^n \Paren{\sum_{i=j}^n \sum_{k=s(i)}^{s(i+1)-1} \beta(R_{j,k}, R_{j, k+1})} \leq \sum_{j=1}^n \sum_{i=j}^n \beta(r_{j,i}, r_{j,i+1})\,.\]
    Using that $r_{j,n+1} = (1-\delta)d(j, \SS^\star)$, we also have that the connection cost satisfies $d(j, \SS^\star) = r_{j,n+1} + \delta d(j, \SS^\star)$.
    Then the total cost is upper bounded by
    \[\sum_{j=1}^n \Paren{\sum_{i=j}^{n} \beta(r_{j,i}, r_{j,i+1}) + r_{j,n+1}} + \delta \cdot \sum_{j=1}^n d(j,\SS^\star)\,.\]
\end{proof}

\subsection{Factor-Revealing Program Feasibility}
\label{sec:fr-feasiblity}
In this subsection we consider some $i^\star \in \FF$ and a subset of clients $\Dgood \subseteq D$ such that, for all $j \in \Dgood$, $\alpha_j^\star \geq \rho \cdot d(j, i^\star)$.
Let $q(i) = \min \Set{i' \mid i' > i, i' \in \Dgood \cup \{n+1\}}$.
We will assume throughout that $\hat{f} > 0$ and $\rho \geq 1.01 \gamma$.

We will show that $\{r_{jk} \mid j \leq k, j \in \Dgood, k \in \Dgood \cup \{n + 1\}\}, \{d(j, i^\star) \mid j \in \Dgood\}, \hat f$ is feasible for the factor revealing program in Definition~\ref{def:frp}, except for the normalization constraint (\ref{eq:fr-norm}).

First we will prove that the facility constraints (\ref{eq:fr-facility}) hold.
\begin{lemma}
\label{lem:fr-facility}
    Consider the central facility $i^\star$ and let $k \in \Dgood$. Then,
    \[
        \sum_{\substack{j \ge k\\ j \in \Dgood}} \beta(\alpha^\star_k, d(i^\star,j)) + \sum_{\substack{j < k\\ j \in \Dgood}} \beta(r_{jk}, d(i^\star,j)) \le \hat f\,.
    \]
\end{lemma}
\begin{proof}
    Consider the phase when $k$ goes inactive. 
    $k$ goes inactive either in Stage 1, Stage 2 or Stage 3.
    \paragraph{Stage 1}
    Consider the state $(\SS=\SS^{(k)},A,\theta)$ of the algorithm just before the atomic step in which $k$ becomes inactive.
    If $k$ goes inactive in stage 1 then $\alpha^\star_k = \theta$ and by Lemma \ref{lem:overbidding},
    \begin{align*}
        \hat f & \ge \sum_{j \in A} \beta(\theta, d(i^\star,j)) + \sum_{j \in I} \beta((1-\delta)d(j, \SS), d(i^\star,j)) 
        \\ & = \sum_{j \ge k} \beta(\alpha^\star_k, d(i^\star,j)) + \sum_{j < k} \beta(r_{jk}, d(i^\star,j)) 
        \\ & \ge \sum_{\substack{j \ge k\\ j \in \Dgood}} \beta(\alpha^\star_k, d(i^\star,j)) + \sum_{\substack{j < k\\ j \in \Dgood}} \beta(r_{jk}, d(i^\star,j))\,,
    \end{align*}
    where we use that if $j < k$ but $j \in A$ then we have $r_{jk} = \alpha_j^\star = \theta$.
    \paragraph{Stage 2}
    Again, consider the state $(\SS^{(k)},A,\theta)$ of the algorithm just before the atomic step in which $k$ becomes inactive.

    The case $\alpha_k^\star = \theta$ is easy:
    Note that $r_{jk} = \alpha_j^\star = \theta$ for all $j < k$ with $j \in A$, as $\alpha_j^\star \geq \theta$ and $\alpha_j^\star \leq \alpha_k^\star$ (because if $j < k$ and $j \in A$ then $j$ must become inactive in the same atomic step as $k$).
    Then by Lemma \ref{lem:overbidding},
    \begin{align*}
        \hat f & \ge \sum_{j \in A} \beta(\theta, d(i^\star,j)) + \sum_{j \in I} \beta((1-\delta)d(j, \SS), d(i^\star,j)) 
        \\ & = \sum_{j \ge k} \beta(\alpha^\star_k, d(i^\star,j)) + \sum_{j < k} \beta(r_{jk}, d(i^\star,j)) 
        \\ & \ge \sum_{\substack{j \ge k\\ j \in \Dgood}} \beta(\alpha^\star_k, d(i^\star,j)) + \sum_{\substack{j < k\\ j \in \Dgood}} \beta(r_{jk}, d(i^\star,j))\,,
    \end{align*}
    where we use that if $j < k$ but $j \in A$ then we have $r_{jk} = \alpha_j^\star = \theta$.
    
    Else, $\alpha_k^\star > \theta$.
    Assume for contradiction,
    \[
        \sum_{\substack{j \ge k\\ j \in \Dgood}} \beta(\alpha^\star_k, d(i^\star,j)) + \sum_{\substack{j < k\\ j \in \Dgood}} \beta(r_{jk}, d(i^\star,j)) >  \hat f\,.
    \]
    Observe that $\alpha^\star_k \leq (1+\epsilon^2)\theta$. Moreover, for $j < k$ and $j \in A$ we have $r_{jk} = \alpha^\star_j \leq (1+\epsilon^2)\theta$. Then, $i^\star$ is a candidate in Stage 2 since
    \begin{align*}
        & \sum_{j\in A}\beta((1+\epsilon^2)\theta, (1- \delta/2)d(i^\star,j)) + \sum_{j\in I}\beta((1-\delta)\dist(j,\SS^{(k)}), (1-\delta/2)d(i^\star,j)) \\
        \ge & \sum_{j \ge k} \beta(\alpha^\star_k, d(i^\star,j)) + \sum_{j < k} \beta(r_{jk}, d(i^\star,j)) \\
        \ge & \sum_{\substack{j \ge k\\ j \in \Dgood}} \beta(\alpha^\star_k, d(i^\star,j)) + \sum_{\substack{j < k\\ j \in \Dgood}} \beta(r_{jk}, d(i^\star,j)) > \hat f\,.
    \end{align*}
    Here, the first inequality holds by the monotonicity of $\beta$.
    So $k \in \Dgood$ becomes inactive in Stage 2 while $i^\star$ is a candidate. Thus, by Lemma \ref{lem:freezing} we know that all clients in $A \cap \Dgood$ become inactive in the same atomic step as $k$ and $\alpha^\star_j = \tau_j$ for all $j \in A\cap \Dgood$. In particular, by property (iv) in $2$-openable,
    \begin{align*}
        \hat f & \ge \sum_{j \in A} \beta(\tau_j, d(i^\star,j)) + \sum_{j \in I} \beta((1-\delta)d(j, \SS^{(k)}), d(i^\star,j)) \\
        & \ge \sum_{j \in A \cap \Dgood} \beta(\alpha^\star_j, d(i^\star,j)) + \sum_{j \in I \cap \Dgood} \beta((1-\delta)d(j, \SS^{(k)}), d(i^\star,j)) \\
        & \ge \sum_{\substack{j \ge k\\ j \in \Dgood}} \beta(\alpha^\star_k, d(i^\star,j)) + \sum_{\substack{j < k\\ j \in \Dgood}} \beta(r_{jk}, d(i^\star,j))\,.
    \end{align*}
    For the last inequality we used that $\beta$ is monotone together with $\alpha^\star_j \ge \alpha^\star_k$ for all $j \ge k, j \in A$ since both $k$ and $j$ become inactive in the same atomic step and are hence ordered by $\alpha^\star$. Moreover, for $j < k$, we set $r_{jk} = \alpha^\star_j$ if $j \in A$ and $r_{jk} = (1-\delta)d(j,\SS^{(k)})$ if $j \in I$.
    \paragraph{Stage 3}
    In this case $k$ goes inactive when we increase the $\alpha$-values in Stage 3. Let $(\SS,A,\theta)$ be the state of the algorithm at the end of Stage 3. By Lemma \ref{lem:overbidding},
    \[
        \hat f \ge \sum_{j \in A} \beta(\theta, d(i^\star,j)) + \sum_{j \in I} \beta((1-\delta)d(j, \SS), d(i^\star,j))  \ge \sum_{\substack{j \ge k\\ j \in \Dgood}} \beta(\alpha^\star_k,  d(i^\star,j)) + \sum_{\substack{j < k\\ j \in \Dgood}} \beta(r_{jk}, d(i^\star,j))\,.
    \]
    Let us verify that the last inequality holds. First consider some $j \ge k$. For $j \in A$ note that $\theta \ge \alpha^\star_k$. Otherwise, $j \in I$ and $j$ became inactive in Stage 3 as well. In that case $(1-\delta)d(j,\SS) = \alpha^\star_j \ge \alpha^\star_k$ because $j$ and $k$ become inactive in the same atomic step and are hence ordered by $\alpha^\star$. Thus, we can use the monotonicity of $\beta$. Next, if $j < k$ then $j \in I$ since $k \in I$.
    If $j$ became inactive in the same atomic step as $k$, we have $r_{jk} = \alpha_j^\star = (1-\delta)d(j, \SS)$.
    Else, by definition $r_{jk} = (1-\delta)d(j, \SS)$. Here, $\SS = \SS^{(k)}$ since we do not open any new facilities in Stage 3. 
\end{proof}

\begin{lemma}
\label{lem:fr-r-mono}
    For all $1 \le j \leq k \le n$ we have $r_{jk} \ge r_{j k+1}$.
\end{lemma}
\begin{proof}
    We consider different cases depending on when $j$ becomes inactive.
    First, consider the case where $j$ becomes inactive strictly before $k$.
    Then $r_{jk} = (1-\delta)d(j,\SS^{(k)}) \ge (1-\delta)d(j,\SS^{(k+1)}) = r_{jk+1}$ since $d(j, \SS)$ is monotonically decreasing with $\SS$ and $\SS^{(k)} \subseteq \SS^{(k+1)}$.
    Next, consider the case where $j$ becomes inactive together with $k$ but strictly before $k+1$ then $r_{jk} = \alpha^\star_j \ge (1-\delta)d(j,\SS^{(k+1)}) = r_{jk+1}$.
    In the remaining case when $j$, $k$ and $k+1$ become inactive at the same time we have $r_{jk} = r_{jk+1} = \alpha^\star_j$.  
\end{proof}

\begin{lemma}
\label{lem:fr-triangle}
    For all $1 \le j < k \le n$ we have $\alpha^\star_k \le r_{jk} + d_j + d_k$.
\end{lemma}
\begin{proof}
    First, consider the case when client $j$ becomes inactive strictly before client $k$.
    Then $r_{jk} = (1-\delta)d(j,\SS^{(k)}) = (1-\delta)d(j,i_0)$ for some facility $i_0 \in \SS^{(k)}$.
     The client $k$ becomes inactive either in Stage 1 with $\alpha_k^\star = \theta \leq (1-\delta) d(k,\SS^{(k)}) \leq (1-\delta)d(k, i_0)$, in Stage 2 with $\alpha_k^\star = \tau_k \leq (1-\delta)d(k,\SS^{(k)}) \leq (1-\delta) d(k, i_0)$ or in Stage 3 with $\alpha_k^\star = (1-\delta)d(k,\SS) \leq (1-\delta) d(k, i_0)$.
    Then, by triangle inequality,
    \[
        \alpha^\star_k \le (1-\delta)d(k,i_0) \le (1-\delta)(d(k, i^\star) + d(i^\star, j) + d(j, i_0)) \le r_{jk} + d_j + d_k\,.
    \]
    In the remaining case $k$ and $j$ become inactive at the same time and $r_{jk} = \alpha^\star_j$.
    If $k$ and $j$ become inactive in Stage 1 then $\alpha^\star_k = \alpha^\star_j = \theta$.
    If $k$ and $j$ become inactive in Stage 2 then $\alpha^\star_k \le \alpha^\star_j + d(k,j) \le \alpha^\star_j + d_k + d_j$ by condition $(v)$ in Definition~\ref{def:openable-mod-2}.
    Otherwise $k$ and $j$ become inactive in Stage 3 and 
    \[
        \alpha^\star_k = (1-\delta)d(k,\SS^{(k)}) \le (1-\delta)d(k,i_0) \le (1-\delta)(d(k,i^\star) + d(i^\star,j) + d(j,i_0)) \le \alpha^\star_j + d_k + d_j
    \]
    where $i_0$ is the facility that $j$ connects to, $\alpha^\star_j = (1-\delta)d(j,\SS^{(k)}) = (1-\delta)d(j,i_0)$. 
\end{proof}

\begin{lemma}
\label{lem:fr-r-metric}
    For all $1 \le j,\ell < k \le n + 1$ we have $r_{jk} \le r_{\ell k} + d_j + d_\ell$.
\end{lemma}
\begin{proof}
    First, consider the case where $j$ becomes inactive strictly before $k$. Then $r_{jk} = (1-\delta)d(j,\SS^{(k)}) = (1-\delta)d(j,i_1)$ for some open facility $i_1 \in \SS^{(k)}$.
    If $\ell$ becomes inactive strictly before $k$, then $r_{\ell k} = (1-\delta)d(\ell,\SS^{(k)}) = (1-\delta)d(\ell,i_2)$ for some open facility $i_2 \in \SS^{(k)}$, and because $i_1$ minimizes $(1-\delta)d(j,\SS^{(k)})$, 
    \[
        r_{jk} = (1-\delta)d(j,i_1) \le (1-\delta)d(j,i_2) \le (1-\delta)(d(j,i^\star) + d(i^\star,\ell) + d(\ell,i_2)) \le r_{\ell k} + d_j + d_\ell\,.
    \]
    Now if $\ell$ becomes inactive with $k$ then $r_{\ell k} = \alpha^\star_\ell$. 
    Assume $j$ becomes inactive in Stage 2 when opening a facility $i_0$ with bids $(\tau_j)_{j \in A}$. By $(v)$ in Definition~\ref{def:openable-mod-2} we have
    \[
        r_{jk} = (1-\delta)d(j,\SS^{(k)}) \le (1-\delta)d(j,i_0) \le \tau_j \le \tau_\ell + d(j,\ell) \le \alpha^\star_\ell + d_j + d_\ell\,,
    \]
    where $(1-\delta) d(j, i_0) \leq \tau_j$ because if $\tau_j > \theta$ then we must have $d(j, i_0) \leq \epsilon \theta$ and if $\tau_j = \theta$ then we must have $(1-\delta) d(j, i_0) \leq \theta$ because $j$ was active but became inactive after $i_0$ was opened.
    
    Otherwise $j$ becomes inactive in Stage 1 or Stage 3 and
    \[
        r_{jk} = (1-\delta)d(j,\SS^{(k)}) \le \alpha^\star_j \le \alpha^\star_\ell = r_{\ell k}
    \]
    since $\ell$ becomes inactive strictly after $j$: this implies that in Stage 1 $\alpha^\star_j$ is equal to the $\theta$ at that stage and in Stage 3 $\alpha^\star_j$ is equal to at most $(1+\eps^2)\theta$ for the $\theta$ at that stage, but then $\alpha^\star_\ell$ is in both cases at least as large.
    
    Next, consider the case where $j$ becomes inactive with $k$. Then $r_{jk} = \alpha^\star_j$. If $\ell \ge j$ then $j$, $\ell$ and $k$ become inactive in the same atomic step and $r_{\ell k} = \alpha^\star_\ell$. By the ordering of $\alpha^\star$ we know that 
    \[
        r_{jk} = \alpha^\star_j \le \alpha^\star_\ell = r_{\ell k}\,.
    \]
    Thus, we can assume that $\ell < j$. Observe that $r_{\ell k} = r_{\ell j}$ since $k$ and $j$ become inactive in the same atomic step. Notice that now the proof follows from Lemma $\ref{lem:fr-triangle}$ as $r_{jk} = \alpha^\star_{j}$ and $r_{\ell k} = r_{\ell j}$.

    Finally, if $k = n + 1$, we have $r_{jk} = (1-\delta)d(j, \SS^\star)$ and $r_{\ell k} = (1-\delta) d(\ell, \SS^\star)$, and then 
    \[r_{jk} \leq (1-\delta)(d(j, i^\star) + d(\ell, i^\star) + d(\ell, \SS^\star)) \leq r_{\ell k} + d_j + d_\ell\,.\]
\end{proof}

\subsection{Proof of Theorem~\ref{theorem:alg-lmp}}
\label{subsec:lmp-proof}
Before we prove the theorem, we will need the following gluing result that bounds the cost of the algorithm in terms of the upper bound on the factor-revealing program.

\begin{lemma}
\label{lem:dgood-cost}
    Let $\rho \geq 1.01 \gamma$ be an upper bound on the optimal value of the factor-revealing program in Definition~\ref{def:frp}.
    For some $i^\star \in \FF$, consider any subset of clients $\Dgood \subseteq D$ such that, for all $j \in \Dgood$, $\alpha_j^\star \geq \rho \cdot d(j, i^\star)$.
    Then
    \[\sum_{j \in \Dgood} \Paren{\sum_{i=j}^n \beta(r_{j,i}, r_{j,i+1}) + r_{j,n+1}} \leq \hat{f} + \rho \cdot \sum_{j \in \Dgood} d(j, i^\star)\,.\]
\end{lemma}
\begin{proof}
If $\Dgood = \emptyset$, the claim is trivial, so assume $\Dgood \neq \emptyset$.
We know that the variables $\{r_{jk} \mid j \leq k, j \in \Dgood, k \in \Dgood \cup \{n + 1\}\}$, $\{d(j, i^\star) \mid j \in \Dgood\}$, $\hat f$ (see Section~\ref{sec:fr-feasiblity}), with the clients reindexed to $1, \ldots, |\Dgood|$  in the same order as the original indices and $n+1$ reindexed to $|\Dgood|+1$, satisfy the constraints of the factor-revealing program in Definition~\ref{def:frp} (note that we identify $r_{jj} = \alpha^\star_j$), except for the normalization constraint (\ref{eq:fr-norm}): this follows by Lemma~\ref{lem:fr-facility}, Lemma~\ref{lem:fr-r-mono}, Lemma~\ref{lem:fr-triangle}, Lemma~\ref{lem:fr-r-metric}, and because all these variables are non-negative.

Let $\dsum = \sum_{j \in \Dgood} d(j, i^\star)$, which we know is strictly positive because $d(j, i) \geq 1$ for all $j \in D$, $i \in \FF$ by assumption \eqref{eq:dist-bound}.
Then, because the factor-revealing program is scale-invariant (except for the normalization constraint), it follows that $\{r_{jk} / \dsum \mid j \leq k, j \in \Dgood, k \in \Dgood \cup \{n + 1\}\}$, $\{d(j, i^\star) / \dsum \mid j \in \Dgood\}$, $\hat f / \dsum$ satisfy all the constraints of the factor-revealing program in Definition~\ref{def:frp}.

Let $q(i) = \min \Set{i' \mid i' > i, i' \in \Dgood \cup \{n+1\}}$.
Then the upper bound $\rho$ on the value of the factor revealing program implies that
\[\sum_{j \in \Dgood} \Paren{\sum_{\substack{i \in \Dgood\\i \geq j}} \beta\Paren{\frac{r_{j,i}}{\dsum}, \frac{r_{j,q(i)}}{\dsum}} + \frac{r_{j,n+1}}{\dsum} } \leq \frac{\hat{f}}{\dsum} + \rho\,,\]
so multiplying by $\dsum$, 
\[\sum_{j \in \Dgood} \Paren{\sum_{\substack{i \in \Dgood\\i \geq j}} \beta\Paren{r_{j,i}, r_{j,q(i)}} + r_{j,n+1}} \leq \hat{f} + \rho \sum_{j \in \Dgood} d(j, i^\star)\,.\]
Finally, because $r_{j,i} \geq r_{j,i+1}$ for all valid $j,i$, by the superadditivity of $\beta$ the above also implies that
\[\sum_{j \in \Dgood} \Paren{\sum_{i=j}^n \beta(r_{j,i}, r_{j,i+1}) + r_{j,n+1}} \leq \hat{f} + \rho \cdot \sum_{j \in \Dgood} d(j, i^\star)\,.\]
\end{proof}

We are ready to prove the main theorem of this section.

\begin{proof}[Proof of Theorem~\ref{theorem:alg-lmp}]
    Consider an optimal solution of cost $\OPT$ to the UFL problem instance with facility cost $f = \hat{f}/\rho$, and let $k^\star$ denote the number of facilities opened in this solution.
    Let $\{i_1^\star, \ldots, i_{k^\star}^\star\}$ be the facilities opened by the optimal solution, and let $\{D_1^\star, \ldots, D_{k^\star}^\star\}$ be such that $D_\ell^\star$ is the set of clients connected to facility $i_\ell^\star$ in the optimal solution.
    
    Let $C_{\mathrm{SOL}}$ and $F_{\mathrm{SOL}}$ be the connection cost and opening cost (with respect to $f$) of our algorithm, and let $C_{\OPT}$ and $F_{\OPT}$ be defined analogously for $\OPT$.
    By Lemma~\ref{lemma:alg-cost}, which bounds the cost with respect to $\hat{f}$, and using that $f = \frac{\hat{f}}{\rho}$, we know that 
    \[C_{\mathrm{SOL}} + \rho \cdot F_{\mathrm{SOL}} \leq \sum_{j=1}^n \Paren{\sum_{i=j}^{n} \beta(r_{j,i}, r_{j,i+1}) + r_{j,n+1}} + \delta \cdot \sum_{j=1}^n d(j,\SS^\star)\,.\]
    Note that $\sum_{j=1}^n d(j, \SS^\star) = C_{\mathrm{SOL}}$, so we can rewrite
    \[(1-\delta) \cdot C_{\mathrm{SOL}} + \rho \cdot F_{\mathrm{SOL}} \leq \sum_{j=1}^n \Paren{\sum_{i=j}^{n} \beta(r_{j,i}, r_{j,i+1}) + r_{j,n+1}}\,.\]

    Fix now some $\ell \in \{1, \ldots, k^\star\}$ and set $i^\star = i^\star_\ell$ and $D^\star = D^\star_{\ell}$.
    Define the subset of $D^\star$
    \[\Dgood = \Set{j \in D^\star \Bigg| \alpha_j^\star \geq \rho \cdot d(j, i^\star)}\,.\]
    For the clients in $\Dgood$, we have by Lemma~\ref{lem:dgood-cost} that
    \[\sum_{j \in \Dgood} \Paren{\sum_{i=j}^n \beta(r_{j,i}, r_{j,i+1}) + r_{j,n+1}} \leq \hat{f} + \rho \cdot \sum_{j \in \Dgood} d(j, i^\star)\,.\]
    Furthermore, for the clients $j \in D^\star \setminus \Dgood$, because $r_{j,j} = \alpha_j^\star$, $r_{j,n+1} = \beta(r_{j,n+1}, 0)$, and $r_{j,i} \geq r_{j,i+1}$ for all valid $j,i$, we have by the superadditivity of $\beta$ that
    \[\sum_{j \in D^\star \setminus \Dgood} \Paren{\sum_{i=j}^n \beta(r_{j,i}, r_{j,i+1}) + r_{j,n+1}} \leq \sum_{j \in D^\star \setminus \Dgood} \alpha^\star_j \leq \rho \cdot \sum_{j \in D^\star \setminus \Dgood} d(j, i^\star)\,.\]
    Overall, then,
    \[\sum_{j \in D^\star} \Paren{\sum_{i=j}^n \beta(r_{j,i}, r_{j,i+1}) + r_{j,n+1}} \leq \hat{f} + \rho \cdot \sum_{j \in D^\star} d(j, i^\star)\,.\]
    
    Note that $D_1^\star, \ldots, D^\star_{k^\star}$ partition $[n]$. Then, summing the above over all $(i_{\ell}^\star, D^\star_{\ell})$,
    \begin{align*}
    \sum_{j=1}^n \Paren{\sum_{i=j}^{n} \beta(r_{j,i}, r_{j,i+1}) + r_{j,n+1}}
    &\leq k^\star \cdot \hat{f} + \rho \cdot C_{\OPT}\\
    &= \rho \cdot F_{\OPT} + \rho \cdot C_{\OPT}
    = \rho \cdot \OPT\,.
    \end{align*}
    Hence, 
    \[(1-\delta) \cdot C_{\mathrm{SOL}} + \rho \cdot F_{\mathrm{SOL}} \leq \rho \cdot \OPT\,,\]
    so 
    \[C_{\mathrm{SOL}} \leq \frac{\rho}{1-\delta} \Paren{\OPT - F_{\mathrm{SOL}}} \leq (1+2\delta) \rho \Paren{\OPT - F_{\mathrm{SOL}}}\,.\]
    Then the conclusion follows because $F_{\mathrm{SOL}} = |S^\star| \cdot f$.
\end{proof}

\subsection{$k$-Median Algorithm Opening $k + O(\log n)$ Centers}
\label{subsec:k-median}

By the techniques developed in~\cite{MR4928455-Cohen25}, we can expand the ideas in this section to obtain a polynomial-time approximation algorithm for $k$-median that opens at most $k + O(\log n)$ centers.
As in their paper, this requires us to develop a robust version of the log-adaptive algorithm, and to slightly adapt their merging procedure to our setting.
We state below our main result for $k$-median.

\begin{theorem}
\label{thm:k-median}
Let $\gamma \geq 1$ and let $\beta(\alpha,d) = [\alpha - \gamma \cdot d]^+$. Let $\rho \geq 1.01 \gamma$ be an upper bound on the optimal value of the factor-revealing program in Definition~\ref{def:frp} over all $n \geq 1$.
For any $\delta > 0$ small enough, there exists a polynomial-time approximation algorithm for the $k$-median problem that outputs a solution with at most $k + O(\log n / \delta^2)$ centers and cost at most $(1+\delta)\rho \cdot \sopt$.
\end{theorem}

As these steps are very similar to the steps in~\cite{MR4928455-Cohen25}, we defer the description and analysis of the algorithm in Theorem~\ref{thm:k-median} to Section~\ref{sec:merging}.

\section{Upper Bound on the Factor-Revealing Program}\label{sec:195-dual-fitting}

We prove Theorem~\ref{thm:1.95-frp} by constructing, for every fixed bidding matrix $b$, a feasible solution of value at most $2-(\gamma-1)^2$ to the corresponding inner dual program.
By Lemma~\ref{lemma:mathematical-program-1.99}, these certificates give the required upper bound on the factor-revealing program.

We interpret the dual as a credit-allocation process involving clients and a common pool. The variables $g_{i,j}$ generate credit, while $s_{i,j}$ and $t_{i,j}$ describe transfers to and from the pool. The client constraints require sufficient credit at each iteration, and the pool constraints require that cumulative withdrawals never exceed cumulative deposits. The objective bounds the charge associated with each client. This interpretation guides the assignment rules below.

The construction begins with a nonnegative profile $\hat g$ that bounds the entries of each row of $g$. We seek a profile satisfying
\begin{equation*}
    R_0 := \sum_{i=1}^n \hat g_i = 1.
\end{equation*}
Its recursive definition depends on boundary data supplied as advice $\Gamma$. Our proof separates two tasks: showing that successful advice produces a feasible certificate of the desired value, and showing that successful advice exists. For the second task, fractional boundary memberships allow us to interpolate between successive choices. We prove that the resulting value of $R_0$ varies continuously along a path whose endpoint values lie on opposite sides of $1$.

Section~\ref{sec:proof-thm-195} states the two main lemmas and derives the theorem from them. Section~\ref{sec:general-case-procedure} defines the construction, and Section~\ref{sec:general-case-proof} proves feasibility whenever it succeeds.
Finally, Section~\ref{sec:good-aux-existence} establishes the existence of successful advice.

\subsection{Proving a Tight Upper Bound}\label{sec:proof-thm-195}

The goal is to design feasible solutions of value $\Mrho$ for every inner linear program in the factor-revealing dual as defined in Definition~\ref{def:frdual}, since the upper-bound relation between the factor revealing programs explained in Lemma~\ref{lemma:mathematical-program-1.99} makes it sufficient. Let us restate the program here for convenience.
\deffrdual*

\smallskip First, we need to observe that we do not need to treat all the factor-revealing cases, such is the $n=1$ case which is trivially upper bounded.
\begin{observation}
\label{obs:n-not-zero}
For $n=1$ the value of the program in Lemma~\ref{lemma:mathematical-program-1.99} is upper bounded by $\Mcoeff$.
For the rest of the proof we assume that $n \geq 2$.
\end{observation}
Next, to simplify our arguments, there are some strategies that are dominant over others, so we can assume w.l.o.g. the following observation.
\begin{observation}
\label{obs:b1}
We assume without loss of generality that $b_{n,j} = 1$ for all $j = 1, \ldots, n$ because it makes the constraints in Lemma~\ref{lemma:mathematical-program-1.99} stronger.
\end{observation}
Then, constrained to these cases, the proof follows equally as it was explained in Section~\ref{sec:approx_guarantee}. 
In Section~\ref{sec:general-case-procedure}, we define a procedure that adaptively computes a feasible solution based on some input. More specifically, the procedure takes as input the values of $b_{i,j}$ and some advice variables $\Gamma$ and returns, for some choices of $\Gamma$, an empty output ($\perp$), and for others, a feasible solution of value $\Mrho$ to the dual program in Definition~\ref{def:frdual} for the $b_{i,j}$ values.
We say that the advice is ``correct'' if the first event does not happen. This is seen in the following lemma:

\begin{restatable}{lemma}{generalcaseprocedure}\label{lem:general-case-procedure}
    For $n \geq 2$, an assignment of the $b_{i,j}\in\{0,1\}$ values with $b_{n,j} = 1$ for all $j = 1, \ldots, n$, a correct advice $\Gamma$ and a value of $\Mcoeff\in(1,1.224]$, the procedure described in Section~\ref{sec:general-case-procedure} generates a $\Mrho$ feasible solution for the corresponding inner program of the factor-revealing dual from Definition~\ref{def:frdual}.
\end{restatable}

To be able to use this previous lemma, we will also prove that there always exists at least one correct set of advice variables:

\begin{restatable}{lemma}{goodtauexistence}\label{lem:good-tau-existence}
    For $n \geq 2$, any assignment of the $b_{i,j}\in\{0,1\}$ with $b_{n,j} = 1$ for all $j = 1, \ldots, n$ and a value of $\gamma \in (1,1.224]$, there always exists a correct advice $\Gamma$ with respect to the procedure in Section~\ref{sec:general-case-procedure}.
\end{restatable}

Thus, if we combine together these two lemmas we can conclude that there always exists a solution.
Now the theorem's proof is direct.

\thmupperboundfrp*

\begin{proof}[Proof of Theorem~\ref{thm:1.95-frp}]
    Using Observation~\ref{obs:n-not-zero}, for $n=1$ we have a bound of $\Mrho$ since $\gamma\leq \Mrho$ for our range. From Observation~\ref{obs:b1}, we see that we can focus on only the cases with $b_{n,j}=1$ for all $j$. Then, together with Lemma~\ref{lem:general-case-procedure} and Lemma~\ref{lem:good-tau-existence}, we get directly that since there exists a $\Mrho$ feasible solution for each of the inner linear programs, it implies the whole program is upper-bounded by $\Mrho$. Thus, combining this with Lemma~\ref{lemma:mathematical-program-1.99} proves the theorem for $\gamma\neq 1$. The proof of the end-point $\gamma=1$ reduces to the analysis of the original greedy, whose LMP ratio is bounded by 2 \cite{JainMMSV03}.
\end{proof}

What is left is to prove the two final lemmas which we do in the following sections, Lemma~\ref{lem:general-case-procedure} in Section~\ref{sec:general-case-proof} and Lemma~\ref{lem:good-tau-existence} in Section~\ref{sec:good-aux-existence}. The definition of the procedure follows in the next section.

\subsection{Constructing a dual certificate from advice}\label{sec:general-case-procedure}

For the definition of the general case's procedure, we need to introduce some notation and definitions. We assume that we are given some input data, which consists of the lower triangular bidding matrix $\Set{b_{i,j} \mid i \geq j \in \{1, 2, ..., n\}}$, and some advice variables $\Gamma$. What follows is the description of the procedure.

First, let us define the \emph{first nonzero bidding time}, which is the only information out of the values of $b_{i,j}$ to make our algorithm adaptive.
\begin{definition}
    For each client $j\in\{1,\dots,n\}$, define its \emph{first nonzero bidding time}
    \begin{equation*}
        f_j := \min\{i\in\{j,\dots,n\} : b_{i,j}=1\}\,.
    \end{equation*}
\end{definition}
This quantity is well defined by Observation~\ref{obs:b1}. Thus, the algorithm will not change with respect to how $b_{i,j}$ column looks after the row $f_j$. Using this definition, we can take the following simplifying bounds on the right hand side of the $(r_{j,i})$ constraint:
\begin{observation}\label{obs:rhs-rij-upperbound}
    The right hand side of the $(r_{j,i})$ constraint is upper bounded by:
    \begin{equation*}
        b_{i,j} - (\Mcoeff-1) \sum_{p=j}^{i-1} b_{p,j} \leq \begin{cases}
            0  & \text{for }i<f_j\\
            1 & \text{for }i=f_j\\
            2-\gamma &\text{for }i>f_j
        \end{cases}
    \end{equation*}
\end{observation}
and it will be using this bound that we will prove our results. Then, to ease the notation, we need to define the set of clients $j$ with the same value of $f_j$:
\begin{definition}
    For each $i\in\{1,\dots,n\}$, write
    \begin{equation*}
        \ones{i} := \{j : f_j=i\}\,.
    \end{equation*}
    Thus $\ones{i}$ is the set of clients whose first non-zero bid appears at time $i$, and by lower triangularity we have $\ones{i}\subseteq\{1,\dots,i\}$.
\end{definition}

Our procedure first creates a profile for $\hg_i$ depending on these values of $f_j$. It works by first finding a specific set of clients (which we call $D$ later on) that satisfy $\hg_i=0$ for all $i$ in these sets, except a degenerate case for $j=n$ (and $t_3^\star=n$). To define this set, we predefine a threshold $t_3^\star$ so that all $j$ that has $f_j>t_3^\star$ is in this set, including some other $j$ with $f_j=t_3^\star$. This variable and the ones that define the cases with $f_j=t_3^\star$ is what we call the advice variables and we define them as the following:

\begin{definition}[Advice Variables]\label{def:auxiliary-input}
    For a fixed bidding matrix $\{b_{i,j}\mid i\ge j\}$, an \emph{advice} for the procedure consists of
    \begin{equation*}
        t_3^\star \in \{1,\dots,n\}\,, \qquad
        \{\mu_j^\star\}_{j\in\ones{t_3^\star}}\,,
    \end{equation*}
    and, in the special case $t_3^\star=n$, an additional scalar
    \begin{equation*}
        \hg_n^\star \in [0,1]\,.
    \end{equation*}
\end{definition}

The $\{\mu_j^\star\}_{j\in\ones{t_3^\star}}$ are the (fractional) variables that denote the (fractional) membership to this group $D$. They must satisfy the following properties, necessary for the construction to work and which will be used in the proof of correctness in Section~\ref{sec:general-case-proof}.

\begin{property}[Admissible frontier input]\label{property:admissible-frontier-input}
    Fix $t_3^\star\in\{1,\dots,n\}$ and write $\ones{t_3^\star}=\{j_1<\dots<j_k\}$. We call a family $\{\mu_j^\star\}_{j\in\ones{t_3^\star}}$ \emph{admissible} if it is non-increasing in $j$, has at most one fractional entry, and, when $t_3^\star=n$, satisfies $\mu_j^\star=1$ for all $j\in\ones{n}$.

    Equivalently, there exist $r\in\{0,\dots,k\}$ and $\lambda\in[0,1]$ such that
    \begin{equation*}
        \mu_{j_p}^\star=
        \begin{cases}
            1, & p\le r\,,\\
            \lambda, & p=r+1\le k\,,\\
            0, & p>r+1\,.
        \end{cases}
    \end{equation*}
    where $r=k$ in the case of $t_3^\star=n$. 
\end{property}

To give some intuition on the procedure, given the value of the advice variables, we will iteratively construct the values of the profile $\hg_i$, from which we will get some thresholds $t_1$ and $t_2$, apart from the one we get as input $t_3^\star$. We separate and treat the clients $j$ differently according to their position with respect to two thresholds $t_1$ and $t_2$, generated by the procedure itself, and the input one $t_3^\star$. Putting it very simply, each client $j$ will belong to one of the following four groups:
\begin{itemize}
    \item Group $A$: clients with $j < t_1$ and $f_j< t_3^\star$.
    \item Group $B$: clients with $t_1 < j \le t_2$ and $f_j< t_3^\star$.
    \item Group $C$: clients with $t_2 < j \le t_3^\star$ and $f_j< t_3^\star$.
    \item Group $D$: clients with $f_j > t_3^\star$.
\end{itemize}
Notice that we did not specify what happens to clients with $f_j = t_3^\star$ or to the $j=t_1$ client; their behaviors relate to the values of $\mu_j^\star$, where they might belong to more than one group at the same time. To generalize this concept of groups and work with different memberships at the same time, we will use the following notation.
\begin{definition}[Membership weights]
    For each client $j\in\{1,\dots,n\}$, the construction produces four auxiliary scalars
    \begin{equation*}
        \mu_j^A,\mu_j^B,\mu_j^C,\mu_j^D\in[0,1]\,,
    \end{equation*}
    called the \emph{membership weights} of client $j$ in the classes $A,B,C,D$. For every client $j$, they satisfy
    \begin{equation*}
        \mu_j^A+\mu_j^B+\mu_j^C+\mu_j^D=1\,.
    \end{equation*}
\end{definition}
Here, we choose $\mu_j^D=\mu_j^\star$ for $f_j=t_3^\star$ so that it is part of the input. The rest of the mass is assigned to $A$ (if $j<t_1$), to $B$ (if $t_1<j\le t_2$), and/or to $C$ (if $t_2<j\le t_3^\star$). In particular, if $t_3^\star=n$, we will have $\mu_j^D=1$ for all clients with $f_j=t_3^\star$, meaning that all those clients belong to group $D$.

Through these inputs, the procedure described in Algorithm~\ref{alg:general-case-procedure-hgi} in the end generates a feasible solution for the dual of the factor-revealing LP using the function $\beta_{\Mcoeff}(\alpha,d) = [\alpha - \Mcoeff d]^+$ if the advice is correct. It first generates the $\hg_i$ values in Algorithm~\ref{alg:general-case-procedure-hgi} and then the $g_{i,j}, \bar s_{i,j}, t_{i,j}$ values follow from their definition in Table~\ref{tab:general-case-procedure}. The matrix $\bar s$ in the table is a preliminary assignment that allow for $\bar{s}_{i,i}>0$ to simplify the argument, whose diagonal mass is shifted one step forward at the end of the proof of Lemma~\ref{lem:general-case-procedure} to obtain $s$. If the advice is not correct, the procedure can fail and return an empty output ($\perp$).

There will be specific constants that are necessary. We will use the following definitions:
\begin{definition}
    We define the following shorthand variables which will help in our analysis, all a function of the value of $\Mcoeff$:
    \begin{align*}
        \Mcoeffmone \coloneqq \Mcoeff - 1\,,
        \qquad
        \MRttwo \coloneqq \Mcoeffmone + \Mcoeffmone^2\,,
        \qquad
        \Mtone \coloneqq 2\cdot \Mcoeffmone + \Mcoeffmone^2\,,
        \qquad
        \MRtone \coloneqq 1-\Mtone
    \end{align*}
\end{definition}
Here, $\Mcoeffmone$ is just the extra slope in the bidding function compared the original $\gamma=1$, and it will be the extra slack that one can use after $f_j$ (see Observation~\ref{obs:rhs-rij-upperbound} where $2-\gamma = 1-\Mcoeffmone$). The constant $\MRttwo$ is important because it is exactly the constant that defines the threshold $t_2$, with $R_{t_2}\leq \MRttwo \leq R_{t_2-1}$, where $R_i$ will be defined later and is simply just the suffix sum of the values of $\hg_i$. Finally, the constant $\MRtone$ is the equivalent one but for $t_1$, and $\Mtone$ is just the complementary to it.

From these constants and the choice of $\gamma$, the following technical properties follow, which will be necessary for the proof, but which need not much explanation.
\begin{property}\label{property:coefficients}
    For $\Mcoeff\in(1,1.224]$, the following properties hold:
    \begin{equation*}
        \Mcoeff \Mcoeffmone = \MRttwo, \qquad \Mcoeff + 1 - \MRttwo = \Mrho, \qquad \Mfrac \leq 1\,.
    \end{equation*}
\end{property}
\begin{property}\label{property:thresholds}
    For $\Mcoeff \in (1,1.224]$, we have $\Mtone \in [0,0.5)$ and thus $\Mcoeffmone \leq \MRttwo \leq \Mtone < 0.5 < \MRtone$.
\end{property}
\begin{property}\label{property:function-tau}
    For $\Mcoeff \in (1,1.224]$, we have that for $x\in[0,1-\Mtone]$, the following holds:
    \begin{equation*}    
        \frac{1-\Mtone-x}{\Mtone + x} \geq \left[1-\frac{x}{\MRttwo}\right]^+\,.
    \end{equation*}
\end{property}
Do note that this last technical property is the bottleneck of our technical analysis and thus it is the bottleneck for the choice of $\Mcoeff$ and the approximation factor $\Mrho$ that we can achieve through our technique.

We will also use the suffix sums to express:
\begin{equation*}
    R_i \coloneqq \sum_{\ell=i+1}^n \hg_\ell, \qquad\text{and}\qquad R_i^{(2)} \coloneqq \sum_{\ell=i+1}^n \hg_\ell^{(2)}\,.
\end{equation*}
In particular, $R_n=R_n^{(2)}=0$, and $R_{i-1}=R_i+\hg_i$, $R_{i-1}^{(2)}=R_i^{(2)}+\hg_i^{(2)}$.

Given the advice, we first construct an auxiliary profile $\hat g^{(2)}$ backwards, which can be seen in Algorithm~\ref{alg:general-case-procedure-hgi}. Its suffix sums determine the thresholds $t_1$ and $t_2$. We then replace the early part of this profile, using a fractional split at $t_1$ to make the transition continuous. Once the final profile is fixed, we assign the group memberships and the remaining dual variables.

\begin{algorithm}[H]
\caption{Construction of the profile $\hg$}\label{alg:general-case-procedure-hgi}
\begin{algorithmic}[1]
    \Require $\{b_{i,j} : i\ge j\in\{1,\dots,n\}\}$, an index $t_3^\star\in\{1,\dots,n\}$, an admissible frontier input $\{\mu_j^\star : j\in\ones{t_3^\star}\}$, and $\hg_n^\star\in[0,1]$ when $t_3^\star=n$

    \State Set
    $\mu_j^D=
    \begin{cases}
        1, & f_j>t_3^\star\,,\\
        \mu_j^\star, & j\in\ones{t_3^\star}\,,\\
        0, & f_j<t_3^\star\,.
    \end{cases}$

    \State Set $R_{t_3^\star}^{(2)}\gets 0$ and $\hg_i^{(2)}\gets R_i^{(2)}\gets 0$ for every $i>t_3^\star$

    \For{each $t_3^\star \ge i \ge 1$} \label{line:phase-2-loop}
        \State $\hg_i^{(2)}\gets
        \begin{cases}
            \hg_n^\star, & \text{if $i=t_3^\star=n$}\,,\\
            (1-\mu_i^D)\dfrac{R_i^{(2)}+\Mtone}{n-i}\,, & \text{otherwise}
        \end{cases}$
        \State $R_{i-1}^{(2)}\gets R_i^{(2)}+\hg_i^{(2)}$.
    \EndFor

    \State $t_1\gets \max\{i : R_{i-1}^{(2)}\ge \MRtone\}$ \label{line:R_i-1-smaller-0.79-condition}
    \State Set $(\hg_i,R_i)\gets (\hg_i^{(2)},R_i^{(2)})$ for every $i>t_1$ (or all $i$ if undefined) and $R_{t_1}\gets R_{t_1}^{(2)}$ (if defined).
    \If{$t_1$ does not exist} \label{line:t_1-perp-condition}
        \State \Return $\perp$
    \EndIf

    \State $t_2\gets \max\{i : R_{i-1}^{(2)}\ge \MRttwo\}$
    \State $z\gets \sum_{\ell=t_2+1}^{n}\left[1-\dfrac{R_{\ell-1}^{(2)}}{\MRttwo}\right]^+\mu_\ell^D$

    \State Define $\alpha$, $\hg_{t_1}^A$, $\hg_{t_1}^B$, $\hg_{t_1}$, and $R_{t_1-1}$ according to Definition~\ref{def:t1-boundary-split}

    \For{each $t_1 > i \geq 1$}
        \State $\hg_i \gets (1-\mu_i^D)\cdot\dfrac{1-[\MRttwo-R_{f_i}]^+}{n-i-\Mfrac\cdot z}$
        \State $R_{i-1} \gets R_i + \hg_i$
    \EndFor

    \If{$R_0 \neq 1$}
        \State \Return $\perp$
    \EndIf
\end{algorithmic}
\end{algorithm}

\begin{definition}[Boundary Split at $t_1$]\label{def:t1-boundary-split}
Assume that $t_1$ has been defined by Algorithm~\ref{alg:general-case-procedure-hgi}.
If $t_1 < n$, define
\begin{equation*}
    \alpha \coloneqq \frac{\MRtone - R_{t_1}^{(2)}}{\hg_{t_1}^{(2)}}\,,
    \qquad
    \hg_{t_1}^B \coloneqq \MRtone - R_{t_1}^{(2)}\,,
    \qquad
    \hg_{t_1}^A \coloneqq
    (1-\alpha)(1-\mu_{t_1}^D)\cdot
    \frac{1-[\MRttwo-R_{f_{t_1}}]^+}
         {n-t_1+\alpha(1-\mu_{t_1}^D)-\Mfrac\cdot z}\,.
\end{equation*}
If $t_1=n$, define instead
\begin{equation*}
    \alpha \coloneqq 1\,,
    \qquad
    \hg_n^A \coloneqq 0\,,
    \qquad
    \hg_n^B \coloneqq \hg_n^{(2)}=\hg_n^\star\,.
\end{equation*}
In both cases, we finally set
\begin{equation*}
    \hg_{t_1} \coloneqq \hg_{t_1}^A+\hg_{t_1}^B\,,
    \qquad
    R_{t_1-1} \coloneqq R_{t_1}+\hg_{t_1}\,.
\end{equation*}
\end{definition}
It is important to note the following about $\alpha$:
\begin{property}
We have $0 \leq \alpha \leq 1$.
\end{property}

With these values, we set the memberships as follows:
\begin{equation}\label{eq:membership-assignment}
    \begin{aligned}
        &\mu_j^A=
        \begin{cases}
            1-\mu_j^D, & j<t_1\,,\\
            (1-\alpha)(1-\mu_{t_1}^D), & j=t_1\,,\\
            0, & j>t_1\,,
        \end{cases}
        \qquad
        &&\mu_j^B=
        \begin{cases}
            \alpha(1-\mu_{t_1}^D), & j=t_1\,,\\
            1-\mu_j^D, & t_1<j\le t_2\,,\\
            0, & \text{otherwise}\,,
        \end{cases}
        \\
        &\mu_j^C=
        \begin{cases}
            1-\mu_j^D, & t_2<j\le t_3^\star\,,\\
            0, & \text{otherwise}\,,
        \end{cases} 
        \qquad
        &&\mu_j^D=
        \begin{cases}
            1, & f_j>t_3^\star,\\
            \mu_j^\star, & j\in\ones{t_3^\star}\,,\\
            0, & f_j<t_3^\star\,.
        \end{cases}
    \end{aligned}
\end{equation}
In particular, the only non-integral memberships are the boundary input value on $\ones{t_3^\star}$ and the split of the client $t_1$ between $A$ and $B$.

Once $\hat g$ and $R$ are fixed, we abbreviate
\begin{equation*}
    \chi_j := \left[1-\frac{R_{j-1}}{\theta}\right]^+\,.
\end{equation*}

The matrices $g_{i,j}$, $\bar s_{i,j}$, and $t_{i,j}$ are then assigned column-wise according to Table~\ref{tab:general-case-procedure}. Every entry not listed there is set to $0$. We call $\bar s$ the \emph{preliminary sent-credit matrix}; it is convenient for the credit calculations, but it may have nonzero diagonal entries. The final sent-credit matrix is obtained from it in Lemma~\ref{lem:diagonal-repair}. The variables $\delta_{i,j}$ and $\sigma_{i,j}$ are auxiliary variables that need to satisfy Property~\ref{property:delta} and Property~\ref{property:sigma}, respectively, and whose existence will be proved in the proof of Lemma~\ref{lem:general-case-procedure}.

\begin{table}[H]
\centering
\small
\caption{Column-wise assignment of $(g_{i,j},\bar s_{i,j},t_{i,j})$ once $\hat g$ and $R$ are fixed.}
\label{tab:general-case-procedure}
\renewcommand{\arraystretch}{1.25}
\begin{tabularx}{\textwidth}{@{}>{\RaggedRight\arraybackslash}p{0.14\textwidth}>{\RaggedRight\arraybackslash}p{0.27\textwidth}>{\RaggedRight\arraybackslash}p{0.28\textwidth}>{\RaggedRight\arraybackslash}X@{}}
\toprule
Group Membership & Generated credit $g_{i,j}$ & Pool transfer $\bar s_{i,j}$ & $t$-entry and extra update \\
\midrule
$j\in A$ &
$g_{i,j}^A = \hg_i-\delta_{i,j}$ &
$\bar s_{i,j}^A = \hg_i-\sigma_{i,j}$ for $i\ge j$ &
$t_{f_j,j}^A = 0$ \\
\addlinespace[0.4em]
$j\in B$ &
$g_{i,j}^B = \hg_i$ &
\begin{tabular}{@{}l@{}}
    $\bar s_{i,j}^B = \hg_i$ for $i>j$ \\
    $\bar s_{f_j+1,j}^B=\hg_{f_j+1} + \Mcoeffmone$
\end{tabular}
& $t_{f_{j},j}^B= [1-R_{j-1} - \Mtone]^+$
\\
\addlinespace[0.4em]
$j\in C$ &
$g_{i,j}^C = \hg_i$ &
\begin{tabular}{@{}l@{}}
    $\bar s_{i,j}^C = \hg_i$ for $i>j$ \\
    $\bar s_{f_j+1,j}^C=\hg_{f_j+1} + \Mcoeffmone$
\end{tabular}
&$t_{f_{j},j}^C=1-R_{j-1} - \Mtone$
\\
\addlinespace[0.4em]
$j\in D$ &
$g_{i,j}^D = \begin{cases}
    \hg_i - \chi_j\Mfrac\hg_i, & i<t_1,\\
    \hg_i - \chi_j\Mfrac\hg_{t_1}^A, & i=t_1,\\
    \hg_i, & i>t_1
\end{cases}$ &
$\bar s_{i,j}^D = 0$ for $i>j$ &
$t_{f_{j},j}^D=1-R_{j-1}$
\\
\bottomrule
\end{tabularx}
\end{table}

For the cases that the client is fractionally assigned to more than one group, the generated $g_{i,j}$, $\bar s_{i,j}$, and $t_{i,j}$ values are the corresponding convex combinations of the values generated by each group:
\begin{equation}\label{formula:fractional-assignments-convex-combination}
    \begin{aligned}
        g_{i,j} &= \mu_j^A g_{i,j}^A + \mu_j^B g_{i,j}^B + \mu_j^C g_{i,j}^C + \mu_j^D g_{i,j}^D\,,\\
        \bar s_{i,j} &= \mu_j^A \bar s_{i,j}^A + \mu_j^B \bar s_{i,j}^B + \mu_j^C \bar s_{i,j}^C + \mu_j^D \bar s_{i,j}^D\,,\\
        t_{i,j} &= \mu_j^A t_{i,j}^A + \mu_j^B t_{i,j}^B + \mu_j^C t_{i,j}^C + \mu_j^D t_{i,j}^D\,.
    \end{aligned}
\end{equation}

As stated before, there are two variables $\delta_{i,j}$ and $\sigma_{i,j}$ that need to satisfy some properties:

\begin{property}\label{property:delta}
    The values of $\delta_{i,j}$ for $\mu_j^A>0$ need to satisfy the following properties:
    \begin{itemize}
        \item $\forall \ i, j$: $0 \le \delta_{i,j} \le \hg_i$.
        \item For $i\leq f_j$: $\delta_{i,j}=0$.
        \item $\sum_{i=f_j+1}^{t_3^\star} \delta_{i,j} = \Mcoeffmone \cdot \left[1 -\frac{[\MRttwo-R_{f_j}]^+}{\MRttwo}\right]$.
    \end{itemize}
\end{property}

\begin{property}\label{property:sigma}
    The values of $\sigma_{i,j}$ for $\mu_j^A>0$ need to satisfy the following properties:
    \begin{itemize}
        \item $\forall \ i, j$: $0 \le \sigma_{i,j} \le \hg_i$.
        \item For $i < j$ or $i > f_j$: $\sigma_{i,j}=0$.
        \item $\sum_{i=j}^{f_j} \sigma_{i,j} = [\MRttwo-R_{f_j}]^+$.
    \end{itemize}
\end{property}

\subsection{Feasibility with correct advice}\label{sec:general-case-proof}

This section is the proof of Lemma~\ref{lem:general-case-procedure}. 
Fix a bidding matrix $b$ and admissible advice $\Gamma$ for which the procedure succeeds. We prove that its output yields a feasible dual solution of value $\Mrho$. The existence of such advice is established in Section~\ref{sec:good-aux-existence}.

This implies that for this choice of auxiliary variables, Algorithm~\ref{alg:general-case-procedure-hgi} finishes without returning $\perp$. 
We first analyze the preliminary sent-credit matrix $\bar s$ from Table~\ref{tab:general-case-procedure}, and at the end shift its diagonal mass to obtain a feasible solution with $\rho=\Mrho$.

For the first step, we need to prove that the algorithm is valid and can return a solution. The main problem is proving that the properties for $\delta$ and $\sigma$ can be satisfied. We thus need to prove the following two lemmas:
\begin{lemma}\label{lem:delta-existence}
It is possible to assign values $\delta_{i,j}$ that satisfy Property~\ref{property:delta}.
\end{lemma}
\begin{lemma}\label{lem:sigma-existence}
It is possible to assign values $\sigma_{i,j}$ that satisfy Property~\ref{property:sigma}.
\end{lemma}
To prove them, we need some observation around the procedure. Algorithm~\ref{alg:general-case-procedure-hgi} computes a profile for $\hg_i$ and computes a pair of thresholds $t_1$ and $t_2$ using the input variables. For such thresholds, we can observe the following:
\begin{observation}
\label{obs:thresholds}
We have $R_{t_3^\star} = 0$, $R_{t_2 - 1} \geq \MRttwo$, $R_{t_1 - 1} \geq \MRtone$, and $R_0 = 1$. Also, if $t_2 \neq n$ then $R_{t_2} < \MRttwo$, and if $t_1 \neq n$ then $R_{t_1} < \MRtone$.
Also, we have that $R_j$ is non-increasing in $j$.
\end{observation}
\begin{observation}
We have $\hg_{t_1} > 0$ and $\hg_{t_2} > 0$.
\end{observation}
The former observation follows trivially from the definitions in Algorithm~\ref{alg:general-case-procedure-hgi} and, for the latter, the claims for $t_1 < n$ and $t_2 < n$ follow trivially from Observation~\ref{obs:thresholds} and the claim for $t_1 = n$ and the claim for $t_2 = n$ are trivially checked. With this, we are able to prove the lemmas:
\begin{proof}[Proof of Lemmas~\ref{lem:delta-existence} and \ref{lem:sigma-existence}]
For $\delta$, it suffices to prove that, for all $1 \leq j \leq t_1$,
\[\sum_{i=f_j+1}^{t_3^\star} \hg_i \geq \Mcoeffmone \cdot \left[1 -\frac{[\MRttwo-R_{f_j}]^+}{\MRttwo}\right]\,.\]
Note that the left-hand side is equal to $R_{f_j} - R_{t_3^\star} = R_{f_j}$.
Then we consider two cases.
First, if $R_{f_j} \geq \MRttwo$, we need to prove that $R_{f_j} \geq \Mcoeffmone$, which is trivially true since $\Mcoeffmone\leq \MRttwo$ from Property~\ref{property:thresholds}.
Second, if $R_{f_j} < \MRttwo$, we need to prove that $R_{f_j} \geq \Mcoeffmone \cdot \frac{R_{f_j}}{\MRttwo}$, which is again trivially true.
Therefore one can choose $\delta_{i,j}$ greedily to satisfy the desired properties.

For $\sigma$, it suffices to prove that, for all $1 \leq j \leq t_1$,
\[\sum_{i=j}^{f_j} \hg_i \geq \left[\MRttwo - R_{f_j}\right]^+\,.\]
Note that the left-hand side is equal to $R_{j - 1} - R_{f_j}$.
Then note that, because $j \leq t_1$, we have $R_{j - 1} \geq \MRtone$ from the observations, so $R_{j-1} - R_{f_j} \geq \MRtone - R_{f_j}\geq \MRttwo - R_{f_j}$, and the claim follows trivially as the left-hand side is nonnegative.
Therefore one can choose $\sigma_{i,j}$ greedily to satisfy the desired properties.
\end{proof}

We now prove that the preliminary assignment satisfies every constraint of the factor-revealing dual except possibly the diagonal convention $\bar s_{i,i}=0$. 

Non-negativity of the variables is satisfied in our algorithm, but we defer the proof in Appendix~\ref{sec:dual-deferred-proofs}.
\begin{restatable}{lemma}{lemnonneg}\label{lem:nonneg}
We have $\hg_i, g_{i,j}, t_{i,j}, \bar s_{i,j} \geq 0$ for all valid $i, j$.
\end{restatable}
To work with the rest of the constraints, we first present some intuition about the problem.
We look at the program above from Definition~\ref{def:frdual} as a process consisting of $n$ "clients" and a central pool, each having some credit that is generated and transferred around through some specific mechanism.
This credit evolves in this mechanism through $n$ iterations, where in each iteration each client has a constraint on the minimal amount of credit that they must have at that point.
With these constraints set, the goal is to minimize an objective function that depends on how the credit is generated and transferred.

Let us first define what we mean by credit for the clients and for the central pool:
\begin{restatable}{definition}{defcreditij}\label{def:credit-ij}
    The amount of credit a client $j$ has at iteration $i$ for $j \leq i$ is:
    \begin{equation*}
        \text{credit}_{i,j} = \sum_{j'=j+1}^n g_{j,j'} + \sum_{i' = j}^{i} \left(g_{i',j} + t_{i',j} - \bar s_{i',j}\right)\,.
    \end{equation*}
    Note that this is the left-hand side of the $(r_{j,i})$ constraint.
\end{restatable}
\begin{definition}\label{def:pool-credit}
    The amount of credit that the central pool has at iteration $i$ is:
    \begin{equation*}
        \text{pool\_credit}_i = \sum_{i'=1}^i \sum_{j=1}^{i'} \left(\bar s_{i',j} - t_{i',j}\right)\,.
    \end{equation*}
    Note that this is the left-hand side of the $(m_i)$ constraint.
\end{definition}  
Every time step, these quantities change by the following amount:
\begin{observation}
    The difference in credit between two different timesteps $i-1$ and $i$ for a client $j<i$ for $i \geq 2$ is:
    \begin{equation*}
        \Delta \text{credit}_{i,j} = \text{credit}_{i,j} - \text{credit}_{i-1,j} = g_{i,j} - \bar s_{i,j} + t_{i,j}\,.
    \end{equation*}
\end{observation}
Similarly, we can make the following observation for $\text{pool\_credit}_i$:
\begin{observation}
    The difference in the central pool credit between timesteps $i-1$ and $i$ for $i \geq 1$ is:
    \begin{equation*}
        \Delta\text{pool\_credit}_i = \text{pool\_credit}_i - \text{pool\_credit}_{i-1} = \sum_{j=1}^{i} \left(\bar s_{i,j} - t_{i,j}\right)\,.
    \end{equation*}
\end{observation}
We interpret $\bar s_{i,j}$ as credit stored by client $j$ in the central pool at timestep $i$, and $t_{i,j}$ as credit taken by client $j$ from the central pool at timestep $i$. On the other hand, the variables $g_{i,j}$ correspond to the generated credit, the amount $\sum_{j'=j}^n g_{j,j'}$ being an initial amount given to $j$, and the terms $g_{i,j}$ for $j<i$ being the credit at timestep $i$ that goes to client $j$. All the generated credit comes just from the $g_{i,j}$ variables.

The procedure defined in Section~\ref{sec:general-case-procedure} is one way to assign values to the variables $g_{i,j}$, $\bar s_{i,j}$ and $t_{i,j}$ in a way that bounds the objective value to $\Mrho$. It works by separating clients into different groups $A$, $B$, $C$ and $D$ and treating each group differently. There are two exceptions, two clients that might be part of a group and part of another. We define the memberships through $\mu_j^X$ to define how much one belongs to group $X$. 
From Equation~\ref{eq:membership-assignment}, these groups are created based on the thresholds $t_1$, $t_2$ and $t_3^\star$ for which we can observe the following:
\begin{observation}\label{obs:mu-sum}
If $j < t_1$, then $\mu_j^A + \mu_j^D = 1$.
If $j = t_1$, then $\mu_j^A + \mu_j^B + \mu_j^D = 1$.
If $t_1 < j \leq t_2$, then $\mu_j^B + \mu_j^D = 1$.
If $t_2 < j \leq t_3^\star$, then $\mu_j^C + \mu_j^D = 1$.
If $j > t_3^\star$, then $\mu_j^D = 1$.
\end{observation}
More specifically, for clients in group $D$, we have that:
\begin{observation}
If $j<n$, we have $\hg_j = 0$ if and only if $\mu_j^D = 1$.
\end{observation}
\begin{proof}
For $j<t_3^\star$ and $j=t_3^*<n$, it follows by the assignment formulas, and for $j>t_3^\star$, all $\hg_j=0$ and all $j$ have $\mu_j^D=1$; proving the statement.
\end{proof}
This last proof is important because it implies that at such a step $j$, since $g_{j,k}\leq \hg_j$ by the $(\gamma_{i,j}\text{ and }\beta_{i,j})$ constraint, we have that no new credit is generated at that iteration. This constraint states that per each iteration $i$, there is a maximum amount of credit $\hg_i$ that every client can generate, and it is a constraint that our algorithm does indeed follow:
\begin{lemma}\label{lem:gamma-beta-constraints}
    The $(\gamma_{i,j} / \beta_{i,j})$ constraints are satisfied for all $1 \leq i,j \leq n$, i.e., $g_{i,j} \leq \hat{g}_i$ for all $i,j$.
\end{lemma}
\begin{proof}
    The claim follows trivially from the assignments of $g_{i,j}$ in Table~\ref{tab:general-case-procedure} for integral clients, and the rest is proven by convexity.
\end{proof}
This constraint is strong when combined with the $(f)$ one, which constrains the total maximum amount of generated credit, capping it at $1$ unit in total throughout all iterations; which we satisfy:
\begin{lemma}\label{lem:f-constraint}
    The $(f)$ constraint is satisfied, i.e., $\sum_{i=1}^n \hg_i \leq 1$. 
\end{lemma}
\begin{proof}
The claim follows trivially from the fact that $\sum_{i=1}^n \hg_i = R_0 = 1$.
\end{proof}

For the credit in our procedure, we can exactly pinpoint the amount of credit that each client has at every point, which has a different behavior for the different groups.

\begin{restatable}{proposition}{obscreditijvalues-early}\label{prop:credit-ij-values-early}
    For $j \leq i < f_j$, the credit that client $j$ has at iteration $i$ is:
    \begin{equation*}
        \text{\emph{credit}}_{i,j} = 
        \ \begin{cases}
            1 - [\MRttwo - R_{f_j}]^+ + \sum_{i'=j}^i \sigma_{i',j}, 
            & \text{if $j \in A$}\,,\\
            \Mtone + \min\{R_{j-1},\MRtone\}, 
            & \text{if $j \in B$}\,,\\
            \Mtone + R_{j-1}, 
            & \text{if $j \in C$}\,,\\
            R_{j-1} - R_i,
            & \text{if $j \in D$}\,,\\
        \end{cases}
    \end{equation*}
    where the expression for a fractional split client $j$ is the corresponding convex combination of the expressions for the integral cases, i.e.,
    \begin{equation*}
        \text{\emph{credit}}_{i,j}
        = \mu_j^A \text{\emph{credit}}_{i,j}^A
        + \mu_j^B \text{\emph{credit}}_{i,j}^B
        + \mu_j^C \text{\emph{credit}}_{i,j}^C
        + \mu_j^D \text{\emph{credit}}_{i,j}^D\,.
    \end{equation*}
    where $\text{\emph{credit}}_{i,j}^A$, $\text{\emph{credit}}_{i,j}^B$, $\text{\emph{credit}}_{i,j}^C$, and $\text{\emph{credit}}_{i,j}^D$ are the expressions defined in the first four cases.
\end{restatable}
\begin{proof}
The four integral cases are exactly Lemmas \ref{lem:credit-A-before}, \ref{lem:credit-B-before}, \ref{lem:credit-C-before}, and \ref{lem:credit-D-before}. The proof for the fractional cases being a convex combination is in Lemma \ref{lem:credit-fractional-before}. All these lemmas and their proofs are deferred to Appendix~\ref{sec:dual-deferred-proofs}.
\end{proof}

\begin{restatable}{proposition}{obscreditijvalues}\label{prop:credit-ij-values}
    For $i \geq f_j$, the credit that client $j$ has at iteration $i$ is:
    \begin{equation*}
        \text{\emph{credit}}_{i,j} = 
        \ \begin{cases}
            1 - \sum_{i'=f_j+1}^i \delta_{i',j},
            & \text{if $j \in A$}\,,\\
            1 - \ind[i > f_j] \cdot \Mcoeffmone,
            & \text{if $j \in B$}\,,\\
            1 - \ind[i > f_j] \cdot \Mcoeffmone,
            & \text{if $j \in C$}\,,\\
            1,
            & \text{if $j \in D$}\,,\\
        \end{cases}
    \end{equation*}
    where the expression for a fractional split client $j$ is the corresponding convex combination of the expressions for the integral cases, i.e.,
    \begin{equation*}
        \text{\emph{credit}}_{i,j}
        = \mu_j^A \text{\emph{credit}}_{i,j}^A
        + \mu_j^B \text{\emph{credit}}_{i,j}^B
        + \mu_j^C \text{\emph{credit}}_{i,j}^C
        + \mu_j^D \text{\emph{credit}}_{i,j}^D\,.
    \end{equation*}
    where $\text{\emph{credit}}_{i,j}^A$, $\text{\emph{credit}}_{i,j}^B$, $\text{\emph{credit}}_{i,j}^C$, and $\text{\emph{credit}}_{i,j}^D$ are the expressions defined in the first four cases.
\end{restatable}
\begin{proof}
The four integral cases are exactly Lemmas \ref{lem:credit-A-after}, \ref{lem:credit-B-after}, \ref{lem:credit-C-after}, and \ref{lem:credit-D-after}. The proof for the fractional cases being a convex combination is in Lemma \ref{lem:credit-fractional-after}. All these lemmas and their proofs are deferred to Appendix~\ref{sec:dual-deferred-proofs}.
\end{proof}
From where we can observe that:
\begin{restatable}{observation}{obscreditnonnegative}\label{obs:credit-non-negative}
    The credit that each client $j$ has at iteration $i$ (as defined in Definition~\ref{def:credit-ij}) is always non-negative.
\end{restatable}
\begin{proof}
For $i \geq f_j$, the claim follows from the proof of Proposition~\ref{prop:credit-ij-values}.
For $i < f_j$, the proof follows by Proposition~\ref{prop:credit-ij-values-early}.
\end{proof}
With these observations, we can now check that the requirements on the clients' credits are satisfied, i.e., we can check the $(r_{j,i})$ constraints.
\begin{lemma}\label{lem:rij-constraints}
    The $(r_{j,i})$ constraints are satisfied for all $1 \leq j \le i \leq n$, i.e. 
    \begin{equation*}
        \text{\emph{credit}}_{i,j} \ge b_{i,j} - \Mcoeffmone \cdot \sum_{p=j}^{i-1} b_{p,j}\,.
    \end{equation*}
\end{lemma}
\begin{proof}
    From Observation~\ref{obs:credit-non-negative}, we can see that if the right hand side is $0$ (which happens when $i<f_j$), then the constraint is satisfied. Thus, we only need to check the case $i \geq f_j$. We will split the cases into the different groups $A$, $B$, $C$, and $D$ of clients using Proposition \ref{prop:credit-ij-values}. 
    
    We will have to analyze two cases, when $i=f_j$ so that the right hand side is $1$ (since $b_{f_j,j} = 1$ and $b_{p,j} = 0$ for $p < f_j$), and when $i > f_j$ so that the right hand side is at most $1-\Mcoeffmone$ (since $b_{f_j,j} = 1$ and $b_{i,j} \leq 1$).

    For an integral client $j\in A$, for $i=f_j$, the left-hand side of the lemma's inequality is $1$ from Proposition \ref{prop:credit-ij-values}, since $\sum_{i'=f_j+1}^i \delta_{i',j} = 0$ as $i=f_j$. The right hand side of the lemma's inequality is also $1$ as stated in the first paragraph, which satisfies the constraint. For $i > f_j$, we have that $\sum_{i'=f_j+1}^i \delta_{i',j} \leq \Mcoeffmone \cdot \left[1 -\frac{[\MRttwo-R_{f_j}]^+}{\MRttwo}\right]$ by Property \ref{property:delta}, which itself is at most $\Mcoeffmone$, so that the left-hand side is at least $1 - \Mcoeffmone$. On the other hand, the right hand side is at most $1 - \Mcoeffmone$ and thus the constraint is satisfied.

    For $j\in B$ or $j\in C$, for $i=f_j$, the left-hand side of the lemma's inequality is $1$ from Proposition \ref{prop:credit-ij-values}, since $\ind[i > f_j] = 0$ as $i=f_j$. The right hand side of the lemma's inequality is also $1$ as stated in the first paragraph, which satisfies the constraint. For $i > f_j$, we have that $\ind[i > f_j] = 1$, so that the left-hand side is at least $1 - \Mcoeffmone$. On the other hand, the right hand side is at most $1 - \Mcoeffmone$ and thus the constraint is satisfied.

    For $j\in D$, for both $i=f_j$ and $i > f_j$, the left-hand side of the lemma's inequality is $1$ from Proposition \ref{prop:credit-ij-values}, thus the constraint is satisfied as it is bigger than the right hand side which is at most $1$ at any point. 

    Finally, for fractionally assigned clients, the left-hand side of the lemma's inequality is a convex combination of the left-hand sides for the integral cases, and the right hand side is the same in all cases, so that the constraint is satisfied as well.
\end{proof}

For the rest of the proof, we will also be interested in how much credit has been generated by client $j$ during all timesteps $i$, whose behavior also differs depending on the membership of $j$. Do note that, by the mechanism of the dual, not all this credit is kept by the client: for iterations $i<j$ this credit goes to the clients $k<j$ and after that, it can be that the client sends it to the central pool.
\begin{restatable}{lemma}{obssumofgeneratedcreditcolumn}\label{obs:sum-of-generated-credit-column}
    Let
    \begin{equation*}
        \sum_{i=1}^n g_{i,j} = 
        \ \begin{cases}
            1 - \Mcoeffmone \cdot \left(1 - \frac{[\MRttwo - R_{f_j}]^+}{\MRttwo}\right),
            & \text{if $j \in A$}\,,\\
            1,
            & \text{if $j \in B$}\,,\\
            1,
            & \text{if $j \in C$}\,,\\
            1 - \Mcoeffmone \cdot \left[1 - \frac{R_{j-1}}{\MRttwo}\right]^+.
            & \text{if $j \in D$}\,,\\
        \end{cases}
    \end{equation*}
    where the expression for a fractional split client $j$ is the corresponding convex combination of the expressions for the integral cases, i.e.,
    \begin{equation*}
        \sum_{i=1}^n g_{i,j} = \mu_j^A G_j^A + \mu_j^B G_j^B + \mu_j^C G_j^C + \mu_j^D G_j^D\,.
    \end{equation*}
    with $G^A$, $G^B$, $G^C$, and $G^D$ being the expressions defined in the first four cases for $\sum_{i=1}^n g_{i,j}$.
\end{restatable}
\begin{proof}
We analyze each specific case in Lemmas \ref{lem:gen-A}, \ref{lem:gen-B}, \ref{lem:gen-C}, and \ref{lem:gen-D}, and the fractional case is then argued in Lemma \ref{lem:gen-fractional}. All these lemmas and their proofs are deferred to Appendix~\ref{sec:dual-deferred-proofs}.
\end{proof}

Using these observations, we can already prove the objective constraint. 
\begin{lemma}\label{lem:d-i-constraints}
    The $(d_{j})$ constraints are satisfied for all $1 \leq j \leq n$, i.e.:
    \begin{equation*}
        \Mcoeff\cdot \sum_{i=1}^n g_{i,j} + \sum_{i=j}^n (\bar s_{i,j} + t_{i,j}) \leq \rho = \Mrho\,.
    \end{equation*}
\end{lemma}
\begin{proof}
    The proof follows by Lemmas \ref{lem:budget-A}, \ref{lem:budget-B}, \ref{lem:budget-C} and \ref{lem:budget-D} for the cases that the client $j$ belongs integrally to $A$, $B$, $C$ or $D$ respectively. For a fractionally assigned client $j$, the proof follows from Lemma \ref{lem:budget-fractional}.
\end{proof}

So, the preliminary assignment now satisfies the nonnegativity, generation, client-credit, and objective constraints. It remains to prove that the pool credit is nonnegative at every iteration, which is the lengthiest part of the proof, and then to enforce the zero-diagonal convention.

To do so, it will also be interesting to define a specific credit profile $\hg_i$ that is generated in the second phase for clients in $B$ and $C$, which we will call $\hg'_i$. This will be very useful with the analysis of the $(m_i)$ constraint.
\begin{definition}
We define $\hg'_i = \hg_i$ if $i \neq t_1$ and $\hg'_{t_1} = \hg_{t_1}^B$ (except in the special case $t_1 = n$, where we define $\hg'_{n} = \hat{g}_n^\star$). We also define $R'_i = \sum_{i'=i+1}^n \hg'_{i'}$.
\end{definition}

Although Proposition~\ref{prop:credit-ij-values} gives exact formulas for client credit, doing the same for the pool is a lot more involved. However, we can still analyze the behavior of the credit of the central pool through the iterations, which is what is used to prove the next lemma. Its proof is the most involved in this section.
\begin{lemma}\label{lem:mi-constraints}
    The $(m_{i})$ constraints are satisfied for all $1 \leq i \leq n$:
    \begin{equation*}
        \text{\emph{pool\_credit}}_i \geq 0\,.
    \end{equation*}
\end{lemma}
\begin{proof}    
    For a general case, we can rewrite $\text{pool\_credit}_{i}$ as
    \begin{equation}\label{eq:mi-rewrite}
        \begin{aligned}
            \text{pool\_credit}_{i} &= 
            \sum_{p=1}^i \sum_{j=1}^p \bar s_{p,j} - \sum_{p=1}^i \sum_{j=1}^p t_{p,j}
            = \sum_{j=1}^i \sum_{p=j}^i \bar s_{p,j} - \sum_{j=1}^i t_{f_j, j} \cdot \ind [f_j \leq i]
        \end{aligned}
    \end{equation}
    using the fact that $t_{p,j} \neq 0$ only for $p = f_j$. We will split up our analysis depending on where this $i$ is related to the threshold values.

    \smallskip
    \noindent\textbf{Case 1: $i<t_1$.} For $i < t_1$, since we have $\mu_j^A + \mu_j^D = 1$ by Observation \ref{obs:mu-sum} we can substitute the value of $t_{f_j,j}$: 
    \begin{equation*}
        t_{f_j, j} \cdot \ind [f_j \leq i] = \mu_j^D \cdot (1-R_{j-1}) \cdot \ind [f_j \leq i] = 0
    \end{equation*}
    because $\mu_j^D>0$ only if $f_j \geq t_3^\star \geq t_1 > i$, so in that case $\ind [f_j \leq i] =0$. So, the resulting inequality only has the term $\sum_{j=1}^i \sum_{p=j}^i \bar s_{p,j}$ which is non-negative, and thus the $(m_i)$ constraint for $i<t_1$ is satisfied.

    \smallskip
    \noindent\textbf{Case 2: $i=t_1$.} For $i = t_1$, we have the same except for $j=t_1$ for which we get that $t_{f_{t_1},t_1} = \mu_{t_1}^B [1-R_{t_1-1}-\Mtone]^+ + \mu_{t_1}^D (1-R_{t_1-1})$. For such an index we have $R_{t_1-1}\geq \MRtone$ by Observation \ref{obs:thresholds}, so its contribution to $t$ from $B$ is still $0$. Thus, we still have for all $j\leq t_1$: $t_{f_{j}, j} \cdot \ind[f_j \leq i] = \mu_j^D \cdot (1 - R_{j - 1}) \cdot \ind[f_j \leq i]$. 
    Such a term can now be non-zero, but only if $f_j \geq t_3^\star \geq t_1 = i$, which allows $\mu_j^D > 0$.
    But then $\ind[f_j \leq i] \neq 0$ only if $t_1 = t_2 = t_3^\star$, and we will postpone the analysis for such a case when analyzing directly $t_3^\star$ below in case 4.

    \smallskip
    \noindent\textbf{Case 3: $t_1<i<t_3^\star$.} For $t_1 < i < t_3^\star$, starting from Equation~\ref{eq:mi-rewrite}, we first separate the sums into two parts, one for $p < t_1$ and another for $p \geq t_1$:
    \begin{align*}
        \sum_{j=1}^i \sum_{p=j}^i \bar s_{p,j} - \sum_{j=1}^i t_{f_j, j} \cdot \ind [f_j \leq i] 
        &= \left(\sum_{j=1}^{t_1 - 1} \sum_{p=j}^i \bar s_{p,j} - \sum_{j=1}^{t_1 - 1} t_{f_j, j} \cdot \ind [f_j \leq i]\right) + \sum_{j=t_1}^i \sum_{p=j}^i \bar s_{p,j} - \sum_{j=t_1}^i t_{f_j, j} \cdot \ind [f_j \leq i] \\
        &\geq\sum_{j=t_1}^i \sum_{p=j}^i \bar s_{p,j} - \sum_{j=t_1}^i t_{f_j, j} \cdot \ind [f_j \leq i]\,.
    \end{align*}
    We can drop the first term because for $j < t_1$ we have $t_{f_j, j} \cdot \ind[f_j \leq i] = \mu_j^D \cdot (1 - R_{j - 1}) \cdot \ind[f_j \leq i] = 0$, which is because $\mu_j^D = 0$ unless $f_j \geq t_3^\star$ but $i < t_3^\star$ so $\ind[f_j \leq i] = 0$.
    
    Next, we substitute the values of $\bar s_{i,j}$ and $t_{i,j}$. We first lower-bound the $\bar s$-term. For $t_1\leq j<t_3^\star$ and $p>j$ we have $\bar s_{p,j} \geq (\mu_j^B + \mu_j^C) \cdot (\hg_p + \Mcoeffmone\cdot \ind[p=f_j+1])$. Thus, we have:
    \begin{align*}
        \sum_{j=t_1}^i \sum_{p=j}^i \bar s_{p,j} 
        &\geq \sum_{j=t_1}^{i-1} \sum_{p=j+1}^{i} (\mu_j^B + \mu_j^C) \cdot \left(\hg_{p} + \Mcoeffmone\cdot \ind[p=f_j+1]\right)  \\
        &\geq \sum_{j=t_1}^{i-1} (\mu_j^B + \mu_j^C) \cdot \sum_{p=j+1}^{i} \hg_{p}
        = \sum_{j=t_1}^{i-1} (\mu_j^B + \mu_j^C) \cdot (R_j-R_i)\,.
    \end{align*}
    Then, we upper-bound the $t$-term. We first observe that $t_{f_{t_1}, t_1} \cdot \ind[f_{t_1} \leq i] = (1-R_{t_1-1}) \cdot \mu_{t_1}^D \cdot \ind[f_{t_1} \leq i] = 0$ because $\mu_{t_1}^B [1-R_{t_1-1}-\Mtone]^+ =0$ (as seen in the $i=t_1$ part of the proof) and $\mu_{t_1}^D \cdot \ind[f_{t_1} \leq i]=0$ because $t_1<t_3^\star$. Then, we have:
    \begin{align*}
        \sum_{j=t_1}^i t_{f_j, j} \cdot \ind [f_j \leq i]
        &= \sum_{j=t_1+1}^i \left((\mu_j^B+\mu_j^C)\cdot \left(1-\sum_{p=j}^{n} \hg_{p}-\Mtone\right) + \mu_j^D \cdot (1-R_{j-1})\right) \cdot \ind [f_j \leq i] \\
        &= \sum_{j=t_1+1}^i \left(1- \sum_{p=j}^{n} \hg_{p}-\Mtone \right)\cdot(\mu_j^B+\mu_j^C)\cdot \ind [f_j \leq i]\\
        &\leq \sum_{j=t_1+1}^i \left(1- \sum_{p=j}^{n} \hg_{p}-\Mtone \right)\cdot(\mu_j^B+\mu_j^C) \\
        &= \sum_{j=t_1+1}^i \sum_{p=t_1}^{j-1} \hg'_p \cdot(\mu_j^B+\mu_j^C)
        = \sum_{p=t_1}^{i-1} \hg'_p \sum_{j=p+1}^{i} (\mu_j^B+\mu_j^C)\,,
    \end{align*}
    where we used that $\mu_j^D\cdot \ind[f_j \leq i]=0$ for $i<t_3^\star$, that $1- \sum_{p=j}^{n} \hg_{p}-\Mtone = \MRtone - R_{j-1} > 0$ for $j \geq t_1 + 1$, and that $R'_{t_1 - 1} = R_{t_1} + \hg_{t_1}^B = \MRtone$.

    Thus, the following inequality is sufficient to satisfy the $(m_i)$ constraint in these ranges:
    \begin{equation*}
        \sum_{j=t_1}^{i-1} \left((\mu_j^B + \mu_j^C) \cdot (R_j-R_i) - \hg'_j \sum_{p=j+1}^{i} (\mu_p^B+\mu_p^C) \right) \geq 0\,.
    \end{equation*}
    To prove the inequality, we prove that each $j$-th term is indeed non-negative. Rearranging, this is equivalent to proving:
    \begin{equation}\label{eq:mi-phase2-inequality}
        (R_i+\Mtone) \leq (R_j+\Mtone) - \hg'_j \cdot \frac{\sum_{p=j+1}^{i} (\mu_p^B + \mu_p^C)}{\mu_j^B+\mu_j^C}\,.
    \end{equation}
    Note that $\hg'_{t_1} = (\mu_{t_1}^B + \mu_{t_1}^C) \cdot \frac{R_{t_1}+\Mtone}{n-t_1}$.
    We divided by $\mu_j^B + \mu_j^C$ because we have assumed that this value is greater than $0$. If it is in fact $0$, one can directly check that the inequality (without dividing by $\mu_j^B+\mu_j^C$) is satisfied because $\hg'_j \propto \mu_j^B + \mu_j^C$ for $j \geq t_1$ by the definition of the algorithm.

    Substituting the value of $\hg'_j = (\mu_j^B+\mu_j^C)\cdot \frac{R_j+\Mtone}{n-j}$, and rearranging we get that we want to prove:
    \begin{equation*}
        (R_i+\Mtone) \leq (R_j+\Mtone) \cdot \frac{n-j - \sum_{p=j+1}^{i} (\mu_p^B+\mu_p^C)}{n-j}\,.
    \end{equation*}

    Note that $i < t_3^\star$, so in particular $i < n$.
    We will now use the following claim:
    \begin{restatable}{claim}{claimphasetworiproduct}\label{claim:phase-2-ri-product}
        For $t_1\leq j\leq i \leq \min(t_3^\star, n - 1)$ we have the following equality that relates $R_j$ to $R_i$:
        \begin{equation*}
            R_j+\Mtone = (R_i + \Mtone) \cdot \frac{n-j - \sum_{\ell=j+1}^i \mu_\ell^D \cdot \frac{R_j+\Mtone}{R_{\ell-1}+\Mtone}}{n-i}\,.
        \end{equation*}
    \end{restatable}
    The proof can be found in Section~\ref{sec:dual-deferred-proofs}.
    Using Claim~\ref{claim:phase-2-ri-product}, we have that the following is true:
    \begin{equation*}
        (R_j + \Mtone)\cdot(n-i) = (R_i+\Mtone) \cdot (n-j) - \sum_{\ell=j+1}^i \mu_\ell^D \cdot \frac{(R_j+\Mtone)\cdot (R_i+\Mtone)}{R_{\ell-1}+\Mtone}\,.
    \end{equation*}
    Rearranging gives:
    \begin{align*}
        R_i+\Mtone 
        &= (R_j + \Mtone)\cdot\frac{n-i+\sum_{\ell=j+1}^i \mu_\ell^D \cdot \frac{R_i+\Mtone}{R_{\ell-1}+\Mtone}}{n-j} 
        \leq (R_j + \Mtone)\cdot\frac{n-i+\sum_{\ell=j+1}^i \mu_\ell^D}{n-j}\\
        &= (R_j + \Mtone)\cdot\frac{n-j-\sum_{\ell=j+1}^i (1-\mu_\ell^D)}{n-j} 
        = (R_j + \Mtone)\cdot\frac{n-j-\sum_{\ell=j+1}^i (\mu_\ell^B+\mu_\ell^C)}{n-j}\,,
    \end{align*}
    where we have used that $R_i\leq R_{\ell-1}$ if $\ell\leq i$, and that $1-\mu_\ell^D = \mu_\ell^B + \mu_\ell^C$ by Observation \ref{obs:mu-sum}. This is exactly the inequality we wanted to prove,
    which proves that the $(m_i)$ constraint is satisfied for $t_1 < i < t_3^\star$.
    
    \smallskip
    \noindent\textbf{Case 4: $i=t_3^\star$.} For $i=t_3^\star$, we want to prove that:
    \begin{equation*}
        \text{pool\_credit}_{t_3^\star} = \sum_{p=1}^{t_3^\star} \sum_{j=1}^p \left(\bar s_{p,j} - t_{p,j}\right) \geq 0\,.
    \end{equation*}
    Let us now add, for all $j\leq t_3^\star$, $\text{credit}_{t_3^\star,j}$ which we know the exact value of thanks to Proposition~\ref{prop:credit-ij-values-early} and Proposition~\ref{prop:credit-ij-values}. Also note that if $f_j > t_3^\star$ then $\mu_j^D = 1$. 
    Using the propositions, then we have:
    \begin{align*}
        \sum_{j=1}^{t_3^\star} \text{credit}_{t_3^\star,j} &= \sum_{j=1}^{t_3^\star}\Paren{ \sum_{q=j+1}^{n} g_{j,q} + \sum_{p=j}^{t_3^\star} \left(g_{p,j} - \bar s_{p,j} + t_{p,j}\right)} \\
        &= \sum_{j=1}^{t_3^\star} \text{credit}_{t_3^\star,j} \cdot \ind[f_j \leq t_3^\star] + \sum_{j=1}^{t_3^\star} \text{credit}_{t_3^\star,j} \cdot \ind[f_j > t_3^\star]\\
        &= \sum_{j=1}^{t_3^\star} \Paren{1 - \Mcoeffmone\cdot \left[1-\frac{[\MRttwo-R_{f_j}]^+}{\MRttwo}\right] \cdot \mu_j^A - \Mcoeffmone\cdot \ind[f_j < t_3^\star] \cdot ( \mu_j^B +  \mu_j^C)} \cdot \ind[f_j \leq t_3^\star]\\
        &\quad + \sum_{j=1}^{t_3^\star} \Paren{R_{j-1} - R_{t_3^\star}} \cdot \ind[f_j > t_3^\star]\\
        &= t_3^\star - \Mcoeffmone\cdot \sum_{j=1}^{t_1} \left[1-\frac{[\MRttwo-R_{f_j}]^+}{\MRttwo}\right] \cdot \mu_j^A
         - \Mcoeffmone \cdot \sum_{j=t_1}^{t_3^\star}(\mu_j^B+\mu_j^C)\cdot \ind[f_j<t_3^\star] \\
        &\quad - \sum_{j=1}^{t_3^\star}(1-R_{j-1}) \cdot \ind[f_j > t_3^\star]\,.
    \end{align*}
    
    Adding this to our inequality, since all the $\bar s_{i,j}$ and $t_{i,j}$ terms cancel out, we get that we need to prove:
    \begin{align*}
        \sum_{j=1}^{t_3^\star} \left(\sum_{q=j+1}^{n} g_{j,q} + \sum_{p=j}^{t_3^\star} g_{p,j} \right)&\geq t_3^\star - \Mcoeffmone\cdot \sum_{j=1}^{t_1}\left[1-\frac{[\MRttwo-R_{f_j}]^+}{\MRttwo}\right] \cdot \mu_j^A 
        - \Mcoeffmone \cdot \sum_{j=t_1}^{t_3^\star}(\mu_j^B+\mu_j^C)\cdot \ind[f_j<t_3^\star] \\
        & - \sum_{j=1}^{t_3^\star}(1-R_{j-1}) \cdot \ind[t_3^\star<f_j]\,.
    \end{align*}
    
    Focusing on the left-hand side, notice that $g_{i,j}\leq \hg_i = 0$ for $i>t_3^\star$, and thus, we can state the following:
    \begin{align*}
        \sum_{j=1}^{t_3^\star} &\left(\sum_{q=j+1}^{n} g_{j,q} + \sum_{p=j}^{t_3^\star} g_{p,j} \right) 
        =\sum_{j=1}^n \left(\sum_{q=j+1}^{n} g_{j,q} + \sum_{p=j}^{n} g_{p,j} \right) 
        = \sum_{j=1}^n \sum_{p=1}^{n} g_{p,j}= \\
        &= n - \Mcoeffmone\cdot \sum_{j=1}^{t_1} \mu_j^A \cdot\left[1 -\frac{(\MRttwo-R_{f_j})^+}{\MRttwo}\right] - \Mcoeffmone\cdot \sum_{j=t_1+1}^n \mu_j^D \cdot\left[1 -\frac{R_{j-1}}{\MRttwo}\right]^+
    \end{align*}
    where we have used Lemma~\ref{obs:sum-of-generated-credit-column} to get the correct values. Thus, the previous inequality reduces to:
    \begin{align*}
        n - \Mcoeffmone\cdot &\sum_{j=1}^{t_1} \mu_j^A \cdot\left[1 -\frac{(\MRttwo-R_{f_j})^+}{\MRttwo}\right]- \Mcoeffmone\cdot \sum_{j=t_1+1}^n \mu_j^D \cdot\left[1 -\frac{R_{j-1}}{\MRttwo}\right]^+ \geq \\
        &\geq t_3^\star - \Mcoeffmone\cdot \sum_{j=1}^{t_1}\left[1-\frac{[\MRttwo-R_{f_j}]^+}{\MRttwo}\right] \cdot \mu_j^A 
        - \Mcoeffmone \cdot \sum_{j=t_1}^{t_3^\star}(\mu_j^B+\mu_j^C)\cdot \ind[f_j<t_3^\star] 
        - \sum_{j=1}^{t_3^\star}(1-R_{j-1}) \cdot \ind[t_3^\star<f_j]\,.
    \end{align*}
    Using the fact that $n-t_3^\star \geq \Mcoeffmone\cdot \sum_{j=t_3^\star+1}^n \mu_j^D \cdot\left[1 -\frac{R_{j-1}}{\MRttwo}\right]^+$ because $\Mcoeffmone \leq 1$ and that $\sum_{j=1}^{t_3^\star}(1-R_{j-1})\cdot \ind[t_3^\star<f_j] \geq \Mcoeffmone\cdot \sum_{j=1}^{t_3^\star} \mu_j^D \cdot\left[1 -\frac{R_{j-1}}{\MRttwo}\right]^+ \cdot \ind[t_3^\star<f_j]$ because $(1-R_{j-1})\geq \Mcoeffmone \cdot\left[1 -\frac{R_{j-1}}{\MRttwo}\right]^+$ as $\Mcoeffmone\leq \MRttwo$, if we cancel things out, using also that $f_j < t_3^\star$ implies $\mu_j^D =0$, it is sufficient to prove that:
    \begin{align*}
        \sum_{j=1}^{t_3^\star}(\mu_j^B+\mu_j^C)\cdot \ind[f_j<t_3^\star] \geq 
         \sum_{j=1}^n \mu_j^D \cdot\left[1 -\frac{R_{j-1}}{\MRttwo}\right]^+ \cdot \ind[f_j=t_3^\star]\,.
    \end{align*}
    At this point we use that inputs satisfy the admissibility property defined in Property~\ref{property:admissible-frontier-input}, which is the following:
    Let $S = \Set{j \in [n] \mid f_j = t_3^\star}$. Then $\mu_j^D \in \{0, 1\}$ for all $j$ except at most one of them, and also if $j_1, j_2 \in S$ with $j_1 \leq j_2$ then $\mu_{j_1}^D \geq \mu_{j_2}^D$. Let $p^\star$ be the maximum index in $S$ such that $\mu_{p^\star}^D > 0$, and $0$ if no such index exists. Note that $p^\star \leq t_3^\star$.
    If $p^\star = 0$, then the RHS is zero, so the claim follows trivially.
    Let then $p^\star > 0$.
    For the LHS we have 
    \begin{align*}
    \sum_{j=1}^{t_3^\star}(\mu_j^B+\mu_j^C)\cdot \ind[f_j<t_3^\star]
    &\geq \sum_{j=1}^{p^\star} (\mu_j^B+\mu_j^C)\cdot \ind[f_j<t_3^\star]
    = \sum_{j=1}^{p^\star - 1} (\mu_j^B+\mu_j^C) + (\mu_{p^\star}^B + \mu_{p^\star}^C) \cdot \ind[f_{p^\star} < t_3^\star]\\
    &= \sum_{j=1}^{p^\star - 1} (\mu_j^B+\mu_j^C)
    \end{align*}
    where we used that if $j < p^\star$ and $f_j \geq t_3^\star$ then $\mu_j^D = 1$ and hence $\mu_j^B = \mu_j^C = 0$, and also that $f_{p^\star} = t_3^\star$ since $p^\star \in S$.
    For the RHS we have
    \begin{align*}
    \sum_{j=1}^n \mu_j^D \cdot\left[1 -\frac{R_{j-1}}{\MRttwo}\right]^+ \cdot \ind[f_j=t_3^\star]
    &= \sum_{j=1}^{p^\star} \mu_j^D \cdot\left[1 -\frac{R_{j-1}}{\MRttwo}\right]^+ \cdot \ind[f_j=t_3^\star]
    \leq \sum_{j=1}^{p^\star} \mu_j^D \cdot\left[1 -\frac{R_{j-1}}{\MRttwo}\right]^+\,.
    \end{align*}
    Therefore, it suffices to prove that 
    \[\sum_{j=1}^{p^\star - 1} (\mu_j^B+\mu_j^C) \geq \sum_{j=1}^{p^\star} \mu_j^D \cdot\left[1 -\frac{R_{j-1}}{\MRttwo}\right]^+\,.\]
    We note that if $j \leq t_2$ then $R_{j-1} \geq \MRttwo$ so we can assume that $p^\star \geq t_2 + 1$, and then what we need to prove is by the same reason equivalent to
    \[\sum_{j=1}^{p^\star - 1} (\mu_j^B+\mu_j^C) \geq \sum_{j=t_2+1}^{p^\star} \mu_j^D \cdot\left[1 -\frac{R_{j-1}}{\MRttwo}\right]^+\,.\]
    We use the following two claims.
    
    \begin{restatable}{claim}{claimphasetworiproductprime}\label{claim:phase-2-ri-product-prime}
    For $t_1 - 1\leq j\leq i \leq \min(t_3^\star, n - 1)$ we have the following equality that relates $R'_j$ to $R'_i$:
    \begin{equation*}
        R'_j+\Mtone = (R'_i + \Mtone) \cdot \frac{n-j - \sum_{\ell=j+1}^i (1- \mu_{\ell}^B - \mu_{\ell}^C) \cdot \frac{R'_j+\Mtone}{R'_{\ell-1}+\Mtone}}{n-i}\,.
    \end{equation*}
    \end{restatable}
    We prove it in Section~\ref{sec:dual-deferred-proofs}.
    Then the following claim proves this result.
    \begin{claim}\label{clm:b-size2}
    If $p^\star \geq t_2 + 1$, then
    \[\sum_{j=1}^{p^\star - 1} (\mu_j^B+\mu_j^C) \geq \sum_{j=t_2+1}^{p^\star} \mu_j^D \cdot\left[1 -\frac{R_{j-1}}{\MRttwo}\right]^+\,.\]
    \end{claim}
    \begin{proof}
        We analyze the evolution of the remaining budget during Phase 2.
        By Claim~\ref{claim:phase-2-ri-product-prime} and choosing $j=t_1 - 1$ and $i=p^\star - 1$, we have that:
        \begin{equation*}
            R'_{t_1 - 1} + \Mtone = (R'_{p^\star - 1} + \Mtone) \cdot \frac{n-t_1+1-\sum_{\ell=t_1}^{p^\star - 1} (1-\mu_{\ell}^B-\mu_{\ell}^C) \cdot \frac{R'_{t_1-1}+\Mtone}{R'_{\ell-1}+\Mtone}}{n-p^\star+1}\,.
        \end{equation*}
        For the numerator, we can do the following:
        \begin{align*}
            &n-t_1+1-\sum_{\ell=t_1}^{p^\star-1} (1-\mu_{\ell}^B-\mu_{\ell}^C) \cdot \frac{R'_{t_1-1}+\Mtone}{R'_{\ell-1}+\Mtone}\\
            &= n-p^\star+1 + \sum_{\ell=t_1}^{p^\star-1}(\mu_{\ell}^B+\mu_{\ell}^C)-\sum_{\ell=t_1}^{p^\star-1} (1-\mu_{\ell}^B-\mu_{\ell}^C) \cdot \frac{R'_{t_1-1}-R'_{\ell-1}}{R'_{\ell-1}+\Mtone}\,.
        \end{align*}
        
        Let us prove that $\frac{R'_{t_1 - 1}-R'_{\ell-1}}{R'_{\ell-1}+\Mtone} \geq \left[1-\frac{R'_{\ell-1}}{\MRttwo}\right]^+$ for $t_1 \leq \ell\leq p^\star$. Using the fact that $R'_{t_1 - 1}=\MRtone=1-\Mtone$, we have that the left-hand side is: $\frac{1-\Mtone-R'_{\ell-1}}{R'_{\ell-1}+\Mtone}$. Then it is easily checked that this is lower bounded by $\left[1-\frac{R'_{\ell-1}}{\MRttwo}\right]^+$ for $R'_{\ell - 1} \leq 1-\Mtone$: i.e., that $\frac{1-\Mtone-x}{\Mtone+x} \geq [1-\frac{x}{\MRttwo}]^+$ for $0 \leq x \leq 1-\Mtone$, which is true by Property~\ref{property:function-tau}.
        
        Then we get the following:
        \begin{align*}
            \sum_{\ell=t_1}^{p^\star-1}\mu_\ell^B +\sum_{\ell=t_1}^{p^\star-1}\mu_\ell^C
            &\geq \frac{R'_{t_1 - 1}-R'_{p^\star-1}}{R'_{p^\star-1}+\Mtone} \cdot (n-p^\star+1) + \sum_{\ell=t_1}^{p^\star-1} (1-\mu_{\ell}^B-\mu_{\ell}^C)\cdot\left[1-\frac{R'_{\ell-1}}{\MRttwo}\right]^+\\
            &\geq \frac{R'_{t_1 - 1}-R'_{p^\star-1}}{R'_{p^\star-1}+\Mtone} + \sum_{\ell=t_1}^{p^\star-1} (1-\mu_{\ell}^B-\mu_{\ell}^C)\cdot\left[1-\frac{R'_{\ell-1}}{\MRttwo}\right]^+\\
            &\geq \left[1-\frac{R'_{p^\star-1}}{\MRttwo}\right]^+ + \sum_{\ell=t_1}^{p^\star-1} (1-\mu_{\ell}^B-\mu_{\ell}^C)\cdot\left[1-\frac{R'_{\ell-1}}{\MRttwo}\right]^+\\
            &\geq \sum_{\ell=t_2+1}^{p^\star} \mu_{\ell}^D \cdot\left[1-\frac{R'_{\ell-1}}{\MRttwo}\right]^+\,.
        \end{align*}
        Then, because $R_{\ell - 1} = R'_{\ell - 1}$ for all $\ell \geq t_2 + 1$, the conclusion follows.
    \end{proof}

    \smallskip
    \noindent\textbf{Case 5: $t_3^\star<i$.} For the case $t_3^\star <i\leq n$, let us observe that the $\text{pool\_credit}_i$ is non-increasing with respect to $i$, i.e.:
    \begin{equation*}
        \text{pool\_credit}_i \geq \text{pool\_credit}_{i+1} \geq \dots \geq \text{pool\_credit}_{n}\,.
    \end{equation*}
    This is because the positive $\bar s_{i,j}$ terms are zero for $i>t_3^\star+1$, because from $t_3^\star$ onward $\hg_{i} = 0$ and, since $f_j+1 \leq t_3^\star+1$ for $\mu_j^D<1$, the extra $\Mcoeffmone$ added to $\bar s_{f_j+1,j}$ by elements in $B$ and $C$ has already happened. On the other hand, the negative sum of $t_{i,j}$ terms can have new non-zero elements for increasing $i$ after $t_3^\star+1$. Thus, proving the inequality for $\text{pool\_credit}_{n}$ is sufficient:
    \begin{equation*}
        \text{pool\_credit}_{n} = \sum_{j=1}^n \sum_{p=j}^n \left(\bar s_{p,j} - t_{p,j}\right) \geq 0\,.
    \end{equation*}
    Let us now add, for all $j$, $\text{credit}_{n,j}$, for which we know the exact value thanks to Proposition~\ref{prop:credit-ij-values}. As a refresher, this is:
    \begin{align*}
        \sum_{j=1}^{n} \text{credit}_{n,j} &= \sum_{j=1}^{n}\Paren{ \sum_{q=j+1}^{n} g_{j,q} + \sum_{p=j}^{n} \left(g_{p,j} - \bar s_{p,j} + t_{p,j}\right)} \\
    &= \sum_{j=1}^{n} \Paren{1 - \Mcoeffmone\cdot \left[1-\frac{[\MRttwo-R_{f_j}]^+}{\MRttwo}\right] \cdot \mu_j^A - \Mcoeffmone\cdot \ind[f_j < n] \cdot ( \mu_j^B +  \mu_j^C)}\\
    &= n - \Mcoeffmone\cdot \sum_{j=1}^{t_1} \left[1-\frac{[\MRttwo-R_{f_j}]^+}{\MRttwo}\right] \cdot \mu_j^A - \Mcoeffmone \cdot \sum_{j=t_1}^{t_3^\star}(\mu_j^B+\mu_j^C)\cdot \ind[f_j<n]\,.
    \end{align*}
    Thus, since all the $\bar s_{i,j}$ and $t_{i,j}$ terms cancel out, we get that it is equivalent to prove:
    \begin{equation*}
        \sum_{j=1}^n \left(\sum_{q=j+1}^{n} g_{j,q} + \sum_{p=j}^{n} g_{p,j} \right)\geq n - \Mcoeffmone\cdot \sum_{j=1}^{n}\left[1-\frac{[\MRttwo-R_{f_j}]^+}{\MRttwo}\right] \cdot \mu_j^A- \Mcoeffmone \cdot \sum_{j=t_1}^{t_3^\star}(\mu_j^B+\mu_j^C)\cdot \ind[f_j<n]\,.
    \end{equation*}
    Now, if we focus on the left-hand side, we can reorganize the $g_{i,j}$ terms to see that we in fact have:
    \begin{align*}
        \sum_{j=1}^n &\left(\sum_{q=j+1}^{n} g_{j,q} + \sum_{p=j}^{n} g_{p,j} \right) 
        = \sum_{j=1}^n \sum_{p=1}^{n} g_{p,j}= 
         n - \Mcoeffmone\cdot \sum_{j=1}^{t_1} \mu_j^A \cdot\left[1 -\frac{(\MRttwo-R_{f_j})^+}{\MRttwo}\right] - \Mcoeffmone\cdot \sum_{j=t_1+1}^n \mu_j^D \cdot\left[1 -\frac{R_{j-1}}{\MRttwo}\right]^+
    \end{align*}
    where we have used Lemma~\ref{obs:sum-of-generated-credit-column} to get the correct values. So we end up having to prove the following:
    \begin{align*}
        n - \Mcoeffmone\cdot \sum_{j=1}^{t_1} \mu_j^A \cdot&\left[1 -\frac{(\MRttwo-R_{f_j})^+}{\MRttwo}\right]- \Mcoeffmone\cdot \sum_{j=t_1+1}^n \mu_j^D \cdot\left[1 -\frac{R_{j-1}}{\MRttwo}\right]^+ \\
        &\geq n - \Mcoeffmone\cdot \sum_{j=1}^{t_1} \mu_j^A \cdot\left[1 -\frac{(\MRttwo-R_{f_j})^+}{\MRttwo}\right] - \Mcoeffmone \cdot \sum_{j=t_1}^{t_3^\star}(\mu_j^B+\mu_j^C)\cdot \ind[f_j<n]
    \end{align*}
    which by using the fact that for $j\leq t_1$ we have that $R_{j-1}\geq \MRtone\geq \theta$, we can add their corresponding $\mu_j^D \cdot\left[1 -\frac{R_{j-1}}{\MRttwo}\right]^+$ as they are $0$, and thus, by simplifying this becomes:
    \begin{equation*}
        \sum_{j=1}^n \mu_j^D \cdot\left[1 -\frac{R_{j-1}}{\MRttwo}\right]^+ \leq \sum_{j=t_1}^{t_3^\star}(\mu_j^B+\mu_j^C)\cdot \ind[f_j<n]\,.
    \end{equation*}
    Note that $i > t_3^\star$, so $t_3^\star = n$ is impossible.
    Then we have $(\mu_j^B+\mu_j^C)\cdot \ind[f_j=n] = 0$ because $f_j > t_3^\star$, so equivalently we need to prove that 
    \begin{equation*}
        \sum_{j=1}^n \mu_j^D \cdot\left[1 -\frac{R_{j-1}}{\MRttwo}\right]^+ \leq \sum_{j=t_1}^{t_3^\star}(\mu_j^B+\mu_j^C)\,.
    \end{equation*}
    Through the following claim, this result is true.
    \begin{claim}\label{clm:b-size}
        We have
        \[\sum_{j=t_1}^{t_3^\star}(\mu_j^B+\mu_j^C) \geq \sum_{j=1}^n \mu_j^D \cdot\left[1 -\frac{R_{j-1}}{\MRttwo}\right]^+\,.\]
    \end{claim}

    \begin{proof}
        We analyze the evolution of the remaining budget during Phase 2.
        By Claim~\ref{claim:phase-2-ri-product-prime} and choosing $j=t_1 - 1$ and $i=t_3^\star$, we have that:
        \begin{equation*}
            R'_{t_1 - 1} + \Mtone = (R'_{t_3^\star} + \Mtone) \cdot \frac{n-t_1+1-\sum_{\ell=t_1}^{t_3^\star} (1 - \mu_{\ell}^B - \mu_{\ell}^C) \cdot \frac{R'_{t_1-1}+\Mtone}{R'_{\ell-1}+\Mtone}}{n-t_3^\star}\,.
        \end{equation*}
        For the numerator, we can do the following:
        \begin{align*}
            &n-t_1+1-\sum_{\ell=t_1}^{t_3^\star} (1 - \mu_{\ell}^B - \mu_{\ell}^C) \cdot \frac{R'_{t_1-1}+\Mtone}{R'_{\ell-1}+\Mtone}\\
            &= n-t_3^\star + \sum_{\ell=t_1}^{t_3^\star}(\mu_{\ell}^B+\mu_{\ell}^C)-\sum_{\ell=t_1}^{t_3^\star} (1 - \mu_{\ell}^B - \mu_{\ell}^C) \cdot \frac{R'_{t_1 - 1}-R'_{\ell-1}}{R'_{\ell-1}+\Mtone}\,.
        \end{align*}

        Let us prove that $\frac{R'_{t_1 - 1}-R'_{\ell-1}}{R'_{\ell-1}+\Mtone} \geq \left[1-\frac{R'_{\ell-1}}{\MRttwo}\right]^+$ for $t_1 \leq \ell\leq t_3^\star$. Using the fact that $R'_{t_1 - 1}=\MRtone=1-\Mtone$, we have that the left-hand side is: $\frac{1-\Mtone-R'_{\ell-1}}{R'_{\ell-1}+\Mtone}$. Then it is easily checked that this is lower bounded by $\left[1-\frac{R'_{\ell-1}}{\MRttwo}\right]^+$ for $R'_{\ell - 1} \leq 1-\Mtone$: i.e., that $\frac{1-\Mtone-x}{\Mtone+x} \geq [1-\frac{x}{\MRttwo}]^+$ for $0 \leq x \leq \MRtone$, which is true by Property~\ref{property:function-tau}.
        
        Then, rearranging, we get the following:
        \begin{equation*}
            \sum_{\ell=t_1}^{t_3^\star}(\mu_{\ell}^B+\mu_{\ell}^C) \geq \frac{R'_{t_1 - 1}-R'_{t_3^\star}}{R'_{t_3^\star}+\Mtone} \cdot (n-t_3^\star) + \sum_{\ell=t_1}^{t_3^\star} (1 - \mu_{\ell}^B - \mu_{\ell}^C) \cdot\left[1-\frac{R'_{\ell-1}}{\MRttwo}\right]^+\,.
        \end{equation*}
        Since $R'_{t_1-1}=\MRtone$ and $R'_{t_3^\star}=0$, the first term on the right is equal to $\frac{\MRtone}{\Mtone}(n-t_3^\star)$, and hence is at least $n-t_3^\star$. Therefore the previous bound implies in particular that
        \begin{equation*}
            \sum_{\ell=t_1}^{t_3^\star}(\mu_{\ell}^B+\mu_{\ell}^C) \geq (n-t_3^\star) + \sum_{\ell=t_1}^{t_3^\star} (1 - \mu_{\ell}^B - \mu_{\ell}^C) \cdot\left[1-\frac{R'_{\ell-1}}{\MRttwo}\right]^+\,.
        \end{equation*}
        We are nearly done. First, let us see that for $\ell>t_3^\star$, $R'_{\ell-1}=0$, and thus, $\left[1-\frac{R'_{\ell-1}}{\MRttwo}\right]^+=1$, making:
        \begin{equation*}
            \sum_{\ell=t_3^\star+1}^{n} (1 - \mu_{\ell}^B - \mu_{\ell}^C) \cdot\left[1-\frac{R_{\ell-1}}{\MRttwo}\right]^+=n-t_3^\star\,.
        \end{equation*}
        Then, for $\ell\leq t_1$, we have that $R_{\ell-1}\geq \MRtone\geq \MRttwo$, and thus, 
        \begin{equation*}
            \sum_{\ell=1}^{t_1} (1 - \mu_{\ell}^B - \mu_{\ell}^C) \cdot\left[1-\frac{R_{\ell-1}}{\MRttwo}\right]^+=0\,.
        \end{equation*}
        Also $R_{\ell - 1} = R'_{\ell - 1}$ for all $\ell \geq t_1 + 1$.
        Thus, we have that
        \begin{align*}
            (n-t_3^\star) + \sum_{\ell=t_1+1}^{t_3^\star} (1 - \mu_{\ell}^B - \mu_{\ell}^C) \cdot\left[1-\frac{R'_{\ell-1}}{\MRttwo}\right]^+
            &= \sum_{\ell=1}^{n} (1 - \mu_{\ell}^B - \mu_{\ell}^C) \cdot\left[1-\frac{R_{\ell-1}}{\MRttwo}\right]^+
            \geq \sum_{\ell=1}^{n} \mu_{\ell}^D \cdot\left[1-\frac{R_{\ell-1}}{\MRttwo}\right]^+\,.
        \end{align*}

    \end{proof}

    Thus, we have proven all the different possible cases, and thus, the constraint $(m_i)$ holds.
\end{proof}

It remains to enforce the diagonal convention on the preliminary matrix $\bar s$. The following lemma shows that this can be done without changing any of the preceding calculations or the objective value.

\begin{lemma}[Diagonal repair]\label{lem:diagonal-repair}
Let $\bar s$ be the preliminary sent-credit matrix produced by Table~\ref{tab:general-case-procedure}. Define $s$ by
\begin{equation}\label{eq:diagonal-repair}
    s_{i,j} \coloneqq
    \begin{cases}
        0, & i=j,\\
        \bar s_{j+1,j}+\bar s_{j,j}, & i=j+1 \text{ and } j<n,\\
        \bar s_{i,j}, & \text{otherwise}.
    \end{cases}
\end{equation}
Replacing $\bar s$ by $s$ gives a feasible solution of the factor-revealing dual with the same value of $\rho$.
\end{lemma}
\begin{proof}
Only the $A$-assignment in Table~\ref{tab:general-case-procedure} can contribute to $\bar s_{j,j}$. Since $\mu_j^A=0$ for $j>t_1$, a positive diagonal entry can occur only for $j\leq t_1$. Moreover, $\bar s_{n,n}=0$: this is immediate if $t_1<n$, while if $t_1=n$, Definition~\ref{def:t1-boundary-split} gives $\alpha=1$ and hence $\mu_n^A=0$. Thus, every diagonal entry that may be positive can be shifted to the following row as in Equation~\eqref{eq:diagonal-repair}.

We first record that the case $t_1=t_3^\star<n$ is impossible. Indeed, the definition of $t_1$ and the initialization $R_{t_3^\star}^{(2)}=0$ give
\begin{equation*}
    \MRtone
    \leq R_{t_1-1}^{(2)}
    = R_{t_3^\star-1}^{(2)}
    =\hg_{t_3^\star}^{(2)}
    =(1-\mu_{t_3^\star}^D)\frac{\Mtone}{n-t_3^\star}
    \leq \frac{\Mtone}{n-t_3^\star}
    \leq \Mtone,
\end{equation*}
contradicting $\Mtone<\MRtone$ from Property~\ref{property:thresholds}.

We next check the pool constraints. For every $i$, all shifts from columns $j<i$ have already re-entered the pool by time $i$, whereas the mass $\bar s_{i,i}$ is delayed until time $i+1$. Consequently,
\begin{equation}\label{eq:repaired-pool-credit}
    \sum_{p=1}^i\sum_{j=1}^p(s_{p,j}-t_{p,j})
    =\text{pool\_credit}_i-\bar s_{i,i}.
\end{equation}
If $\bar s_{i,i}=0$, non-negativity follows from Lemma~\ref{lem:mi-constraints}. Suppose then that $\bar s_{i,i}>0$, so $i\leq t_1$. No $t$-mass has been withdrawn by time $i$ unless possibly $i=t_1=t_3^\star$. To see this, clients $j<t_1$ have only $A$- and $D$-membership; the $A$-assignment has no $t$-mass, while a positive $D$-membership implies $f_j\geq t_3^\star$. At $j=t_1$, the $A$-part again has no $t$-mass, the $D$-part can only withdraw at $f_{t_1}\geq t_3^\star$, and the $B$-part is zero because
\begin{equation*}
    [1-R_{t_1-1}-\Mtone]^+=0
\end{equation*}
by $R_{t_1-1}\geq\MRtone=1-\Mtone$. The exceptional case $i=t_1=t_3^\star<n$ was ruled out above, and if $i=t_1=t_3^\star=n$, then $\bar s_{i,i}=\bar s_{n,n}=0$. Hence, whenever $\bar s_{i,i}>0$, the preliminary pool credit at time $i$ is a sum of nonnegative $\bar s$-entries containing $\bar s_{i,i}$. Equation~\eqref{eq:repaired-pool-credit} is therefore nonnegative.

For a client constraint $(r_{j,i})$, when $i=j$ replacing $\bar s$ by $s$ increases its left-hand side by $\bar s_{j,j}$. When $i\geq j+1$, the cumulative transfer is unchanged because
\begin{equation*}
    \sum_{p=j}^i s_{p,j}=\sum_{p=j}^i \bar s_{p,j}.
\end{equation*}
Thus all client constraints remain satisfied. The same identity with $i=n$ shows that each $(d_j)$ constraint, and hence the value of $\rho$, is unchanged. Finally, $s$ is nonnegative, has zero diagonal by construction, and all constraints not involving $s$ are unaffected. This proves the lemma.
\end{proof}

The proof of Lemma~\ref{lem:general-case-procedure} now follows:

\begin{proof}[Proof of Lemma~\ref{lem:general-case-procedure}]
    Lemmas \ref{lem:f-constraint}, \ref{lem:mi-constraints}, \ref{lem:d-i-constraints}, \ref{lem:rij-constraints}, \ref{lem:gamma-beta-constraints}, and \ref{lem:nonneg} show that the preliminary assignment satisfies all constraints except possibly $\bar s_{i,i}=0$, with objective value $\Mrho$. Lemma~\ref{lem:diagonal-repair} transforms it into a feasible solution satisfying the diagonal convention and preserves the objective value. Thus the lemma holds.
\end{proof}

\newpage
\subsection{Existence of correct advice}\label{sec:good-aux-existence}

This section is the proof of Lemma~\ref{lem:good-tau-existence}. We must show that there exists a set of advice variables $\Gamma$ as defined in Definition~\ref{def:auxiliary-input} and Property~\ref{property:admissible-frontier-input} such that the outcome of the procedure is not a null result. 
It will not be a null result if and only if $t_1(\Gamma)$ is defined ($t_1(\Gamma)\neq \perp$) and $R_0(\Gamma)=1$ for the chosen variables, where $R_0(\Gamma)$ is the value of $R_0$ at the termination of the procedure when $\Gamma$ are the input variables.

First, let us see that the condition that $t_1(\Gamma)\neq \perp$ is redundant. In the procedure, the condition of $t_1=\perp$ in Line~\ref{line:t_1-perp-condition} means that the procedure got to the $i=1$ element in the Phase 2 loop in Line~\ref{line:phase-2-loop} and from Line~\ref{line:R_i-1-smaller-0.79-condition} we know that no $i$ had $R_{i-1}^{(2)}\ge \MRtone$. This last condition is also true for $i=1$, thus $R_0^{(2)}< \MRtone$. Thus, getting some variables that have $R_0=1$ is already a sufficient and necessary condition.

Our proof strategy will be the following. We will first find two different advices $\Gamma_0$ and $\Gamma_1$ such that one satisfies $R_0(\Gamma_0)\leq 1$ and the other $R_0(\Gamma_1)\geq 1$. Then, we will define a path $\Gamma(\lambda)$ that goes from the first set of variables $\Gamma_0$ to the other $\Gamma_1$, parameterized by some real continuous parameter $\lambda$. By proving that the resulting function $R_0(\Gamma(\lambda))$ is continuous, with a direct application of the intermediate value theorem, we will be able to state that there must exist a $\Gamma^\star$ such that $R_0(\Gamma^\star)=1$, implying that there exists a correct advice.

Thus, let us first define the following two points.
\begin{definition}
    Let the first point $\Gamma_0$ be one such that all clients are considered part of $D$:
    \begin{equation*}
        \Gamma_0 = \{t_3^\star=1, \{\mu_j^\star=1, \ \forall j \in \ones{1}\}\}\,.
    \end{equation*}
\end{definition}
\begin{definition}
    Let the second point $\Gamma_1$ be one such that all clients with $f_j<n$ are not considered part of $D$:
    \begin{equation*}
        \Gamma_1 = \{t_3^\star = n, \{\mu_j^\star=1, \ \forall j \in \ones{n}\}, \ \hg_n^\star=1\}\,.
    \end{equation*}
\end{definition}

For these two points, let us state the following observations:
\begin{observation}\label{obs:oracle-R0-gamma-0}
    The total sum of produced credit for $\Gamma_0$ is $R_0(\Gamma_0)\leq 1$.
\end{observation}
\begin{observation}\label{obs:oracle-R0-gamma-1}
    The total sum of produced credit for $\Gamma_1$ is $R_0(\Gamma_1)\geq 1$.
\end{observation}
\begin{proof}[Proof of Observation~\ref{obs:oracle-R0-gamma-0} and Observation~\ref{obs:oracle-R0-gamma-1}]
    For $\Gamma_0$, since $t_3^\star=1$, we have $R_1(\Gamma_0)=0$, and thus $R_0(\Gamma_0)=\hg_1$. For this case, since $\mu_1^D=1$ (because either $1\in\ones{1}$ implying $\mu_1^\star=1$ or $f_1>t_3^\star=1$ and thus $\mu_1^D=1$), we have that $\hg_1(\Gamma_0)\propto (1-\mu_1^D)=0$, thus proving that $R_0(\Gamma_0)\leq 1$. For $\Gamma_1$, it is trivial to see as $R_0(\Gamma_1) \geq \hg_n(\Gamma_1) = \hg_n^\star = 1$.
\end{proof}

Let us now define the path that goes from one point to the other. This path must go through each value of $t_3^\star$ from $1$ to $n$, and within each $t_3^\star$, we have to continuously change the values of $\mu_j^D$. 

Let us define the section of the path for a fixed $t_3^\star$:
\begin{definition}
    For a fixed $t_3^\star=\ell < n$, let $\Gamma^{\ell}(\lambda)$ be the path defined with $\lambda\in[0,1]$ as:
    \begin{equation*}
        \Gamma^{\ell}(\lambda) = \{t_3^\star=\ell, \{\mu_j^\star=\left[1-\left[n\lambda - (n-j)\right]^+\right]^+, \ \forall j\in \ones{\ell} \}\}\,.
    \end{equation*}
    For $t_3^\star=\ell=n$, this is:
    \begin{equation*}
        \Gamma^{n}(\lambda) = \{t_3^\star=n, \{\mu_j^\star=1, \ \forall j\in \ones{n} \}, \hg_n^\star = \lambda\}\,.
    \end{equation*}
\end{definition}
And finally, we can define the full path going from $\Gamma_0$ to $\Gamma_1$ as:
\begin{definition}
    For the above defined points $\Gamma_0$ and $\Gamma_1$, let the path $\Gamma(\lambda)$ that joins these points for $\lambda\in[0,n]$ be:
    \begin{equation*}
        \Gamma(\lambda) = \Gamma^{\left\lfloor\lambda\right\rfloor + 1}\left(\lambda - \left\lfloor\lambda\right\rfloor\right)
    \end{equation*}
    and with $\Gamma(0) = \Gamma^1(0)$ and $\Gamma(n) = \Gamma^n(1)$.
\end{definition}

As we need the advice variables to satisfy Property~\ref{property:admissible-frontier-input}, we make the following observation:
\begin{observation}\label{obs:oracle-valid-input}
    For any $\lambda \in [0,n]$, the chosen $\mu_j^\star\in\Gamma(\lambda)$ values are non-increasing in $j$ and have at most one fractional entry, and if $t_3^\star=n$ we have $\mu_j^\star=1$.
\end{observation}
\begin{proof}[Proof of Observation~\ref{obs:oracle-valid-input}]
    First, for the case where $t_3^\star = n$, the path $\Gamma^n(\lambda)$ explicitly sets $\mu_j^\star = 1$ for all $j \in \ones{n}$, satisfying the final property directly.

    For the case $t_3^\star = \ell < n$, the fractional memberships are defined as $\mu_j^\star = \left[1-\left[n\lambda - (n-j)\right]^+\right]^+$. Let $x_j = n\lambda - n + j$. As the index $j$ increases, $x_j$ strictly increases. Consequently, the inner term $[x_j]^+$ is non-decreasing with respect to $j$, which implies that $1 - [x_j]^+$ is non-increasing. Applying the outer non-negativity bound $[\cdot]^+$ preserves this monotonicity, meaning $\mu_j^\star$ is non-increasing in $j$.

    Finally, an entry evaluates to a fractional value (i.e., $0 < \mu_j^\star < 1$) if and only if $0 < [x_j]^+ < 1$, which strictly requires $0 < x_j < 1$. Because distinct indices differ by at least $1$, the distance between terms is at least $x_{j_2} - x_{j_1} = j_2 - j_1 \geq 1$. Since the open interval $(0, 1)$ has a length of exactly $1$, it can capture at most one such $x_j$ for any given $\lambda$. Therefore, there can be at most one fractional entry for any state in the path.
\end{proof}
Thus, using any point $\Gamma$ in the path is a valid input to the procedure.
We will now prove that while traveling through this path, the produced $R_0(\lambda)$ is continuous. First, we prove the continuity for each section $\Gamma^\ell(\lambda)$; this is the most technical proof in this section and we leave it to the end of the section.

\begin{lemma}\label{lem:oracle-continuous-section}
    Any given section $\Gamma^{\ell}(\lambda)$ of the path $\Gamma(\lambda)$ has a continuous value of $R_0(\Gamma^{\ell}(\lambda))$ with respect to $\lambda$.
\end{lemma}

Then, we need to stitch together the paths to prove that the general path is continuous:

\begin{observation}\label{obs:oracle-path-stitched}
    For $1 \leq \ell<n$, the jump from section $\ell$ to section $\ell+1$ is continuous:
    \begin{equation*}
        R_0(\Gamma^{\ell}(1)) = R_0(\Gamma^{\ell+1}(0))\,.
    \end{equation*}
\end{observation}
\begin{proof}[Proof of Observation~\ref{obs:oracle-path-stitched}]
    This is a direct proof by substitution of the value of $\lambda$ and analyzing the execution of the procedure for both points. 
    
    For the first point $\Gamma^\ell(1)$, we substitute $\lambda=1$ into the path definition. For all $j \in \ones{\ell}$, we have:
    \begin{equation*}
        \mu_j^\star = \left[1 - \left[n(1) - (n-j)\right]^+\right]^+ = \left[1 - j\right]^+ = 0\,.
    \end{equation*}
    Thus, for all $j$ with $f_j\leq \ell$ we have $\mu_j^D=0$; and for all $j$ with $f_j>\ell$ we have $\mu_j^D=1$.
    
    For the second point $\Gamma^{\ell+1}(0)$, we substitute $\lambda=0$ into the path definition. For all $j \in \ones{\ell+1}$, we have:
    \begin{equation*}
        \mu_j^\star = \left[1 - \left[n(0) - (n-j)\right]^+\right]^+ = \left[1 - 0\right]^+ = 1\,.
    \end{equation*}
    This is true even in the case $\ell+1=n$, for which we also have that $\hg_n^\star=0$. Thus, just the same as before, for all $j$ with $f_j\leq \ell$ we have $\mu_j^D=0$; and for all $j$ with $f_j>\ell$ we have $\mu_j^D=1$.
    
    During the execution for $\Gamma^\ell(1)$, we have $t_3^\star = \ell$. In the initialization (Phase 3), the procedure sets $\hg_{\ell+1} = 0$, $\mu_{\ell+1}^D=1$, and $R_\ell = 0$. Then, entering Phase 2 at $i = \ell$, it calculates:
    \begin{equation*}
        \hg_\ell = (1 - \mu_\ell^D) \cdot \frac{R_\ell + \Mtone}{n - \ell} = (1 - \mu_\ell^D) \cdot \frac{\Mtone}{n - \ell}
    \end{equation*}
    and it computes the values of $\hg_i$ for $i<\ell$ based on this.
    
    During the execution for $\Gamma^{\ell+1}(0)$, we have $t_3^\star = \ell+1$. The procedure enters Phase 2 directly at $i = \ell+1$, setting $R_{\ell+1}=0$. If $\ell+1=n$, then by definition of $\Gamma^{n}(0)$ we have that $\hg_n=\hg_n^\star=0$. Otherwise, using $\mu_{\ell+1}^D = 1$ because either $f_{\ell+1}>t_3^\star$ or $f_{\ell+1}=t_3^\star$ and thus $\mu_{\ell+1}^\star=1$, it calculates:
    \begin{equation*}
        \hg_{\ell+1} = (1 - 1) \cdot \frac{R_{\ell+1} + \Mtone}{n - (\ell+1)} = 0\,.
    \end{equation*}
    This yields $R_\ell = R_{\ell+1} + \hg_{\ell+1} = 0$. Proceeding to iteration $i = \ell$, the variable $\mu_\ell^D$ remains the same as for $\Gamma^{\ell}(1)$ by the previous observation. Therefore, the calculation evaluates exactly as before:
    \begin{equation*}
        \hg_\ell = (1 - \mu_\ell^D) \cdot \frac{R_\ell + \Mtone}{n - \ell} = (1 - \mu_\ell^D) \cdot\frac{\Mtone}{n - \ell}\,.
    \end{equation*}
    
    Thus, at the end of iteration $i=\ell$, both executions yield $R_\ell = 0$ and $\hg_\ell = (1 - \mu_\ell^D) \cdot \frac{\Mtone}{n-\ell}$. Since the remaining iterations $i < \ell$ will evaluate identically for both states as all values of $\mu_j^D$ are the same, the final generated credit $R_0$ must be identical, which finishes the proof.
\end{proof}

Finally, we prove that the start and end points correspond to the original ones we had before.

\begin{observation}\label{obs:oracle-path-endpoints}
    We have that the two extreme points of $\Gamma(\lambda)$ are $\Gamma_0$ and $\Gamma_1$:
    \begin{equation*}
        \Gamma(0) = \Gamma_0, \ \ \Gamma(n)= \Gamma_1\,.
    \end{equation*}
\end{observation}
\begin{proof}[Proof of Observation~\ref{obs:oracle-path-endpoints}]
    This is a direct proof by substitution of the value of $\lambda$. For $\lambda=0$, we have:
    \begin{equation*}
        \Gamma(0) = \Gamma^1\left(0\right) = \{t_3^\star=1, \{\mu_j^\star=\left[1-\left[0 - (n-j)\right]^+\right]^+, \ \forall j\in \ones{1} \}\}
    \end{equation*}
    which is clearly $\Gamma_0$ as the middle term is equal to $1$. For $\lambda=n$, we have:
    \begin{equation*}
        \Gamma(n) = \Gamma^n(1) = \{t_3^\star=n, \{\mu_j^\star=1, \ \forall j\in \ones{n} \}, \hg_n^\star = 1\} = \Gamma_1
    \end{equation*}
    which finishes the proof.
\end{proof}

\begin{corollary}\label{cor:oracle-continuity-of-global-path}
    The path $\Gamma(\lambda)$ goes from $\Gamma_0$ (for $\lambda=0$) to $\Gamma_1$ (for $\lambda=n$) and satisfies that $R_0(\Gamma(\lambda))$ is continuous with respect to $\lambda$. 
\end{corollary}
\begin{proof}[Proof of Corollary~\ref{cor:oracle-continuity-of-global-path}]
    First, from Lemma~\ref{lem:oracle-continuous-section} and Observation~\ref{obs:oracle-path-stitched} we know that we have a path $\Gamma(\lambda)$ with a continuous $R_0$ with respect to $\lambda$. Using Observation~\ref{obs:oracle-path-endpoints}, we know that this goes from $\Gamma_0$ to $\Gamma_1$, thus proving the corollary.
\end{proof}

Thus, we are now ready for the proof of the main section's lemma.

\begin{proof}[Proof of Lemma~\ref{lem:good-tau-existence}]
    As stated before, we just want to prove that there is a point $\Gamma^\star$ such that $R_0(\Gamma)=1$. By the continuity of $R_0(\Gamma(\lambda))$ proven in Corollary~\ref{cor:oracle-continuity-of-global-path}, we can apply the intermediate value theorem using Observation~\ref{obs:oracle-R0-gamma-0} and Observation~\ref{obs:oracle-R0-gamma-1} to get that there must be such a point. Thus, the proof is complete.
\end{proof}

\vspace{0.25in}

We end this section with proving the deferred proof of Lemma~\ref{lem:oracle-continuous-section}.

\begin{proof}[Proof of Lemma~\ref{lem:oracle-continuous-section}]
    For $\ell < n$, the path $\Gamma^\ell(\lambda)$ defines a continuous progression where the fractional membership variables $\mu_j^\star$ are sequentially decreased from $1$ to $0$.
    For $\ell = n$, the path $\Gamma^n(\lambda)$ fixes $\mu_j^\star = 1$ for all $j \in \ones{n}$ and continuously varies $\hg_n^\star = \lambda$ from $0$ to $1$.
    In both cases, we must show that the bookkeeping value $R_0(\Gamma^\ell(\lambda))$ generated by the execution of Algorithm~\ref{alg:general-case-procedure-hgi} is continuous with respect to $\lambda$.

    Let $R_i^{(2)}(\lambda)$ be the sequence of budget variables generated by purely executing Phase 2.
    Initialized at $R_{\ell}^{(2)} = 0$, for $i \le \ell$:
    \begin{align*}
        \hg_i^{(2)}(\lambda) &= \begin{cases}
            \hg_n^\star(\lambda) & \text{if } i = \ell = n\,, \\
            (1 - \mu_i^D(\lambda)) \cdot \frac{R_i^{(2)}(\lambda) + \Mtone}{n - i} & \text{otherwise}\,.
        \end{cases} \\
        R_{i-1}^{(2)}(\lambda) &= R_i^{(2)}(\lambda) + \hg_i^{(2)}(\lambda)
    \end{align*}
    For this quantity that follows the above formulas, we can use the following claim:
    \begin{claim}\label{claim:oracle-recurrence-R-phase-2}
        Each term $R_i^{(2)}$ can be expressed as:
        \begin{equation*}
            R_{i}^{(2)}(\lambda) = (\hg_n(\lambda) +\Mtone)\cdot\prod_{j=i+1}^{n-1}\left(1+\frac{1-\mu_j^D(\lambda)}{n-j}\right) - \Mtone
        \end{equation*}
        except for $R_n^{(2)}=0$ and where we remind that $\hg_n(\lambda)=0$ if $t_3^\star<n$.
    \end{claim}
    \begin{proof}[Proof of Claim~\ref{claim:oracle-recurrence-R-phase-2}]
        We can prove this by expanding the recurrence relation defined by the Phase 2 updates. For any index $k < n$, the algorithm computes the budget as:
        \begin{equation*}
            R_{k-1}^{(2)} = R_k^{(2)} + \hg_k^{(2)} = R_k^{(2)} + (1 - \mu_k^D) \cdot \frac{R_k^{(2)} + \Mtone}{n - k}\,.
        \end{equation*}
        Adding $\Mtone$ to both sides and factoring yields:
        \begin{equation*}
            R_{k-1}^{(2)} + \Mtone = (R_k^{(2)} + \Mtone) \left( 1 + \frac{1 - \mu_k^D}{n - k} \right)\,.
        \end{equation*}
        By substituting $k$ with indices from $i+1$ up to $n-1$ and repeatedly applying this relation, we can unroll the recurrence into a product:
        \begin{equation*}
            R_{i}^{(2)} + \Mtone = (R_{n-1}^{(2)} + \Mtone) \prod_{j=i+1}^{n-1} \left( 1 + \frac{1 - \mu_j^D}{n - j} \right)\,.
        \end{equation*}
        Now, we evaluate the base term $R_{n-1}^{(2)}$ by splitting into two cases based on $t_3^\star = \ell$:
        
        \textbf{Case 1: $\ell = n$.} The sequence is initialized at $R_n^{(2)} = 0$, and the generation function sets $\hg_n^{(2)} = \hg_n^\star = \hg_n$. Therefore, $R_{n-1}^{(2)} = R_n^{(2)} + \hg_n^{(2)} = \hg_n$. Substituting $R_{n-1}^{(2)} = \hg_n$ into our product gives exactly the claimed formula.
        
        \textbf{Case 2: $\ell < n$.} The sequence is initialized at $R_\ell^{(2)} = 0$, and by definition, $\hg_n = 0$. Let us evaluate the right-hand side of the claim's formula at $i = \ell$. For all $j > \ell$, the path definition ensures $\mu_j^D = 1$, meaning $1 - \mu_j^D = 0$. Consequently, every multiplier in the product for $j \in (\ell, n-1]$ evaluates exactly to $1$. Using the claim's formula for $i=\ell$ with $\hg_n=0$, we get:
        \begin{equation*}
            R_\ell^{(2)} + \Mtone = (0 + \Mtone) \cdot 1 \implies R_\ell^{(2)} = 0\,.
        \end{equation*}
        This matches the initialization $R_\ell^{(2)} = 0$ perfectly. Because the structural zeros in $1-\mu_j^D$ for $j > \ell$ neutralize the product, the sequence seamlessly backward-propagates from an implicit $R_{n-1}^{(2)} = \hg_n = 0$. 
        
        Thus, for all cases, subtracting $\Mtone$ from both sides of the unrolled recurrence gives the exact mathematical expression in the claim.
    \end{proof}

    We can make the two following observations on this quantity using this claim:
    \begin{observation}\label{obs:oracle-continuous-phase-2}
        The quantity $R_{i}^{(2)}(\lambda)$ is continuous with respect to $\lambda$.
    \end{observation}
    \begin{observation}\label{obs:oracle-non-decreasing-phase-2}
        The quantity $R_{i}^{(2)}(\lambda)$ is non-decreasing with respect to $\lambda$.
    \end{observation}
    \begin{proof}[Proof of Observation~\ref{obs:oracle-continuous-phase-2}]
        Looking at the expression in Claim~\ref{claim:oracle-recurrence-R-phase-2}, since for every index in the product $j$ we have $j\leq n-1$, we note that all denominators in the claim are positive. Thus, since there are no discontinuities with respect to $\mu_j^D$ or with respect to $\hg_n(\lambda)$, we can state that $R_i^{(2)}(\lambda)$ is continuous with respect to any continuous change in $\mu_j^D$ and $\hg_n$ for any $i$ and $j$; and thus, it is also continuous with respect to $\lambda$ as it continuously changes $\mu_j^D$ or $\hg_n$ depending on $\ell$. 
    \end{proof}
    \begin{proof}[Proof of Observation~\ref{obs:oracle-non-decreasing-phase-2}]
        On the other hand, we can also prove that $R_i^{(2)}(\lambda)$ is not only continuous with respect to $\lambda$ but also that it is non-decreasing. 
        For $\ell=n$ this is clear, as $\mu_j^D$ do not change with respect to $\lambda$, and only $\hg_n^\star=\lambda$ changes, which leaves $R_i^{(2)}=c\cdot (\lambda+\Mtone)-\Mtone$ that is increasing with respect to $\lambda$. 
        On the other hand, for $\ell<n$, we have that $\mu_j^D$ are non-increasing with respect to $\lambda$ as they are either fixed or equal to $\mu_j^\star=\left[1-\left[n\lambda - (n-j)\right]^+\right]^+$ which is indeed non-increasing on lambda. 
        The expression in Claim~\ref{claim:oracle-recurrence-R-phase-2} is non-increasing in $\mu_j^D$ (since $\mu_j^D\leq 1$); and thus, the whole thing is non-decreasing in $\lambda$.
    \end{proof}
    
    We now make the case analysis in the proof explicit.

    \smallskip
    \noindent\textbf{Case 0: $R_0^{(2)}(\lambda)<\MRtone$.}
    From Observation~\ref{obs:oracle-continuous-phase-2}, any index $i$ purely in Phase 2 has a continuous $R_{i-1}(\lambda)=R_{i-1}^{(2)}(\lambda)$ with respect to $\lambda$. From this first analysis, this also implies that in the case that $R_0^{(2)}<\MRtone$, then no $t_1$ exists, so only Phase 2 would be run and thus $R_0=R_0^{(2)}$. As we know that any $R_i^{(2)}$ is continuous, this already proves it for the case $R_0^{(2)}<\MRtone$. We assume from now on that $t_1$ does exist.

    Before analyzing Phase 1 assignments and the transition, we need to consider the variable $z$.
    The algorithm defines $t_2 := \max\Set{i \mid R_{i-1}^{(2)} \geq \MRttwo}$ and sets $z\coloneqq \sum_{m=t_2+1}^n [1-R_{m-1}^{(2)}/\MRttwo]^+\mu_m^D$.
    Because $R_{m-1}^{(2)} \ge \MRttwo$ for all $m \le t_2$, we have $\left[1 - \frac{R_{m-1}^{(2)}}{\MRttwo}\right]^+ = 0$.
    Thus, $z$ can be identically written as $z = \sum_{m=1}^n \left[1 - \frac{R_{m-1}^{(2)}}{\MRttwo}\right]^+ \cdot \mu_m^D$ as $t_2\geq 1$.
    The discrete jump associated with $t_2$ changing only adds or removes terms evaluating to $0$, making $z$ continuous with respect to the sequence $R_i^{(2)}$ and thus with respect to $\lambda$ by what we proved earlier.

    \smallskip
    \noindent\textbf{Case 1: $t_1<n$.}
    Next, we analyze the transition from Phase 2 to Phase 1, for now assuming that $t_1<n$. The algorithm defines $t_1$ as the largest index $i$ such that $R_{i-1}^{(2)} = R_i^{(2)}+\hg_i^{(2)} \ge \MRtone$, which implies that it has $R_{i}^{(2)}<\MRtone$, otherwise $t_1=i+1$. For a given $\lambda$ that has a specific value of $t_1(\lambda)$, we need to look at the right and left limits $\lambda^+=\lambda+\varepsilon$ and $\lambda^-=\lambda-\varepsilon$. By Observation~\ref{obs:oracle-non-decreasing-phase-2}, we know that $R_{i-1}^{(2)}(\lambda^-)\leq R_{i-1}^{(2)}(\lambda)\leq R_{i-1}^{(2)}(\lambda^+)$ for $i=t_1(\lambda)$.

    First, let us analyze the case of $\lambda$. For its $t_1\coloneqq t_1(\lambda)$, from Definition~\ref{def:t1-boundary-split} we have that the interpolation factor is $\alpha(\lambda) = \frac{\MRtone-R_{t_1}^{(2)}(\lambda)}{\hg_{t_1}^{(2)}(\lambda)}$, and that its fractional updates are: $\hg_{t_1}^B = \MRtone - R_{t_1}^{(2)}$ and $\hg_{t_1}^A = (1 - \alpha(\lambda))(1 - \mu_{t_1}^D) \cdot \Phi_{t_1}(\alpha(\lambda),\lambda)$, where $\Phi_i(\alpha(\lambda),\lambda)$ is the transition to Phase 1 generation function: 
    $\Phi_i(\alpha,\lambda) = \frac{1-[\MRttwo-R_{f_i}^{(2)}(\lambda)]^+}{n-i+\alpha\cdot(1-\mu_i^D(\lambda))-\Mfrac\cdot z(\lambda)}$, where the proof for using $R_{f_i}^{(2)}$ instead of $R_{f_i}$ is given in the next paragraph. 
    Thus, adding $R_{t_1}$ with the updates in Definition~\ref{def:t1-boundary-split}, the value for $R_{t_1-1}$ is:
    \begin{equation*}
        R_{t_1-1}(\lambda) = \MRtone + (1 - \alpha(\lambda))(1 - \mu_{t_1}^D(\lambda)) \cdot \Phi_{t_1}(\alpha(\lambda),\lambda)\,.
    \end{equation*}
    Some important observations about $\Phi_i(\alpha,\lambda)$ are the following. Note that the denominator is always positive for $i\leq t_1<n$ since $z\leq n-t_2 \leq n-i$ for $i\leq t_1\leq t_2$ and $\alpha \in [0,1]$, and that the function $\Phi_i(\alpha,\lambda)$ is both continuous in $\alpha$ and on $\lambda$, the latter one being because $R_{f_i}^{(2)}(\lambda)$, $\mu_i^D$ and $z(\lambda)$ are all continuous in $\lambda$. 
    Also note that we are using $R_{f_i}^{(2)}$ instead of $R_{f_i}$ because $R_{f_i}=R_{f_i}^{(2)}$ if $R_{f_i}\leq \MRttwo$ as $f_i$ would still be in Phase 2, and if they are different then $f_i$ would be fully in phase $1$, but it would already imply that both $R_{f_i}^{(2)}\geq \MRttwo$ and $R_{f_i}\geq \MRttwo$, thus $[\MRttwo-R_{f_i}^{(2)}(\lambda)]^+=[\MRttwo-R_{f_i}(\lambda)]^+=0$.
    Thus, proving that $\alpha(\lambda)$ is continuous in $\lambda$ proves that $\Phi_i(\alpha(\lambda),\lambda)$ is also continuous in $\lambda$.
    
    \smallskip
    \noindent\textbf{Case 1.1: Right limit $\lambda^+$.}
    Let us then look at the right limit $\lambda^+$, for a sufficiently small $\varepsilon$ such that $R_{t_1(\lambda)}^{(2)}(\lambda^+)<\MRtone$ is still true. This can be done thanks to the continuity from Observation~\ref{obs:oracle-continuous-phase-2}. By the non-decreasingness of $R_i^{(2)}(\lambda)$ from Observation~\ref{obs:oracle-non-decreasing-phase-2}, we know that $R_{t_1(\lambda)-1}^{(2)}(\lambda^+) \geq R_{t_1(\lambda)-1}^{(2)}(\lambda)\geq \MRtone$; implying that $\hg_{t_1}(\lambda^+)>0$ and that $t_1(\lambda^+)=t_1(\lambda)= t_1$.
    The interpolation factor is $\alpha(\lambda^+) = \frac{\MRtone-R_{t_1}^{(2)}(\lambda^+)}{\hg_{t_1}^{(2)}(\lambda^+)}$.
    The fractional updates are $\hg_{t_1}^B(\lambda^+) = \MRtone - R_{t_1}^{(2)}(\lambda^+)$ and $\hg_{t_1}^A(\lambda^+) = (1 - \alpha(\lambda^+))(1 - \mu_{t_1}^D(\lambda^+) ) \cdot \Phi_{t_1}(\alpha(\lambda^+), (\lambda^+))$.
    
    As $\varepsilon\to 0$ meaning $\lambda^+\to\lambda$, we have that $R_{{t_1}-1}^{(2)}(\lambda^+) \to R_{{t_1}-1}^{(2)}(\lambda)$ and thus $\hg_{t_1}^{(2)}(\lambda^+)\to \hg_{t_1}^{(2)}(\lambda)> 0$ by Observation~\ref{obs:oracle-continuous-phase-2}. Thus, $\alpha(\lambda^+) \to \frac{\MRtone-R_{t_1}^{(2)}(\lambda)}{\hg_{t_1}^{(2)}(\lambda)} = \alpha(\lambda)$ implying that $\alpha$ is right-continuous with respect to $\lambda$. Thus, the budget becomes:
    \begin{align*}
        R_{t_1-1}(\lambda^+)&=\MRtone + (1 - \alpha(\lambda^+))(1 - \mu_{t_1}^D(\lambda^+)) \cdot \Phi_{t_1}(\alpha(\lambda^+),\lambda^+) \\
        &\to  \MRtone + (1 - \alpha(\lambda))(1 - \mu_{t_1}^D(\lambda)) \cdot \Phi_{t_1}(\alpha(\lambda),\lambda) = R_{t_1-1}(\lambda)\,,
    \end{align*}
    where we used the continuity of $\Phi_i(\alpha,\lambda)$ with respect to $\alpha$ and $\lambda$. Having proven that the  limit is continuous even at $t_1$, now what is left for this part is to go up to $R_0$. Notice that the assignment formulas for indices at Phase 1 are $\hg_i(\lambda) = (1-\mu_i^D(\lambda))\cdot \Phi_i(0,\lambda)$, one can check that such a statement is correct by substitution and comparison. Since we have already proven that $\Phi_i(0,\lambda)$ is continuous in $\lambda$ and $\alpha$, and here we fixed $\alpha$ to 0, it implies that $\hg_i(\lambda^+)\to (1-\mu_i^D(\lambda))\cdot \Phi_i(0,\lambda)=\hg_i(\lambda)$. Thus, the value for $R_0(\lambda^+)$ is:
    \begin{equation*}
        R_0(\lambda^+) = R_{t_1-1}(\lambda^+) + \sum_{i=1}^{t_1-1}\hg_i(\lambda^+) \to R_{t_1-1}(\lambda) + \sum_{i=1}^{t_1-1}\hg_i(\lambda) = R_0(\lambda)\,,
    \end{equation*}
    which implies that the right limit is indeed equal to the value at $\lambda$.

    \smallskip
    \noindent\textbf{Case 1.2: Left limit $\lambda^-$.}
    Let us now look at the left limit. We split the analysis into two cases, first $R_{t_1(\lambda)-1}^{(2)}>\MRtone$ and then $R_{t_1(\lambda)-1}^{(2)}=\MRtone$. 

    \smallskip
    \noindent\textbf{Case 1.2.a: $R_{t_1(\lambda)-1}^{(2)}>\MRtone$.}
    We first take $R_{t_1(\lambda)-1}^{(2)}>\MRtone$, and do the analysis for a sufficiently small $\varepsilon$ such that $R_{t_1(\lambda)-1}^{(2)}(\lambda^-)\geq \MRtone$ is still true. This can be done thanks to the continuity from Observation~\ref{obs:oracle-continuous-phase-2}. 
    By the non-decreasingness of $R_i^{(2)}(\lambda)$ from Observation~\ref{obs:oracle-non-decreasing-phase-2}, we know that $R_{t_1(\lambda)}^{(2)}(\lambda^-) \leq R_{t_1(\lambda)}^{(2)}(\lambda)<\MRtone$; and, thus, $t_1(\lambda^-)=t_1(\lambda)$. We can see that for this case we can apply exactly the same analysis as the previous right limit; so it can be proven that the left limit is equal to the value at $\lambda$, for the case $R_{t_1(\lambda)-1}^{(2)}>\MRtone$.

    \smallskip
    \noindent\textbf{Case 1.2.b: $R_{t_1(\lambda)-1}^{(2)}=\MRtone$.}
    For the next case, assuming that $R_{t_1(\lambda)-1}^{(2)}(\lambda)=\MRtone$ exactly, if we decrease $\lambda$ by $\varepsilon$, then it could be that we make $R_{t_1(\lambda)-1}^{(2)}(\lambda^-)<\MRtone$ (by Observation~\ref{obs:oracle-non-decreasing-phase-2}), implying that $t_1$ would have to decrease. Thus, this case needs more attention. 
    First, let us assume this does indeed happen, that for all $\varepsilon \in (0,\delta]$ for some arbitrarily small $\delta$ we already have that $R_{t_1(\lambda)-1}^{(2)}(\lambda-\varepsilon)<\MRtone$. If we were to assume the opposite, then the previous analysis would already work.

    Let us find what the value of $t_1^-\coloneqq t_1(\lambda-\varepsilon)$ could be for small $\varepsilon$. Let $k\coloneqq \max\{j<t_1:\mu_j^D(\lambda)<1\}$, if this set is nonempty. Using the product formula for $R_{t_1-1}^{(2)}$, the only $\mu$-entry that changes when moving left and affects $R_{t_1-1}^{(2)}$ lies at an index $\geq t_1$, so all $\mu_j^D$ with $j<t_1$ must be locally constant since $R_{t_1-1}^{(2)}$ does decrease and there can be only one changing value by Observation~\ref{obs:oracle-valid-input}. Hence $k$ is fixed for all small perturbations. Also, for $r\in(k,t_1)$, $\mu_r^D=1$, so $\hg_r^{(2)}=0$, hence $R_k^{(2)}(\lambda')=R_{t_1-1}^{(2)}(\lambda')$ for $\lambda'$ near $\lambda$. Since $R_{t_1-1}^{(2)}(\lambda')\to {\MRtone}$ and $\hg_k^{(2)}>0$ for $\lambda'\to\lambda$, for all sufficiently small $\varepsilon>0$, $0<{\MRtone}-R_{t_1-1}^{(2)}(\lambda-\varepsilon)<\hg_k^{(2)}(\lambda-\varepsilon)$. Therefore $R_k^{(2)}(\lambda-\varepsilon)<{\MRtone}\leq R_{k-1}^{(2)}(\lambda-\varepsilon)$, so $t_1^-=t_1(\lambda-\varepsilon)=k$ for all small enough $\varepsilon$.
    
    \smallskip
    \noindent\textbf{Case 1.2.b.i: $t_1^-\neq \perp$.}
    Then, assuming that we are able to find such a value, i.e., $t_1^-\neq \perp$, we have $R_{t_1^- -1}^{(2)}(\lambda^-)\geq \MRtone$ and $R_{t_1^-}^{(2)}(\lambda^-) < \MRtone$. Because $\hg_\ell = 0$ for all intermediate indices $\ell \in (t_1^-, t_1)$ at both $\lambda$ and $\lambda^-$, the budget remains constant across this gap, meaning $R_{t_1^-}^{(2)}(\lambda^-) = R_{t_1-1}^{(2)}(\lambda^-)$. 
    For this index $t_1^-$, the interpolation factor is $\alpha(\lambda^-) = \frac{\MRtone - R_{t_1^-}^{(2)}(\lambda^-)}{\hg_{t_1^-}^{(2)}(\lambda^-)}$. Let us compute its limit as $\varepsilon \to 0$ (so $\lambda^- \to \lambda$).
    First, the numerator approaches $\MRtone - R_{t_1-1}^{(2)}(\lambda)$. Since we are specifically in the case where $R_{t_1-1}^{(2)}(\lambda) = \MRtone$, the numerator converges to $0$. 
    Second, the denominator $\hg_{t_1^-}^{(2)}(\lambda^-)$ converges to $\hg_{t_1^-}^{(2)}(\lambda) = (1 - \mu_{t_1^-}^D(\lambda)) \frac{R_{t_1^-}^{(2)}(\lambda) + \Mtone}{n - t_1^-}$. Because $\mu_{t_1^-}^D(\lambda) < 1$ and $R_{t_1^-}^{(2)}(\lambda) = \MRtone$, this denominator strictly converges to a positive value. 
    Thus, evaluating the limit yields $\alpha(\lambda^-) \to 0$.

    This implies that the values for the updates are $\hg_{{t_1^-}}^B(\lambda^-) = \MRtone - R_{{t_1^-}}^{(2)}(\lambda^-)$ and $\hg_{{t_1^-}}^A (\lambda^-)= (1 - \alpha(\lambda^-))(1 - \mu_{{t_1^-}}^D(\lambda^-)) \cdot \Phi_{{t_1^-}}(\alpha(\lambda^-),\lambda^-)$. As $\lambda^-\to\lambda$, we have that $\hg_{{t_1^-}}^B(\lambda^-)\to \MRtone - R_{{t_1^-}}^{(2)}(\lambda)=0$, and that $\hg_{{t_1^-}}^A (\lambda^-)\to (1 - 0)(1 - \mu_{{t_1^-}}^D(\lambda)) \cdot \Phi_{{t_1^-}}(0,\lambda)=(1 - \mu_{{t_1^-}}^D(\lambda)) \cdot \Phi_{{t_1^-}}(0,\lambda)=\hg_{t_1^-}(\lambda)$; where we used the fact that $\Phi$ is continuous in both $\alpha$ and $\lambda$, and where we used the definition of $\hg_{t_1^-}(\lambda)$ for $t_1^-$ in Phase 1. As all the $\hg_j$ for $j$ between $t_1^-$ and $t_1$ are $0$ and for $t_1$, we have that $R_{t_1-1}(\lambda^-)\to\MRtone$, we have that $R_{t_1^- -1}(\lambda^-) = R_{t_1-1}(\lambda^-) + \hg_{t_1^-}(\lambda^-) \to \MRtone + \hg_{t_1^-}(\lambda) = R_{t_1^- -1}(\lambda)$. From this point on, we can apply the same argument as for $\lambda^+$ to get that $R_0(\lambda^-)\to R_0(\lambda)$, and thus it is continuous for our cases in which $t_1<n$, $t_1^-$ does exist and where $R_{t_1(\lambda)-1}^{(2)}=\MRtone$ exactly.

    \smallskip
    \noindent\textbf{Case 1.2.b.ii: $t_1^-=\perp$.}
    Now, assuming that $t_1^-=\perp$, this implies that all $j<t_1$ have $\mu_j^D(\lambda-\varepsilon)=1$ for the range of $\varepsilon$ we were considering. This implies that for all $j$, we have that $R_j(\lambda-\varepsilon)=R_j^{(2)}(\lambda-\varepsilon)$ for $\varepsilon\geq 0$ (note that it is also true with equality), which is true because the value of $R_j^{(2)}(\lambda-\varepsilon)\leq \MRtone$. Since $R_0^{(2)}(\lambda)$ is already continuous in $\lambda$, it implies that $R_0^{(2)}(\lambda^-)\to R_0^{(2)}(\lambda)$, which using the previous fact for these values, it implies $R_0(\lambda^-)\to R_0(\lambda)$ which proves what we want.

    \smallskip
    \noindent\textbf{Case 2: $t_1=n$.}
    To finish this transition analysis, we must explicitly analyze the boundary case where $t_1 = n$, which implies that $t_3^\star=n$ and that $R_{n-1}=\hg_n^\star=\lambda\geq \MRtone$.

    \smallskip
    \noindent\textbf{Case 2.a: $\lambda>\MRtone$.}
    If $\lambda > \MRtone$, continuity is easy to prove since we can pick a small enough $\varepsilon$ such that $t_1$ does not change. For such a case, the value for $R_{t_1-1}=\lambda$ is clearly continuous in $\lambda$, and since all other indices $j<n$ are in Phase 1, and all operations in Phase 1 are continuous in $\lambda$, the result of $R_0(\lambda)$ is continuous in $\lambda$.

    \smallskip
    \noindent\textbf{Case 2.b: $\lambda=\MRtone$.}
    For the case of $\lambda=\MRtone$, things are a bit trickier. For the right limit, $\lambda^+=\lambda+\varepsilon$, the same argument as before can apply. It is the left limit that is more difficult. For such a case, $t_1^-=t_1(\lambda^-)$ will not be $n$, but instead some smaller index. Just as before for the case $t_1<n$, for a small enough $\varepsilon$, $t_1^-$ is the biggest index $j<n$ with $\mu_j^D<1$ if we take $\lambda^-=\lambda-\varepsilon$, which simplifies the analysis.

    \smallskip
    \noindent\textbf{Case 2.b.i: $t_1^-=\perp$.}
    If such an index does not exist it both implies $\mu_j^D=1$ for all $j$ and $t_1^-=\perp$. Thus, since $\mu_j^D=1$ we have that $\hg_j=0$ for $j<n$, which implies $R_0 (\lambda^-)= R_{n-1}(\lambda^-)=\lambda^-$ which is indeed continuous on $\lambda$, so $R_0(\lambda^-)= R_{n-1}(\lambda^-)=\lambda^-\to \lambda=\MRtone$. 

    \smallskip
    \noindent\textbf{Case 2.b.ii: $t_1^-\neq \perp$.}
    We now assume that $t_1^-$ does exist. As $R_{t_1^-}(\lambda^-)=R_{n-1}(\lambda^-)=\lambda^-<\MRtone$, we have that $R_{t_1^-}(\lambda^-)=R_{t_1^-}^{(2)}(\lambda^-)$. For such a case, we can compute things just as previously in the case that $t_1^-$ changes for the original $t_1<n$ and we get that the interpolation factor is $\alpha(\lambda^-) = \frac{\MRtone - \lambda^-}{\hat{g}_{t_1^-}^{(2)}}$. 
    As $\lambda^- \to \MRtone$, $\alpha(\lambda^-) \to 0$ since we have that $\hg_{t_1^-}(\lambda')>0$ for $\lambda'\in[\lambda^-,\lambda]$. 
    The fractional updates at $t_1^-$ are $\hat{g}_{t_1^-}^B(\lambda^-) = \MRtone - \lambda^- \to 0$ and $\hat{g}_{t_1^-}^A(\lambda^-) = (1 - \alpha(\lambda^-))(1 - \mu_{t_1^-}^D(\lambda^-)) \cdot \Phi_{t_1^-}(\alpha(\lambda^-),\lambda^-) \to (1 - \mu_{t_1^-}^D(\lambda)) \cdot \Phi_{t_1^-}(0, \lambda)$ using the same argument as every time that $\Phi(\alpha,\lambda)$ is continuous in both $\alpha$ and $\lambda$. 
    The budget becomes $R_{t_1^--1}(\lambda^-) = \lambda^- + \hat{g}_{t_1^-}^B(\lambda^-) + \hat{g}_{t_1^-}^A(\lambda^-) \to \MRtone + (1 - \mu_{t_1^-}^D) \cdot \Phi_{t_1^-}(0,\lambda)$, which we can check that is the same as it is for $R_{t_1^--1}(\lambda)$. Now, using the same repeated arguments in this section, we can get that $R_0(\lambda^-)\to R_0(\lambda)$, and thus, even for this case, the continuity is proven.

    Therefore, since all the possible cases have been looked at, continuity of $R_0(\lambda)$ is proven.
\end{proof}

\bibliographystyle{alpha}
\bibliography{literature}

\newpage
\appendix

\section{The Factor-Revealing Program Bounds the LMP-Approximation Ratio}
\label{sec:proof-lmp-approx-ratio-greedy}
In this section we will prove Theorem \ref{theorem:facility-location} which states that our algorithm is indeed a $\rho$-LMP approximation for facility location if $\rho$ is an upper bound on the optimal value of the factor-revealing program

\theoremfacilitylocation*

For a client $j \in D$, we denote by $\alpha_j^\star$ its final $\alpha_j$ value in Algorithm \ref{alg:dual_generalized}.
For the purposes of the analysis, we order the clients from $1$ to $n=|D|$ in increasing order of $\alpha$-values, $\alpha^\star_1 \le \alpha^\star_2 \le \dots \le \alpha^\star_n$. If $\alpha^\star_k = \alpha^\star_j$ for two clients $j \neq k \in D$ then break ties according to the order of events in which the clients became inactive. Note that with this order if $j \le k$ then $j$ becomes inactive strictly before $k$ or in the same event as $k$.
We denote by $\SS^{(k)}$ the set of open facilities just before the event in which client $k$ becomes inactive.
Moreover, $\SS^{(n+1)} = \SS^\star$ is the set of open facilities at the end of the algorithm. We say a client $j$ becomes inactive in an Event 1 or Event 2 if $j$ is removed from the set of active clients $A$ during this event.
For $1 \le j \leq k \le n$, we set 
\[
    r_{jk} = 
    \begin{cases}
        \alpha^\star_j & \text{$j$ and $k$ become inactive in the same event}, \\
        d(j, \SS^{(k)}) & \text{otherwise}.
    \end{cases}
\]
We also set $r_{j,n+1} = d(j, \SS^\star)$. 

\begin{observation}
\label{obs:r-mono}
We have $r_{j,k} \geq r_{j,k+1}$ for all $1 \leq j \leq k \leq n$.
\end{observation}
\begin{proof}
    We have trivially that $\alpha_j^\star \leq \alpha_j^\star$ and $d(j, \SS^{(k)}) \geq d(j, \SS^{(k')})$ for $k < k'$ as $\SS^{(k)}$ is monotone.
    Moreover, $\alpha_j^\star \geq d(j, \SS^{(k)})$ once $j$ is inactive with respect to $\SS^{(k)}$ so the claim follows.
\end{proof}
First we will upper bound the cost of Algorithm \ref{alg:dual_generalized} by the $r$ variables defined above. 
\begin{lemma}
\label{lemma:alg-cost-ufl}
    The cost (with respect to opening cost $\hat{f}$) of a solution output by Algorithm \ref{alg:dual_generalized} is at most
    \[\sum_{j=1}^n \Paren{\sum_{i=j}^{n} \beta(r_{j,i}, r_{j,i+1}) + r_{j,n+1}}\,.\]
\end{lemma}
\begin{proof}    
For a client $j$, its connection cost is $r_{j,n+1} = d(j,\SS^\star)$. Further, we need to bound the contribution of $j$ to the opening of each facility in $\SS^\star$. When opening a facility $i \in \FF$ we have 
\[
    \sum_{j \in A} \beta(\alpha_j, d(i,j)) + \sum_{j \in D \setminus A} \beta(d(j,\SS),d(i,j)) = \hat{f}.
\]
Let $\emptyset = \SS_0 \subsetneq \SS_1 \subsetneq \ldots \subsetneq \SS_m = \SS^\star$ with $m = |\SS^\star|$ be the sets of opened facilities during the execution of the algorithm. We are interested in the total contributions of client $j$ to facility payments. Name $\ell$ the index when they got connected. Then, their payments (discarding the final connection cost) are:
\begin{equation*}
    \beta(\alpha_j, d(j,\SS_\ell))+  \sum_{i=\ell}^{m-1} \beta(d(j,\SS_i), d(j, \SS_{i+1}))
\end{equation*}
To go from all the distances $d(j,\SS_\ell)$ to just the chosen set $r_{j,i}$, we repeatedly use the 
\emph{superadditivity} of the bidding function as per Definition~\ref{def:admissible-bidding}.
This property implies that, for any two times $p<q$, we can bound $\sum_{\ell=p}^{q-1}\beta(d(j,\SS_\ell), d(j,\SS_{\ell+1}))$ by $\beta(d(j,\SS_p),d(j,\SS_q))$. The proof for the lemma follows directly from this.

First, let us take care of the term with $\alpha_j^\star$. Take the first time $i$ that has $r_{j,i+1}$ not equal to $\alpha_j^\star$. If no such $i$ exists, then all facility payments of $j$ are zero and the bound for this term is trivial. If there are any previous terms $\beta(r_{j,i'},r_{j,i'+1})$ with $i'<i$, both function inputs are the same value and thus all the previous terms evaluate to $0$. Then enumerate all sets $\SS_\ell \subsetneq \ldots \subsetneq \SS_{\ell+k} = \SS^{(i+1)}$. By the \emph{superadditivity}, we can upper bound the terms:
\begin{equation*}
    \beta(\alpha_j^\star, d(j,\SS_\ell)) + \sum_{p=\ell}^{\ell+k-1} \beta(d(j,\SS_p), d(j, \SS_{p+1})) \leq 
    \beta(\alpha_j^\star, d(j, \SS^{(i+1)})) = 
    \beta(r_{j,i}, r_{j,i+1}))
\end{equation*}

Now, let us do the same for the terms that do not include the term $\alpha_j^\star$. For an index $i$ with $r_{j,i}=d(j,\SS^{(i)})$, enumerate all sets $\SS^{(i)} = \SS_\ell \subsetneq \ldots \subsetneq \SS_{\ell+k} = \SS^{(i+1)}$, and let there be $k$ such sets. By the \emph{superadditivity}, we can upper bound the terms:
\begin{equation*}
    \sum_{p=\ell}^{\ell+k-1} \beta(d(j,\SS_p), d(j, \SS_{p+1})) \leq 
    \beta(d(j, \SS^{(i)}), d(j, \SS^{(i+1)})) = 
    \beta(r_{j,i}, r_{j,i+1}))
\end{equation*}
Adding all of these together proves the lemma.
\end{proof}

\begin{lemma}
\label{lem:star-cost}
    Let $\rho \geq 1$ be an upper bound on the optimal value of the factor-revealing program in Definition~\ref{def:frp}. 
    For some facility $i^\star \in \FF$, consider any subset of clients $D^\prime \subseteq D$
    Then
    \[
    \sum_{j \in D^\prime} \Paren{\sum_{i=j}^n \beta(r_{j,i}, r_{j,i+1}) + r_{j,n+1}} \leq \hat{f} + \rho \cdot \sum_{j \in D^\prime} d(j, i^\star)\,.
    \]
\end{lemma}
\begin{proof}
The case $D^\prime = \emptyset$ is immediate, so assume $D^\prime \neq \emptyset$.
First, we will prove that the variables $(r,d,\hat f)$ supported on $D^\prime$ are feasible for the factor-revealing program \ref{def:frp} without the normalization constraint \ref{eq:fr-norm}. Here, we re-index the variables in $D^\prime$ to $1, \dots, |D^\prime|$ and $n+1$ to $|D^\prime|+1$. It is important that we keep the original order of the variables. For $k \in D^\prime$, we set $d_k = d(i^\star,k)$. 
\begin{claim}
    The variables $(r,d,\hat f)$ satisfy constraints \eqref{eq:fr-r-mono} to \eqref{eq:fr-nonneg}.
\end{claim}
\paragraph{Monotonicity \eqref{eq:fr-r-mono}} Holds by Observation \ref{obs:r-mono}.
\paragraph{Triangle constraint \eqref{eq:fr-triangle}} 
We have to show that for $j < k$ we have $r_{kk} \le r_{jk} + d_k + d_j$
By definition $r_{kk} = \alpha^\star_k$. If $k$ and $j$ become inactive in the same event, then $r_{jk} = \alpha^\star_j = \alpha^\star_k = r_{kk}$. Otherwise, $r_{jk} = d(j, \SS^{(k)}) = d(j,i)$ for some facility $i \in \SS^{(k)}$. Since the facility $i$ is open while $k$ is still active,
\[
    r_{kk} = \alpha^\star_k \le d(i,k) \le d(k,i^\star) + d(j,i^\star) + d(j,i) = r_{jk} + d_j + d_k.
\]
\paragraph{Metric inequalities of $r$ \eqref{eq:fr-r-metric}}
We have to show that for $j \neq k < l$ we have $r_{jl} \le r_{kl} + d_j + d_k$. First assume that $j$ becomes inactive strictly before $l$, $r_{jl} = d(j, \SS^{(l)}) = d(j,i_1)$ for some facility $i_1 \in \SS^{(l)}$. Similar, if $k$ becomes inactive strictly before $l$, $r_{kl} = d(k, \SS^{(l)}) = d(k,i_2)$ for some facility $i_2 \in \SS^{(l)}$. Then,
\[
    r_{jl} = d(j,i_1) \le d(j,i_2) \le d(j,i^\star) + d(k, i^\star) + d(k,i_2) = r_{kl} + d_j + d_k.
\]
Now, if $k$ becomes inactive in the same event as $l$ then $r_{kl} = \alpha^\star_k$. However, note that $r_{jl} \le r_{jj} = \alpha^\star_j \le \alpha^\star_{k}$ since $k$ becomes inactive strictly after $j$.
\newline Next, assume that $j$ becomes inactive in the same event as $l$, $r_{jl} = r_{jj} = \alpha^\star_j$. Similar, if $k$ becomes inactive in the same event then,
\[
    r_{jl} = \alpha^\star_j = \alpha^\star_k = r_{kl}.
\]
In the case where $k$ becomes inactive strictly before $j$ we have $r_{kl} = d(k, \SS^{(l)}) = d(k,\SS^{j}) = r_{kj}$ because $j$ and $l$ become inactive in the same event. Now, the inequality follows from the Triangle constraint \eqref{eq:fr-triangle}. 
\paragraph{Facility constraint \eqref{eq:fr-facility}} We have to show that for all $k \in D^\prime$ we have
\[
    \sum_{\substack{j < k \\ j \in D^\prime}} \beta(r_{jk},d_j)
  \;+\;
  \sum_{\substack{j \ge k \\ j \in D^\prime}} \beta(r_{kk},d_j)
  \;\le\; \hat{f}.
\]
To this end consider the event in which $k$ becomes inactive. Just before this event, we have that 
\[
    \text{bid}(j,i^\star) = \beta(\alpha^\star_k ,d(i^\star,j)) = \beta(r_{kk},d(i^\star,j))
\] 
for all active clients $j \in A$. Moreover, 
\[
    \text{bid}(j,i^\star) = \beta(d(j,\SS^{(k)}),d(i^\star,j)) = \beta(r_{jk},d(i^\star,j))
\]
 for all inactive clients $j \in D \setminus A$. Note that if $j$ and $k$ become inactive in the same event then $r_{kk} = r_{jk} = \alpha^\star_k = \alpha^\star_j$. If $i^\star \in \SS^{(k)}$, all bids are zero, because $\alpha^\star_k \le d(j,\SS^{(k)})$ for active clients and $d(j,\SS^{(k)}) \le d(j,i^\star)$ for all clients. If $i^\star$ is not yet open then the inequality is satisfied since the sum of all bids $\sum_{j \in D} \text{bid}(j,i^\star)$ towards an unopened facility $i^\star$ can never exceed $\hat{f}$. 
\newline
\newline
 Next, we will normalize the variables as to satisfy \eqref{eq:fr-norm}. To this end let $\dsum = \sum_{j \in D^\prime} d(j, i^\star) \ge 1$. Since all constraints \eqref{eq:fr-r-mono} to \eqref{eq:fr-nonneg} are scale invariant we have that $(r/\dsum, d/\dsum, \hat{f}/\dsum)$ is now feasible for the factor-revealing program \ref{def:frp}. In order to index according to the original set of clients $D$ define $q(i) = \min \Set{i' \mid i' > i, i' \in D^\prime \cup \{n+1\}}$. Using the $\rho$ is an upper bound on the value of the program we have,
\[
    \frac{\sum_{j \in D^\prime} \Paren{\sum_{\substack{i \in D^\prime\\i \geq j}} \beta\Paren{r_{j,i}, r_{j,q(i)}} + r_{j,n+1}} - \hat{f}}{\dsum} \le \rho.
\]
Finally, because $r_{j,i} \geq r_{j,i+1}$ for all valid $j,i$, by the superadditivity of $\beta$ the above also implies that
\[\sum_{j \in D^\prime} \Paren{\sum_{i=j}^n \beta(r_{j,i}, r_{j,i+1}) + r_{j,n+1}} \leq \hat{f} + \rho \cdot \sum_{j \in D^\prime} d(j, i^\star)\,.\]
\end{proof}
We are ready to prove Theorem \ref{theorem:facility-location}.
\begin{proof}[Proof of Theorem~\ref{theorem:facility-location}]
 Consider an optimal solution of cost $\OPT$ to the original problem with facility cost $f = \frac{\hat{f}}{\rho}$, and let $k^\star$ denote the number of facilities opened in this solution.
    Let $\{i_1^\star, \ldots, i_{k^\star}^\star\}$ be the facilities opened by the optimal solution, and let $\{D_1^\star, \ldots, D_{k^\star}^\star\}$ be such that $D_\ell^\star$ is the set of clients connected to facility $i_\ell^\star$ in the optimal solution.
    Let $C_{\mathrm{SOL}}$ and $F_{\mathrm{SOL}}$ be the connection cost and opening cost (with respect to $f$) of our algorithm, and let $C_{\OPT}$ and $F_{\OPT}$ be defined analogously for $\OPT$.
    By Lemma~\ref{lemma:alg-cost-ufl}, which bounds the cost with respect to $\hat{f}$, and using that $f = \frac{\hat{f}}{\rho}$, we know that 
    \[
        C_{\mathrm{SOL}} + \rho \cdot F_{\mathrm{SOL}} \leq \sum_{j=1}^n \Paren{\sum_{i=j}^{n} \beta(r_{j,i}, r_{j,i+1}) + r_{j,n+1}}\,.
    \]
    For $\ell \in \{1, \ldots, k^\star\}$ consider $(i^\star_\ell, D_\ell)$. By Lemma~\ref{lem:star-cost} we have that
    \[\sum_{j \in D_\ell} \Paren{\sum_{i=j}^n \beta(r_{j,i}, r_{j,i+1}) + r_{j,n+1}} \leq \hat{f} + \rho \cdot \sum_{j \in D_\ell} d(j, i^\star_\ell)\,.\]
    Note that $D_1^\star, \ldots, D^\star_{k^\star}$ partition $D$. Then, summing the above over all $(i_{\ell}^\star, D^\star_{\ell})$,
    \begin{align*}
    C_{\mathrm{SOL}} + \rho \cdot F_{\mathrm{SOL}} & \le \sum_{j=1}^n \Paren{\sum_{i=j}^{n} \beta(r_{j,i}, r_{j,i+1}) + r_{j,n+1}}\\
    &\leq k^\star \cdot \hat{f} + \rho \cdot C_{\OPT}\\
    &= \rho \cdot F_{\OPT} + \rho \cdot C_{\OPT}
    = \rho \cdot \OPT\,.
    \end{align*}
    The theorem follows because $F_{\mathrm{SOL}} = |\SS^\star| \cdot f$.
\end{proof}

\newpage

\section{Deriving the Dual}\label{sec:appendix-deriving-dual}

In this section, we prove that the factor revealing dual program as defined in Definition~\ref{def:frdual} is indeed equivalent to the original one, namely, we prove Lemma~\ref{lemma:mathematical-program-1.99}.
We restate the exact definition of the factor-revealing dual and its lemma here:

\deffrdual*

\lemmamathematicalprogram*

To prove this lemma, we first define an intermediate version of the factor-revealing primal different from the one in Definition~\ref{def:frp}.

Differently from the previous factor revealing, we rewrite the metric inequalities. We define the variable $m_i$ as the smallest distance from the central facility to an already open facility at time $i$; for $i=1$, before any earlier client is available, $m_1$ is only used as a slack variable. 
Using this definition, the metric inequalities can be rewritten as follows:
\begin{equation*}
    \forall \ 1 \le j \le i \le n: \quad r_{j,i} \le m_i + d_j
\end{equation*}
Without enforcing the exact definition of $m_i$, we only need $m_i$ to follow the following two structural constraints:
\begin{align*}
    \forall \ 1 \le j < i \le n:& \quad m_i \le r_{j,i} + d_j \\
    \forall \ 1 \le i < n:& \quad m_i \ge m_{i+1}
\end{align*}
The first comes by an application of the metric inequalities, directly going from central facility to an open facility should be faster than going through a client $j$ that was already connected before time $i$. The second is due to the monotonic non-increasingness of the distances, just like it is the case with $r_{j,i}$ through $i$; opening a new facility can only decrease the distance.

With this defined and using also the $b_{i,j}$ variables defined in the main text, we can now define the intermediate factor-revealing problem:
\begin{definition}[Factor-revealing primal]\label{def:intermediate-frp}
    Fix $n\geq 1$.
    \begin{equation*}
        \begin{aligned}
            \smash{\max_{\substack{\\ \forall b_{i,j}\in\{0,1\} \\ \forall i,j: 1\leq j\leq i \leq n}}} \quad 
            \max \quad 
            & \sum_{j=1}^{n} \smash{\sum_{i=j}^{n} b_{i,j} \cdot (r_{j,i} -\Mcoeff \cdot r_{j,i+1}) - f} \span\span\span\span\\
            \text{\emph{s.t.:}} \quad
                &&&\sum_{i=1}^{n} d_i \le 1 && \text{($\rho$: Limit on distances)}\\
                & \forall \ 1 \le j < i \le n: 
                \quad && m_i \le r_{j,i} + d_j && \text{($s_{i,j}$: Definition of $m_i$)}\\
                & \forall \ 1 \le j \leq i \le n:
                \quad && r_{j,i} \le m_i + d_j && \text{($t_{i,j}$: Metric constraint)}\\
                & \forall \ 1 \le j \le i \le n: 
                \quad &&r_{j,i} \ge r_{j,i+1} && \text{($y_{i,j}$: Monotonicity of $r_{j,i}$)}\\
                & \forall \ 1 \le i < n: 
                \quad && m_i \ge m_{i+1} && \text{($y'_i$: Monotonicity of $m_i$)}\\
                & \forall \ 1\le j <i\le n: 
                \quad && r_{j,i} - \Mcoeff \cdot d_j \le \gamma_{i,j} && \text{($g^a_{i,j}$: Bidding linearization)}\\
                & \forall \ 1\le i\le j\le n: 
                \quad && r_{i,i} -\Mcoeff \cdot d_j \le \beta_{i,j} && \text{($g^b_{i,j}$: Bidding linearization)}\\
                & \forall \ 1\le i\le n: \quad && \sum_{j=1}^{i-1} \gamma_{i,j} + \sum_{j=i}^{n} \beta_{i,j} \le f && \text{($\hg_i$: No overbidding)}\\
                & \forall \ 1\le i, j\le n + 1: && r_{j,i}, d_{j}, m_i, \beta_{i,j}, \gamma_{i,j}, f \ge 0 \span\span
        \end{aligned}
    \end{equation*}
\end{definition}
\begin{lemma}\label{lemma:generalized-factor-revealing-problem-metric-rewrite}
    The result from the factor-revealing program defined in Definition~\ref{def:intermediate-frp}, defined as a maximization over a family of primal LPs, is an upperbound to the one presented in Definition~\ref{def:frp} for the specific bidding function $\beta_{\Mcoeff}(\alpha, d) = [\alpha - \Mcoeff \cdot d]^+$.
\end{lemma}
\begin{proof}    
We only prove a one-way relation: every optimal solution of Definition~\ref{def:frp} extends to a feasible solution with no smaller objective value to Definition~\ref{def:intermediate-frp} for one specific linear program above.

Take an optimal solution to the original program and name its variables $r_{j,i}'$, $d_j'$ and $f'$. Let us now create the variables $r_{j,i}$, $d_j$ and $f$ by copying the previous ones, except that for $r_{j,n+1}$ we set it to $0$ for all $j$.

Second, let us choose the specific linear program that we will be working with. Let us take $b_{i,j}=\ind[r_{j,i}-\Mcoeff \cdot r_{j,i+1}\ge0]$ for $i<n$, and let us take $b_{n,j}=1$ for all $j$. The expression for the new objective terms with $i<n$ become:
\begin{align*}
    b_{i,j} \cdot (r_{j,i} -\Mcoeff \cdot r_{j,i+1}) &= \ind[r_{j,i}-\Mcoeff \cdot r_{j,i+1}\ge0] \cdot (r_{j,i} -\Mcoeff \cdot r_{j,i+1}) \\
    &= [r_{j,i} -\Mcoeff \cdot r_{j,i+1}]^+ = \beta_{\Mcoeff} (r_{j,i},r_{j,i+1})
\end{align*}
which makes the objective function equal to:
\begin{equation*}
    \sum_{j=1}^{n} \left(\sum_{i=j}^{n-1} \beta_{\Mcoeff}(r_{j,i}, r_{j,i+1}) + \beta_{\Mcoeff}(r_{j,n},0)\right)- f 
\end{equation*}
Knowing that $r_{j,i}=r_{j,i}'$ for all the $r_{j,i}$ shown in this equation and that $f'=f$, subtracting the objective function from the program in Definition~\ref{def:frp} on this equation gives:
\begin{equation*}
    \sum_{j=1}^{n} \left(\beta_{\Mcoeff}(r_{j,n},0)-\beta_{\Mcoeff}(r_{j,n},r'_{j,n+1}) - r'_{j,n+1}\right) = \sum_{j=1}^{n} \left(r_{j,n}-([r_{j,n}-\gamma \cdot r'_{j,n+1}]^+ + r'_{j,n+1})\right) \geq 0
\end{equation*}
where we have substituted the expression of $\beta_{\gamma}(\alpha,d)=[\alpha-\gamma d]^+$ and we have used that $[\alpha-\gamma d]^+ + d \leq \alpha$ for $\gamma \geq 1$ and $0 \leq d \leq \alpha$.
Thus, with these variables we get that the objective value for the program in Definition~\ref{def:intermediate-frp} is an upper bound on the original one found in Definition~\ref{def:frp}.

Third, for $m_i$, we set them as:
\begin{equation*}\label{eq:mi-exact-definition}
    m_i = \min_{1\le j < i} (r_{j,i} + d_j) \qquad \text{for } i \ge 2.
\end{equation*}
Indeed, the inequalities $m_i \le r_{j,i} + d_j$ are immediate from this assignment for $j<i$, and $r_{j,i} \le m_i + d_j$ is also direct since we had $r_{j,i} \le r_{k,i}+d_k + d_j$ for all earlier $k<i$. For the monotonicity constraints $m_i \ge m_{i+1}$ with $i\ge2$, if $\ell$ is the index that attains the minimum at time $i$, then the following holds $m_{i+1} \le r_{\ell,i+1} + d_\ell \le r_{\ell,i} + d_\ell = m_i$. Finally, $m_1$ has no upper-bound constraint, so we choose it large enough to satisfy $r_{1,1}\le m_1+d_1$ and, when $n\ge2$, $m_1\ge m_2$.

Finally, we set the values for $\gamma_{i,j}$ and $\beta_{i,j}$. These will be $\gamma_{i,j} = [r_{j,i}-\Mcoeff \cdot d_j]^+  =\beta_{\Mcoeff}(r_{j,i},d_j)$ and $\beta_{i,j} = [r_{i,i}-\Mcoeff\cdot d_j]^+=\beta_{\Mcoeff}(r_{i,i},d_j)$, which clearly satisfy the $g^a_{i,j}$ and $g^b_{i,j}$ constraints as well as the $\hg_i$ because the original facility constraint \ref{eq:fr-facility} was satisfied. Thus, all $(g)$ constraints are satisfied.

Therefore the feasible region of the original program is contained in the feasible region of the relaxed program, so any upper bound on the relaxed program is also an upper bound on the original program.
\end{proof}

However, the dual represented in Definition~\ref{def:frdual} is not the direct dual one gets from the Linear Programs in Definition~\ref{def:intermediate-frp}, there are still some extra steps. Let us first take this direct dual of the inner program for any $b_{i,j}$, using the notation that is on the right of each inequality to refer to each constraint:
\begin{definition}[Dual of inner program]\label{def:intermediate-dual}
    The linear dual of any inner program of the intermediate factor-revealing program is:
    \begin{equation}\label{eq:direct-dual-lp}
        \begin{aligned}
        \text{minimize}& \quad \rho \\
        \text{subject to:}& 
            &&  \sum_{i=1}^{n} \hat{g}_i \leq 1 && \text{($f$)} \\
            & \forall \ 1 \le i, j \le n: 
            && \hat{g}_i \geq g_{i,j} && \text{($\gamma_{i,j}$, $\beta_{i,j}$)}\\
            & \forall \ 1\le j \le n: 
            && \sum_{q=j+1}^n g_{j,q} + \left(g_{j,j} + t_{j,j} - s_{j,j} \right) - y_{j,j} \geq b_{j,j}
            && \text{($r_{j,j}$)}\\
            & \forall \ 1\le j < i < n: 
            && y_{i-1,j} + \left(g_{i,j} + t_{i,j} - s_{i,j}\right) - y_{i,j} \geq b_{i,j} - \Mcoeff \cdot b_{i-1,j}
            && \text{($r_{j,i}$)}\\
            & \forall \ 1\le j < i = n: 
            && y_{n-1,j} + \left(g_{n,j} + t_{n,j} - s_{n,j}\right) \quad \ \ \ \geq b_{n,j} - \Mcoeff \cdot b_{n-1,j}
            && \text{($r_{j,n}$)}\\
            & \ \ \ i=1: 
            && \ \ \ \qquad \sum_{j=1}^i s_{i,j} - \sum_{j=1}^i t_{i,j} - y_i' \ge 0 && \text{($m_1$)}\\
            & \forall \ 1< i < n: 
            && y_{i-1}' + \sum_{j=1}^i s_{i,j} - \sum_{j=1}^i t_{i,j} - y_i'\ge 0 && \text{($m_i$)}\\
            & \ \ \ 1< i = n: 
            && y_{n-1}' + \sum_{j=1}^n s_{i,j} - \sum_{j=1}^n t_{i,j} \qquad \ge 0 && \text{($m_n$)}\\
            & \forall \ 1\leq j\leq n:
            && \rho \geq \ \Mcoeff\cdot \sum_{i=1}^n g_{i,j} + \sum_{i=j}^n s_{i,j} + \sum_{i=j}^n t_{i,j} && \text{($d_j$)} \\
            & \forall \ 1\leq i\leq n:
            && s_{i,i}=0 && \text{(diagonal convention)}\\
            & \forall \ 1 \le i,j \le n: && \hat{g}_{i}, \ g_{i,j}, \ t_{i,j}, \ s_{i,j}, \ y_{i,j},\ y_i' \geq 0
        \end{aligned}
    \end{equation}
    The variables $s_{i,i}$ are dummy variables fixed to zero so that the summation notation matches Definition~\ref{def:frdual}. We have also abused the notation and called both $g_{i,j}^a$ and $g_{i,j}^b$ as $g_{i,j}$, which is valid because the indices never overlap.
\end{definition}

Now let us prove the following lemma:
\begin{lemma}\label{lem:fr-intermediate-eq-dual}
    The inner program of the factor-revealing dual defined in Definition~\ref{def:frdual} and the linear program in Definition~\ref{def:intermediate-dual} are equivalent.
\end{lemma}
\begin{proof}
    The main difference between the two programs is that in the second there are variables $y_{i,j}$ and $y_i'$. To move to the other representation in Definition~\ref{def:frdual}, we will interpret these variables as slack ones, which carry unused slack from $(r_{j,i})$ to the next timestep $(r_{j,i+1})$ for the case of $y_{i,j}$ or from $(m_i)$ to $(m_{i+1})$ for $y'_i$. We now prove that a feasible solution for the one in Definition~\ref{def:frdual} implies a feasible solution with the same cost in Definition~\ref{def:intermediate-dual}, and the reverse as well.
    
    First, given a feasible solution $(\rho, \hat{g}, g, t, s)$ for Definition~\ref{def:frdual}, we construct a feasible solution for the direct dual in Definition~\ref{def:intermediate-dual} by keeping the same variables and explicitly defining the slack variables $y_{i,j}$ and $y'_i$ as follows:
    \begin{align*}
        y_{i,j} &= \sum_{q=j+1}^n g_{j,q} + \sum_{p=j}^{i} \left(g_{p,j} + t_{p,j} - s_{p,j}\right) - b_{i,j} + \Mcoeffmone \sum_{p=j}^{i-1} b_{p,j} \\
        y'_i &= \sum_{p=1}^i \sum_{j=1}^p s_{p,j} - \sum_{p=1}^i \sum_{j=1}^p t_{p,j}
    \end{align*}
    By the $(r_{j,i})$ and $(m_i)$ constraints of Definition~\ref{def:frdual}, it trivially holds that $y_{i,j} \geq 0$ for all $1 \le j \le i \le n$ and $y'_i \geq 0$ for all $1 \le i \le n$. We now verify that these definitions satisfy the constraints of Definition~\ref{def:intermediate-dual}.
    
    For $i=j$, substituting $y_{j,j}$ in the $(r_{j,j})$ constraint gives:
    \begin{align*}
        \sum_{q=j+1}^n g_{j,q} + (g_{j,j} + t_{j,j} - s_{j,j}) - y_{j,j}
        &=  \sum_{q=j+1}^n g_{j,q} + (g_{j,j} + t_{j,j} - s_{j,j}) \\
        &\quad - \left(\sum_{q=j+1}^n g_{j,q} + (g_{j,j} + t_{j,j} - s_{j,j}) - b_{j,j}\right) \\
        &= b_{j,j}
    \end{align*}
    which satisfies it with equality. For $i>j$, the difference $y_{i-1,j} - y_{i,j}$ algebraically evaluates to:
    \begin{align*}
        y_{i-1,j} - y_{i,j} 
        &= -(g_{i,j} + t_{i,j} - s_{i,j}) + b_{i,j} - b_{i-1,j} - \Mcoeffmone b_{i-1,j} \\
        &= -(g_{i,j} + t_{i,j} - s_{i,j}) + b_{i,j} - \Mcoeff\cdot b_{i-1,j}
    \end{align*}
    Rearranging this expression gives:
    \begin{equation*}
        y_{i-1,j} + (g_{i,j} + t_{i,j} - s_{i,j}) - y_{i,j} = b_{i,j} - \Mcoeff\cdot b_{i-1,j}
    \end{equation*}
    This also fulfills the $(r_{j,i})$ constraint with equality. For $(r_{j,n})$, let us substitute the value of $y_{n-1,j}$:
    \begin{align*}
        b_{n,j} - \Mcoeff b_{n-1,j}
        &\leq y_{n-1,j} + (g_{n,j} + t_{n,j} - s_{n,j}) \\
        &= \sum_{q=j+1}^n g_{j,q} + \sum_{p=j}^{n} \left(g_{p,j} + t_{p,j} - s_{p,j}\right) - b_{n-1,j} + \Mcoeffmone \sum_{p=j}^{n-2} b_{p,j}
    \end{align*}
    Rearranging this expression gives:
    \begin{equation*}
        b_{n,j} - \Mcoeffmone \sum_{p=j}^{n-1} b_{p,j} \leq \sum_{q=j+1}^n g_{j,q} + \sum_{p=j}^{n} \left(g_{p,j} + t_{p,j} - s_{p,j}\right)
    \end{equation*}
    which is satisfied as the feasibility of the first program because of its $(r_{j,i})$ constraint.
    Similarly, we evaluate the $(m_i)$ constraints using the variables $y'_i$. For $i=1$:
    \begin{equation*}
        s_{1,1} - t_{1,1} - y'_1 = s_{1,1} - t_{1,1} - (s_{1,1} - t_{1,1}) = 0
    \end{equation*}
    satisfying the $(m_1)$ constraint.
    For $1 < i \le n$:
    \begin{align*}
        y'_{i-1} - y'_i &= \sum_{p=1}^{i-1} \sum_{j=1}^p (s_{p,j} - t_{p,j}) - \sum_{p=1}^i \sum_{j=1}^p (s_{p,j} - t_{p,j}) \\
        &= - \sum_{j=1}^i (s_{i,j} - t_{i,j})
    \end{align*}
    Thus, $y'_{i-1} + \sum_{j=1}^i (s_{i,j} - t_{i,j}) - y'_i = 0$, satisfying the $(m_i)$ constraint with equality.
    For the final $(m_n)$ constraint, we have $y'_{n-1} + \sum_{j=1}^n (s_{n,j} - t_{n,j}) = \sum_{p=1}^{n} \sum_{j=1}^p (s_{p,j} - t_{p,j})$. Since the feasibility of the first program ensures that this quantity is at least $0$ because of its ($m_{n})$ constraint, the inequality $y'_{n-1} + \sum_{j=1}^n (s_{n,j} - t_{n,j}) \ge 0$ naturally holds. All other constraints that do not involve $y_{i,j}$ and $y'_i$ remain trivially satisfied as they are identical across both programs.
    
    Conversely, suppose we have a feasible solution for the direct dual in Definition~\ref{def:intermediate-dual}. We show that this implies feasibility for Definition~\ref{def:frdual} by telescoping the inequalities.
    For the $(r_{j,i})$ constraints, the direct dual gives:
    \begin{align*}
        \sum_{q=j+1}^n g_{j,q} + (g_{j,j} + t_{j,j} - s_{j,j}) - y_{j,j} &\geq b_{j,j} \\
        y_{k-1,j} + (g_{k,j} + t_{k,j} - s_{k,j}) - y_{k,j} &\geq b_{k,j} - (1+\Mcoeffmone) b_{k-1,j} \quad \text{for } k > j
    \end{align*}
    For $i<n$, summing these inequalities from $k=j$ to $i$, the slack variables $y_{k,j}$ telescope perfectly, and we obtain:
    \begin{align*}
        \sum_{q=j+1}^n g_{j,q} + \sum_{k=j}^i (g_{k,j} + t_{k,j} - s_{k,j}) - y_{i,j} &\geq b_{j,j} + \sum_{k=j+1}^i \big( b_{k,j} - (1+\Mcoeffmone) b_{k-1,j} \big) \\
        &= b_{i,j} - \Mcoeffmone \sum_{k=j}^{i-1} b_{k,j}
    \end{align*}
    Since $y_{i,j} \geq 0$, we can drop $-y_{i,j}$ from the left side and get:
    \begin{equation*}
        \sum_{q=j+1}^n g_{j,q} + \sum_{k=j}^i (g_{k,j} + t_{k,j} - s_{k,j}) \geq b_{i,j} - \Mcoeffmone \sum_{k=j}^{i-1} b_{k,j}
    \end{equation*}
    which is exactly the $(r_{j,i})$ constraint in Definition~\ref{def:frdual}. For $i=n$, we sum the constraints $(r_{j,j})$, $(r_{j,j+1}), \ldots, (r_{j,n-1})$ together with the last constraint $(r_{j,n})$. The $y_{k,j}$ variables still telescope, but no final slack remains, giving directly
    \begin{equation*}
        \sum_{q=j+1}^n g_{j,q} + \sum_{k=j}^n (g_{k,j} + t_{k,j} - s_{k,j}) \geq b_{n,j} - \Mcoeffmone \sum_{k=j}^{n-1} b_{k,j},
    \end{equation*}
    which is the desired constraint for $i=n$.
    
    Applying the same telescoping sum to the $(m_k)$ constraints from $k=1$ to $i$:
    \begin{align*}
        (s_{1,1} - t_{1,1}) - y'_1 &\ge 0 \\
        y'_{k-1} + \sum_{j=1}^k (s_{k,j} - t_{k,j}) - y'_k &\ge 0 \quad \text{for } 1 < k \le i
    \end{align*}
    For $i<n$, summing these yields:
    \begin{equation*}
        \sum_{p=1}^i \sum_{j=1}^p (s_{p,j} - t_{p,j}) - y'_i \geq 0
    \end{equation*}
    Since $y'_i \geq 0$, it follows that $\sum_{p=1}^i \sum_{j=1}^p (s_{p,j} - t_{p,j}) \geq 0$, which is the $(m_i)$ constraint in Definition~\ref{def:frdual}. For $i=n$, summing $(m_1),\ldots,(m_{n-1})$ together with the last constraint $(m_n)$ gives directly
    \begin{equation*}
        \sum_{p=1}^n \sum_{j=1}^p (s_{p,j} - t_{p,j}) \geq 0,
    \end{equation*}
    which is the last $(m_i)$ constraint. This concludes the proof that the two representations are equivalent.

    We have therefore shown that, for every fixed choice of the binary variables $b_{i,j}$, the formulation in Definition~\ref{def:frdual} and \ref{def:intermediate-dual} have the same optimum value. 
\end{proof}

With this, the proof becomes direct, as it is just an application of strong duality.

\begin{proof}[Proof of Lemma~\ref{lemma:mathematical-program-1.99}]
    For every fixed choice of the binary variables $b_{i,j}$, the linear program in Definition~\ref{def:intermediate-dual} is exactly the LP dual of the corresponding primal linear program in Definition~\ref{def:intermediate-frp} by definition. Hence, by strong duality, these two inner programs have the same optimum value. 
    
    By now using Lemma~\ref{lem:fr-intermediate-eq-dual}, we know that the inner program in Definition~\ref{def:frdual} and the one in Definition~\ref{def:intermediate-frp} for one same fixed set of $b_{i,j}$ are equivalent.
    Taking the outer maximum over all feasible choices of $b_{i,j}\in\{0,1\}$ preserves this equality, so the two overall formulations are equivalent.

    Thus, now applying Lemma~\ref{lemma:generalized-factor-revealing-problem-metric-rewrite}, we get to our desired result that the solution of the program from Definition~\ref{def:intermediate-dual} is an upper-bound to the program from Definition~\ref{def:frp}, which proves our lemma.
\end{proof}

\newpage

\section{Dual Fitting Deferred Proofs}\label{sec:dual-deferred-proofs}

This section collects the technical proofs deferred from Section~\ref{sec:195-dual-fitting}. We first derive the formulas for the credit of each client, separating the cases $i \geq f_j$ and $i < f_j$; we then prove the corresponding expressions for total generation and transfer, which are combined to obtain the budget inequalities required in the proof of Lemma~\ref{lem:general-case-procedure}. Throughout, we use the same case distinction as in the main text: we analyze clients in groups $A$, $B$, $C$, and $D$ separately, and handle fractional memberships by writing them as convex combinations of the integral cases. At the end, there are the rest of the remaining proofs that were deferred.

\paragraph{Credit calculation for \texorpdfstring{$i \geq f_j$}{i >=f\_j}.}

Let us restate the definition of credit, and then prove various lemmas about it.

\defcreditij*

\begin{lemma}
    \label{lem:credit-A-after}
    If client $j$ belongs integrally to group $A$, then for $i \geq f_j$
    \begin{equation*}
        \text{\emph{credit}}_{i,j} = 1 - \sum_{i'=f_j+1}^i \delta_{i',j}.
    \end{equation*}
\end{lemma}
\begin{proof}
    Because $j \in A$, we have $j < t_1$ by how the memberships are assigned in Equation~\ref{eq:membership-assignment}. Also, $\delta_{j,j'}=0$ for every $j' > j$, so every later $A/B/C$-part contributes $\hg_j$ to the initial term, while every later $D$-part contributes $\hg_j \cdot (1-\Mfrac \chi_{j'})$. Hence
    \begin{equation*}
        \sum_{j'=j+1}^n g_{j,j'} = \left(n-j-\Mfrac \cdot \sum_{j'=j+1}^n \mu_{j'}^D \cdot \chi_{j'}\right) \cdot \hg_j =  (n-j-\Mfrac \cdot z) \cdot \hg_j = 1 - [\MRttwo - R_{f_j}]^+.
    \end{equation*}
    where $z$ appears because $z=\sum_{j'=t_2+1}^{n}\chi_{j'} \mu_{j'}^D$ and for $j' \le t_2$ we have $\chi_{j'}=0$, thus having $\sum_{j'=j+1}^n \mu_{j'}^D \cdot \chi_{j'} = z$ since $j \leq t_1 \leq t_2$.
    Moreover,
    \begin{equation*}
        \sum_{i'=j}^{i} g_{i',j} = \sum_{i'=j}^{i} (\hg_{i'}-\delta_{i',j}) = R_{j-1}-R_i-\sum_{i'=f_j+1}^{i} \delta_{i',j},
    \end{equation*}
    \begin{equation*}
        \sum_{i'=j}^{i} t_{i',j} = 0,
        \qquad
        \sum_{i'=j}^{i} \bar s_{i',j} = \sum_{i'=j}^{i} (\hg_{i'}-\sigma_{i',j}) = R_{j-1}-R_i-[\MRttwo-R_{f_j}]^+,
    \end{equation*}
    because $\sigma_{i',j}=0$ for $i'>f_j$ and $\sum_{i'=j}^{f_j} \sigma_{i',j}=[\MRttwo-R_{f_j}]^+$.
    Substituting into the definition of $\text{credit}_{i,j}$ gives the claim.
\end{proof}

\begin{lemma}
\label{lem:credit-B-after}
    If client $j$ belongs integrally to group $B$, then for $i \geq f_j$
    \begin{equation*}
        \text{\emph{credit}}_{i,j} = 1 - \ind[i > f_j] \cdot \Mcoeffmone.
    \end{equation*}
\end{lemma}
\begin{proof}
    For an integral $B$-client we have that $j \geq t_1$ from the membership assignment in Equation~\ref{eq:membership-assignment}. 
    At $j=t_1$, we have $\hg_{t_1}^A=0$.
    Hence, for every $j' > j$ we have $g_{j,j'} = \hg_j$ because they belong to $B$, $C$ or $D$ (since $j \geq t_1$ and $t_1$ is the boundary between $A$ and $B$); and thus
    \begin{equation*}
        \sum_{j'=j+1}^n g_{j,j'} = (n-j) \cdot \hg_j = R_j + \Mtone,
    \end{equation*}
    Also,
    \begin{equation*}
        \sum_{i'=j}^{i} g_{i',j} = R_{j-1}-R_i,
        \qquad
        \sum_{i'=j}^{i} t_{i',j} = [1-R_{j-1}-\Mtone]^+.
    \end{equation*}
    Since an integral $B$-client satisfies $R_{j-1} \le \MRtone$, the $t$-term is $\MRtone-R_{j-1}$.
    Finally,
    \begin{equation*}
        \sum_{i'=j}^{i} \bar s_{i',j} =
        \begin{cases}
            R_j-R_{f_j}, & i=f_j,\\
            R_j-R_i+\Mcoeffmone, & i>f_j.
        \end{cases}
    \end{equation*}
    Therefore, if $i=f_j$ then
    \begin{equation*}
        \text{credit}_{i,j} = (R_j+\Mtone) + (R_{j-1}-R_{f_j}) + (\MRtone-R_{j-1}) - (R_j-R_{f_j}) = 1,
    \end{equation*}
    and if $i>f_j$ then
    \begin{equation*}
        \text{credit}_{i,j} = (R_j+\Mtone) + (R_{j-1}-R_i) + (\MRtone-R_{j-1}) - (R_j-R_i+\Mcoeffmone) = 1-\Mcoeffmone.
    \end{equation*}
    This is exactly the stated formula.
\end{proof}

\begin{lemma}
\label{lem:credit-C-after}
    If client $j$ belongs integrally to group $C$, then for $i \geq f_j$
    \begin{equation*}
        \text{\emph{credit}}_{i,j} = 1 - \ind[i > f_j] \cdot \Mcoeffmone.
    \end{equation*}
\end{lemma}
\begin{proof}
    For a $C$-client we again have, through the same arguments as in the previous lemma, that for every $j' > j$ we have $g_{j,j'} = \hg_j$ because they belong to $C$ or $D$ (since $j>t_2\geq t_1$ and $t_2$ is the boundary between $B$ and $C$ and $t_1$ between $A$ and $B$); and thus
    \begin{equation*}
        \sum_{j'=j+1}^n g_{j,j'} = (n-j) \cdot \hg_j = R_j + \Mtone,
        \qquad
        \sum_{i'=j}^{i} g_{i',j} = R_{j-1}-R_i.
    \end{equation*}
    Moreover,
    \begin{equation*}
        \sum_{i'=j}^{i} t_{i',j} = 1-R_{j-1}-\Mtone = \MRtone-R_{j-1},
    \end{equation*}
    and
    \begin{equation*}
        \sum_{i'=j}^{i} \bar s_{i',j} =
        \begin{cases}
            R_j-R_{f_j}, & i=f_j,\\
            R_j-R_i+\Mcoeffmone, & i>f_j.
        \end{cases}
    \end{equation*}
    The same calculation as in the proof of Lemma \ref{lem:credit-B-after} gives $\text{credit}_{i,j}=1$ for $i=f_j$ and $\text{credit}_{i,j}=1-\Mcoeffmone$ for $i>f_j$.
\end{proof}

\begin{lemma}
    \label{lem:credit-D-after}
    If client $j$ belongs integrally to group $D$, then for $i \geq f_j$
    \begin{equation*}
        \text{\emph{credit}}_{i,j} = 1.
    \end{equation*}
\end{lemma}
\begin{proof}
    If $j=t_3^\star=n$, then $i=f_j=n$, $g_{n,n}=\hg_n^\star$, and $t_{n,n}=1-\hg_n^\star$, so $\text{credit}_{n,n}=1$.
    Assume now that $j \neq t_3^\star=n$.
    Then $j \in D$ implies $\hg_j=0$, hence $\sum_{j'=j+1}^n g_{j,j'}=0$.
    Also, if $j \le t_1$ then $\chi_j=0$, while if $j>t_1$ then the $D$-formula already gives $g_{i',j}=\hg_{i'}$ for all $i' \ge j$.
    Thus, in all cases,
    \begin{equation*}
        \sum_{i'=j}^{i} g_{i',j} = R_{j-1}-R_i.
    \end{equation*}
    Moreover, $\sum_{i'=j}^{i} \bar s_{i',j}=0$ and $\sum_{i'=j}^{i} t_{i',j}=1-R_{j-1}$.
    Finally, $j \in D$ implies $f_j \ge t_3^\star$, so for $i \ge f_j$ we have $R_i=0$.
    Therefore
    \begin{equation*}
        \text{credit}_{i,j} = 0 + (R_{j-1}-R_i) + (1-R_{j-1}) = 1.
    \end{equation*}
\end{proof}

\begin{lemma}\label{lem:credit-fractional-after}
    If client $j$ is fractionally assigned, then for every $i \geq f_j$ we have
    \begin{equation*}
        \text{\emph{credit}}_{i,j}
        = \mu_j^A \text{\emph{credit}}_{i,j}^A 
        + \mu_j^B \text{\emph{credit}}_{i,j}^B 
        + \mu_j^C \text{\emph{credit}}_{i,j}^C
        + \mu_j^D \text{\emph{credit}}_{i,j}^D \ge 0,
    \end{equation*}
    where $\text{credit}_{i,j}^A$, $\text{credit}_{i,j}^B$, $\text{credit}_{i,j}^C$, and $\text{credit}_{i,j}^D$ are the contributions to the credit from groups $A$, $B$, $C$, and $D$ as defined in Lemmas \ref{lem:credit-A-after}, \ref{lem:credit-B-after}, \ref{lem:credit-C-after}, and \ref{lem:credit-D-after}, respectively.
\end{lemma}
\begin{proof}
    The integral cases are exactly from Lemmas \ref{lem:credit-A-after} to \ref{lem:credit-D-after}. It remains to prove the formula for the fractional clients.

    If $j<t_1$, then $\mu_j^B=\mu_j^C=0$ and $\mu_j^A+\mu_j^D=1$. Also,
    \begin{equation*}
        \sum_{j'=j+1}^{n} g_{j,j'}=(n-j-\Mfrac\, z)\hg_j=
        \mu_j^A\bigl(1-[\MRttwo-R_{f_j}]^+\bigr),
    \end{equation*}
    and by Equation~\ref{formula:fractional-assignments-convex-combination} the remaining term is the corresponding convex combination of the $A$- and $D$-contributions. Thus
    \begin{equation*}
        \text{credit}_{i,j}=
        \mu_j^A\text{credit}_{i,j}^A+
        \mu_j^D\text{credit}_{i,j}^D.
    \end{equation*}

    If $t_1<j\le t_2$, then $\mu_j^A=\mu_j^C=0$, $\mu_j^B+\mu_j^D=1$, and
    \begin{equation*}
        \sum_{j'=j+1}^{n} g_{j,j'}=(n-j)\hg_j=
        \mu_j^B(R_j+\Mtone),
    \end{equation*}
    so again
    \begin{equation*}
        \text{credit}_{i,j}=
        \mu_j^B\text{credit}_{i,j}^B+
        \mu_j^D\text{credit}_{i,j}^D.
    \end{equation*}

    If $t_2<j\le t_3^\star$, then $\mu_j^A=\mu_j^B=0$, $\mu_j^C+\mu_j^D=1$, and
    \begin{equation*}
        \sum_{j'=j+1}^{n} g_{j,j'}=(n-j)\hg_j=
        \mu_j^C(R_j+\Mtone),
    \end{equation*}
    which gives
    \begin{equation*}
        \text{credit}_{i,j}=
        \mu_j^C\text{credit}_{i,j}^C+
        \mu_j^D\text{credit}_{i,j}^D.
    \end{equation*}

    It remains to consider the boundary client $j=t_1$. Here $\mu_j^C=0$,
    \begin{equation*}
        \mu_j^A=(1-\alpha)(1-\mu_j^D),
        \qquad
        \mu_j^B=\alpha(1-\mu_j^D),
        \qquad
        \hg_j=\hg_j^A+\hg_j^B.
    \end{equation*}
    We have $i\ge f_j$, so
    \begin{equation*}
        \sum_{i'=j}^{i} g_{i',j}=R_{j-1}-R_i-\mu_j^A\sum_{i'=f_j+1}^{i}\delta_{i',j},
        \qquad
        \sum_{i'=j}^{i} t_{i',j}=\mu_j^D(1-R_{j-1}),
    \end{equation*}
    where we used that $[1-R_{j-1}-\Mtone]^+=0$ at $j=t_1$.
    Moreover,
    \begin{equation*}
        \sum_{i'=j}^{i} \bar s_{i',j}
        = \mu_j^A\left(R_{j-1}-R_i-[\MRttwo-R_{f_j}]^+\right)
        + \mu_j^B\left(R_j-R_i+\ind[i>f_j]\cdot\Mcoeffmone\right).
    \end{equation*}
    Using
    \begin{equation*}
        \sum_{j'=j+1}^{n} g_{j,j'}=(n-j-\Mfrac\, z)\hg_j^A+(n-j)\hg_j^B,
    \end{equation*}
    with
    \begin{equation*}
        (n-j-\Mfrac\, z)\hg_j^A=
        \mu_j^A\bigl(1-[\MRttwo-R_{f_j}]^+\bigr)-\mu_j^B\hg_j^A,
        \qquad
        (n-j)\hg_j^B=
        \mu_j^B(R_j+\Mtone),
    \end{equation*}
    we obtain
    \begin{align*}
        \text{credit}_{i,j}
        &= (n-j-\Mfrac\, z)\hg_j^A+(n-j)\hg_j^B + R_{j-1}-R_i - \mu_j^A\sum_{i'=f_j+1}^{i}\delta_{i',j} + \mu_j^D(1-R_{j-1}) \\
        &\qquad - \mu_j^A\left(R_{j-1}-R_i-[\MRttwo-R_{f_j}]^+\right)
        - \mu_j^B\left(R_j-R_i+\ind[i>f_j]\cdot\Mcoeffmone\right) \\
        &= \mu_j^A\left(1-\sum_{i'=f_j+1}^{i}\delta_{i',j}\right)
        + \mu_j^B(R_j+\Mtone+\hg_j^B)
        + \mu_j^D(1-R_i)
        - \mu_j^B\ind[i>f_j]\cdot\Mcoeffmone.
    \end{align*}
    Finally, $\hg_j^B=\MRtone - R_j$, so $R_j+\Mtone+\hg_j^B=1$. Also, if $\mu_j^D>0$, then $f_j\geq t_3^\star$, so $i\ge f_j$ implies $R_i=0$. Therefore
    \begin{equation*}
        \text{credit}_{i,j}
        = \mu_j^A\left(1-\sum_{i'=f_j+1}^{i}\delta_{i',j}\right)
        + \mu_j^B\left(1-\ind[i>f_j]\cdot\Mcoeffmone\right)
        + \mu_j^D\cdot 1,
    \end{equation*}
    which then implies:
    \begin{equation*}
        \text{credit}_{i,j}
        = \mu_j^A\text{credit}_{i,j}^A
        + \mu_j^B\text{credit}_{i,j}^B
        + \mu_j^D\text{credit}_{i,j}^D,
    \end{equation*}
    which is exactly the claimed convex combination given that $\mu_j^C=0$ for $j=t_1$.
\end{proof}

\paragraph{Credit calculation for \texorpdfstring{$i < f_j$}{i < f\_j}.}

Using the same definition of credit, we now treat the case $i < f_j$.

\defcreditij*

\begin{lemma}
    \label{lem:credit-A-before}
    If client $j$ belongs integrally to group $A$, then for $j \leq i < f_j$
    \begin{equation*}
        \text{\emph{credit}}_{i,j} = 1 - [\MRttwo - R_{f_j}]^+ + \sum_{i'=j}^i \sigma_{i',j} \ge 0.
    \end{equation*}
\end{lemma}
\begin{proof}
    As in the proof of Lemma~\ref{lem:credit-A-after},
    \begin{equation*}
        \sum_{j'=j+1}^n g_{j,j'} = 1 - [\MRttwo - R_{f_j}]^+.
    \end{equation*}
    Since $i<f_j$, we also have $\delta_{i',j}=0$ and $t_{i',j}=0$ for all $j \le i' \le i$, hence
    \begin{equation*}
        \sum_{i'=j}^{i} g_{i',j} = R_{j-1}-R_i,
        \qquad
        \sum_{i'=j}^{i} t_{i',j}=0,
        \qquad
        \sum_{i'=j}^{i} \bar s_{i',j} = R_{j-1}-R_i-\sum_{i'=j}^{i} \sigma_{i',j}.
    \end{equation*}
    Substituting gives
    \begin{equation*}
        \text{credit}_{i,j} = 1 - [\MRttwo - R_{f_j}]^+ + \sum_{i'=j}^{i} \sigma_{i',j}.
    \end{equation*}
    This is non-negative because $[\MRttwo-R_{f_j}]^+ \le \MRttwo < 1$ and $\sigma_{i',j} \ge 0$.
\end{proof}

\begin{lemma}
    \label{lem:credit-B-before}
    If client $j$ belongs integrally to group $B$, then for $j \leq i < f_j$
    \begin{equation*}
        \text{\emph{credit}}_{i,j} = \Mtone + \min\{R_{j-1},\MRtone\} \ge 0.
    \end{equation*}
\end{lemma}
\begin{proof}
    For a client $j$ integrally in $B$, we have that $j \geq t_1$ from the membership assignment in Equation~\ref{eq:membership-assignment}. We have that $R_{j-1} \le \MRtone$, with equality if $j = t_1$. Thus, showing that the credit is $\Mtone + R_{j-1}$ suffices to conclude the claim. As in the proof of Lemma~\ref{lem:credit-B-after},
    \begin{equation*}
        \sum_{j'=j+1}^n g_{j,j'} = R_j+\Mtone,
        \qquad
        \sum_{i'=j}^{i} g_{i',j} = R_{j-1}-R_i,
        \qquad
        \sum_{i'=j}^{i} t_{i',j}=0.
    \end{equation*}
    Because $i<f_j$, the extra $\Mcoeffmone$ has not been sent yet, so
    \begin{equation*}
        \sum_{i'=j}^{i} \bar s_{i',j} = \sum_{i'=j+1}^{i} \hg_{i'} = R_j-R_i.
    \end{equation*}
    Therefore
    \begin{equation*}
        \text{credit}_{i,j} = (R_j+\Mtone) + (R_{j-1}-R_i) - (R_j-R_i) = \Mtone + R_{j-1} \ge 0.
    \end{equation*}
\end{proof}

\begin{lemma}
    \label{lem:credit-C-before}
    If client $j$ belongs integrally to group $C$, then for $j \leq i < f_j$
    \begin{equation*}
        \text{\emph{credit}}_{i,j} = \Mtone + R_{j-1} \ge 0.
    \end{equation*}
\end{lemma}
\begin{proof}
    The calculation is identical to the one in the proof of Lemma~\ref{lem:credit-B-before}: we have
    \begin{equation*}
        \sum_{j'=j+1}^n g_{j,j'} = R_j+\Mtone,
        \qquad
        \sum_{i'=j}^{i} g_{i',j} = R_{j-1}-R_i,
        \qquad
        \sum_{i'=j}^{i} t_{i',j}=0,
    \end{equation*}
    and, since $i<f_j$,
    \begin{equation*}
        \sum_{i'=j}^{i} \bar s_{i',j} = R_j-R_i.
    \end{equation*}
    Thus
    \begin{equation*}
        \text{credit}_{i,j} = (R_j+\Mtone) + (R_{j-1}-R_i) - (R_j-R_i) = \Mtone + R_{j-1} \ge 0.
    \end{equation*}
\end{proof}

\begin{lemma}
    \label{lem:credit-D-before}
    If client $j$ belongs integrally to group $D$, then for $j \leq i < f_j$
    \begin{equation*}
        \text{\emph{credit}}_{i,j} = R_{j-1}-R_i \ge 0.
    \end{equation*}
\end{lemma}
\begin{proof}
    Because $j \in D$, we have $\hg_j=0$, hence $\sum_{j'=j+1}^n g_{j,j'}=0$.
    Also, if $j \le t_1$ then $\chi_j=0$, while if $j>t_1$ then the $D$-formula gives $g_{i',j}=\hg_{i'}$ for all $i' \ge j$.
    Therefore
    \begin{equation*}
        \sum_{i'=j}^{i} g_{i',j} = R_{j-1}-R_i.
    \end{equation*}
    Since $i<f_j$, we also have $\sum_{i'=j}^{i} t_{i',j}=0$, and for group $D$ we always have $\sum_{i'=j}^{i} \bar s_{i',j}=0$.
    Hence
    \begin{equation*}
        \text{credit}_{i,j} = R_{j-1}-R_i \ge 0.
    \end{equation*}
\end{proof}

\begin{lemma}\label{lem:credit-fractional-before}
    If client $j$ is fractionally assigned, then for every $j \leq i < f_j$ we have
    \begin{equation*}
        \text{\emph{credit}}_{i,j} = 
        \mu_j^A\text{\emph{credit}}_{i,j}^A +
        \mu_j^B\text{\emph{credit}}_{i,j}^B +
        \mu_j^C\text{\emph{credit}}_{i,j}^C +
        \mu_j^D\text{\emph{credit}}_{i,j}^D \ge 0,
    \end{equation*}
    where $\text{credit}_{i,j}^A$, $\text{credit}_{i,j}^B$, $\text{credit}_{i,j}^C$, and $\text{credit}_{i,j}^D$ are the contributions to the credit from groups $A$, $B$, $C$, and $D$ as defined in Lemmas \ref{lem:credit-A-before}, \ref{lem:credit-B-before}, \ref{lem:credit-C-before}, and \ref{lem:credit-D-before}, respectively.
\end{lemma}
\begin{proof}
    The integral cases are exactly from Lemmas \ref{lem:credit-A-before} to \ref{lem:credit-D-before}. It remains to prove the formula for the fractional clients.

    If $j<t_1$, then by Observation~\ref{obs:mu-sum} we have $\mu_j^B=\mu_j^C=0$ and $\mu_j^A+\mu_j^D=1$. Moreover,
    \begin{equation*}
        \sum_{j'=j+1}^{n} g_{j,j'}=(n-j-\Mfrac\, z)\hg_j=
        \mu_j^A\bigl(1-[\MRttwo-R_{f_j}]^+\bigr),
    \end{equation*}
    and, by Equation~\ref{formula:fractional-assignments-convex-combination}, the term
    \begin{equation*}
        \sum_{i'=j}^{i} \bigl(g_{i',j}+t_{i',j}-\bar s_{i',j}\bigr)
    \end{equation*}
    is the same convex combination of the $A$- and $D$-contributions. Hence
    \begin{equation*}
        \text{credit}_{i,j}=
        \mu_j^A\text{credit}_{i,j}^A+
        \mu_j^D\text{credit}_{i,j}^D.
    \end{equation*}

    If $t_1<j\le t_2$, then $\mu_j^A=\mu_j^C=0$ and $\mu_j^B+\mu_j^D=1$. Also,
    \begin{equation*}
        \sum_{j'=j+1}^{n} g_{j,j'}=(n-j)\hg_j=
        \mu_j^B(R_j+\Mtone),
    \end{equation*}
    and again the remaining term is the corresponding convex combination of the $B$- and $D$-contributions. Thus
    \begin{equation*}
        \text{credit}_{i,j}=
        \mu_j^B\text{credit}_{i,j}^B+
        \mu_j^D\text{credit}_{i,j}^D.
    \end{equation*}

    If $t_2<j\le t_3^\star$, the proof is identical: now $\mu_j^A=\mu_j^B=0$, $\mu_j^C+\mu_j^D=1$, and
    \begin{equation*}
        \sum_{j'=j+1}^{n} g_{j,j'}=(n-j)\hg_j=
        \mu_j^C(R_j+\Mtone),
    \end{equation*}
    so
    \begin{equation*}
        \text{credit}_{i,j}=
        \mu_j^C\text{credit}_{i,j}^C+
        \mu_j^D\text{credit}_{i,j}^D.
    \end{equation*}

    It remains to consider the boundary client $j=t_1$. Here $\mu_j^C=0$,
    \begin{equation*}
        \mu_j^A=(1-\alpha)(1-\mu_j^D),
        \qquad
        \mu_j^B=\alpha(1-\mu_j^D),
        \qquad
        \hg_j=\hg_j^A+\hg_j^B.
    \end{equation*}
    Since $i<f_j$, we have $\delta_{i',j}=0$ and $t_{i',j}=0$ for $j\le i'\le i$. Therefore
    \begin{equation*}
        \sum_{i'=j}^{i} g_{i',j}=R_{j-1}-R_i,
    \end{equation*}
    while
    \begin{equation*}
        \sum_{i'=j}^{i} \bar s_{i',j}
        =\mu_j^A\left(R_{j-1}-R_i-\sum_{i'=j}^{i}\sigma_{i',j}\right)
        +\mu_j^B(R_j-R_i).
    \end{equation*}
    For the initial term we have
    \begin{equation*}
        \sum_{j'=j+1}^{n} g_{j,j'}=(n-j-\Mfrac\, z)\hg_j^A+(n-j)\hg_j^B.
    \end{equation*}
    Using
    \begin{equation*}
        (n-j-\Mfrac\, z)\hg_j^A=
        \mu_j^A\bigl(1-[\MRttwo-R_{f_j}]^+\bigr)-\mu_j^B\hg_j^A
    \end{equation*}
    and
    \begin{equation*}
        (n-j)\hg_j^B=
        \mu_j^B(R_j+\Mtone),
    \end{equation*}
    we obtain
    \begin{align*}
        \text{credit}_{i,j}
        &= (n-j-\Mfrac\, z)\hg_j^A+(n-j)\hg_j^B + (R_{j-1}-R_i) - \sum_{i'=j}^{i} \bar s_{i',j} \\
        &= \mu_j^A\bigl(1-[\MRttwo-R_{f_j}]^+\bigr) + \mu_j^B(R_j+\Mtone-\hg_j^A) \\
        &\qquad + \mu_j^B(R_{j-1}-R_i) + \mu_j^D(R_{j-1}-R_i) - \mu_j^B(R_j-R_i) + \mu_j^A\sum_{i'=j}^{i}\sigma_{i',j} \\
        &= \mu_j^A\left(1-[\MRttwo-R_{f_j}]^+ + \sum_{i'=j}^{i}\sigma_{i',j}\right)
        + \mu_j^B(R_j+\Mtone+\hg_j^B)
        + \mu_j^D(R_{j-1}-R_i).
    \end{align*}
    Finally, $\hg_j^B=\MRtone - R_j$, so $R_j+\Mtone+\hg_j^B=1$. Hence
    \begin{equation*}
        \text{credit}_{i,j}
        = \mu_j^A\left(1-[\MRttwo-R_{f_j}]^+ + \sum_{i'=j}^{i}\sigma_{i',j}\right)
        + \mu_j^B\cdot 1
        + \mu_j^D(R_{j-1}-R_i),
    \end{equation*}
    which is exactly
    \begin{equation*}
        \text{credit}_{i,j}
        = \mu_j^A\text{credit}_{i,j}^A
        + \mu_j^B\text{credit}_{i,j}^B
        + \mu_j^D\text{credit}_{i,j}^D,
    \end{equation*}
    since for $j=t_1$ we have $\text{credit}_{i,j}^B=\Mtone+\min\{R_{j-1},\MRtone\}=1$ as $R_{j-1}\geq \MRtone$.

    The non-negativity follows immediately, since each of the integral expressions $\text{credit}_{i,j}^A$, $\text{credit}_{i,j}^B$, $\text{credit}_{i,j}^C$, and $\text{credit}_{i,j}^D$ is non-negative.
\end{proof}

\paragraph{Production calculation.}

\begin{lemma}
    \label{lem:gen-A}
    If client $j$ belongs integrally to group $A$, then
    \begin{equation*}
        \sum_{i=1}^n g_{i,j} = 1 - \Mcoeffmone \cdot \left(1 - \frac{[\MRttwo - R_{f_j}]^+}{\MRttwo}\right).
    \end{equation*}
\end{lemma}
\begin{proof}
    We prove a stronger statement, that for every $j$ with $\mu_j^A > 0$ we have the previous statement but substituting $g_{i,j}^A$ for $g_{i,j}$, where $g_{i,j}^A$ is defined as in the $A$-assignment used in Equation~\ref{formula:fractional-assignments-convex-combination}. 
    If $j \in A$, then
    \begin{equation*}
        \sum_{i=1}^n g_{i,j}^A = \sum_{i=1}^n (\hg_i-\delta_{i,j}) = 1 - \sum_{i=f_j+1}^{t_3^\star} \delta_{i,j}
        = 1 - \Mcoeffmone \cdot \left(1 - \frac{[\MRttwo - R_{f_j}]^+}{\MRttwo}\right)
    \end{equation*}
    by Property~\ref{property:delta}.
\end{proof}

\begin{lemma}
    \label{lem:gen-B}
    If client $j$ belongs integrally to group $B$, then
    \begin{equation*}
        \sum_{i=1}^n g_{i,j} = 1.
    \end{equation*}
\end{lemma}
\begin{proof}
    We prove a stronger statement, that for every $j$ with $\mu_j^B > 0$ we have the previous statement but substituting $g_{i,j}^B$ for $g_{i,j}$, where $g_{i,j}^B$ is defined as in the $B$-assignment used in Equation~\ref{formula:fractional-assignments-convex-combination}.
    If $j \in B$, then $g_{i,j}^B=\hg_i$ for every $i$, hence $\sum_{i=1}^n g_{i,j}^B=\sum_{i=1}^n \hg_i=1$.
\end{proof}

\begin{lemma}
    \label{lem:gen-C}
    If client $j$ belongs integrally to group $C$, then
    \begin{equation*}
        \sum_{i=1}^n g_{i,j} = 1.
    \end{equation*}
\end{lemma}
\begin{proof}
    We prove a stronger statement, that for every $j$ with $\mu_j^C > 0$ we have the previous statement but substituting $g_{i,j}^C$ for $g_{i,j}$, where $g_{i,j}^C$ is defined as in the $C$-assignment used in Equation~\ref{formula:fractional-assignments-convex-combination}.
    If $j \in C$, then $g_{i,j}^C=\hg_i$ for every $i$, so $\sum_{i=1}^n g_{i,j}^C=\sum_{i=1}^n \hg_i=1$.
\end{proof}

\begin{lemma}\label{lem:gen-D}
    If client $j$ belongs integrally to group $D$, then
    \begin{equation*}
       \sum_{i=1}^n g_{i,j} = 1 - \Mcoeffmone \cdot \left[1 - \frac{R_{j-1}}{\MRttwo}\right]^+.
    \end{equation*}
\end{lemma}
\begin{proof}
    We prove a stronger statement, that for every $j$ with $\mu_j^D > 0$ we have the previous statement but substituting $g_{i,j}^D$ for $g_{i,j}$, where $g_{i,j}^D$ is defined as in the $D$-assignment used in Equation~\ref{formula:fractional-assignments-convex-combination}.
    For $j \in D$, by the $D$-assignment,
    \begin{equation*}
        \sum_{i=1}^n g_{i,j}^D = 1 - \chi_j \cdot \Mfrac \cdot \left(\sum_{i=1}^{t_1-1} \hg_i + \hg_{t_1}^A\right).
    \end{equation*}
    For $j\leq t_1$, the equality is direct because $\chi_j=\left[1 - \frac{R_{j-1}}{\MRttwo}\right]^+=0$, since $R_{j-1}\geq R_{t_1-1}\geq \MRtone > \MRttwo$. For $j>t_1$:
    \begin{equation*}
        \sum_{i=1}^{t_1-1} \hg_i + \hg_{t_1}^A = 1 - R_{t_1 - 1} + (\hg_{t_1} - \hg_{t_1}^B) = 1 - R_{t_1} - (\MRtone - R_{t_1}) = \Mtone\,.
    \end{equation*}
    we obtain
    \begin{equation*}
        \sum_{i=1}^n g_{i,j}^D = 1 - \chi_j \cdot \Mcoeffmone
        = 1 - \Mcoeffmone \cdot \left[1 - \frac{R_{j-1}}{\MRttwo}\right]^+.
    \end{equation*}
\end{proof}

\begin{lemma}\label{lem:gen-fractional}
    If client $j$ belongs fractionally to groups $A,B,C,D$ with coefficients $\mu_j^A, \mu_j^B, \mu_j^C, \mu_j^D$, then
    \begin{equation*}
        \sum_{i=1}^n g_{i,j} = 1 - \Mcoeffmone \cdot \mu_j^A \cdot \left(1 - \frac{[\MRttwo - R_{f_j}]^+}{\MRttwo}\right) - \Mcoeffmone \cdot \mu_j^D \cdot \left[1 - \frac{R_{j-1}}{\MRttwo}\right]^+
    \end{equation*}
\end{lemma}
\begin{proof}
    The proof is a direct calculation using the definition of $g_{i,j}$ as a convex combination of the $g_{i,j}^X$ for $X \in \{A,B,C,D\}$ using the previous lemmas' proofs which used no property of the membership being integral.
    \begin{align*}
        \sum_{i=1}^n g_{i,j}
        &= \sum_{i=1}^n \sum_{X\in\{A,B,C,D\}} \mu_j^X g_{i,j}^X
        = \sum_{X\in\{A,B,C,D\}} \mu_j^X \sum_{i=1}^n g_{i,j}^X \\
        &= \mu_j^A \cdot \left(1 - \Mcoeffmone \cdot \left(1 - \frac{[\MRttwo - R_{f_j}]^+}{\MRttwo}\right)\right) + (\mu_j^B + \mu_j^C) \cdot 1 + \mu_j^D \cdot \left(1 - \Mcoeffmone \cdot \left[1 - \frac{R_{j-1}}{\MRttwo}\right]^+\right)\\
        &= 1 - \Mcoeffmone \cdot \mu_j^A \cdot \left(1 - \frac{[\MRttwo - R_{f_j}]^+}{\MRttwo}\right) - \Mcoeffmone \cdot \mu_j^D \cdot \left[1 - \frac{R_{j-1}}{\MRttwo}\right]^+.
    \end{align*}
\end{proof}

\paragraph{Transfer calculation.}

\begin{lemma}
    \label{lem:tran-A}
    If client $j$ belongs integrally to group $A$, then
    \begin{equation*}
        \sum_{i=1}^n (\bar s_{i,j} + t_{i,j}) = R_{j-1} - [\MRttwo - R_{f_j}]^+.
    \end{equation*}
\end{lemma}
\begin{proof}
    We prove a stronger statement, that for every $j$ with $\mu_j^A > 0$ we have the previous statement but substituting $\bar s_{i,j}^A$ and $t_{i,j}^A$ for $\bar s_{i,j}$ and $t_{i,j}$, where $\bar s_{i,j}^A$ and $t_{i,j}^A$ are defined as in the $A$-assignment used in Equation~\ref{formula:fractional-assignments-convex-combination}.
    For group $A$ we have $\sum_{i=1}^n t_{i,j}^A=0$, and
    \begin{equation*}
        \sum_{i=1}^n \bar s_{i,j}^A = \sum_{i=j}^n (\hg_i-\sigma_{i,j}) = R_{j-1} - \sum_{i=j}^{f_j} \sigma_{i,j}
    = R_{j-1} - [\MRttwo - R_{f_j}]^+
    \end{equation*}
    by Property~\ref{property:sigma}.
\end{proof}

\begin{lemma}
    \label{lem:tran-B}
    If client $j$ belongs integrally to group $B$, then
    \begin{equation*}
        \sum_{i=1}^n (\bar s_{i,j} + t_{i,j}) \leq \Mttwo.
    \end{equation*}
\end{lemma}
\begin{proof}
    We prove a stronger statement, that for every $j$ with $\mu_j^B > 0$ we have the previous statement but substituting $\bar s_{i,j}^B$ and $t_{i,j}^B$ for $\bar s_{i,j}$ and $t_{i,j}$, where $\bar s_{i,j}^B$ and $t_{i,j}^B$ are defined as in the $B$-assignment used in Equation~\ref{formula:fractional-assignments-convex-combination}.
    We have
    \begin{equation*}
        \sum_{i=1}^n \bar s_{i,j}^B \leq \sum_{i=j+1}^n \hg_i + \Mcoeffmone = R_j + \Mcoeffmone,
    \end{equation*}
    where the inequality is only because the bookkeeping entry at $f_j+1$ may lie outside the LP index set when $f_j=n$.
    Also,
    \begin{equation*}
        \sum_{i=1}^n t_{i,j}^B = [1-R_{j-1}-\Mtone]^+ \leq [1-R_{j}-\Mtone]^+ = \MRtone - R_j,
    \end{equation*}
    because $R_{j-1} \ge R_j$ and any $B$-client (whether integrally or partially) satisfies $R_{j} \le \MRtone$.
    Therefore
    \begin{equation*}
        \sum_{i=1}^n (\bar s_{i,j}^B+t_{i,j}^B) \le (R_j+\Mcoeffmone) + (\MRtone-R_j) = \Mttwo.
    \end{equation*}
    because $\Mcoeffmone + \MRtone = \Mttwo$.
\end{proof}

\begin{lemma}
    \label{lem:tran-C}
    If client $j$ belongs integrally to group $C$, then
    \begin{equation*}
        \sum_{i=1}^n (\bar s_{i,j} + t_{i,j}) \leq \Mttwo.
    \end{equation*}
\end{lemma}
\begin{proof}
    We prove a stronger statement, that for every $j$ with $\mu_j^C > 0$ we have the previous statement but substituting $\bar s_{i,j}^C$ and $t_{i,j}^C$ for $\bar s_{i,j}$ and $t_{i,j}$, where $\bar s_{i,j}^C$ and $t_{i,j}^C$ are defined as in the $C$-assignment used in Equation~\ref{formula:fractional-assignments-convex-combination}.
    Then, exactly as above,
    \begin{equation*}
        \sum_{i=1}^n \bar s_{i,j}^C \le R_j + \Mcoeffmone.
    \end{equation*}
    For group $C$ we also have
    \begin{equation*}
        \sum_{i=1}^n t_{i,j}^C = 1-R_{j-1}-\Mtone = \MRtone-R_{j-1} \le \MRtone-R_j.
    \end{equation*}
    Hence
    \begin{equation*}
        \sum_{i=1}^n (\bar s_{i,j}^C+t_{i,j}^C) \le (R_j+\Mcoeffmone) + (\MRtone-R_j) = \Mttwo.
    \end{equation*}
    because $\Mcoeffmone + \MRtone = \Mttwo$.
\end{proof}

\begin{lemma}
    \label{lem:tran-D}
    If client $j$ belongs integrally to group $D$, then
    \begin{equation*}
        \sum_{i=1}^n (\bar s_{i,j} + t_{i,j}) = 1 - R_{j-1}.
    \end{equation*}
\end{lemma}
\begin{proof}
    We prove a stronger statement, that for every $j$ with $\mu_j^D > 0$ we have the previous statement but substituting $\bar s_{i,j}^D$ and $t_{i,j}^D$ for $\bar s_{i,j}$ and $t_{i,j}$, where $\bar s_{i,j}^D$ and $t_{i,j}^D$ are defined as in the $D$-assignment used in Equation~\ref{formula:fractional-assignments-convex-combination}.
    For group $D$ we always have $\bar s_{i,j}^D=0$, and
    \begin{equation*}
        \sum_{i=1}^n t_{i,j}^D = t_{f_j,j}^D = 1-R_{j-1}.
    \end{equation*}
    The result follows.
\end{proof}

\begin{lemma}
    \label{lem:tran-fractional}
    If client $j$ belongs fractionally to groups $A,B,C,D$ with coefficients $\mu_j^A, \mu_j^B, \mu_j^C, \mu_j^D$, then
    \begin{equation*}
        \sum_{i=1}^n (\bar s_{i,j} + t_{i,j}) \leq \mu_j^A \cdot (R_{j-1} - [\MRttwo - R_{f_j}]^+) + (\mu_j^B + \mu_j^C) \cdot (\Mttwo) + \mu_j^D \cdot (1 - R_{j-1}).
    \end{equation*}
\end{lemma}
\begin{proof}
    The proof is a direct calculation using the definition of $\bar s_{i,j}$ and $t_{i,j}$ as convex combinations of the $\bar s_{i,j}^X$ and $t_{i,j}^X$ for $X \in \{A,B,C,D\}$ and using the previous lemmas' proofs which did not use integrality of the memberships.
    \begin{align*}
        \sum_{i=1}^n (\bar s_{i,j} + t_{i,j})
        &= \sum_{i=1}^n \sum_{X\in\{A,B,C,D\}} \mu_j^X (\bar s_{i,j}^X + t_{i,j}^X)
        = \sum_{X\in\{A,B,C,D\}} \mu_j^X \sum_{i=1}^n (\bar s_{i,j}^X + t_{i,j}^X) \\
        &\leq \mu_j^A \cdot (R_{j-1} - [\MRttwo - R_{f_j}]^+) + \mu_j^B \cdot (\Mttwo) + \mu_j^C \cdot (\Mttwo) + \mu_j^D \cdot (1 - R_{j-1}),
    \end{align*}
\end{proof}

\paragraph{Budget calculation.}

We use $\MRttwo = \Mcoeffmone + \Mcoeffmone^2$ and $\Mcoeff = \Mcoeffmone + 1$, hence $\MRttwo = \Mcoeffmone \cdot \Mcoeff$ and $\Mcoeff + 1 - \MRttwo = 2- \Mcoeffmone^2 = \Mrho$, as seen in Property~\ref{property:coefficients}.

\begin{lemma}
    \label{lem:budget-A}
    If client $j$ belongs integrally to group $A$, then
    \begin{equation*}
        \Mcoeff \cdot \sum_{i=1}^n g_{i,j} + \sum_{i=1}^n (\bar s_{i,j} + t_{i,j}) \leq \Mrho.
    \end{equation*}
\end{lemma}
\begin{proof}
    By Lemmas \ref{lem:gen-A} and \ref{lem:tran-A},
    \begin{align*}
        \Mcoeff \cdot \sum_{i=1}^n g_{i,j} + \sum_{i=1}^n (\bar s_{i,j} + t_{i,j})
        &= \Mcoeff \cdot \left(1 - \Mcoeffmone \cdot \left(1 - \frac{[\MRttwo - R_{f_j}]^+}{\MRttwo}\right)\right) + R_{j-1} - [\MRttwo - R_{f_j}]^+ \\
        &= \Mcoeff - \MRttwo + R_{j-1} \\
        &\le \Mcoeff + 1 - \MRttwo = \Mrho,
    \end{align*}
    using $\Mcoeff \cdot \Mcoeffmone = \MRttwo$, $R_{j-1} \le 1$ and $\Mcoeff + 1 - \MRttwo = \Mrho$ which follow from the definitions of $\Mcoeff$ and $\MRttwo$ as reminded above.
\end{proof}

\begin{lemma}
    \label{lem:budget-B}
    If client $j$ belongs integrally to group $B$, then
    \begin{equation*}
        \Mcoeff \cdot \sum_{i=1}^n g_{i,j} + \sum_{i=1}^n (\bar s_{i,j} + t_{i,j}) \leq \Mrho.
    \end{equation*}
\end{lemma}
\begin{proof}
    By Lemmas \ref{lem:gen-B} and \ref{lem:tran-B},
    \begin{equation*}
        \Mcoeff \cdot \sum_{i=1}^n g_{i,j} + \sum_{i=1}^n (\bar s_{i,j} + t_{i,j}) \le \Mcoeff + \Mttwo = \Mrho.
    \end{equation*}
    using $\Mcoeff + \Mttwo = \Mrho$ which follows from the definitions of $\Mcoeff$ and $\MRttwo$ as reminded above.
\end{proof}

\begin{lemma}
    \label{lem:budget-C}
    If client $j$ belongs integrally to group $C$, then
    \begin{equation*}
        \Mcoeff \cdot \sum_{i=1}^n g_{i,j} + \sum_{i=1}^n (\bar s_{i,j} + t_{i,j}) \leq \Mrho.
    \end{equation*}
\end{lemma}
\begin{proof}
    By Lemmas \ref{lem:gen-C} and \ref{lem:tran-C},
    \begin{equation*}
        \Mcoeff \cdot \sum_{i=1}^n g_{i,j} + \sum_{i=1}^n (\bar s_{i,j} + t_{i,j}) \le \Mcoeff + \Mttwo = \Mrho.
    \end{equation*}
    using $\Mcoeff + \Mttwo = \Mrho$ which follows from the definitions of $\Mcoeff$ and $\MRttwo$ as reminded above.
\end{proof}

\begin{lemma}
    \label{lem:budget-D}
    If client $j$ belongs integrally to group $D$, then
    \begin{equation*}
        \Mcoeff \cdot \sum_{i=1}^n g_{i,j} + \sum_{i=1}^n (\bar s_{i,j} + t_{i,j}) \leq \Mrho.
    \end{equation*}
\end{lemma}
\begin{proof}
    By Lemmas \ref{lem:gen-D} and \ref{lem:tran-D},
    \begin{align*}
        \Mcoeff \cdot \sum_{i=1}^n g_{i,j} + \sum_{i=1}^n (\bar s_{i,j} + t_{i,j})
        &= \Mcoeff \cdot \left(1 - \Mcoeffmone \cdot \left[1 - \frac{R_{j-1}}{\MRttwo}\right]^+\right) + 1 - R_{j-1} \\
        &= \Mcoeffpone - R_{j-1} - \MRttwo \cdot \left[1 - \frac{R_{j-1}}{\MRttwo}\right]^+.
    \end{align*}
    If $R_{j-1} \ge \MRttwo$, this is at most $\Mcoeffpone-\MRttwo = \Mrho$.
    Otherwise, the bracket is $1-R_{j-1}/\MRttwo$, and the value is exactly $\Mcoeffpone-\MRttwo = \Mrho$.
    We have used $\Mcoeff \cdot \Mcoeffmone = \MRttwo$ and $\Mcoeffpone - \MRttwo = \Mrho$ which follow from the definitions of $\Mcoeff$ and $\MRttwo$ as reminded above.
\end{proof}

\begin{lemma}
    \label{lem:budget-fractional}
    If client $j$ belongs fractionally to groups $A,B,C,D$ with coefficients $\mu_j^A, \mu_j^B, \mu_j^C, \mu_j^D$, then
    \begin{equation*}
        \Mcoeff \cdot \sum_{i=1}^n g_{i,j} + \sum_{i=1}^n (\bar s_{i,j} + t_{i,j}) \leq \Mrho\,.
    \end{equation*}
\end{lemma}
\begin{proof}
    The proof is a direct calculation using Lemmas \ref{lem:gen-fractional} and \ref{lem:tran-fractional}, and then repeating the same arguments as with the previous integral lemmas to conclude that the resulting expression is at most $\Mrho$.
    \begin{align*}
        \Mcoeff \sum_i g_{i,j} + \sum_i(\bar s_{i,j}+t_{i,j})
        &= \sum_X \mu_j^X \left(\Mcoeff \sum_i g_{i,j}^X + \sum_i(\bar s_{i,j}^X+t_{i,j}^X)\right) \\
        &\le \sum_X \mu_j^X (\Mrho)
        = \Mrho.
    \end{align*}
    where the inequality is by the previous lemmas' analysis in the integral case, which did not use integrality of the memberships and has the same expressions as the value from $\Mcoeff \sum_i g_{i,j}^X + \sum_i(\bar s_{i,j}^X+t_{i,j}^X)$ for $X \in \{A,B,C,D\}$; from where the inequality follows by linearity.
\end{proof}

\paragraph{Other proofs.}

\begin{proof}[Proof of Lemma~\ref{lem:nonneg}]
We restate the lemma here:
\lemnonneg*
For $\hg_i$, let us go phase by phase.
In Phase 3 and Phase 2 the claim is trivial.
Next, note that $z \leq n - t_2$.
If $t_1 < n$, for the boundary assignment for $t_1$, we have that $\hg_{t_1}^A \geq 0$ because $n - t_1 \geq \Mfrac \cdot z$ and that $\hg_{t_1}^B \geq 0$ because $R_{t_1} < \MRtone$. 
Finally, in Phase 1 the claim follows again because $n - i \geq \Mfrac \cdot z$ for $i < t_1$.

For $g_{i,j}$, let us go group by group.
For group A clients, we have $g_{i,j} = \hg_i - \delta_{i,j}$ where $\delta_{i,j} \leq \hg_i$ by definition, so $g_{i,j} \geq 0$. For group B and C clients, we have $g_{i,j} = \hg_i$ so it follows from the non-negativity of $\hg_i$ that $g_{i,j} \geq 0$. For group D clients, we have $g_{i,j} \geq \hg_i\cdot(1 - \chi_j \Mfrac)$, and since $\chi_j \leq 1$ and $\Mfrac \leq 1$ because $\Mcoeffmone \leq \Mtone$, we have $g_{i,j} \geq 0$.
For non-integral clients, it follows from the convex combinations assignment of $g_{i,j}$.

For $t_{i,j}$, observe that it is $0$ for all $i$ except for $i=f_j$. We now explore only the case $i=f_j$ as otherwise it is clear that it is nonnegative. For group A clients, we have that $t_{f_j,j}$ is $0$, thus being non-negative. For group B clients, we have that $t_{f_j,j}$ is $[1-R_{j-1}-\Mtone]^+$ which is non-negative by definition. For group C clients, we have that $t_{f_j,j}$ is $1 - (R_{j-1}+\Mtone)$, but since $R_{j-1} \leq \MRtone$ because $j > t_1$, we have that $t_{f_j,j} \geq 1 - (\MRtone + \Mtone) = 0$.
Finally, for group D clients, we have that $t_{f_j,j} = 1-R_{j-1}$, and since $R_{j-1} \leq 1$, we have that $t_{f_j,j} \geq 0$.
For non-integral clients, it follows from the convex combinations assignment of $t_{i,j}$.

For $\bar s_{i,j}$, for group A clients, we have $\bar s_{i,j} = \hg_i - \sigma_{i,j}$ where $\sigma_{i,j} \leq \hg_i$ by definition, so $\bar s_{i,j} \geq 0$. For group B and C clients, we have $\bar s_{i,j} = \hg_i$ (with $+\Mcoeffmone$ for $i=f_j+1$ or $0$ for $i=j$) so it follows from the non-negativity of $\hg_i$ that $\bar s_{i,j} \geq 0$. For group D clients, we have $\bar s_{i,j} = 0$, so it is non-negative.
For non-integral clients, we have that $\bar s_{i,j}$ is a convex combination of non-negative values, so it is non-negative.
\end{proof}

\begin{proof}[Proof of Claim~\ref{claim:phase-2-ri-product}]
    We restate the Claim here:
    \claimphasetworiproduct*
    The case $j = i$ is trivial. Consider next the case $j = i - 1$. Then
    \begin{align*}
    R_{i-1} + \Mtone
    &= R_i + \Mtone + \hg_i\\
    &= R_i + \Mtone + (1 - \mu_i^D) \cdot \frac{R_i + \Mtone}{n - i}\\
    &= \Paren{R_i + \Mtone} \cdot \frac{n - (i - 1) - \mu_i^D}{n - i}\,,
    \end{align*}
    matching the claim. For general $t_1 \leq j \leq i \leq \min(t_3^\star, n-1)$ with $i - j \geq 2$, we prove the claim by induction on $i - j$, using as base case the case $j=i-1$ above. We have
    \begin{align*}
    R_{j} + \Mtone
    &= \Paren{R_{i - 1} + \Mtone} \cdot \frac{n - j - \sum_{\ell = j + 1}^{i - 1} \mu_{\ell}^D \cdot \frac{R_j + \Mtone}{R_{\ell - 1} + \Mtone}}{n - (i - 1)}\,.
    \end{align*}
    By applying the result for the case $j = i-1$, we have
    \begin{align*}
    R_{j} + \Mtone
    &= \Paren{R_{i} + \Mtone} \cdot \frac{n - (i - 1) - \mu_i^D}{n - i} \cdot \frac{n - j - \sum_{\ell = j + 1}^{i - 1} \mu_{\ell}^D \cdot \frac{R_j + \Mtone}{R_{\ell - 1} + \Mtone}}{n - (i - 1)}\\
    &= \Paren{R_{i} + \Mtone} \cdot \Paren{1 - \frac{\mu_i^D}{n - (i - 1)}} \cdot \frac{n - j - \sum_{\ell = j + 1}^{i - 1} \mu_{\ell}^D \cdot \frac{R_j + \Mtone}{R_{\ell - 1} + \Mtone}}{n - i}\,,
    \end{align*}
    and it remains to prove that
    \begin{equation*}\Paren{1 - \frac{\mu_i^D}{n - (i - 1)}} \Paren{n - j - \sum_{\ell = j + 1}^{i - 1} \mu_{\ell}^D \cdot \frac{R_j + \Mtone}{R_{\ell - 1} + \Mtone}} = n - j - \sum_{\ell = j + 1}^{i} \mu_{\ell}^D \cdot \frac{R_j + \Mtone}{R_{\ell - 1} + \Mtone}\,.\end{equation*}
    Denote $\zeta = n - j - \sum_{\ell = j + 1}^{i-1} \mu_{\ell}^D \cdot \frac{R_j + \Mtone}{R_{\ell - 1} + \Mtone}$. Then we need to prove that
    \begin{equation*}\Paren{1 - \frac{\mu_i^D}{n - (i - 1)}} \cdot \zeta = \zeta - \mu_i^D \cdot \frac{R_j + \Mtone}{R_{i - 1} + \Mtone}\,,\end{equation*}
    equivalent to
    \begin{equation*}\frac{\mu_i^D}{n - (i - 1)} \cdot \zeta = \mu_i^D \cdot \frac{R_j + \Mtone}{R_{i - 1} + \Mtone}\,.\end{equation*}
    If $\mu_i^D = 0$ the claim is trivial. Otherwise we divide by $\mu_i^D$, so we need to prove that
    \begin{equation*}\frac{ n - j - \sum_{\ell = j + 1}^{i-1} \mu_{\ell}^D \cdot \frac{R_j + \Mtone}{R_{\ell - 1} + \Mtone}}{n - (i - 1)} = \frac{R_{j}+\Mtone}{R_{i-1}+\Mtone}\,.\end{equation*}
    But this is equivalent to the statement of the claim we are doing induction on for $j$ and $i-1$:
    \begin{equation*}R_{j}+\Mtone = \Paren{R_{i-1}+\Mtone} \cdot \frac{n - j - \sum_{\ell=j+1}^{i-1} \mu_\ell^D \cdot \frac{R_j+\Mtone}{R_{\ell-1}+\Mtone}}{n - (i - 1)}\,.\end{equation*}
    Hence, the claim is proved.
\end{proof}

\begin{proof}[Proof of Claim~\ref{claim:phase-2-ri-product-prime}]
    We restate the claim here:
    \claimphasetworiproductprime*
    If $t_1=n$, then $j=i=n-1$ and the claim is trivial.
    Else, the proof is identical to the proof of Claim~\ref{claim:phase-2-ri-product}, using that $\hg'_{t_1} = (\mu_{t_1}^B + \mu_{t_1}^C) \cdot \frac{R_{t_1}+\Mtone}{n-t_1}$ and that $\hg_i = \hg'_i$ and $R_i = R'_i$ for $i > t_1$.
\end{proof}

\newpage 

\section{$k$-Median Algorithm Opening $k + O(\log n)$ Centers}
\label{sec:merging}

In this section we expand on the ideas in Section~\ref{sec:logad}, based on the techniques developed in~\cite{MR4928455-Cohen25}, in order to obtain an algorithm for $k$-median that opens at most $k + O(\log n)$ centers.
As in their paper, this requires us to develop a robust version of the Log-Adaptive Algorithm, and to slightly adapt their merging procedure to our setting.
We remark that the definitions and notation in this section follow closely those in~\cite{MR4928455-Cohen25}.

First, we generalize the notions of $1$-openability and $2$-openability in Definitions~\ref{def:openable-mod-1} and~\ref{def:openable-mod-2} to approximate notions that we call $(1,\eta)$-openability and $(2,\eta)$-openability.
Note that these definitions of openable only depend on $\theta$ and $\SS$.
We will choose $\eta$ throughout the analysis as $c \frac{\delta \epsilon}{n} 2^{-n \cdot |\FF|}$ for a small enough absolute constant $c > 0$.

\begin{definition}[$(1, \eta)$-openable facility]
\label{def:openable-mod-eta-1}
Consider a time $\theta$ and a set $\SS$ of open facilities.
Let $I = \{j \in D \mid \theta \geq (1-\delta) d(j, \SS)\}$ and $A = D \setminus I$.
A facility $i\in \FF \setminus \SS$ is $(1,\eta)$-\emph{openable} w.r.t.\ $(\SS,\theta)$ if
satisfying:
\[
\sum_{j\in A}\beta(\theta, (1-\delta)\dist(i,j))
\;+\;
\sum_{j\in I}\beta((1-\delta)\dist(j,\SS), (1-\delta)\dist(i,j))
\;\ge\; \hat f - n \eta.
\]
\end{definition}

\begin{definition}[$(2, \eta)$-openable facility]
\label{def:openable-mod-eta-2}
Consider a time $\theta$ and a set $\SS$ of open facilities.
Let $I = \{j \in D \mid \theta \geq (1-\delta) d(j, \SS)\}$ and $A = D \setminus I$.
A facility $i\in \FF\setminus \SS$ is $(2,\eta)$-\emph{openable} w.r.t.\ $(\SS,\theta)$ if there exist bids $(\tau_j)_{j\in A}$
satisfying:
\begin{enumerate}
\item[\textbf{(i)}] (\textbf{Only nearby clients may increase}) 
\[
\tau_j=\theta \qquad \forall j\in A\setminus \Ball{i}{\epsilon\theta}.
\]
\item[\textbf{(ii)}] (\textbf{Nearby clients increase only slightly})
\[
\theta \le \tau_j \leq \min\Bigl\{(1-\delta)\dist(j,\SS),\, (1+\epsilon^2) \theta \Bigr\}
\qquad \forall j\in A\cap \Ball{i}{\epsilon\theta}.
\]
\item[\textbf{(iii)}] (\textbf{Facility $i$ is paid for in the $(1-\delta/2)$-scaled metric})
\[
\sum_{j\in A}\beta(\tau_j, (1-\delta/2)\dist(i,j))
\;+\;
\sum_{j\in I}\beta((1-\delta)\dist(j,\SS), (1-\delta/2)\dist(i,j))
\;\ge\; \hat f - n \eta.
\]
\item[\textbf{(iv)}] (\textbf{No overbidding for regular facilities after updating to $\tau$})
\[
\forall i_0\in \FF:\qquad
\sum_{j\in A}\beta(\tau_j,\dist(i_0,j))
\;+\;
\sum_{j\in I}\beta((1-\delta)  \dist(j,\SS), \dist(i_0,j))
\;\le\; \hat f.
\]
\item[\textbf{(v)}] (\textbf{Triangle inequality})
\[
\tau_j \le \tau_k + \dist(j,k)
\qquad \forall j,k\in A.
\]
\end{enumerate}
\end{definition}

As in \cite{MR4928455-Cohen25}, we also allow the introduction of ``free facilities'' whose opening costs are not paid.
Let $\SS_{\mathrm{free}} \subseteq \SS$ be the multiset of free facilities that have been opened.
For each free copy $\tilde{i} \in \SS_{\mathrm{free}}$, we also define a parameter $u(\tilde{i})$ which represents the distance between the free facility $\tilde{i}$ and its original copy $i$.
For any point $x$ in the metric space, the distance is given by $d(\tilde{i}, x) = d(i, x) + u(\tilde{i})$.
Also let $\SS_{\mathrm{reg}} = \SS \setminus \SS_{\mathrm{free}}$ be the set of open regular facilities.

For $\nu \in \{1, 2\}$, we define $(\nu,\eta)$-valid sequences.
\begin{definition}[$(\nu, \eta)$-valid sequence]
\label{def:valid-sequence-eta}
Consider a time $\theta$ and let $\SS$ be a set of opened facilities. 
A sequence $(h_1, \ldots, h_{\ell})$ of facilities is $(\nu,\eta)$-valid if the following conditions hold.
\begin{itemize}
    \item \textbf{(Validity)} Each $h_t$ is either free or $(\nu,\eta)$-openable w.r.t. the time $\theta$ and a superset $\SS' \supseteq \SS \cup \{h_1, \ldots, h_{t-1}\}$.
    \item \textbf{(Maximality)} The sequence is maximal, i.e., there is no $(\nu,0)$-openable facility w.r.t. the time $\theta$ and $\SS \cup \{h_1, \ldots, h_{\ell}\}$.
\end{itemize}
\end{definition}

We note that Algorithm~\ref{alg:log-adaptive-greedy} generates $(\nu,0)$-valid sequences:

\begin{lemma}
\label{lem:log-valid-1}
Any sequence generated in Stage 1 of Algorithm~\ref{alg:log-adaptive-greedy} is $(1,0)$-valid.
\end{lemma}
\begin{proof}
Both properties follow from the definition of the algorithm.
For the first property, a facility $h_t$ is added only if it is $(1,0)$-openable w.r.t. the time $\theta$ and $\SS \cup \{h_1, \ldots, h_{t - 1}\}$.
For the second property, the loop in Stage 1 continues as long as there exists a $(1,0)$-openable facility w.r.t. $\theta$ and the set of facilities opened so far.
\end{proof}

\begin{lemma}
\label{lem:log-valid-2}
Any sequence generated in Stage 2 of Algorithm~\ref{alg:log-adaptive-greedy} is $(2,0)$-valid.
\end{lemma}
\begin{proof}
Both properties follow from the definition of the algorithm.
For the first property, a facility $h_t$ is added only if it is $(2,0)$-openable w.r.t. the time $\theta$ and $\SS \cup \{h_1, \ldots, h_{t - 1}\}$.
For the second property, the loop in Stage 2 continues as long as there exists a $(2,0)$-openable facility w.r.t. $\theta$ and the set of facilities opened so far.
\end{proof}

For simplicity of notation, let $\nu : \mathbb{N} \to \{1,2\}$ satisfy $\nu(x) = 1$ for odd integers and $\nu(x)=2$ for even integers.
We define now an $\eta$-valid sequence of sequences, which consists of $(\nu,\eta)$-valid sequences and in addition the first sequence is required to be of type $1$ and the sequences are required to alternate between type $1$ and type $2$.

\begin{definition}
\label{def:valid-seq-of-seq}
For a sequence of sequences $\mathcal{H} = (H_1, \ldots, H_L)$, we say it consists of $\eta$-valid sequences if $H_{t}$ is $(\nu(t), \eta)$-valid w.r.t. $\theta = (1+\epsilon^2)^{\lfloor (t-1)/2 \rfloor}$ and $\SS = \cup_{q \in [t - 1]} H_q$, for all $1 \leq t \leq L$.
\end{definition}
Furthermore, we say that $\mathcal{H}$ is a solution if $L$ is even and $A = \emptyset$ at the end, meaning that at time $\theta=(1+\epsilon^2)^{|\mathcal{H}| / 2}$, the set of inactive clients $I = \{j \in D \mid \theta \geq (1-\delta) d(j, \SS)\}$ is equal to the set of all clients.
Otherwise, we say that $\mathcal{H}$ is a partial solution.

We prove that the following approximation guarantee holds for solutions $\mathcal{H}$:

\begin{theorem}
\label{theorem:alg-lmp-robust}
Let $\gamma \geq 1$ and let $\beta(\alpha,d) = [\alpha - \gamma \cdot d]^+$. Let $\rho \geq 1.01 \gamma$ be an upper bound on the optimal value of the factor-revealing program in Definition~\ref{def:frp} over all $n \geq 1$.
Let $\delta > 0$ be small enough compared to $1/\gamma$ and let $\eps > 0$ be small enough compared to $1/\gamma$ and $\delta$,
and let $\OPT$ be the optimal cost to the facility location problem with facility cost $f = \frac{\hat{f}}{\rho + 6 \delta}$ for $\hat f > n \eta$.
If $\mathcal{H}$ is an $\eta$-valid solution with respect to these parameters, then the set $\SS^\star$ of facilities opened by $\mathcal{H}$ satisfies
\[\sum_{j \in D} d(j, \SS^\star) \leq (1+8\delta) \rho \cdot (\OPT - |\SS^\star_{\mathrm{reg}}| \cdot f) + 2 |\SS^\star_{\mathrm{reg}}| \cdot n \eta\,.\]
\end{theorem}

We will prove Theorem~\ref{theorem:alg-lmp-robust} in Section~\ref{sec:logad-robust}.
The proof of the theorem follows the same set of steps as that of Theorem~\ref{theorem:alg-lmp}, but becomes more technically intricate because of the more general definition of $(\nu,\eta)$-openable.
An important technical challenge comes from the fact that a facility may be opened with respect to a \emph{superset} of the set of open facilities; this ends up making it more challenging to embed a solution into a feasible solution for the factor-revealing program, and requires us to use a more complicated assignment of values to the variables $r$ in the factor-revealing program.

\subsection{Merging Procedure}

Next, we will start with some definitions required in the merging procedure.

We define now two procedures $\mathrm{CompleteSolution}_{\hat f,u}$ and $\mathrm{CompleteSequence}_{\hat f,u}$, that are extensively used in the merging algorithm.
These procedures are defined analogously to the same procedures in~\cite{MR4928455-Cohen25}:

\begin{itemize}
    \item $\mathrm{CompleteSolution}_{\hat f,u}(\mathcal{H})$ takes a partial $\eta$-valid solution $\mathcal{H}$ as input and returns a solution by running Algorithm~\ref{alg:log-adaptive-greedy} using opening cost $\hat f$ and distance $u$ for free facilities already opened.
    If $|\mathcal{H}|$ is odd, it starts at Stage 2 with $\theta = (1+\eps^2)^{\lfloor (|\mathcal{H}| - 1) / 2 \rfloor}$, and if $|\mathcal{H}|$ is even, it starts at Stage 1 with $\theta = (1+\eps^2)^{|\mathcal{H}| / 2}$.
    Note that $\mathrm{CompleteSolution}$ can only terminate after a Stage 3 in Algorithm~\ref{alg:log-adaptive-greedy}, so the returned solution has even length.
    \item $\mathrm{CompleteSequence}_{\hat f,u}(\mathcal{H}=(H_1, \ldots, H_{p-1}),H_p = (h_1, \ldots, h_\ell))$ where $\mathcal{H}$ is a partial $\eta$-valid solution and $H_p$ is a $(\nu(p), \eta)$-valid sequence, except that it may not be maximal.
    If $p$ is odd, $\mathrm{CompleteSequence}$ runs Stage 1 in Algorithm~\ref{alg:log-adaptive-greedy} using opening cost $\hat f$ and distance $u$ for free facilities, and if $p$ is even it runs Stage 2 in Algorithm~\ref{alg:log-adaptive-greedy} using opening cost $\hat f$ and distance $u$ for free facilities, to return a maximal sequence that contains $H_p$ as a prefix.
\end{itemize}

The definitions of $\mathrm{CompleteSolution}$ and $\mathrm{CompleteSequence}$ imply that the solutions and the sequences generated by them are valid:

\begin{observation}
\label{obs:alg-valid}
Any output of $\mathrm{CompleteSolution}_{\hat f,u}(\mathcal{H})$ on a valid input is an $\eta$-valid solution,
and any output of $\mathrm{CompleteSequence}_{\hat f,u}(\mathcal{H}=(H_1, \ldots, H_{p-1}),H_p = (h_1, \ldots, h_\ell))$ on a valid input produces a $(\nu(p),\eta)$-valid partial sequence that contains $H_p$ as a prefix.
\end{observation}

We will also need a non-overbidding property to be satisfied at the beginning of each sequence in a solution generated by $\mathrm{CompleteSolution}$:

\begin{lemma}
\label{lem:seq-overbidding}
Let $\mathcal{H} = (H_1, \ldots, H_{p})$ be an $\eta$-valid partial solution.
Suppose that for $\SS' = \bigcup_{H \in \mathcal{H}} H$, $I' = \{j \in D \mid \theta \geq (1-\delta) d(j, \SS')\}$, and $A' = D \setminus I'$ it holds for all $i_0 \in \FF$ that
    \[\sum_{j\in A'}\beta(\theta,\dist(i_0,j))
\;+\;
\sum_{j\in I'}\beta((1-\delta)  \dist(j,\SS'), \dist(i_0,j))
\;\le\; \hat f.\]
Then, for all $H_t$ with $t > p$ in the output of $\mathrm{CompleteSolution}(\mathcal{H})_{\hat f, u}$, for $\SS = \SS' \cup \bigcup_{q=p+1}^t H_q$, $I = \{j \in D \mid \theta \geq (1-\delta) d(j, \SS)\}$, and $A = D \setminus I$, we have for all $i_0 \in \FF$ that 
    \[\sum_{j\in A}\beta(\theta,\dist(i_0,j))
\;+\;
\sum_{j\in I}\beta((1-\delta)  \dist(j,\SS), \dist(i_0,j))
\;\le\; \hat f.\]
\end{lemma}
\begin{proof}
The claim is captured by Lemma~\ref{lem:overbidding-robust}, which proves that the property holds even after the opening of each facility in each sequence.
\end{proof}

We now discuss the merging procedure itself.
Informally, the merging procedure takes two $\eta$-valid solutions $\mathcal{H}$ and $\mathcal{H}'$ that open less and more than $k$ regular facilities, respectively, and merges them to obtain an $\eta$-valid solution that opens exactly $k$ regular facilities.
This merging procedure is essentially identical to that in~\cite{MR4928455-Cohen25}, except that (1) it uses \emph{our} definitions of $\completesol$ and $\completesequence$ above and (2) it has a slightly different initialization.

Let us describe our initialization.
We perform bisection search on $\hat f$ in the interval $[\delta/n, 4\gamma n\Maxdist]$ and run $\completesol_{\hat f, u}(\emptyset)$ with empty $u$.
Let us start by analyzing the endpoints of this interval:
\begin{itemize}
    \item $\hat f = \delta/n$: At $\hat f = \delta/n$ every client $j \in D$ satisfies $d(j, \SS^\star) \leq \frac{1+\eps^2}{1-\delta} (\gamma d(j, F) + \delta/n)$ because each client may open its closest facility by itself by the time $\theta \geq \gamma d(j, F) + \delta/n$.
    Then the resulting solution is a $\frac{1+\epsilon^2}{1-\delta}(\gamma+\delta)$-approximation, which by our choice of parameters is also a $\rho$-approximation.
    If this solution opens at most $k$ centers, we are done.
    Else, it gives a solution that opens more than $k$ centers.
    \item $\hat f = 4 \gamma n \Maxdist$: At $\hat f = 4 \gamma n \Maxdist$ the algorithm opens exactly one facility, because when first facility is opened then every other client becomes immediately inactive at that step.
    If $k=1$, we are done.
    Else, it gives a solution that opens less than $k$ centers.
\end{itemize}
Then we perform bisection search until we find two $\hat f > \hat f'$ such that $|\hat f - \hat f'| \leq \eta$ and the corresponding solutions $\mathcal{H}$ and $\mathcal{H}'$ open strictly less and strictly more than $k$ facilities, respectively.
(If one happens to open exactly $k$ facilities, we return it immediately.)

Then, as in~\cite{MR4928455-Cohen25}, the procedure maintains two solutions $\mathcal{H}$ and $\mathcal{H}'$ whose number of opened regular facilities sandwich $k$.
It also ensures that $\mathcal{H}$ is $\eta$-valid w.r.t. some $\hat f, u$ and $\mathcal{H}'$ is $\eta$-valid w.r.t. some $\hat f', u'$ such that $\hat f, u$ and $\hat f', u'$ differ in at most one ``entry'', by at most $\eta$: either they are identical, or $\hat f \neq \hat f'$ and $|\hat f - \hat f'| \leq \eta$, or there exists a single free facility $\tilde{i}$ with $u(\tilde i) \neq u'(\tilde i)$  and $|u(\tilde i) - u'(\tilde i)| \leq \eta$.
The procedure modifies $\mathcal{H}$ and $\mathcal{H}'$ in steps, such that after the $p$-th step of the procedure either it has found a solution that opens exactly $k$ facilities or it ensures that the first $p$ sequences of $\mathcal{H}$ and $\mathcal{H}'$ are identical.
Also, each step of the procedure opens at most three new free facilities.

For $\hat f \leq 4\gamma n \Maxdist$, at the latest when $\theta$ reaches $6 \gamma \Maxdist$ a facility will open and every client will become inactive, so $\theta$ is upper bounded by $(1+\eps^2) 6 \gamma \Maxdist$.
Because $\theta$ starts at $1$ and increases exponentially in the number of sequences (specifically, it is multiplied by $(1+\eps^2)$ every two sequences), this implies that the number of sequences in any $\eta$-valid solution produced by our procedure with $\hat f \leq 4 \gamma n \Maxdist \leq \poly(n)$  is $O(\log n / \eps^2)$.
This implies that the procedure finds a solution within $O(\log n / \eps^2)$ steps, because otherwise it would reach the contradiction that $\mathcal{H}$ and $\mathcal{H}'$ are identical but open different numbers of regular facilities.

\begin{proof}[Proof sketch of Theorem~\ref{thm:k-median}]
The merging procedure generates an $\eta$-valid solution that opens exactly $k$ regular facilities and at most $O(\log n/ \eps^2)$ free facilities. 
Then, by Theorem~\ref{theorem:alg-lmp-robust},
\begin{align*}
\sum_{j \in D} d(j, \SS^\star)
&\leq (1+8\delta) \rho \cdot (\OPT_f - |\SS^\star_{\mathrm{reg}}| \cdot f) + 2 |\Sr^\star| \cdot n \eta\\
&= (1+8\delta) \rho \cdot \left( \OPT_f - f \cdot k \right)  + 2 |\Sr^\star| \cdot n \eta\\
&\leq \Paren{(1+8\delta)\rho + 2 |\Sr^\star| \cdot n \eta} \cdot \sopt\\
&\leq (1+9\delta)\rho \cdot \sopt\,,
\end{align*}
where $\sopt$ is the cost of the optimal $k$-median solution.
We used above that $\sopt \geq \OPT_f - f \cdot k$, that $\sopt \geq 1$ (because every distance is at least $1$), and also that $\eta = c \frac{\delta \epsilon}{n} 2^{-n \cdot |\FF|}$.
The statement of Theorem~\ref{thm:k-median} is captured by renaming $\delta \gets 9\delta$ and noticing that we only choose $\eps$ to be a constant factor smaller than $\delta$.

Strictly speaking, $\SS^\star$ may not be a valid solution to the $k$-median problem because it also contains free facilities that are not members of $\FF$.
However, each free facility $\tilde{i}$ is a copy of a regular facility $i \in \FF$, so we consider in fact the solution that for each free facility $\tilde{i} \in \SS^\star$ opens the regular facility $i$ that it is a copy of.
This cannot increase the cost of the solution: for all $j \in D$, we have $d(j, i) = d(j, \tilde{i}) - u(\tilde{i}) \leq d(j, \tilde{i})$.
Hence, this solution also has cost upper bounded by $(1+9\delta)\rho \cdot \sopt$.
The number of regular facilities opened is still bounded by $k + O(\log n)$.
\end{proof}

\subsection{Merging Procedure Analysis}

The analysis of the merging procedure depends on the following two lemmas that are essential to argue that, during the merging procedure, the sequences produced are maximal.
An analogous result is implicit in~\cite{MR4928455-Cohen25}.

\begin{lemma}
\label{lem:pay-to-open-1}
Suppose $\hat f > 0$.
Let $H_p = (h_1, \ldots, h_\ell)$ be a $(1,\eta)$-valid sequence with respect to time $\theta$ and open facilities $\SS'$, except that it may not be maximal.
For $I' = \{j \in D \mid \theta \geq (1-\delta) d(j, \SS')\}$, $A' = D \setminus I'$, suppose that for all $i_0 \in \FF$,
    \[\sum_{j\in A'}\beta(\theta,\dist(i_0,j))
\;+\;
\sum_{j\in I'}\beta((1-\delta)  \dist(j,\SS'), \dist(i_0,j))
\;\le\; \hat f.\]
Let $\SS = \SS' \cup \{h_1, \ldots, h_{\ell}\}$, and let $I = \{j \in D \mid \theta \geq (1-\delta) d(j, \SS)\}$ and $A = D \setminus I$.
Suppose furthermore that
\[\sum_{j \in A} \beta(\theta, (1-\delta) d(i, j)) + \sum_{j \in I}\beta((1-\delta)d(j,\SS), (1-\delta)d(j, i)) \geq \hat{f}\]
for some $i \in \FF$. Then there exists some $i^\star \in \FF\setminus \SS$ that is $(1,0)$-openable with respect to $\theta$ and $\SS$.
\end{lemma}
\begin{proof}
We start by noting that $i \in \SS$ is impossible unless $\hat{f}=0$, because then $\theta \leq (1-\delta) d(j,\SS) \leq (1-\delta)d(i,j)$ for all $j \in A$ and $d(j, \SS) \leq d(i,j)$ for all $j \in [n]$,  which implies that the left-hand side of the budget assumption is zero.
Therefore $i \notin \SS$.
Then $i$ itself is $(1,0)$-openable, because the condition in Definition~\ref{def:openable-mod-eta-1} is identical to the second inequality in our hypothesis.
\end{proof}

\begin{lemma}
\label{lem:pay-to-open-2}
Suppose $\hat f > 0$.
Let $H_p = (h_1, \ldots, h_\ell)$ be a $(2,\eta)$-valid sequence with respect to time $\theta$ and open facilities $\SS'$, except that it may not be maximal.
For $I' = \{j \in D \mid \theta \geq (1-\delta) d(j, \SS')\}$, $A' = D \setminus I'$, suppose that for all $i_0 \in \FF$,
    \[\sum_{j\in A'}\beta(\theta,\dist(i_0,j))
\;+\;
\sum_{j\in I'}\beta((1-\delta)  \dist(j,\SS'), \dist(i_0,j))
\;\le\; \hat f.\]
Let $\SS = \SS' \cup \{h_1, \ldots, h_{\ell}\}$, and let $I = \{j \in D \mid \theta \geq (1-\delta) d(j, \SS)\}$ and $A = D \setminus I$.
Suppose furthermore that there exist values $(\tau_j)_{j \in A}$ such that $\tau_j = \theta$ for $j \in A \setminus B(i, \epsilon \theta)$, $\theta \leq \tau_j \leq \min((1-\delta)d(j,\SS), (1+\epsilon^2)\theta)$ for $j \in A \cap B(i, \epsilon \theta)$, $\tau_j \leq \tau_k + d(j, k)$ for all $j, k \in A$, and
\[\sum_{j \in A} \beta(\tau_j, (1-\delta/2) d(i, j)) + \sum_{j \in I}\beta((1-\delta)d(j,\SS), (1-\delta/2)d(j, i)) \geq \hat{f}\]
for some $i \in \FF$. Then there exists some $i^\star \in \FF\setminus \SS$ that is $(2,0)$-openable with respect to $\theta$ and $\SS$.
\end{lemma}
\begin{proof}
We start by noting that $i \in \SS$ is impossible unless $\hat f=0$, because then $\tau_j \leq  (1-\delta) d(i,j)$ for all $j \in A$ and $d(j, \SS) \leq d(j, i)$ for all $j \in [n]$, which implies that the left-hand side of the budget assumption is zero.

Fix $(\tau_j)_{j \in A}$ such that $\tau_{\max} = \max_{j \in A} \tau_j$ (with $\tau_{\max} = \theta$ if $A = \emptyset$) is minimized and the assumptions are satisfied with respect to some $i \in \FF\setminus \SS$.
We will prove that $i$ is $(2,0)$-openable.
The first three constraints and the fifth constraint are satisfied by assumption.

\textbf{Fourth constraint for $i_0 \in \FF\setminus \SS$.}
For the fourth constraint, suppose for the sake of contradiction that there exists some $i_0 \in \FF\setminus \SS$ such that
\[\sum_{j \in A} \beta(\tau_j, d(i_0, j)) + \sum_{j \in I}\beta((1-\delta)d(j,\SS), d(i_0, j)) > \hat{f}\,.\]

\textbf{Existence of non-zero active bid.}
Let us prove that there exists $j \in A$ such that $\tau_j > \theta$ and $\beta(\tau_{j}, d(i_0, {j})) > 0$.
We have that
\[\sum_{j \in A'} \beta(\theta, d(i_0, j)) + \sum_{j \in I'}\beta((1-\delta)d(j,\SS'), d(i_0, j)) \leq \hat{f}\,,\]
so because $d(j, \SS) \leq d(j, \SS')$ also 
\[\sum_{j \in A'} \beta(\theta, d(i_0, j)) + \sum_{j \in I'}\beta((1-\delta)d(j,\SS), d(i_0, j)) \leq \hat{f}\,.\]
Let us now consider the clients $j \in A' \cap I$.
Because $j \in I$, we have $\theta \geq (1-\delta) d(j,\SS)$, so we also have $\beta(\theta, d(i_0, j)) \geq \beta((1-\delta)d(j,\SS), d(i_0, j))$.
Hence, overall, we also have that 
\[\sum_{j \in A} \beta(\theta, d(i_0, j)) + \sum_{j \in I}\beta((1-\delta)d(j,\SS), d(i_0, j)) \leq \hat{f}\,.\]
Then, for this inequality to be violated when switching from $\theta$ to $(\tau_j)_{j \in A}$, there must exist some $j \in A$ such that $\tau_{j} > \theta$ and $\beta(\tau_{j}, d(i_0, {j})) > 0$.

\textbf{Finalizing fourth constraint for $i_0 \in \FF\setminus \SS$.}
For some $\zeta > 0$ small enough, consider a new set of values $\{\tau'_j\}_{j \in A}$ such that $\tau'_j = \tau_j - \zeta$ if $\tau_j > \theta$ and $\tau'_j = \tau_j$ otherwise.
Then we still have $\tau'_j \geq \theta$ for all $j \in A$ and 
\[\sum_{j \in A} \beta(\tau'_j, d(i_0, j)) + \sum_{j \in I}\beta((1-\delta)d(j,\SS), d(i_0, j)) > \hat{f}\,.\]
Let us prove now that this leads to a contradiction.
Let $\tau'_{\max} = \max_{j \in A} \tau'_j$. We have $\theta < \tau'_{\max} < \tau_{\max}$.
Let $\alpha'_j = \min(\tau'_{\max}, (1-\delta)d(j,\SS))$, and note that $\alpha'_j \geq \tau'_{j}$ for all $j \in A$, because if $\alpha'_j = (1-\delta) d(j,\SS)$ we have $(1-\delta) d(j,\SS) \geq \tau_j \geq \tau'_j$.
Then we also have
\[\sum_{j \in A} \beta(\alpha'_j, d(i_0, j)) + \sum_{j \in I}\beta((1-\delta)d(j,\SS), d(i_0, j)) > \hat{f}\,.\]
Let $\tilde{\tau}_j = \min(\alpha'_j, (1+\epsilon) \theta - d(i_0, j))$ for $j \in A \cap B(i_0, \epsilon \theta)$ and $\tilde{\tau}_j = \theta$ for $j \in A \setminus B(i_0, \epsilon \theta)$.
We claim that 
\begin{align*}
&\sum_{j \in A} \beta(\tilde\tau_j, (1-\delta/2) d(i_0, j)) + \sum_{j \in I}\beta((1-\delta)d(j,\SS), (1-\delta/2)d(j, i_0))\\
=& \sum_{j \in A \cap B(i_0, \eps\theta)}  \beta(\min(\alpha'_j, (1+\epsilon) \theta - d(i_0, j)), (1-\delta/2) d(i_0, j)) + \sum_{j \in A \setminus B(i_0, \eps\theta)}  \beta(\theta, (1-\delta/2) d(i_0, j))\\
&\quad + \sum_{j \in I}\beta((1-\delta)d(j,\SS), (1-\delta/2)d(j, i_0))\\
\geq& \sum_{j \in A \cap B(i_0, \eps\theta)}  \beta(\alpha'_j, d(i_0, j)) + \sum_{j \in A \setminus B(i_0, \eps\theta)}  \beta(\alpha'_j, d(i_0, j)) + \sum_{j \in I}\beta((1-\delta)d(j,\SS), d(j, i_0))\,. 
\end{align*}
For the inequality above, for the first sum, the claim is trivial if $\alpha_j'$ is the minimizer, so let us assume that $\alpha_j' \geq (1+\eps)\theta - d(i_0, j)$. Note that $\alpha_j' \leq (1+\eps^2)\theta$.
Then, using that $d(i_0, j) \leq \epsilon \theta$ we have
\begin{align*}
(1+\epsilon) \theta - d(i_0, j)-(\gamma (1-\delta/2)) d(i_0, j)
&= (1+\epsilon) \theta - (1-\eps)d(i_0, j)-(\gamma (1-\delta/2) + \eps) d(i_0, j)\\
&\geq (1+\eps^2) \theta - (\gamma (1-\delta/2) + \eps)d(i_0,j)\\
&\geq (1+\eps^2)\theta - \gamma d(i_0, j)\\
&\geq \alpha_j' - \gamma d(i_0, j)\,,
\end{align*}
where we used that $\gamma \geq \gamma (1-\delta/2) + \eps$ for $\delta \gg \eps$.

For the second sum, using that $d(i_0, j) > \epsilon \theta$ we have
\[\theta - \gamma (1 - \delta/2) d(i_0, j) \geq (1+\eps^2) \theta - \gamma d(i_0, j) \geq \alpha_j' - \gamma d(i_0, j)\,,\]
where we used that $\gamma \delta d(i_0, j) / 2 \geq \gamma \delta \eps \theta / 2 \geq \eps^2 \theta$ for $\delta \gg \eps$ and $\alpha_j' \leq (1+\eps^2) \theta$.
The third sum follows trivially.
Then
\[\sum_{j \in A} \beta(\tilde\tau_j, (1-\delta/2) d(i_0, j)) + \sum_{j \in I}\beta((1-\delta)d(j,\SS), (1-\delta/2)d(j, i_0)) \geq \hat{f}\,.\]

We prove now that the other assumptions of the current lemma are also satisfied for $i_0$ and $\{\tilde\tau\}_{j \in A}$, thus reaching a contradiction because $\max_{j \in A}\tilde\tau_j \leq \tau'_{\max} < \tau_{\max}$.
We have by definition $\tilde\tau_j = \theta$ for $j \in A \setminus B(i_0, \eps \theta)$ and $\theta \leq \tilde\tau_j \leq \min((1-\delta)d(j,\SS),(1+\eps^2)\theta)$ for $j \in A \cap B(i_0, \eps \theta)$.

It remains to check the triangle inequalities $\tilde\tau_j \leq \tilde\tau_k + d(j, k)$.
The triangle inequalities between $j, k$ with $\tilde\tau_k \geq \tilde\tau_j$ are trivially satisfied, so we assume $\tilde\tau_j > \tilde\tau_k \geq \theta$.
Then $\tilde\tau_j = \min(\alpha'_j, (1+\epsilon) \theta - d(i_0, j))$, where $d(i_0, j) \leq \eps \theta$.
First, consider the case $d(i_0, k) > \eps \theta$, which implies $\tilde\tau_k = \theta$.
Then
\[\tilde\tau_j \leq (1+\eps)\theta - d(i_0, j) \leq (1+\eps)\theta - (d(i_0, k) - d(j, k)) \leq \theta + d(j, k) = \tilde\tau_k + d(j, k)\,.\]
Second, consider the case $d(i_0, k) \leq \eps \theta$, which implies $\tilde\tau_k = \min(\alpha'_k, (1+\epsilon) \theta - d(i_0, k))$, where also recall that $\alpha'_k = \min(\tau'_{\max}, (1-\delta)d(k, \SS))$.
If $\tilde\tau_k = \tau'_{\max}$, we are done.
If $\tilde\tau_k = (1-\delta)d(k, \SS)$, we have 
\[\tilde\tau_j \leq (1-\delta) d(j,\SS) \leq (1-\delta)(d(j,k)+d(k,\SS)) \leq \tilde\tau_k + d(j,k)\,.\]
Finally, if $\tilde\tau_k = (1+\epsilon) \theta - d(i_0, k)$, we have
\[\tilde\tau_j \leq (1+\epsilon) \theta - d(i_0, j) \leq (1+\epsilon) \theta - (d(i_0, k) - d(j, k)) = \tilde\tau_k + d(j,k)\,.\]
Hence, we satisfy all assumptions of the current lemma and we reach a contradiction.

\textbf{Fourth constraint for $i_0 \in \SS$.}
Finally, the remaining possibility is that there exists some $i_0 \in \SS$ for which the fourth constraint is violated.
However, in that case, we note that $\tau_j \leq (1-\delta) d(j, \SS) \leq d(i_0, j)$, so the left-hand side of the fourth constraint is zero, so it is impossible for the constraint to be violated.
\end{proof}

Given Lemmas~\ref{lem:pay-to-open-1} and~\ref{lem:pay-to-open-2}, the analysis of the merging procedure is nearly identical to that in~\cite{MR4928455-Cohen25}.
To prove that the merging procedure maintains the maximality of the sequences produced they (implicitly) use the fact that a ``payable'' facility is an ``openable'' facility.
This property holds in our case by Lemmas~\ref{lem:pay-to-open-1} and~\ref{lem:pay-to-open-2}.
The only remaining claim whose proof we need to change is that of Claim $2$ in~\cite{MR4928455-Cohen25}, which proves that if a sequence $H'_t$ is $(\nu, 0)$-valid with respect to some $\hat f', u'$, then it is also $(\nu, \eta)$-valid with respect to $\hat f, u$ that are $\eta$-close to $\hat f', u'$.
We prove below a version of this claim that applies to our case, by starting with a verbatim copy of their proof and adapting it to our setting.
(This proof will also use Lemmas~\ref{lem:pay-to-open-1} and~\ref{lem:pay-to-open-2}.)
The rest of the analysis is identical up to cosmetic changes.

\begin{claim}[See Claim $2$ in~\cite{MR4928455-Cohen25}]
\label{claim:param-ident}
    Let $H'_p$ be a $(\nu(p), 0)$-valid sequence with respect to some $\hat f', u'$, and suppose that the non-overbidding inequality in Lemma~\ref{lem:seq-overbidding} is satisfied after its execution with respect to $\hat f', u'$.
    Let $\hat f, u$ satisfy that either $\hat f, u$ and $\hat f', u'$ are identical, or $\hat f > \hat f'$ with $|\hat f - \hat f'| \leq \eta$, or there exists a single free facility $\tilde{i}$ with $u(\tilde i) < u'(\tilde i)$  and $|u(\tilde i) - u'(\tilde i)| \leq \eta$.
    Then $H'_p$ is a $(\nu(p),\eta)$-valid sequence with respect to $\hat f, u$.
    \label{claim:1ststep_etavalid}
\end{claim}
\begin{proof}
Suppose first that $\hat f\neq \hat f'$ is the difference parameter.
Because each facility in $H'_p$ was $(\nu(p),0)$-openable with respect to parameters $\hat f', u'$, then it is also $(\nu(p),\eta)$-openable with respect to parameters $\hat f, u$ because the only difference is that $\hat f$ is larger than $\hat f'$ by at most $\eta$.
Furthermore, $H'_p$ is maximal with respect to $\hat f, u$ by the following argument based on Lemmas~\ref{lem:pay-to-open-1} and~\ref{lem:pay-to-open-2}: Since $H'_p$ was maximal with respect to $\hat f', u'$, it means the assumptions of Lemmas~\ref{lem:pay-to-open-1} (if $\nu(p)=1$) or~\ref{lem:pay-to-open-2} (if $\nu(p)=2$) are not satisfied with respect to $\hat f', u'$.
The first assumption is satisfied by the no-overbidding hypothesis, so the second assumption must be the one not satisfied.
However, their second assumptions become even stronger for $\hat f$ larger than $\hat f'$, so they are still not satisfied, which implies that there still does not exist any $(\nu(p),0)$-openable facility, because otherwise these assumptions would trivially be satisfied.

Next, consider the case when $u(\tilde i) \neq u'(\tilde i)$ is the difference parameter.
Then, by the same argument as above, we also have that each facility in $H'_p$ is $(\nu(p),\eta)$-openable with respect to parameters $\hat f, u$.
To see this more formally, let $H'_p = \langle h_1, \dots, h_{\ell} \rangle$ and consider $h_t$ for some $t \in [\ell]$, which is $(\nu(p), 0)$-openable with respect to $\hat f', u', \theta = (1 + \eps^2)^{\lfloor \frac{p - 1}{2} \rfloor}$ and some superset $\SS \supseteq (\cup_{q \in [p-1]} H_q) \cup \{ h_1, \dots, h_{t-1} \}$.
When we switch from $\hat f', u'$ to $\hat f, u$, because the only difference is $u(\tilde i) < u'(\tilde i)$, the set of active clients $A'$ for $\hat f', u'$ is a superset of the set of active clients $A$ for $\hat f, u$.
Also, let $d'(\cdot, \cdot)$ and $d(\cdot, \cdot)$ be the distances with respect to $u'$ and $u$, respectively. 
For any $h_t \in H'_p$, we have that $h_t$ was $(\nu(p),0)$-openable with respect to the former setting with bids $(\tau'_j)_{j \in A'}$.
Then let $(\tau_j)_{j \in A'}$ be the new bids defined as $\tau_j = \min(\tau'_j, (1-\delta)d(j, \SS))$ for $j \in A'$ (these bids are only used on $A$, but we define them on the superset $A'$ for the purposes of the analysis). 
Let us verify that $h_t$ is $(\nu(p),\eta)$-openable with respect to $\hat f, u$ with the bids $(\tau_j)_{j \in A}$.
We need to check the corresponding conditions in Definitions~\ref{def:openable-mod-eta-1} (if $\nu(p)=1$) and~\ref{def:openable-mod-eta-2} (if $\nu(p)=2$).
We focus for simplicity on the case $\nu(p)=2$, which is the more intricate of the two.
The case $\nu(p)=1$ follows by similar arguments.

\begin{itemize}
    \item $\tau_j = \theta$ for every $j \in A - B(h_t, \eps \theta)$: Since it was true for $A' \supseteq A$, it remains true. 
    \item $\theta \leq \tau_j \leq \min\{(1-\delta)d(j,\SS),(1+\eps^2)\theta\}$ for every $j \in A \cap B(h_t, \eps \theta)$: Fix $j \in A \cap B(h_t, \eps \theta)$. 
    Since $A \subseteq A'$, $\tau'_j \leq \min\{(1-\delta)d'(j,\SS),(1+\eps^2)\theta\}$. Since $\tau_j$ is defined as $\min(\tau'_j, (1-\delta)d(j,\SS))$, the condition is satisfied.
    \item $
    \sum_{j \in A} \beta(\tau_j, (1-\delta/2)d(h_t,j)) + \sum_{j \in I} \beta((1-\delta)d(j,\SS), (1-\delta/2)d(h_t,j)) \geq \hat f - n \eta$: Let us argue that the contribution of each client $j$ to the left-hand side changes by at most $\eta$ between $(A', d', \tau')$ and $(A, d, \tau)$.
    We know that 
    \[\sum_{j \in A'} \beta(\tau'_j, (1-\delta/2)d(h_t,j)) + \sum_{j \in I'} \beta((1-\delta)d'(j,\SS), (1-\delta/2)d(h_t,j)) \geq \hat f\,.\]
    For $j \in A'$, note that $\tau_j < \tau'_j$ implies $\tau_j = (1-\delta)d(j, \SS) \geq (1-\delta)(d'(j, \SS) - \eta) \geq \tau'_j - \eta$, so $\tau_j \geq \tau'_j - \eta$ in any case.
    Then 
    \[\sum_{j \in A'} \beta(\tau_j + \eta, (1-\delta/2)d(h_t,j)) + \sum_{j \in I'} \beta((1-\delta)d'(j,\SS), (1-\delta/2)d(h_t,j)) \geq \hat f\,.\]
    For $j \in I'$, note that $d(j, \SS)$ and $d'(j, \SS)$ only differ by at most $\eta$.
    So also 
    \[\sum_{j \in A'} \beta(\tau_j + \eta, (1-\delta/2)d(h_t,j)) + \sum_{j \in I'} \beta((1-\delta)(d(j,\SS) + \eta), (1-\delta/2)d(h_t,j)) \geq \hat f\,.\]
    Finally, for $j \in A' \cap I$, $(1-\delta)d(j, \SS) \geq \tau_j$, so $\beta(\tau_j + \eta, (1-\delta/2)d(h_t, j)) \leq \beta((1-\delta)d(j,\SS)+\eta, (1-\delta/2)d(h_t, j))$, so
    \[\sum_{j \in A} \beta(\tau_j + \eta, (1-\delta/2)d(h_t,j)) + \sum_{j \in I} \beta((1-\delta) d(j,\SS) + \eta, (1-\delta/2) d(h_t,j)) \geq \hat f\,,\]
    from which it follows that
    \[\sum_{j \in A} \beta(\tau_j, (1-\delta/2)d(h_t,j)) +  \sum_{j \in I} \beta((1-\delta)d(j,\SS), (1-\delta/2)d(h_t,j)) \geq \hat f - n \eta\,.\]
    \item For every regular facility $i_0$, $
    \sum_{j \in A} \beta(\tau_j, d(i_0,j)) + \sum_{j \in I} \beta((1-\delta)d(j,\SS),d(i_0,j)) \leq \hat f$: We have
    \[
    \sum_{j \in A'} \beta(\tau'_j, d'(i_0,j)) + \sum_{j \in I'} \beta((1-\delta)d'(j,\SS), d'(i_0,j))
    \leq \hat f.
    \]
    We have trivially that $d'(j, \SS) \geq d(j, \SS)$ and $\tau_j' \geq \tau_j$.
    Because $i_0$ is a regular facility, we also have $d'(i_0, j) = d(i_0, j)$.
    Also, for $j \in A' \cap I$, because $j \in I$ implies that $\tau_j' \geq \theta \geq (1-\delta)d(j, \SS)$, we have $\beta(\tau'_j, d(i_0, j)) \geq \beta((1-\delta) d(j,\SS), d(i_0, j))$, so we also have
    \[
    \sum_{j \in A} \beta(\tau_j, d(i_0,j)) + \sum_{j \in I} \beta((1-\delta)d(j,\SS), d(i_0,j))
    \leq \hat f.
    \]
    \item $\tau_j \leq \tau_k + d(j, k)$ for every $j, k \in A$: For $j, k \in A'$, we had $\tau'_j \leq \tau'_k + d(j,k)$.
    Recall that $\tau_j = \min(\tau_j', (1-\delta)d(j,\SS))$ for all $j \in A$.
    Suppose $\tau_k < \tau'_k$, as otherwise the inequality holds trivially.
    Then $\tau_k = (1-\delta) d(k, \SS)$, and also $\tau_j \leq (1-\delta) d(j, \SS)$, so
    \[\tau_j \leq (1-\delta) (d(j, k) + d(k, \SS)) = \tau_k + (1-\delta) d(j, k) \leq \tau_k + d(j, k)\,.\]
\end{itemize}
Furthermore, $H'_p$ is maximal w.r.t. $\hat f, u$ by the following argument based on Lemmas~\ref{lem:pay-to-open-1} and~\ref{lem:pay-to-open-2}: Since $H'_p$ was maximal w.r.t. $\hat f', u'$, it means the assumptions of Lemmas~\ref{lem:pay-to-open-1} (if $\nu(p)=1$) or~\ref{lem:pay-to-open-2} (if $\nu(p)=2$) are not satisfied w.r.t. $\hat f', u'$.
We know that their first assumption (about $\leq \hat{f}$) is satisfied by the no-overbidding hypothesis, so the second assumption must be the one not satisfied.
Let us argue that these second assumptions are still not satisfied w.r.t. $\hat f, u$:
\begin{itemize}
    \item If $\nu(p)=1$, for Lemma~\ref{lem:pay-to-open-1}, we know that for all regular facilities $i \in \FF\setminus \SS$
    \[\sum_{j \in A'} \beta(\theta, (1-\delta)d(i,j)) + \sum_{j \in I'} \beta((1-\delta)d'(j, \SS), (1-\delta)d(j,i)) < \hat{f}\,.\]
    We have $d'(j, \SS) \geq d(j, \SS)$, and also for $j \in A' \cap I$ we have that $\theta \geq (1-\delta) d(j, \SS)$, so it easily follows that 
    \[\sum_{j \in A} \beta(\theta, (1-\delta)d(i,j)) + \sum_{j \in I} \beta((1-\delta)d(j, \SS), (1-\delta)d(j,i)) < \hat{f}\,.\]
    \item If $\nu(p)=2$, for Lemma~\ref{lem:pay-to-open-2}, we know that for all regular facilities $i \in \FF\setminus \SS$ and all $\{\tilde{\tau}_j\}_{j \in A'}$ such that $\tilde{\tau}_j = \theta$ for $j \in A' \setminus B(i, \epsilon \theta)$, $\theta \leq \tilde{\tau}_j \leq \min((1-\delta)d'(j,\SS), (1+\epsilon^2)\theta)$ for $j \in A' \cap B(i, \epsilon \theta)$, $\tilde{\tau}_j \leq \tilde{\tau}_k + d(j, k)$ for all $j, k \in A'$, we have
    \[\sum_{j \in A'} \beta(\tilde\tau_j, (1-\delta/2)d(i,j)) + \sum_{j \in I'} \beta((1-\delta)d'(j, \SS), (1-\delta/2)d(j,i)) < \hat{f}\,.\]
    We will proceed by contradiction. 
    Suppose there exists a regular facility $i \in \FF\setminus \SS$ and some $\{\tilde{\tau}_j\}_{j \in A}$ such that $\tilde{\tau}_j = \theta$ for $j \in A \setminus B(i, \epsilon \theta)$, $\theta \leq \tilde{\tau}_j \leq \min((1-\delta)d(j,\SS), (1+\epsilon^2)\theta)$ for $j \in A \cap B(i, \epsilon \theta)$, and $\tilde{\tau}_j \leq \tilde{\tau}_k + d(j, k)$ for all $j, k \in A$, with
    \[\sum_{j \in A} \beta(\tilde\tau_j, (1-\delta/2)d(i,j)) + \sum_{j \in I} \beta((1-\delta)d(j, \SS), (1-\delta/2)d(j,i)) \geq \hat{f}\,.\]
    Let us extend $\tilde\tau_j = \theta$ for all $j \in A' \setminus A$, which trivially satisfies $\tilde\tau_j \leq \min((1-\delta) d'(j, \SS), (1+\eps^2)\theta)$ for $j \in A' \cap B(i, \eps \theta)$.
    Note also that $d(j, \SS) \leq d'(j, \SS)$.
    Furthermore, for $j \in A' \setminus A$, we have $\theta \geq (1-\delta)d(j,\SS)$, so
    \[\sum_{j \in A'} \beta(\tilde\tau_j, (1-\delta/2)d(i,j)) + \sum_{j \in I'} \beta((1-\delta)d'(j, \SS), (1-\delta/2)d(j,i)) \geq \hat{f}\,.\]
    To reach a contradiction, it remains to prove that $\tilde\tau_j \leq \tilde\tau_k + d(j,k)$ for all $j, k \in A'$.
    We already know this for $j, k \in A$, and the claim is also trivial if $j \in A' \setminus A$ because $\tilde\tau_j = \theta \leq \tilde\tau_k$.
    So let us consider the case $j \in A$ and $k \in A' \setminus A$.
    Note that in this case $\tilde\tau_j \leq (1-\delta) d(j,\SS)$ and $\tilde\tau_k = \theta \geq (1-\delta) d(k, \SS)$.
    Then $\tilde\tau_j \leq (1-\delta) (d(j,k)+d(k,\SS)) \leq \tilde\tau_k + d(j,k)$, so we reach a contradiction.
\end{itemize}
Therefore, Lemmas~\ref{lem:pay-to-open-1} and~\ref{lem:pay-to-open-2} are still not satisfied, which implies that there still does not exist any $(\nu(p),0)$-openable facility, because otherwise these assumptions would trivially be satisfied.
\end{proof}

\newpage

\section{Proof of Theorem~\ref{theorem:alg-lmp-robust}}
\label{sec:logad-robust}

The goal of this section is to prove Theorem~\ref{theorem:alg-lmp-robust}.
We will prove the theorem in the following sections.
The proof is concluded in Section~\ref{subsec:lmp-proof-robust}.
We note that this section depends on the definitions in Section~\ref{sec:merging} that appear up to statement of Theorem~\ref{theorem:alg-lmp-robust}.

\subsection{Preliminaries}
Throughout the analysis, we will assume that $\delta > 0$ is small enough compared to $1/\gamma$ and $\epsilon > 0$ is small enough compared to $1/\gamma$ and $\delta$.

\paragraph{Execution of a solution.}
As in~\cite{MR4928455-Cohen25}, we associate an $\eta$-valid solution $\mathcal{H} = \{ H_1, \dots, H_L \}$ with an algorithm.
We define the {\em execution} of $\mathcal{H}$ as follows: Start with $\SS = \emptyset$ and $\theta = 1$, iterate over $p = 1, \dots, L$, and for each $p$ open the facilities in $H_p$ in the order in which they appear.
When facility $i$ is opened, it is added to $\SS$.
After each sequence $H_p$ with $p$ even, $\theta$ is increased to $(1+\eps^2) \theta$.
At any step, the set of open facilities $\SS$ can be partitioned to be the set of the open free facilities $\Sf$ and the set of the open regular facilities $\Sr$. 
We define an \emph{atomic step} in the algorithm as the opening of a facility or the increase of $\theta$ to $(1+\eps^2)\theta$.

We note that in a $(\nu,\eta)$-valid sequence it is only required for a regular facility to be $(\nu,\eta)$-openable with respect to a \emph{superset} of the set of open facilities $\SS$.
For the purposes of analysis, we will fix for each regular facility an arbitrary \emph{minimal} superset of $\SS$ with respect to which it was opened and a corresponding choice of bids $\tau_j$ in case the sequence is opened in a sequence of type $2$.
In particular, if the facility is $(\nu,\eta)$-openable with respect to $\SS$, then the superset will be $\SS$.

\paragraph{Active and inactive clients.}
At a step with time $\theta$ with open facilities $\SS$, we let\linebreak $I_{\SS} = \Set{j \in D \mid \theta \geq (1-\delta) d(j, \SS)}$ and $A_{\SS} = D \setminus I_{\SS}$ denote the sets of inactive clients and active clients, respectively. We drop the subscript $\SS$ when it is clear from context.
We say that a client $j$ \emph{becomes inactive} at an atomic step if $j \in I_{\SS}$ for the new set of open facilities $\SS$ and the $\theta$ at the end of that atomic step and if it never became inactive before.
We refer to such a client as inactive in all future steps, and otherwise as active.

For a client $j$, we denote by $\alpha_j^\star$ its ``value'' when it becomes inactive: Specifically, say that client $j$ becomes inactive at an atomic step in which a new facility was opened.
If this happened in a sequence of type $1$, or if the facility was free, we let $\alpha_j^\star = \theta$.
If this happened in a sequence of type $2$, if $j$ was inactive w.r.t. the superset w.r.t. which this facility was opened we let again $\alpha_j^\star = \theta$; otherwise we let $\alpha_j^\star = \tau_j$ as determined by the definition of $(2,\eta)$-openable.
Finally, if client $j$ becomes inactive due to the increase of $\theta$ to $(1+\eps^2)\theta$, we let $\alpha_j^\star = (1-\delta) d(j, \SS)$ for the set of open facilities at that step.

For the purposes of the analysis, we order the clients from $1$ to $n$ by the order of the atomic steps in which they become inactive. 
In case two clients become inactive in the same atomic step, we order them non-decreasingly by their $\alpha^\star$ value.

\paragraph{Variables $\SS^{(k)}$, $r_{jk}$ and $\zeta_{jk}$.}
We denote by $\SS^{(k)}$ the set of open facilities at the beginning of the atomic step in which client $k$ becomes inactive.
We also denote by $\tilde{\SS}^{(k)}$ the superset of $\SS^{(k)}$ with respect to which a facility is opened at this atomic step. (We let $\tilde{\SS}^{(k)} = \SS^{(k)}$ if this atomic step is one in which $\theta$ increases to $(1+\eps^2)\theta$ or if the opened facility is free.)
We denote by $\SS^{(n+1)} = \SS^\star$ the set of open facilities at the end of the algorithm.

\begin{observation}
\label{obs:s-mono-robust}
We have $\SS^{(k)} \subseteq \SS^{(k+1)}$ for all $1 \leq k \leq n$.
\end{observation}
\begin{proof}
    If $k=n$, the claim follows because $\SS^{(n+1)} = \SS^\star$, so we assume $k < n$.
    Because $\SS$ is monotone, we have either $\SS^{(k)} \subseteq \SS^{(k+1)}$ or $\SS^{(k)} \supsetneq \SS^{(k+1)}$.
    Suppose $k$ and $k+1$ become inactive in different atomic steps, as otherwise $\SS^{(k)} = \SS^{(k+1)}$.
    Then we argue that $\SS^{(k)} \supsetneq \SS^{(k+1)}$ cannot be the case.
    If that is the case, because $\SS$ is increased monotonically by the algorithm, it means that the atomic step in which $k$ became inactive was after the atomic step in which $k+1$ became inactive.
    But this violates the order that we defined, so it must then be the case that $\SS^{(k)} \subseteq \SS^{(k+1)}$.
\end{proof}

For $1 \le j \leq k \le n$, we set 
\[
    r_{jk} =
    \begin{cases}
        \alpha^\star_j &
        \begin{array}[t]{@{}l@{}}
            \text{if $j$ is active w.r.t. $\SS^{(k)}$ at the beginning of the atomic step}\\
            \text{in which $k$ becomes inactive},
        \end{array}
        \\[2mm]
       (1-\delta)~d(j, \SS^{(k)}) & \text{otherwise},
    \end{cases}
\]
and
\[
    \zeta_{jk} = 
    \begin{cases}
        \tau_j^{(k)} &
        \begin{array}[t]{@{}l@{}}
            \text{if $j$ is active w.r.t. $\tilde{\SS}^{(k)}$ at the beginning of the atomic step}\\
            \text{in which $k$ becomes inactive},
        \end{array}
        \\[2mm]
       (1-\delta)~d(j, \tilde{\SS}^{(k)}) & \text{otherwise}.
    \end{cases}
\]
where $\tau_j^{(k)} = \alpha_j^\star$ if $j$ becomes inactive in the same atomic step in which $k$ becomes inactive and $\theta$ otherwise.
We also let $r_{j,n+1} = \zeta_{j,n+1} = (1-\delta)d(j, \SS^\star)$. 

Let us prove some simple properties of these variables:

\begin{observation}
\label{obs:r-mono-robust}
We have $r_{jk} \geq r_{jk+1}$ for all $1 \leq j \leq k \leq n$.
\end{observation}
\begin{proof}
    We have trivially that $\alpha_j^\star \leq \alpha_j^\star$ and $(1-\delta) d(j, \SS^{(k)}) \geq (1-\delta) d(j, \SS^{(k')})$ for $k < k'$ as $\SS^{(k)}$ is monotone. Now consider the client $\ell$ such that $j$ is active w.r.t. $\SS^{(\ell)}$ and inactive w.r.t. $\SS^{(\ell+1)}$. Then, $r_{j\ell} = \alpha^\star_j$ and $r_{j\ell+1} = (1-\delta)d(j,\SS^{(\ell+1)})$. Moreover, if $\ell \neq j$ then $j$ and $\ell$ become inactive in the same atomic step. We consider different cases depending on whether the atomic step happens during the execution of a sequence of type $1$, type $2$, or when $\theta$ increases to $(1+\eps^2)\theta$.
    \begin{enumerate}
        \item If $j$ becomes inactive in a type $1$ sequence, then $\alpha_j^\star = \theta \ge (1-\delta)d(j,\SS^{(\ell+1)})$.
        \item Consider the case where $j$ becomes inactive in a type $2$ sequence. Then $\alpha^\star_j \ge \theta$ and  $\theta  \ge (1-\delta)d(j,\SS^{(\ell+1)})$.
        \item If $j$ becomes inactive when $\theta$ increases, then we set $\alpha^\star_j = (1-\delta)d(j,\SS^{(\ell)}) \geq (1-\delta)d(j,\SS^{(\ell+1)})$.
    \end{enumerate}   
\end{proof}

\begin{observation}
\label{obs:tau}
If $\tau_j > \theta$ at any point in the execution of a sequence of type $2$, then $\tau_j = \alpha_j^\star$.
\end{observation}
\begin{proof}
If $\tau_j > \theta$, then for the facility $i$ that is opened in that atomic step we must have $d(j, i) \leq \epsilon \theta$, so $j$ becomes inactive when $i$ is opened and hence $\alpha_j^\star = \tau_j$.
\end{proof}

\begin{observation}
\label{obs:rzdom}
$\zeta_{jk} \leq r_{jk}$ for all $1 \leq j \leq k \leq n$.
\end{observation}
\begin{proof}
If $k$ becomes inactive in an atomic step in which $\theta$ increases to $(1+\eps^2)\theta$, then $\SS^{(k)} = \tilde{\SS}^{(k)}$ and $\zeta_{jk} = r_{jk}$ as in this case both branches in the definitions of $\zeta$ and $r$ coincide.
Else, suppose this is not the case.
If $j$ is inactive with respect to $\SS^{(k)}$, then it is also inactive with respect to $\tilde{\SS}^{(k)}$, and the result holds because $(1-\delta) d(j, \tilde{\SS}^{(k)}) \leq (1-\delta) d(j, \SS^{(k)})$.
If $j$ is active with respect to $\SS^{(k)}$ but inactive with respect to $\tilde{\SS}^{(k)}$, then the result holds because $\alpha_j^\star \geq \theta$ but $\theta \geq (1-\delta)d(j, \tilde{\SS}^{(k)})$.
Finally, if $j$ is active with respect to both $\SS^{(k)}$ and $\tilde{\SS}^{(k)}$, then we have if $\tau_j^{(k)} = \theta$ that $\alpha^\star_j \geq \theta$, and otherwise we have by~\ref{obs:tau} that $j$ becomes inactive in the same atomic step as $k$, and hence $\tau_j^{(k)} = \alpha_j^\star$.
\end{proof}

\paragraph{Variables $\SS_k$, $R_{jk}$, and $Z_{jk}$.}
Let $\SS_0 \subsetneq \SS_1 \subsetneq \ldots \subsetneq \SS_m = \SS^\star$ with $m = |\SS^\star|$ be the sets of opened facilities during the execution of the algorithm.
Also let $\tilde{\SS}_k \supseteq \SS_k$ be the superset of facilities with respect to which the $(k+1)$-th facility is opened. (We let $\tilde{\SS}_k = \SS_k$ if the opened facility is free.)
Let $s(j)$ be the index $k$ such that $\SS_k = \SS^{(j)}$ and $s(n+1)=m$.

For $1 \le j \leq n$ and $0 \leq k \le m - 1$ such that $s(j) \leq k$, we set 
\[
    R_{jk} =
    \begin{cases}
        r_{j,j} & \text{if $k = s(j)$}\\
        \alpha^\star_j &
        \begin{array}[t]{@{}l@{}}
            \text{if $j$ is active w.r.t. $\SS_k$ at the beginning of the atomic step}\\
            \text{in which facility $\SS_{k+1}\setminus \SS_k$ is opened},
        \end{array}
        \\[2mm]
        (1-\delta)d(j,\SS_k) & \text{otherwise},
    \end{cases}
\]
and
\[
    Z_{jk} =
    \begin{cases}
        \tau_{j,k} &
        \begin{array}[t]{@{}l@{}}
            \text{if $j$ is active w.r.t. $\tilde \SS_k$ at the beginning of the atomic step}\\
            \text{in which facility $\SS_{k+1}\setminus \SS_k$ is opened},
        \end{array}
        \\[2mm]
        (1-\delta)d(j,\tilde{\SS}_k) & \text{otherwise},
    \end{cases}
\]
where $\tau_{j,k} = \alpha_j^\star$ if $j$ becomes inactive in the same atomic step in which facility $k+1$ is opened and $\theta$ otherwise.
We also define $R_{j,m}= Z_{j,m} = (1-\delta) d(j, \SS^\star)$ for all $j$.
By the same reasoning as in Observations~\ref{obs:r-mono-robust} and~\ref{obs:rzdom}, we also have $R_{j,k} \geq R_{j,k+1}$ and $Z_{j,k} \leq R_{j,k}$ for all valid $j, k$.

We note that the special case $k = s(j)$ in the definition of $R$ is used such that $r_{j,k} = R_{j,s(k)}$ is satisfied even in the special case when $k$ becomes inactive in an atomic step in which $\theta$ increases.
In that case, it is possible for $j$ to be active w.r.t. $\SS^{(k)}$ (resp. $\tilde{\SS}^{(k)}$) at the beginning of the atomic step in which $k$ becomes inactive but for it to be inactive w.r.t. $\SS^{(k)}$ (resp. $\tilde{\SS}^{(k)}$) at the beginning of the atomic step in which the next facility is opened.

\textbf{Candidate facilities.}
Our analysis will also use the notion of a candidate facility:
\begin{definition}
We say that a regular facility $i$ at some point of the algorithm is a candidate if, for some superset $\tilde \SS$ of the set of open facilities $\SS$ and sets of active and inactive clients with respect to it $\tilde A=A_{\tilde \SS}$, $\tilde I = I_{\tilde \SS}$,
\begin{align*}
    \sum_{j\in \tilde A}\beta((1+\eps^2) \theta, (1- \delta/2)\dist(i,j))
    \;+\;
    \sum_{j\in \tilde I}\beta((1-\delta)\dist(j,\tilde \SS), (1-\delta/2)\dist(i,j))
    \;\ge\; \hat f - n \eta.
\end{align*}   
\end{definition}

In fact, if a facility $i$ is a candidate, then it is also a candidate with respect to the set of open facilities $\SS$ directly.
Specifically, because $\tilde A \subseteq A$, $d(j, \tilde \SS) \leq d(j,\SS)$ and for $j \in A \cap \tilde I$ we have $(1-\delta)d(j,\tilde \SS) \leq \theta$, then the fact that $i$ is a candidate with respect to $\tilde \SS$ also implies that
\begin{align*}
    \sum_{j\in A}\beta((1+\eps^2) \theta, (1- \delta/2)\dist(i,j))
    \;+\;
    \sum_{j\in I}\beta((1-\delta)\dist(j,\SS), (1-\delta/2)\dist(i,j))
    \;\ge\; \hat f - n \eta.
\end{align*}

\subsection{Algorithm Properties}

First, we prove that at every step of the algorithm the clients do not overbid toward any facility.

\begin{lemma}[Robust no overbidding]
    \label{lem:overbidding-robust}
    At any point of the execution of the sequence $\mathcal{H} = \{ H_1, \dots, H_L \}$ with time $\theta$ and any superset $\tilde \SS$ of the set of open facilities $\SS$, with $\tilde I = \{j \in D \mid \theta \geq (1-\delta) d(j, \tilde \SS)\}$, $\tilde A = D \setminus \tilde I$, we have for all $i_0 \in \FF$ that 
    \[\sum_{j\in \tilde A}\beta(\theta,\dist(i_0,j))
\;+\;
\sum_{j\in \tilde I}\beta((1-\delta)  \dist(j,\tilde \SS), \dist(i_0,j))
\;\le\; \hat f.\]
\end{lemma}
\begin{proof}
Note that $\tilde A \subseteq A$ and $\tilde I \supseteq I$.
We have $d(j, \tilde \SS) \leq d(j,\SS)$ and for $j \in A \cap \tilde I$ we have $(1-\delta)d(j,\tilde \SS) \leq \theta$,
so
\begin{align*}
\sum_{j\in \tilde A}\beta(\theta,\dist(i_0,j)) + \sum_{j\in \tilde I}\beta((1-\delta)  \dist(j,\tilde \SS), \dist(i_0,j))
\leq \sum_{j\in A}\beta(\theta,\dist(i_0,j)) + \sum_{j\in I}\beta((1-\delta)  \dist(j,\SS), \dist(i_0,j)),
\end{align*}
so it suffices to prove that
\begin{equation}
\label{eq:invariant}
    \sum_{j\in A}\beta(\theta,\dist(i_0,j))
\;+\;
\sum_{j\in I}\beta((1-\delta)  \dist(j,\SS), \dist(i_0,j))
\;\le\; \hat f.
\end{equation}
The inequality holds at the beginning of the execution because $\SS = \emptyset$ so $A=D$ and $\theta = 1 \le d(i_0,j)$ for every client $j$. 
For every $H_p \in \mathcal{H}$, we will show that the invariant still holds after opening a facility $h_t \in H_p$ and after raising $\theta$ to $(1+\epsilon^2)\theta$ at the end of executing $H_p$ if $p$ is even.

\paragraph{Type 1 and type 2}
Opening up a facility $h_t$ decreases the left hand side since $d(j,\SS)$ is monotonically decreasing with $\SS$. The only other change to the left hand side that can happen is that a client $j$ becomes inactive. In that case, however, $\theta \ge (1-\delta)d(j,\SS\cup \{h_t\})$. Thus, the left hand side decreases by the monotonicity of $\beta$.

\paragraph{Increasing $\theta$}
Next, consider a sequence $H_p \in \mathcal{H}$ with even $p$.
Let $A, I$ be the sets of active and inactive clients after raising $\theta$ to $(1+\eps^2)\theta$.
Assume for contradiction that after raising $\theta$ to $(1+\eps^2)\theta$ the invariant is violated for some facility, and let $\FF_0$ be the set of all such facilities, such that for all $i \in \FF_0$
\[
    \sum_{j \in A} \beta((1+\epsilon^2)\theta, d(i,j)) + \sum_{j \in I} \beta((1-\delta)d(j, \SS), d(i,j)) > \hat f.
\]
Note that none of the facilities $i \in \FF_0$ are open because otherwise there would be no client with a positive bid to $i$ (for clients $j \in A$ we would have $(1+\eps^2)\theta < (1-\delta)d(j,\SS) \leq (1-\delta)d(j, i)$, and for clients $j \in I$ we would have $(1-\delta)d(j,\SS) \leq (1-\delta)d(j,i) \leq d(j, i)$).

Also let $A' \supseteq A$ and $I' \subseteq I$ be the active and inactive sets before the raising of $\theta$, i.e., at the end of the execution of the previous sequence of type $2$.
For each facility $i_0 \in \FF_0$ compute $\theta \le \tau^{(i_0)} \le (1+\eps^2)\theta$ such that for $\alpha^\prime_{i_0,j} = \min\bigl(\tau^{(i_0)},\,(1-\delta)\dist(j,\SS)\bigr)$ the invariant holds with equality
\[
    \sum_{j \in A'} \beta(\alpha^\prime_{i_0,j}, d(i_0,j)) + \sum_{j \in I'} \beta((1-\delta)d(j, \SS), d(i_0,j)) = \hat f.
\]
Let us prove that such $\tau^{(i_0)}$ exists: We know by the invariant before this atomic step that
\[\sum_{j \in A'} \beta(\theta, d(i_0,j)) + \sum_{j \in I'} \beta((1-\delta)d(j, \SS), d(i_0,j)) \leq \hat f,\]
so it also holds for $\alpha^\prime_{i_0,j}$ with the choice $\tau^{(i_0)} = \theta$.
We also show that the invariant is violated for $\alpha^\prime_{i_0,j}$ with the choice $\tau^{(i_0)} = (1+\eps^2)\theta$, and then the existence of $\tau^{(i_0)}$ follows by continuity.
Recall that
\[\sum_{j \in A} \beta((1+\epsilon^2)\theta, d(i_0,j)) + \sum_{j \in I} \beta((1-\delta)d(j, \SS), d(i_0,j)) > \hat f.\]
Then, for $j \in A$ we have $(1+\eps^2)\theta < (1-\delta)d(j,\SS)$ and so $\alpha^\prime_{i_0,j} = (1+\eps^2)\theta$, and for $j \in A' \cap I$ we have $(1+\eps^2)\theta \geq (1-\delta) d(j,\SS)$ and so $\alpha^\prime_{i_0,j} = (1-\delta)d(j,\SS)$, so also
\[\sum_{j \in A'} \beta(\alpha^\prime_{i_0,j}, d(i_0,j)) + \sum_{j \in I'} \beta((1-\delta)d(j, \SS), d(i_0,j)) > \hat f.\]
Therefore $\tau^{(i_0)}$ exists by continuity.

Now select the facility $i \in \FF_0$ that minimizes $\tau^{(i)}$. This implies that $\alpha^\prime_{i,j} \leq \alpha^\prime_{i_0,j}$ for all $i_0 \in \FF_0$ and $j \in D$, so
\begin{equation}
    \label{eq:minimal-facility-robust}
    \sum_{j \in A'} \beta(\alpha^\prime_{i,j}, d(i_0,j)) + \sum_{j \in I'} \beta((1-\delta)d(j, \SS), d(i_0,j)) \le \hat f\,.
\end{equation}
From now on, let $\alpha_j^\prime = \alpha_{i,j}^\prime$.
We will show that $i$ was $(2, 0)$-openable with bids $\tau_j = \min\bigl(\tau^{(i)},\,(1-\delta)\dist(j,\SS), (1+\epsilon) \theta - d(j,i)\bigr)$ for all $j \in A' \cap B(i,\epsilon \theta)$ and $\tau_j = \theta$ for all $j \in A'\setminus B(i,\epsilon \theta)$. This is a contradiction to the maximality of the $(2,\eta)$-valid sequence $H_p$.
Let us check the conditions one by one which are satisfied just as in the proof of Lemma \ref{lem:overbidding},
\begin{enumerate}
        \item[\textbf{(i)} - \textbf{(ii)}] Clearly holds by the definition of $\tau$.
        \item[\textbf{(iii)}]
        We have that 
        \begin{align*}
            \hat f & = \sum_{j \in A'} \beta(\alpha_j^\prime, d(i,j)) + \sum_{j \in I'} \beta((1-\delta)d(j, \SS), d(i,j)) \\ 
            & \le \sum_{j \in A'} \beta(\tau_j, (1- \delta/2)\dist(i,j)) + \sum_{j \in I'} \beta((1-\delta)d(j, \SS), (1- \delta/2)\dist(i,j)).
        \end{align*}
        Let us verify that the inequality holds. First, consider $j \in B(i,\epsilon \theta)$. In the case where $\tau_j = \alpha^\prime_j$ we have by the monotonicity of $\beta$ that $\beta(\alpha_j^\prime, d(i,j)) \le \beta(\tau_j,  (1-\delta/2) d(i,j))$. In the other case when $\tau_j = (1+\epsilon) \theta - d(i,j)$ observe that
        \begin{align*}
            \beta(\alpha_j^\prime, d(i,j)) & \le \pos{(1+\epsilon^2)\theta - \gamma  d(i,j)} \\ & \le \pos{(1+\epsilon)\theta - d(i,j) - \gamma (1-\delta/2) d(i,j)}
            \\ & = \beta(\tau_j,  (1-\delta/2) d(i,j))
        \end{align*}
        because using that $\gamma \geq \gamma (1-\delta/2) + \eps$ for $\delta \gg \eps$ we have
    \begin{align*}
        (1+\epsilon) \theta - d(i, j)-(\gamma (1-\delta/2)) d(i, j)
        &= (1+\epsilon) \theta - (1-\eps)d(i, j)-(\gamma (1-\delta/2) + \eps) d(i, j)\\
        &\geq (1+\eps^2) \theta - (\gamma (1-\delta/2) + \eps)d(i,j)\\
        &\geq (1+\eps^2)\theta - \gamma d(i, j)\\
        &\geq \alpha_j' - \gamma d(i, j)\,.
    \end{align*}
        For $j \in A' \setminus B(i,\epsilon \theta)$, using that $\eps^2 \theta \le \epsilon \gamma \delta /2~\theta \le \gamma \delta/2~ d(i,j)$,
        \begin{align*}
            \beta(\alpha_j^\prime, d(i,j)) = \pos{\alpha_j^\prime - \gamma  d(i,j)} \le \pos{\theta - \gamma (1-\delta/2) d(i,j)} = \beta(\tau_j,  (1-\delta/2) d(i,j)).
        \end{align*}
        \item[\textbf{(iv)}] Assume that the no-overbidding inequality is violated for some facility $i_0 \in \FF_0$,
        \begin{align*}
            \hat f & < \sum_{j\in A'}\beta(\tau_j,\dist(i_0,j))
                  \;+\;
                  \sum_{j\in I'}\beta((1-\delta)  \dist(j,\SS), \dist(i_0,j))
                  \\ & \le \sum_{j\in A'} \beta(\alpha^\prime_j,\dist(i_0,j))
                  \;+\;
                  \sum_{j\in I'}\beta((1-\delta)  \dist(j,\SS), \dist(i_0,j)).
        \end{align*}
        The last inequality holds by the monotonicity of $\beta$.
        We arrive at a contradiction to the minimal choice of $i$. In particular, a contradiction to (\ref{eq:minimal-facility-robust}). 
        On the other hand, for $i_0 \in \FF \setminus \FF_0$, the claim follows because we have 
        \[
\sum_{j \in A} \beta((1+\epsilon^2)\theta, d(i_0,j)) + \sum_{j \in I} \beta((1-\delta)d(j, \SS), d(i_0,j)) \leq \hat f
    \]
        and we have $(1+\epsilon^2)\theta \geq \tau_j$ and for $j \in A' \cap I$ we have $(1-\delta)d(j, \SS) \geq \tau_j$ so also
        \[
\sum_{j \in A'} \beta(\tau_j, d(i_0,j)) + \sum_{j \in I'} \beta((1-\delta)d(j, \SS), d(i_0,j)) \leq \hat f\,.
    \]
        \item[\textbf{(v)}] We have to show that
              \[
                  \tau_j \le \tau_k + \dist(j,k)
                  \qquad \forall j,k\in A'
              \]
              which is clearly true if $\tau_j = \theta$ or $\tau_k = \tau^{(i)}$. Moreover, if $\tau_k = (1+\epsilon)\theta - d(i,k)$ then
              \[
                \tau_j \leq (1+\epsilon) \theta - d(i, j) \leq (1+\epsilon) \theta - (d(i, k) - d(j, k)) = \tau_k + d(j,k)\,.
            \]
            From now on we assume that $j \in B(i,\epsilon \theta)$. 
            Consider the two remaining cases where either $\tau_k =\theta$ and $k \in A' \setminus B(i,\epsilon \theta)$ or $\tau_k = (1-\delta)d(k,\SS) = (1-\delta)d(k,i_0)$ for some open facility $i_0$. In the first case,
              \[
                \tau_j \le (1+\epsilon) \theta - d(j,i) \le (1+\epsilon) \theta - (d(k,i) - d(j,k)) \le \theta + d(j,k) = \tau_k + d(j,k).
              \]
              In the second case, 
              \[
                \tau_j \le (1-\delta)d(j,\SS) \le (1-\delta)d(j,i_0) \le (1-\delta) (d(i_0,k) + d(j,k)) \le \tau_k + d(j,k).
              \]
    \end{enumerate}
\end{proof}

Next, we will show that after the execution of sequence of type $1$, all facilities that are candidates during the execution of the following sequence of type $2$ are either far apart or close together. This structural property will be crucial to prove feasibility for the factor revealing program.

\begin{lemma}
\label{lem:structure-robust}
    Assume $\hat{f} > n \eta$.
    Suppose two facilities $i_1$ and $i_2$ are candidates at the same step during the execution of a sequence of type $2$.
    Then either
    $d(i_1,i_2)\le \delta \theta$ or $d(i_1,i_2)\ge \theta / \gamma$.
\end{lemma}
\begin{proof}
If $i_1 = i_2$ then the claim is immediate so we assume $i_1 \neq i_2$.
Let $N_1 = A \cap B(i_1, \epsilon \theta)$ and $N_2 = A \cap B(i_2, \epsilon \theta)$.
Let us assume without loss of generality that $|N_1| \le |N_2|$.
Let us also assume for the sake of contradiction that $\delta \theta < d(i_1, i_2) < \theta/\gamma$, which implies that $N_1$ and $N_2$ are disjoint as $\delta \gg \eps$.

Let 
\[P = \sum_{j\in A}\beta((1+\eps^2) \theta, (1- \delta/2)\dist(i_1,j)) +
\sum_{j\in I}\beta((1-\delta)\dist(j,\SS), (1-\delta/2)\dist(i_1,j))\]
and 
\[Q = \sum_{j\in A}\beta(\theta, (1-\delta)\dist(i_1,j)) +\sum_{j\in I}\beta((1-\delta)\dist(j,\SS), (1-\delta)\dist(i_1,j))\,.\]
Because $i_1$ is a candidate, $P \geq \hat f - n \eta$.
We also claim that $Q < \hat f$.
If $i_1 \not\in \SS$, we claim that this is true because otherwise facility $i_1$ would have been $(1,0)$-openable during the execution of the previous sequence of type $1$, and the sequence would not have been maximal.
To see this, let us inspect the state of the algorithm at the beginning of the execution of the sequence of type $2$.
Let $\SS' \subseteq \SS$ be the set of open facilities at the beginning of the execution of the sequence of type $2$, and let $A', I'$ be the set of active and inactive clients at that time.
Note that $A' \supseteq A$ and $I' \subseteq I$.
For $j \in I'$ we use $(1-\delta) d(j,\SS) \leq (1-\delta) d(j, \SS')$, and for $j \in A' \cap I$ we use $(1-\delta) d(j,\SS) \leq \theta$.
The last inequality is immediate if $\alpha_j^\star = \theta$, and if $\alpha_j^\star > \theta$ then $j$ must have been $\eps \theta$-close to some facility in $\SS$ and hence $d(j,\SS) \leq \eps\theta$.
Then
\begin{align*}
Q &= \sum_{j\in A}\beta(\theta, (1- \delta)\dist(i_1,j)) +
\sum_{j\in I}\beta((1-\delta)\dist(j,\SS), (1-\delta)\dist(i_1,j))\\
&\leq \sum_{j\in A'}\beta(\theta, (1- \delta)\dist(i_1,j)) +
\sum_{j\in I'}\beta((1-\delta)\dist(j,\SS'), (1-\delta)\dist(i_1,j))\,,
\end{align*}
so the last term would have been at least $\hat f$ if $Q \geq \hat f$, and then facility $i_1$ would have been $(1,0)$-openable in the execution of the previous sequence of type $1$.
On the other hand, if $i_1 \in \SS$, the claim is true simply because $\theta < (1-\delta) d(i_1, j)$ for $j \in A$ and $d(j, \SS) \leq d(i_1, j)$ for $j \in I$, so the left-hand side would be zero.
Hence $Q < \hat f$.

We will prove that $Q \geq P + n \eta$, which is impossible because $Q < \hat f \leq P + n \eta$, and thus we reach the desired contradiction.
Let us start by analyzing the active terms in $P$ and $Q$.
We consider three cases:
\begin{itemize}
    \item \textbf{Case $j \in N_1$.} 
    Then, because $d(i_1, j) \leq \eps \theta$, we have that $\beta((1+\eps^2)\theta, (1-\delta/2)d(i_1,j))$ and $\beta(\theta,(1-\delta)d(i_1,j))$ are both positive and
    \begin{align*}
    &[\theta - \gamma(1-\delta)d(i_1,j)]^+ - [(1+\eps^2)\theta - \gamma(1-\delta/2)d(i_1,j)]^+\\
    =& -\eps^2 \theta + (\gamma \delta / 2) d(i_1, j)
    \geq -\eps^2 \theta + n \eta\,.
    \end{align*}
    \item \textbf{Case $j \in N_2$.}
    Then $\delta \theta - \eps \theta \leq d(i_1, j) \leq \theta / \gamma + \eps \theta$, and we have again that $\beta((1+\eps^2)\theta, (1-\delta/2)d(i_1,j))$ and $\beta(\theta,(1-\delta)d(i_1,j))$ are both positive and
    \begin{align*}
    &[\theta - \gamma(1-\delta)d(i_1,j)]^+ - [(1+\eps^2)\theta - \gamma(1-\delta/2)d(i_1,j)]^+\\
    =& -\eps^2 \theta + (\gamma \delta / 2) d(i_1, j)
    \geq -\eps^2 \theta + (\gamma \delta / 2)(\delta-\eps)\theta
    \geq \eps^2 \theta + n\eta\,,
    \end{align*}
    where we used $\delta \gg \eps$.
    \item \textbf{Case $j \in A \setminus (N_1 \cup N_2)$.}
    Then $d(i_1, j) \geq \eps \theta$.
    First, we claim that $\beta((1+\eps^2)\theta, (1-\delta/2)d(i_1,j))$ is positive only if $\beta(\theta,(1-\delta)d(i_1,j))$ is positive:
    Suppose
    \[(1+\eps^2)\theta - \gamma(1-\delta/2)d(i_1, j) \geq 0\,.\]
    Then
    \begin{align*}
    \theta - \gamma (1-\delta) d(i_1, j)
    \geq \frac{\gamma(1-\delta/2)}{1+\eps^2}d(i_1,j)-\gamma(1-\delta)d(i_1,j) \geq 0\,,
    \end{align*}
    where we used $\delta \gg \eps$.
    So it suffices to consider the case when both are positive.
    In that case,
    \begin{align*}
    &[\theta - \gamma(1-\delta)d(i_1,j)]^+ - [(1+\eps^2)\theta - \gamma(1-\delta/2)d(i_1,j)]^+\\
    =& -\eps^2 \theta + (\gamma \delta / 2) d(i_1, j)
    \geq -\eps^2 \theta + (\gamma \delta / 2) \eps \theta
    \geq n \eta\,,
    \end{align*}
    where we used again $\delta \gg \eps$.
\end{itemize}
Summing up these contributions over all $j \in A$, and using that $|N_1| \leq |N_2|$, we get that $Q \geq P$.
In fact, if there exists some $j \in A$ with non-zero contribution to $P$, summing up this contributions gives the stronger bound $P \leq Q - n \eta$.
Note that the inactive term of $P$ is trivially upper bounded by the inactive term of $Q$.
Therefore we reach the desired contradiction as long as there exists some $j \in A$ with non-zero contribution to $P$.

Finally, consider the special case when all clients $j \in A$ have zero contribution to $P$: $\beta((1+\eps^2)\theta, (1-\delta/2)d(i_1, j)) = 0$.
Because $P \geq \hat f - n \eta > 0$, there exists some $j \in I$ such that $\beta((1-\delta)d(j,\SS), (1-\delta/2)d(i_1, j)) > 0$, and then also $\beta((1-\delta)d(j,\SS), (1-\delta)d(i_1, j)) > 0$, i.e., the contribution of $j$ to $Q$ is also positive and at least as large.
Then, in fact, the contribution of $j$ to $Q$ is larger by at least $n \eta$:
\[[(1-\delta)d(j,\SS) - \gamma(1-\delta)d(i_1,j)]^+ - [(1-\delta)d(j,\SS) - \gamma(1-\delta/2)d(i_1,j)]^+ \geq (\gamma\delta/2)d(i_1,j) \geq n\eta\]
because $d(i_1,j) \geq 1$ as discussed in Section~\ref{sec:prelim}.
Thus, we reach again the desired contradiction, and this concludes the proof.
\end{proof}

We also have the following analogue of the structural result in Lemma~\ref{lem:freezing}.

\begin{lemma}
    \label{lem:freezing-robust}
    Assume $\hat{f} > n \eta$ and $\rho \geq 1.01 \gamma$.
    Consider a regular facility $i^\star$ and a subset of clients $D^\star \subseteq [n]$ such that for all $j \in D^\star$ it holds that $\alpha^\star_j \geq \rho \cdot d(j, i^\star)$.
    Suppose that a client $j \in D^\star$ becomes inactive in the execution of a sequence of type $2$ with $\alpha_j^\star > \theta$ and a regular facility $i_0$ is opened.
    Suppose also that facility $i^\star$ was a candidate at the beginning of that atomic step.
    Then, for the facility $i_0$ that was opened at that time, we have $d(i^\star, i_0) \leq \delta \theta$.
    Furthermore, all other clients in $D^\star$  that were active before this atomic step become inactive in this atomic step.
\end{lemma}
\begin{proof}
    Both $i^\star$ and $i_0$ are candidates, so by Lemma~\ref{lem:structure-robust} we have that either $d(i^\star,i_0) \leq \delta \theta$ or $d(i^\star,i_0) \geq \theta/\gamma$.
    Note that, because $\alpha_j^\star > \theta$, we have $d(j, i_0) \leq \epsilon \theta$.
    In the execution of a sequence of type 2 we have $\alpha^\star_j \leq (1+\eps^2) \theta$, so combining with $\alpha_j^\star \geq \rho \cdot d(j, i^\star)$ we get that $(1+\eps^2) \theta \geq \rho \cdot d(j, i^\star)$.

    If $d(i^\star, i_0) \geq \theta/\gamma$, we have that $d(j, i^\star) \geq d(i^\star, i_0) - d(j, i_0) \geq \theta/\gamma - \epsilon\theta \geq 0.999 \cdot \theta/\gamma$ (using $1/\gamma \gg \epsilon$), which violates the assumption that $d(j, i^\star) \leq (1+\eps^2) \theta/\rho$.

    Hence, it must be the case that $d(i^\star, i_0) \leq \delta \theta$.
    Suppose some $k \in D^\star$ that is active before this atomic step does not become inactive at this atomic step.
    This implies that $\theta < (1-\delta) d(k, i_0)$.
    However, after $i_0$ is opened, this client will become inactive at the latest when its $\theta$ value increases to at least $(1-\delta) d(k, i_0)$.
    Then, because $\theta$ increases in multiples of $1+\eps^2$, we have that $\alpha^\star_k \leq (1+\eps^2)(1-\delta)d(k, i_0)$.
    On the other hand,
    \[\alpha^\star_k \geq \rho \cdot d(k, i^\star) \geq \rho \cdot (d(k, i_0) - d(i^\star, i_0)) \geq \rho \cdot d(k, i_0) - \rho \delta \theta \geq \rho \Paren{1 - \delta(1-\delta)} d(k, i_0)\,.\]
    Putting the two together, we get that
    \[\rho \Paren{1 - \delta(1-\delta)} d(k, i_0) \leq (1+\eps^2)(1-\delta)d(k, i_0)\,,\]
    which is a contradiction for $\rho \geq 1.01$ for $\delta \gg \eps > 0$ small enough.
\end{proof}

\subsection{Algorithm Cost}

We now bound the cost of the algorithm in terms of the variables $Z$ and $R$.

\begin{lemma}
\label{lemma:alg-cost-robust}
    The cost (with respect to opening cost $\hat{f}$) of a solution output by our execution is at most
    \[\sum_{j=1}^n \Paren{\sum_{i=s(j)}^{m-1} \beta(Z_{j,i}, R_{j,i+1}) + R_{j, m}} + \delta \cdot \sum_{j=1}^n d(j,\SS^\star) + |\Sr^\star| n \eta\,.\]
\end{lemma}
\begin{proof}
    The $j$-th client has connection cost $d(j, \SS^\star)$.
    We also need to calculate the contribution of $j$ to the opening of each facility in $\SS^\star$.    
    We note that, by definition, if $i$ is a free facility then its opening cost is zero.
    Otherwise, it costs $\hat{f}$ to open a regular facility, and we know that at the atomic step when it was opened we had
\[
\sum_{j\in A}\beta(\theta, (1-\delta)\dist(i,j))
\;+\;
\sum_{j\in I}\beta((1-\delta)\dist(j,\SS), (1-\delta)\dist(i,j))
\;\ge\; \hat f - n \eta.
\]
    if it was opened in a sequence of type $1$ (by Definition~\ref{def:openable-mod-eta-1})
    or 
\[
\sum_{j\in A}\beta(\tau_j, (1-\delta/2)\dist(i,j))
\;+\;
\sum_{j\in I}\beta((1-\delta)\dist(j,\SS), (1-\delta/2)\dist(i,j))
\;\ge\; \hat f - n \eta.
\]
  if it was opened in a sequence of type $2$ (by Definition~\ref{def:openable-mod-eta-2}).
  Therefore, up to an extra cost of $|\Sr^\star| n \eta$, it will suffice to bound these left-hand sides.

    \textbf{Type $1$ sequences.}
    Suppose that $i$ is opened in a sequence of type $1$ with respect to a set $\SS$ of open facilities at the beginning of that atomic step, and with active and inactive sets $A = A_{\SS}$ and $I=I_{\SS}$, respectively.
    (We observe that, if $i$ is $(1,\eta)$-openable w.r.t. a superset $\tilde{\SS}$, then it is also $(1,\eta)$-openable w.r.t. $\SS$ (because $d(j, \SS)$ is monotone in $\SS$ and by Lemma~\ref{lem:overbidding-robust} there is no overbidding), so it suffices to consider $\tilde{\SS} = \SS$, and $\zeta_{jk} = r_{jk}$ for all $j,k$.)
    Then, if $j \in A$ then the contribution is $\beta(\theta, (1-\delta)d(i,j))$ and if $j \in I$ then the contribution is $\beta((1-\delta)d(j,\SS), (1-\delta)d(i, j))$.
    
    We note that if $i$ is opened before the atomic step in which $j$ becomes inactive, we have that the contribution is zero, because $j$ is active after the opening of $i$ and hence $\theta < (1-\delta)d(j,\SS \cup \{i\})$.

    We now bound the contribution of client $j$ to the opening of the facility $\SS_{k+1} \setminus \SS_{k}$, for $k \geq s(j)$, by $\beta(Z_{j,k}, R_{j,k+1})$.
    This follows by the definitions of $Z$ and $R$.
    Specifically, if $j \in A$, then $Z_{j,s(j)} = \alpha^\star_j = \theta$, matching the above.
    If $j \in I$ but $k = s(j)$, which can happen if $j$ becomes inactive in an atomic step when $\theta$ increases, we have $Z_{j,s(j)} = \zeta_{j,j} = (1-\delta)d(j, \SS)$, again matching the above.
    For $k > s(j)$ we have by definition $Z_{j, k} = (1-\delta)d(j, \SS)$, matching the above.
    Finally, for all $k \geq s(j)$, $R_{j,k+1} = (1-\delta) d(j, \SS_{k+1})$, where $d(j, \SS_{k+1})$ is a lower bound on the distance from $j$ to the facility, so together the above suffice to upper bound the contribution by $\beta(Z_{j,k}, R_{j,k+1})$.

    \textbf{Type $2$ sequences.}
    Suppose that $i$ is opened in a sequence of type $2$ with respect to a superset $\tilde{\SS}$ of the set $\SS$ of open facilities at the beginning of that atomic step, and with active and inactive sets with respect to $\tilde{\SS}$ denoted as $\tilde{A} = A_{\tilde{\SS}}$ and $\tilde{I}=I_{\tilde{\SS}}$, respectively.
    If $j \in \tilde{A}$ then the contribution is $\beta(\tau_j, (1-\delta/2)d(j,i))$ and if $j \in \tilde{I}$ then the contribution is $\beta((1-\delta)d(j,\tilde \SS), (1-\delta/2) d(j, i))$.

    Suppose $i$ is opened before the atomic step in which $j$ becomes inactive.
    We argue again that the contribution is zero in this case.
    We have again that $j$ is active after the opening of $i$ and hence $\theta < (1-\delta) d(j, \SS \cup\{i\})$.
    Consider two cases.
    If $j \in \tilde{A}$ then we must have $\tau_j = \theta$, because otherwise by Observation~\ref{obs:tau} we have that $j$ also becomes inactive in the same atomic step.
    But $\theta < (1-\delta) d(j,\SS \cup \{i\})$, so then $\beta(\tau_j, (1-\delta/2)d(j,i))= 0$.
    If $j \in \tilde{I}$ then we have $\theta \geq (1-\delta)d(j,\tilde \SS)$ but $\theta < (1-\delta) d(j, \SS \cup \{i\}) \leq (1-\delta) d(j, i)$, so $\beta((1-\delta)d(j,\tilde \SS), (1-\delta/2) d(j, i)) = 0$ too.

    We now bound the contribution of client $j$ to the opening of the facility $\SS_{k+1} \setminus \SS_{k}$, for $k \geq s(j)$, by $\beta(Z_{j,k}, R_{j,k+1})$.
    This follows by the definitions of $Z$ and $R$.
    Specifically, if $j \in \tilde{A}$, then $Z_{j,s(j)} = \alpha_j^\star = \tau_j$, matching the above.
    If $j \in \tilde{I}$ but $k = s(j)$, we have $Z_{j,s(j)} = (1-\delta)d(j, \tilde{\SS})$, again matching the above.
    For $k > s(j)$ we have by definition $Z_{j, k} = (1-\delta)d(j, \tilde{\SS})$, matching the above.
    Finally, for all $k \geq s(j)$, $R_{j,k+1} = (1-\delta) d(j, \SS_{k+1})$, where $(1-\delta)d(j, \SS_{k+1})$ is a lower bound on $(1-\delta/2)$ times the distance from $j$ to the facility, so together the above suffice to upper bound the contribution by $\beta(Z_{j,k}, R_{j,k+1})$.

    \textbf{Putting things together.}
    The costs above only ensure that each facility is paid $\hat{f} - n \eta$, so we need to add $|\Sr^\star| n \eta$ to the final cost.
    Then the total opening costs are upper bounded by
    \[\sum_{j=1}^n \sum_{i=s(j)}^{m-1} \beta(Z_{j,i}, R_{j, i+1}) + |\Sr^\star|n\eta\,.\]
    Using that $R_{j,m} = (1-\delta)d(j, \SS^\star)$, we also have that the connection cost satisfies $d(j, \SS^\star) = R_{j,m} + \delta d(j, \SS^\star)$.
    Then the total cost is upper bounded by
    \[\sum_{j=1}^n \Paren{\sum_{i=s(j)}^{m-1} \beta(Z_{j,i}, R_{j,i+1}) + R_{j, m}} + \delta \cdot \sum_{j=1}^n d(j,\SS^\star) + |\Sr^\star| n \eta\,.\]
\end{proof}

\subsection{Factor-Revealing Program Feasibility}
\label{sec:fr-feasiblity-robust}

In this subsection we consider some $i^\star \in \FF$ and a subset of clients $\Dgood \subseteq D$ such that, for all $j \in \Dgood$, $\alpha_j^\star \geq \rho \cdot d(j, i^\star)$.
Recall $q(i) = \min \Set{i' \mid i' > i, i' \in \Dgood \cup \{n+1\}}$.
We will assume throughout that $\hat{f} > n \eta$ and $\rho \geq 1.01 \gamma$.

We define a new set of variables $r^\star_{jk}$ for $j \in \Dgood$, $k \in \Dgood \cup \{n + 1\}$ and $k \geq j$, that we will use with the role of $r$ in the factor revealing program in Definition~\ref{def:frp}.
Importantly, we show that the cost of the algorithm associated with clients $j \in \Dgood$ is upper bounded (up to a small error) by the objective of the factor revealing program with $r^\star$ variables.
We proceed to define $r^\star$, which will depend on the definition of an event \specialevent.

\begin{definition}
Let $\klast \in \Dgood$ be the client with the largest index in $\Dgood$. Suppose that:
\begin{itemize}
    \item The atomic step in which $\klast$ becomes inactive corresponds to a sequence of type $2$.
    \item $\alpha^\star_{\klast} > \theta$ for the $\theta$ at the atomic step in which $\klast$ becomes inactive.
    \item In the atomic step in which $\klast$ becomes inactive, a regular facility $i_0^\star$ is opened that is $(2,\eta)$-openable only with respect to a strict superset $\tilde{\SS}^{(\klast)}$ of $\SS^{(\klast)}$.
    \item $i^\star$ is a candidate with respect to $\SS^{(\klast)}$.
\end{itemize}
Then we say that the event \specialevent\ happens.
In that case, we also denote $\kfirst \in \Dgood$ to be the client with the smallest index in $\Dgood$ that becomes inactive in the same atomic step as $\klast$ and has $\alpha^\star_{\kfirst} > \theta$, where $\theta$ is the time of the atomic step in which $\klast$ becomes inactive.
\end{definition}

We can now define $r^\star$ for $j \in \Dgood$ and $k \in \Dgood \cup \{n + 1\}$ with $k \geq j$:
\[
    r^\star_{jk} = 
    \begin{cases}
        0 & \text{if \specialevent\ and } k=n+1\\
       \zeta_{jk} & \text{if \specialevent\ and } k \geq \kfirst \text{ and } d(j,i_0^\star) \leq \theta/3,\\
       \min(r_{j,k},\max(\zeta_{jk}, d(j, i^\star))) & \text{if \specialevent\ and } k \geq \kfirst \text{ and } d(j,i_0^\star) > \theta/3,\\
        r_{jk} & \text{otherwise}.
    \end{cases}
\]

We note that \specialevent\ captures the conditions of Lemma~\ref{lem:freezing-robust}.
In the analysis below we will make use of the following property when \specialevent\ happens:

\begin{lemma}
\label{lem:superset-far}
In the event \specialevent, we have for all $i \in \tilde{\SS}^{\klast}$ that $d(i_0^\star, i) \geq \theta$.
\end{lemma}
\begin{proof}
Suppose there exists some $i \in \tilde{\SS}^{\klast}$ such that $d(i_0^\star, i) < \theta$.
Then, for every $j \in B(i_0^\star, \epsilon \theta)$ we have $d(j, \tilde{\SS}^{\klast}) \leq d(j, i) \leq d(j, i_0^\star) + d(i, i_0^\star) \leq \epsilon \theta + \theta$, so $(1-\delta) d(j, \tilde{\SS}^{\klast}) \leq \theta$ and all such $j$ are either inactive or have $\tau_j = \theta$.

Using the above, we prove that it is impossible for $i_0^\star$ to be $(2,\eta)$-openable only with respect to a strict superset of $\SS^{\klast}$.
Because $i_0^\star$ is $(2,\eta)$-openable with respect to $\tilde{\SS}^{\klast}$ and all clients are either inactive or have $\tau_j=\theta$, we get
\[\sum_{j \in \tilde A} \beta(\theta, (1-\delta/2)d(i_0^\star,j)) + \sum_{j \in \tilde I} \beta((1-\delta)d(j,\tilde{\SS}^{\klast}), (1-\delta/2)d(i_0^\star,j)) \geq \hat{f} - n \eta\,.\]
Note that $\tilde A \subseteq A$ and $\tilde I \supseteq I$.
We have $(1-\delta) d(j, \tilde \SS^{\klast}) \leq (1-\delta) d(j, \SS^{\klast})$ and $(1-\delta)d(j, \tilde{\SS}^{\klast}) \leq \theta$ for all $j \in \tilde{I}$, and this then implies
\[\sum_{j \in A} \beta(\theta, (1-\delta/2)d(i_0^\star,j)) + \sum_{j \in I} \beta((1-\delta)d(j,\SS^{\klast}), (1-\delta/2)d(i_0^\star,j)) \geq \hat{f} - n \eta\,.\]
Thus, $i_0^\star$ is $(2,\eta)$-openable with respect to $\SS$, as it satisfies all constraints with $\tau_j = \theta$ for all $j \in A$:
The first two and the fifth constraints are trivial, the third constraint is the inequality we obtained above, and the fourth constraint follows by Lemma~\ref{lem:overbidding-robust}.
\end{proof}

We can now state the relation between the cost of the algorithm associated with clients $j \in \Dgood$ is upper bounded (up to a small error) and the objective of the factor revealing program with $r^\star$ variables.

\begin{lemma}
\label{lem:r-star-upper-bound}
We have 
\[\sum_{j \in \Dgood} \Paren{\sum_{i=s(j)}^{m-1} \beta(Z_{j,i}, R_{j,i+1}) + R_{j, m}} \leq \sum_{j \in \Dgood} \Paren{\sum_{\substack{i=j\\i \in \Dgood}}^n \beta(r^\star_{j,i}, r^\star_{j,q(i)}) + r^\star_{j,n+1}} + 6 \delta \sum_{j \in \Dgood} d(j, i^\star)\,.\]
\end{lemma}
\begin{proof}
We will prove the inequality term by term for each $j \in \Dgood$, i.e., 
\[\sum_{i=s(j)}^{m-1} \beta(Z_{j,i}, R_{j,i+1}) + R_{j, m} \leq \sum_{\substack{i=j\\i \in \Dgood}}^n \beta(r^\star_{j,i}, r^\star_{j,q(i)}) + r^\star_{j,n+1} + 6 \delta d(j, i^\star)\,.\]
If the event \specialevent\ does not happen, then the bound follows by superadditivity of $\beta$, since $Z_{j,i} \leq R_{j,i}$ for all $j,i$, $R_{j,s(j)} = r_{j,j} = r^\star_{j,j}$, $R_{j,m} = r^\star_{j,n+1}$, and $R_{j, i} \geq R_{j,i+1}$ for all valid $j, i$.
Suppose then that \specialevent\ happens.
Then, using that $Z_{j,i} \leq R_{j,i}$ for all $j, i$,
\[\sum_{i=s(j)}^{m-1} \beta(Z_{j,i}, R_{j,i+1}) + R_{j, m}
\leq \sum_{i=s(j)}^{s(\klast) - 1} \beta(R_{j,i}, R_{j,i+1}) + \sum_{i=s(\klast)}^{m-1} \beta(Z_{j,i}, R_{j,i+1}) + R_{j,m}\,.\]
If $j \geq \kfirst$, we note that $s(j) = s(\klast)$ and the first term is equal to $0$.
Otherwise, for the first term, because $j \in \Dgood$ and $\klast \in \Dgood$, and because $R_{j,s(j)} = r_{j,j}$, $R_{j,s(\klast)} = r_{j,\klast}$, and $R_{j,k} \geq R_{j,k+1}$ for all $j,k$, we can use the superadditivity of $\beta$ to get
\begin{align*}
\sum_{i=s(j)}^{s(\klast) - 1} \beta(R_{j,i}, R_{j,i+1})
&\leq \sum_{\substack{i=j\\i \in \Dgood}}^{\klast - 1} \beta(r_{j,i}, r_{j,q(i)})
\leq \sum_{\substack{i=j\\i \in \Dgood}}^{\klast - 1} \beta(r^\star_{j,i}, r^\star_{j,q(i)})\,,
\end{align*}
where in the last inequality we used that $r^\star_{j,i} = r_{j,i}$ for $i < \kfirst$ and $r_{j,i} = r_{j,q(i)}$ for $\kfirst < q(i) \leq \klast - 1$ (hence in this latter case $\beta(r_{j,i}, r_{j,q(i)}) =0$), and then for the last term $r^\star_{j,\klast} \leq r_{j,\klast}$, which gives an upper bound.
Then we have
\begin{align*}
&\sum_{i=s(j)}^{m-1} \beta(Z_{j,i}, R_{j,i+1}) + R_{j, m}\\
\leq& \sum_{\substack{i=j\\i \in \Dgood}}^{\klast - 1} \beta(r^\star_{j,i}, r^\star_{j,q(i)}) + \beta(Z_{j,s(\klast)}, R_{j,s(\klast)+1}) + \sum_{i=s(\klast) + 1}^{m-1} \beta(R_{j,i}, R_{j,i+1}) + R_{j,m}\\
\leq& \sum_{\substack{i=j\\i \in \Dgood}}^{\klast - 1} \beta(r^\star_{j,i}, r^\star_{j,q(i)}) + \beta(Z_{j,s(\klast)}, R_{j,s(\klast)+1}) + \beta(R_{j,s(\klast)+1}, R_{j,m}) + R_{j,m}\,.
\end{align*}
Also note that $Z_{j,s(\klast)} = \zeta_{j,\klast}$.
For the last three terms, let us consider two cases.

\textbf{Case $d(j, i_0^\star) \leq \theta/3$.}
In this case $r^\star_{j, \klast} = \zeta_{j, \klast}$.
If $j$ is active with respect to $\tilde{\SS}^{(\klast)}$, then $\zeta_{j, \klast} \geq \theta$.
Else, by Lemma~\ref{lem:superset-far}, we have for every $i \in \tilde{\SS}^{(\klast)}$ that $d(i, i_0^\star) \geq \theta$, so $d(j, i) \geq d(i, i_0^\star) - d(j, i_0^\star) \geq 2\theta/3$, so also $\zeta_{j, \klast} \geq (1-\delta)2\theta/3$.
Then in all cases, because $i_0^\star$ is then added to the set of open facilities, we also have $R_{j,m} \leq R_{j,s(\klast)+1} \leq d(j, i_0^\star) \leq \theta/3 \leq \zeta_{j, \klast}$, where we use that any $j \in \Dgood$ must be inactive after the atomic step.
Then, because $\zeta_{j, \klast} \geq R_{j,s(\klast)+1} \geq R_{j,m}$, we can use the superadditivity of $\beta$ to conclude
\begin{align*}
\sum_{i=s(j)}^{m-1} \beta(Z_{j,i}, R_{j,i+1}) + R_{j, m}
&\leq \sum_{\substack{i=j\\i \in \Dgood}}^{\klast - 1} \beta(r^\star_{j,i}, r^\star_{j,q(i)}) + r^\star_{j,\klast}\\
&= \sum_{\substack{i=j\\i \in \Dgood}}^n \beta(r^\star_{j,i}, r^\star_{j,q(i)}) + r^\star_{j,n+1}\,,
\end{align*}
where we also used that $r^\star_{j,n+1}=0$.

\textbf{Case $d(j, i_0^\star) > \theta/3$.} In this other case $r^\star_{j, \klast} = \min(r_{j, \klast}, \max(\zeta_{j, \klast}, d(j, i^\star)))$.
Because $i_0^\star$ is then added to the set of open facilities, we also have $R_{j,m} \leq R_{j,s(\klast)+1} \leq d(j, i_0^\star) \leq d(j, i^\star) + d(i_0^\star, i^\star)$, where we use that any $j \in \Dgood$ must be inactive after the atomic step.
We have by Lemma~\ref{lem:freezing-robust} that $d(i^\star, i_0^\star) \leq \delta \theta$, so then the upper bound above becomes $d(j, i^\star) + \delta \theta$.
We also note that if $r^\star_{j, \klast} \geq R_{j, s(\klast) + 1}$ then we can apply the superadditivity of $\beta$ and the result follows as in the previous case, so we assume that $r^\star_{j, \klast} < R_{j, s(\klast) + 1}$, which also implies that $r^\star_{j, \klast} < r_{j,\klast}$ as $R_{j, s(\klast) + 1} \leq R_{j, s(\klast)} = r_{j,\klast}$.
Then $r^\star_{j, \klast} = \max(\zeta_{j,\klast}, d(j, i^\star)) \geq d(j, i^\star) \geq R_{j,s(\klast)+1} - \delta \theta$.
Then, also using that $r^\star_{j, \klast} \geq \zeta_{j,\klast}$, we conclude
\begin{align*}
\sum_{i=s(j)}^{m-1} \beta(Z_{j,i}, R_{j,i+1}) + R_{j, m}
&\leq \sum_{\substack{i=j\\i \in \Dgood}}^{\klast - 1} \beta(r^\star_{j,i}, r^\star_{j,q(i)}) + \beta(r^\star_{j,\klast}, R_{j,s(\klast)+1}) + \beta(R_{j,s(\klast)+1}, R_{j,m}) + R_{j,m}\\
&= \sum_{\substack{i=j\\i \in \Dgood}}^{\klast - 1} \beta(r^\star_{j,i}, r^\star_{j,q(i)}) + \beta(R_{j,s(\klast)+1}, R_{j,m}) + R_{j,m},
\end{align*}
where we dropped the term $\beta(r^\star_{j,\klast}, R_{j,s(\klast)+1})$ because it equals zero as $r^\star_{j,\klast} < R_{j,s(\klast)+1}$.
Continuing, because $R_{j,m}=\beta(R_{j,m},0)$, by using the superadditivity of $\beta$,
\begin{align*}
\sum_{i=s(j)}^{m-1} \beta(Z_{j,i}, R_{j,i+1}) + R_{j, m}
&\leq \sum_{\substack{i=j\\i \in \Dgood}}^{\klast - 1} \beta(r^\star_{j,i}, r^\star_{j,q(i)}) + R_{j,s(\klast)+1}\\
&\leq \sum_{\substack{i=j\\i \in \Dgood}}^{\klast - 1} \beta(r^\star_{j,i}, r^\star_{j,q(i)}) + r^\star_{j,\klast} + \delta \theta\\
&= \sum_{\substack{i=j\\i \in \Dgood}}^n \beta(r^\star_{j,i}, r^\star_{j,q(i)}) + r^\star_{j,n+1} + \delta \theta\,,
\end{align*}
where in the last step we used that $r^\star_{j,n+1}=0$. We also have that $\theta \leq 3 d(j, i_0^\star) \leq 3 (d(j, i^\star) + d(i^\star,i_0^\star)) \leq 3d(j,i^\star) + 3 \delta \theta$, which implies $\theta \leq 6 d(j, i^\star)$, so we get
\[\sum_{i=s(j)}^{m-1} \beta(Z_{j,i}, R_{j,i+1}) + R_{j, m} \leq \sum_{\substack{i=j\\i \in \Dgood}}^n \beta(r^\star_{j,i}, r^\star_{j,q(i)}) + r^\star_{j,n+1} + 6 \delta d(j, i^\star)\,.\]
\end{proof}

In the rest of this section we will show that $\{r^\star_{jk} \mid j \leq k, j \in \Dgood, k \in \Dgood \cup \{n + 1\}\}, \{d(j, i^\star) \mid j \in \Dgood\}, \hat f$ is feasible for the factor revealing program in Definition~\ref{def:frp}, except for the normalization constraint (\ref{eq:fr-norm}).

First we will prove that the facility constraints (\ref{eq:fr-facility}) hold.
\begin{lemma}
\label{lem:fr-facility-robust}
    Consider the central facility $i^\star$ and let $k \in \Dgood$. Then,
    \[
        \sum_{\substack{j \ge k\\ j \in \Dgood}} \beta(\alpha^\star_k, d(i^\star,j)) + \sum_{\substack{j < k\\ j \in \Dgood}} \beta(r^\star_{jk}, d(i^\star,j)) \le \hat f.
    \]
\end{lemma}
\begin{proof}
    Consider the phase when $k$ goes inactive. 
    $k$ goes inactive either in an atomic step in a sequence of type $1$, type $2$, or when $\theta$ increases.
    \paragraph{Type $1$ sequence}
    Consider the state $(\SS=\SS^{(k)},\theta)$ of the algorithm just before the atomic step in which $k$ becomes inactive.
    If $k$ becomes inactive in a type $1$ sequence then $\alpha^\star_k = \theta$. Thus, by Lemma \ref{lem:overbidding-robust},
    \begin{align*}
        \hat f & \ge \sum_{j \in A \cap \Dgood} \beta(\theta, d(i^\star,j)) + \sum_{j \in I \cap \Dgood} \beta((1-\delta)d(j, \SS), d(i^\star,j)) 
        \\ & \ge \sum_{\substack{j \ge k\\ j \in \Dgood}} \beta(\alpha^\star_k, d(i^\star,j)) + \sum_{\substack{j < k\\ j \in \Dgood}} \beta(r^\star_{jk}, d(i^\star,j)),
    \end{align*}
    where we use that if $j < k$ but $j \in A$ then we have $r^\star_{jk} = \alpha_j^\star = \theta$.
    \paragraph{Type $2$ sequence}
    Again, consider the state $(\SS= \SS^{(k)},\theta)$ of the algorithm just before the atomic step in which $k$ becomes inactive.
    
    If $k$ becomes inactive in a step in which a free facility is opened, we have essentially the same case as in the type $1$ sequence and the same argument works.
    
    The case $\alpha_k^\star = \theta$ is also easy:
    We have $r_{jk}^\star = r_{jk}$, as even if \specialevent\ happens it only affects $k \geq \kfirst$ with $\alpha^\star_{\kfirst} > \theta$.
    Note that $r_{jk} = \alpha_j^\star = \theta$ for all $j < k$ with $j \in A$, as $\alpha_j^\star \geq \theta$ and $\alpha_j^\star \leq \alpha_k^\star$ (because if $j < k$ and $j \in A$ then $j$ must become inactive in the same atomic step as $k$).
    Then by Lemma \ref{lem:overbidding-robust},
    \begin{align*}
        \hat f & \ge \sum_{j \in A \cap \Dgood} \beta(\theta, d(i^\star,j)) + \sum_{j \in I\cap \Dgood} \beta((1-\delta)d(j, \SS), d(i^\star,j)) 
        \\ & \ge \sum_{\substack{j \ge k\\ j \in \Dgood}} \beta(\alpha^\star_k, d(i^\star,j)) + \sum_{\substack{j < k\\ j \in \Dgood}} \beta(r^\star_{jk}, d(i^\star,j)),
    \end{align*}
    where we use that if $j < k$ but $j \in A$ then we have $r^\star_{jk} = \alpha_j^\star = \theta$.
    
    Else, we have $\theta < \alpha_k^\star \leq (1+\eps^2) \theta$. Assume for contradiction that
    \[
        \sum_{\substack{j \ge k\\ j \in \Dgood}} \beta(\alpha^\star_k, d(i^\star,j)) + \sum_{\substack{j < k\\ j \in \Dgood}} \beta(r^\star_{jk}, d(i^\star,j)) >  \hat f.
    \]
    If $j < k$, we claim that $r^\star_{jk} \leq (1-\delta)d(j,\SS)$ if $j \in I$ and $r^\star_{jk} \leq \alpha^\star_j \leq (1+\epsilon^2)\theta$ if $j \in A$.
    These hold because $r^\star_{jk} \leq r_{jk}$, and then for $j \in I$ we have $r_{jk} = (1-\delta) d(j,\SS)$ and for $j \in A$ we have $r_{jk} = \alpha_j^\star$.
    Then, $i^\star$ is a candidate since
    \begin{align*}
        & \sum_{j\in A\cap \Dgood}\beta((1+\epsilon^2)\theta, (1- \delta/2)d(i^\star,j)) + \sum_{j\in I\cap \Dgood}\beta((1-\delta)\dist(j,\SS), (1-\delta/2)d(i^\star,j)) \\
        \ge & \sum_{\substack{j \ge k\\ j \in \Dgood}} \beta(\alpha^\star_k, d(i^\star,j)) + \sum_{\substack{j < k\\ j \in \Dgood}} \beta(r^\star_{jk}, d(i^\star,j)) > \hat f
    \end{align*}
    and all omitted terms are non-negative.
    Here, the inequality holds by the properties of $\beta$, monotonicity and the scaling.
    So $k \in \Dgood$ becomes inactive in a sequence of type $2$ while $i^\star$ is a candidate in a step in which a regular facility is opened. Thus, by Lemma \ref{lem:freezing-robust} we know that all clients in $A \cap \Dgood$ become inactive in the same atomic step as $k$.
    Let $i^\star_0$ be the facility that is opened in this atomic step with respect to a superset $\tilde{\SS} \supseteq \SS^{(k)}$.

    Note that, because $\alpha_k^\star > \theta$ we have $k \in \tilde{A}$.
    Let $\tilde{I} = \{j \in D \mid \theta \geq (1-\delta)d(j,\tilde \SS)\}$ and $\tilde{A} = D \setminus \tilde{I}$.
    We have
    \begin{align*}
    \hat f &\ge \sum_{j \in \tilde A\cap \Dgood} \beta(\tau_j, d(i^\star,j)) + \sum_{j \in \tilde  I\cap \Dgood} \beta((1-\delta)d(j, \tilde  \SS), d(i^\star,j))\,.
    \end{align*}
    We also have that $\alpha_j^\star \geq \alpha_k^\star > \theta$ for all $j \geq k$ in $\Dgood$ since all clients become inactive in the same atomic step and are hence ordered by $\alpha^\star$. Because of this, all these $j$ are also in $\tilde A$ and satisfy $\tau_j = \alpha_j^\star$.
    Also note that, for $j < k$ such that $j \in \tilde A$, we have $r^\star_{j,k} = r_{j,k} = \zeta_{j,k} = \alpha_j^\star$.
    Then
    \[\hat f \geq \sum_{\substack{j \ge k\\ j \in \Dgood}} \beta(\alpha^\star_k, d(i^\star,j)) + \sum_{\substack{j < k\\ j \in \tilde A \cap \Dgood}} \beta(r^\star_{j,k}, d(i^\star,j)) + \sum_{\substack{j < k\\ j \in \tilde I \cap \Dgood}} \beta((1-\delta)d(j, \tilde \SS), d(i^\star,j))\,.\]

    Let us reason now about the $r^\star_{jk}$ in the case $j \in \tilde I$.
    If $\tilde{\SS} = \SS$, then we have that $r^\star_{jk} = r_{jk} = (1-\delta) d(j, \SS)$ and the bound will follow easily, as shown below.
    Else, $\tilde{\SS}$ is a strict superset of $\SS$ and event \specialevent\ happens.
    In this case also $k \geq \kfirst$ since $\alpha_k^\star > \theta$.
    Consider then two cases.
    First, suppose $d(j, i_0^\star) \leq \theta / 3$.
    Then $r_{j,k}^\star = \zeta_{j,k} = (1-\delta)d(j, \tilde \SS)$.
    Second, suppose $d(j, i_0^\star) > \theta / 3$.
    Then $r_{j,k}^\star \leq \max((1-\delta) d(j, \tilde \SS), d(j,i^\star))$.
    The only worry is the case $d(j, i^\star) \geq (1-\delta) d(j,\tilde \SS)$.
    But then $\beta((1-\delta)d(j,\tilde \SS), d(j, i^\star)) = 0$ and $\beta(r^\star_{jk}, d(j, i^\star)) \leq \beta(d(j,i^\star), d(j,i^\star))=0$, so we have trivially that $\beta((1-\delta) d(j,\tilde{\SS}), d(i^\star, j)) \geq \beta(r^\star_{jk}, d(i^\star, j))$.
    Hence, in all these cases,
    \begin{align*}
    \hat f
    &\geq \sum_{\substack{j \ge k\\ j \in \Dgood}} \beta(\alpha^\star_k, d(i^\star,j)) + \sum_{\substack{j < k\\ j \in \Dgood}} \beta(r^\star_{jk}, d(i^\star,j))\,.
    \end{align*}
    \paragraph{Increase of $\theta$}
    In this case $k$ becomes inactive when we increase $\theta$ to $(1+\eps^2)\theta$.
    Let $\SS$ be the set of open facilities at the end of that atomic step, and note that $\SS=\SS^{(k)}$ since no facility is added.
    Let $A$ and $I$ be the sets of active and inactive clients w.r.t. $\SS$ and $(1+\eps^2)\theta$.
    By Lemma \ref{lem:overbidding-robust},
    \begin{align*}
    \hat f
    &\ge \sum_{j \in A\cap \Dgood} \beta((1+\eps^2)\theta, d(i^\star,j)) + \sum_{j \in I\cap \Dgood} \beta((1-\delta)d(j, \SS), d(i^\star,j))\\
    &\ge \sum_{\substack{j \ge k\\ j \in \Dgood}} \beta(\alpha^\star_k,  d(i^\star,j)) + \sum_{\substack{j < k\\ j \in \Dgood}} \beta(r^\star_{jk}, d(i^\star,j)).
    \end{align*}
    Let us verify that the last inequality holds. For $j \geq k$, if $j \in A$ then $(1+\eps^2)\theta \geq \alpha_k^\star$, and if $j \in I$ then $j$ became inactive in the same atomic step as $k$ so $(1-\delta)d(j,\SS) = \alpha_j^\star \geq \alpha_k^\star$, where we used that in this case the clients are ordered by $\alpha^\star$.
    Next, if $j < k$ then $j \in I$ since $k \in I$. 
    If $j$ became inactive in the same atomic step as $k$, we have $r_{jk} = \alpha_j^\star = (1-\delta)d(j, \SS)$.
    Else, by definition $r^\star_{jk} = (1-\delta)d(j, \SS)$.
\end{proof}

\begin{lemma}
\label{lem:fr-r-mono-robust}
    For all $1 \le j \le k \le n$ with $j, k \in \Dgood$ we have $r^\star_{j,k} \ge r^\star_{j,q(k)}$.
\end{lemma}
\begin{proof}
    Let us consider first the case when \specialevent\ does not happen, or it does happen but $q(k) < \kfirst$.
    We consider different cases depending on when $j$ becomes inactive. First, consider the case where $j$ becomes inactive strictly before $k$. Then $r^\star_{jk} = (1-\delta)d(j,\SS^{(k)}) \ge (1-\delta)d(j,\SS^{(q(k))}) = r^\star_{j q(k)}$ since $d(j, \SS)$ is monotonically decreasing with $\SS$ and $\SS^{(k)} \subseteq \SS^{(q(k))}$. Next, consider the case where $j$ becomes inactive together with $k$ but strictly before $q(k)$ then $r^\star_{jk} = \alpha^\star_j \ge (1-\delta)d(j,\SS^{(q(k))}) = r^\star_{jq(k)}$. In the remaining case when $j$, $k$ and $q(k)$ become inactive at the same time we have $r^\star_{jk} = r_{j q(k)} = \alpha^\star_j$.  

    Next, suppose \specialevent\ happens and $q(k) \geq \kfirst$.
    First, consider the border case $q(k) = \kfirst$.
    Then we note that $r^\star_{j,k} = r_{j, k}$, $r_{j, k} \geq r_{j, \kfirst}$, and $r_{j, \kfirst} \geq r^\star_{j,\kfirst}$, so $r^\star_{j,k} \geq r^\star_{j,\kfirst}$ is true.
    For $\kfirst \leq k \leq \klast - 1$, we have $r^\star_{j,k} = r^\star_{j,q(k)}$ because the clients become inactive in the same atomic step, so the claim is also true.
    Finally, $r^\star_{j,n+1}=0$, so $r^\star_{j,\klast} \geq r^\star_{j,n+1}$ is also true.
\end{proof}

\begin{lemma}
\label{lem:fr-triangle-robust}
    For all $1 \le j < k \le n$ with $j, k \in \Dgood$ we have $\alpha^\star_k \le r^\star_{jk} + d_j + d_k$.
\end{lemma}
\begin{proof}
    Let us consider first the case when \specialevent\ does not happen, or it does happen but $k < \kfirst$.
    First, consider the case when $j$ becomes inactive strictly before $k$.
    Then $r^\star_{jk} = (1-\delta)d(j,\SS^{(k)}) = (1-\delta)d(j,i_0)$ for some facility $i_0 \in \SS^{(k)}$.
    If the client $k$ becomes inactive either in a type $1$ sequence or type $2$ sequence at the opening of a free facility or with $k$ inactive with respect to the relevant superset, then $\alpha_k^\star = \theta \leq (1-\delta) d(k,\SS^{(k)}) \leq (1-\delta)d(k, i_0)$.
    Else, in a type $2$ sequence at the opening of a regular facility with $k$ active with respect to the relevant superset, $\alpha_k^\star = \tau_k \leq (1-\delta)d(k,\SS^{(k)}) \leq (1-\delta) d(k, i_0)$.
    Finally, when $\theta$ increases to $(1+\eps^2)\theta$, $\alpha_k^\star = (1-\delta)d(k,\SS) \leq (1-\delta) d(k, i_0)$.
    Then, by triangle inequality,
    \[
        \alpha^\star_k \le (1-\delta)d(k,i_0) \le (1-\delta)(d(k, i^\star) + d(i^\star, j) + d(j, i_0)) \le r^\star_{jk} + d_j + d_k.
    \]
    In the remaining case $k$ and $j$ become inactive at the same time and $r^\star_{jk} = \alpha^\star_j$.
    If $k$ and $j$ become inactive in the execution of a sequence of type $1$, or of a sequence of type $2$ when a free facility is opened, then $\alpha^\star_k = \alpha^\star_j = \theta$.
    Else, if $k$ and $j$ become inactive in the execution of a sequence of type $2$ then we consider a few cases: If $k$ and $j$ are active with respect to $\tilde{\SS}^{(k)}$, then $\alpha^\star_k \le \alpha^\star_j + d(k,j) \le \alpha^\star_j + d_k + d_j$ by condition $(v)$ in Definition~\ref{def:openable-mod-eta-2}.
    If $k$ is active with respect to $\tilde{\SS}^{(k)}$ but $j$ is not, then $\alpha_k^\star \leq (1-\delta)d(k, \tilde{\SS}^{(k)})$ by (i) and (ii) in Definition~\ref{def:openable-mod-eta-2} together with $\theta < (1-\delta)d(k, \tilde{\SS}^{(k)})$, and $\alpha_j^\star = \theta \geq (1-\delta)d(j, \tilde{\SS}^{(k)})$, so
    \[
        \alpha^\star_k \le (1-\delta)d(k,\tilde \SS^{(k)}) \le (1-\delta)(d(k, i^\star) + d(i^\star, j) + d(j, \tilde \SS^{(k)})) \le r^\star_{jk} + d_j + d_k.
    \]
    Lastly, if $k$ is inactive with respect to $\tilde{\SS}^{(k)}$ then $\alpha^\star_k = \theta$, and we have $r^\star_{jk} = \alpha^\star_j \geq \theta$ so the inequality is trivially satisfied.
    
    Otherwise, $k$ and $j$ become inactive when $\theta$ increases to $(1+\eps^2)\theta$ and 
    \[
        \alpha^\star_k = (1-\delta)d(k,\SS^{(k)}) \le (1-\delta)d(k,i_0) \le (1-\delta)(d(k,i^\star) + d(i^\star,j) + d(j,i_0)) \le \alpha^\star_j + d_k + d_j
    \]
    where $i_0$ is the facility that $j$ connects to, $\alpha^\star_j = (1-\delta)d(j,\SS^{(k)}) = (1-\delta)d(j,i_0)$. 

    Next, suppose \specialevent\ happens and $k \geq \kfirst$.
    Because $\alpha_k^\star > \theta$, we know that $k$ is active w.r.t. the superset $\tilde{\SS}^{(k)}$.
    Suppose first that $j$ is inactive w.r.t. $\tilde{\SS}^{(k)}$.
    Then $r^\star_{jk} \geq (1-\delta) d(j, \tilde{\SS}^{(k)})$.
    We also have $\alpha_k^\star \leq (1-\delta) d(k, \tilde{\SS}^{(k)})$ by condition (ii) in Definition~\ref{def:openable-mod-eta-2}, so
    \[\alpha_k^\star \leq (1-\delta) (d(j,i^\star) + d(k,i^\star) + d(j,\tilde{\SS}^{(k)})) \leq r^\star_{jk} + d_j + d_k\,.\]
    Else, $j$ is active w.r.t. $\tilde{\SS}^{(k)}$.
    Then $r^\star_{jk} = \alpha^\star_j$, and the claim follows because of the triangle inequality condition in the definition of $(2,\eta)$-openable: $\alpha^\star_k \leq \alpha^\star_j + d(j, k) \leq r^\star_{jk} + d(j, i^\star) + d(k, i^\star)$.
\end{proof}

\begin{lemma}
\label{lem:fr-r-metric-robust}
    For all $1 \le j,\ell < k \le n + 1$ with $j, \ell, k \in \Dgood \cup \{n+1\}$ we have $r^\star_{jk} \le r^\star_{\ell k} + d_j + d_\ell$.
\end{lemma}
\begin{proof}
    Let us consider first the case when \specialevent\ does not happen, or it does happen but $k < \kfirst$.
    If $k = n + 1$, we have $r^\star_{jk} = (1-\delta)d(j, \SS^\star)$ and $r^\star_{\ell k} = (1-\delta) d(\ell, \SS^\star)$, and then 
    \[r^\star_{jk} \leq (1-\delta)(d(j, i^\star) + d(\ell, i^\star) + d(\ell, \SS^\star)) \leq r^\star_{\ell k} + d_j + d_\ell\,,\]
    so assume $k \leq n$.
    
    First, consider the case where $j$ becomes inactive strictly before $k$. Then $r^\star_{jk} = (1-\delta)d(j,\SS^{(k)}) = (1-\delta)d(j,i_1)$ for some open facility $i_1$. If $\ell$ becomes inactive strictly before $k$, then $r^\star_{\ell k} = (1-\delta)d(\ell,\SS^{(k)}) = (1-\delta)d(\ell,i_2)$ for some open facility $i_2$, and because $i_1$ minimizes $(1-\delta)d(j,\SS^{(k)})$, 
    \[
        r^\star_{jk} = (1-\delta)d(j,i_1) \le (1-\delta)d(j,i_2) \le (1-\delta)(d(j,i^\star) + d(i^\star,\ell) + d(\ell,i_2)) \le r^\star_{\ell k} + d_j + d_\ell.
    \]
    Now if $\ell$ becomes inactive with $k$ then $r^\star_{\ell k} = \alpha^\star_\ell$. 
    Assume $j$ becomes inactive in the execution of a sequence of type $2$ when opening a regular facility $i_0$ and when it is active with respect to the superset with respect to which $i_0$ is opened.
    Suppose $i_0$ opens with bids $(\tau_j)_{j \in A}$.
    Let $\theta_0$ be the time at the atomic step in which $i_0$ is opened.
    Because $j$ becomes inactive when $i_0$ is opened, we must have $(1-\delta) d(j, i_0) \leq \theta_0$.
    Because $i_0$ is still open at the time when $k$ becomes inactive, we also have $r^\star_{jk} = (1-\delta)d(j, \SS^{(k)}) \leq (1-\delta)d(j,i_0)$.
    Hence $r^\star_{jk} \leq \theta_0$, but we must have $\alpha^\star_\ell \geq \theta_0$, so $r^\star_{jk} \leq r^\star_{\ell k}$.
     
    Otherwise $j$ becomes inactive in the execution of a sequence of type $1$, or of type $2$ when a free facility was opened or $j$ was inactive with respect to the relevant superset, or when $\theta$ increases to $(1+\eps^2)\theta$, and
    \[
        r^\star_{jk} = (1-\delta)d(j,\SS^{(k)}) \le \alpha^\star_j \le \alpha^\star_\ell = r^\star_{\ell k}
    \]
    since $\ell$ becomes inactive strictly after $j$: this implies that in a type $1$ sequence or in a type $2$ sequence in which $j$ was inactive with respect to the relevant superset $\alpha^\star_j$ is equal to the $\theta$ at that step and when $\theta$ increases $\alpha^\star_j$ is equal to at most $(1+\eps^2)\theta$ for the $\theta$ at that step, but then $\alpha^\star_\ell$ is in both cases at least as large.
    
    Next, consider the case where $j$ becomes inactive with $k$. Then $r^\star_{jk} = \alpha^\star_j$. If $\ell \ge j$ then $j$, $\ell$ and $k$ become inactive in the same atomic step and $r^\star_{\ell k} = \alpha^\star_\ell$. By the ordering of $\alpha^\star$ we know that 
    \[
        r^\star_{jk} = \alpha^\star_j \le \alpha^\star_\ell = r^\star_{\ell k}
    \]
    Thus, we can assume that $\ell < j$. Observe that $r^\star_{\ell k} = r^\star_{\ell j}$ since $k$ and $j$ become inactive in the same atomic step. Notice that now the proof follows from Lemma $\ref{lem:fr-triangle-robust}$ as $r^\star_{jk} = \alpha^\star_{j}$ and $r^\star_{\ell k} = r^\star_{\ell j}$.

    Next, suppose \specialevent\ happens and $k \geq \kfirst$.
    We also note that, in the case $k=n+1$ the claim holds trivially, because if \specialevent\ happens then $r^\star_{j,n+1} = 0$ for all $j$, so assume $k \leq n$.
    
    \textbf{Case 1.} If both $j$ and $\ell$ are active in the superset $\tilde{\SS}^{(k)}$, we have $r^\star_{jk} = \alpha_j^\star$ and $r^\star_{\ell k} = \alpha_\ell^\star$, and then the inequality $\alpha_j^\star \leq \alpha_\ell^\star + d_j + d_\ell$ follows because of the triangle inequality condition in the definition of $(2,\eta)$-openable: $\alpha^\star_j \leq \alpha^\star_\ell + d(j, \ell) \leq \alpha^\star_\ell + d(j, i^\star) + d(\ell, i^\star)$.
    
    \textbf{Case 2.} If $j$ is active in the superset $\tilde{\SS}^{(k)}$ but $\ell$ is not, we have $r^\star_{jk} = \alpha_j^\star$ and $r^\star_{\ell k} \geq (1-\delta)d(\ell, \tilde{\SS}^{(k)})$.
    Since $j$ is active with respect to $\tilde{\SS}^{(k)}$, we have $\alpha_j^\star \leq (1-\delta)d(j, \tilde{\SS}^{(k)})$ by (i) and (ii) in Definition~\ref{def:openable-mod-eta-2} together with $\theta < (1-\delta)d(j, \tilde{\SS}^{(k)})$.
    Then 
    \[r^\star_{jk} \leq (1-\delta)(d(j, i^\star) + d(\ell, i^\star) + d(\ell, \tilde \SS^{(k)})) \leq r^\star_{\ell k} + d(j, i^\star) + d(\ell, i^\star)\,.\]
    
    \textbf{Case 3.} If $j$ is inactive in the superset but $\ell$ is active, we have $r_{jk}^\star \leq \max((1-\delta)d(j, \tilde \SS^{(k)}), d(i^\star, j))$ and $r^\star_{\ell k} = \alpha^\star_\ell$.
    If $r^\star_{jk} \leq (1-\delta) d(j, \tilde \SS^{(k)})$, we also note that because $j$ is inactive we have $\theta \geq (1-\delta) d(j, \tilde \SS^{(k)})$ and $\alpha^\star_\ell \geq \theta$, so $\alpha^\star_\ell \geq (1-\delta) d(j, \tilde \SS^{(k)})$, so we trivially have $r^\star_{jk} \leq r^\star_{\ell k}$.
    Else $r^\star_{jk} \leq d(i^\star, j)$, and the claim follows trivially as the right-hand side contains a term $d(i^\star, j)$.

    \textbf{Case 4.} If $j$ and $\ell$ are both inactive in the superset, we use again that 
    \[r_{jk}^\star \leq \max((1-\delta)d(j, \tilde \SS^{(k)}), d(i^\star, j)).\]
    If $r^\star_{jk} \leq (1-\delta) d(j, \tilde \SS^{(k)})$, we also note that $r^\star_{\ell k } \geq (1-\delta) d(\ell, \tilde \SS^{(k)})$, so
    \[r^\star_{jk} \leq (1-\delta)(d(j, i^\star) + d(\ell, i^\star) + d(\ell, \tilde \SS^{(k)})) \leq r^\star_{\ell k} + d(j, i^\star) + d(\ell, i^\star)\,.\]
    Else $r^\star_{jk} \leq d(i^\star, j)$, and the claim follows trivially as the right-hand side contains a term $d(i^\star, j)$.

\end{proof}

\subsection{Proof of Theorem~\ref{theorem:alg-lmp-robust}}
\label{subsec:lmp-proof-robust}
Before we prove the theorem, we will need the following gluing result that bounds the cost of the algorithm in terms of the upper bound on the factor-revealing program.

\begin{lemma}
\label{lem:dgood-cost-robust}
    Assume $\hat{f} > n \eta$ and let $\rho \geq 1.01 \gamma$ be an upper bound on the optimal value of the factor-revealing program in Definition~\ref{def:frp}.
    For some $i^\star \in \FF$, consider any subset of clients $\Dgood \subseteq D$ such that, for all $j \in \Dgood$, $\alpha_j^\star \geq \rho \cdot d(j, i^\star)$.
    Then
    \[\sum_{j \in \Dgood} \Paren{\sum_{i=s(j)}^{m-1} \beta(Z_{j,i}, R_{j,i+1}) + R_{j, m}} \leq \hat{f} + (\rho+6\delta) \sum_{j \in \Dgood} d(j, i^\star)\,.\]
\end{lemma}
\begin{proof}
If $\Dgood = \emptyset$, the claim is trivial, so assume $\Dgood \neq \emptyset$.
We know that the variables $\{r^\star_{jk} \mid j \leq k, j \in \Dgood, k \in \Dgood \cup \{n + 1\}\}$, $\{d(j, i^\star) \mid j \in \Dgood\}$, $\hat f$ (see Section~\ref{sec:fr-feasiblity-robust}), with the clients reindexed to $1, \ldots, |\Dgood|$  in the same order as the original indices and $n+1$ reindexed to $|\Dgood|+1$, satisfy the constraints of the factor-revealing program in Definition~\ref{def:frp} (note that we identify $r^\star_{jj} = \alpha^\star_j$), except for the normalization constraint (\ref{eq:fr-norm}): this follows by Lemma~\ref{lem:fr-facility-robust}, Lemma~\ref{lem:fr-r-mono-robust}, Lemma~\ref{lem:fr-triangle-robust}, and Lemma~\ref{lem:fr-r-metric-robust}, and because all these variables are non-negative.

Let $\dsum = \sum_{j \in \Dgood} d(j, i^\star)$, which we know is strictly positive because $d(j, i) \geq 1$ for all $j \in D$, $i \in \FF$.
Then, because the factor-revealing program is scale-invariant (except for the normalization constraint), it follows that $\{r^\star_{jk} / \dsum \mid j \leq k, j \in \Dgood, k \in \Dgood \cup \{n + 1\}\}$, $\{d(j, i^\star) / \dsum \mid j \in \Dgood\}$, $\hat f / \dsum$ satisfy all the constraints of the factor-revealing program in Definition~\ref{def:frp}.

Let $q(i) = \min \Set{i' \mid i' > i, i' \in \Dgood \cup \{n+1\}}$.
Then the upper bound $\rho$ on the value of the factor revealing program implies that
\[\sum_{j \in \Dgood} \Paren{\sum_{\substack{i \in \Dgood\\i \geq j}} \beta\Paren{\frac{r^\star_{j,i}}{\dsum}, \frac{r^\star_{j,q(i)}}{\dsum}} + \frac{r^\star_{j,n+1}}{\dsum} } \leq \frac{\hat{f}}{\dsum} + \rho\,,\]
so multiplying by $\dsum$, 
\[\sum_{j \in \Dgood} \Paren{\sum_{\substack{i \in \Dgood\\i \geq j}} \beta\Paren{r^\star_{j,i}, r^\star_{j,q(i)}} + r^\star_{j,n+1}} \leq \hat{f} + \rho \sum_{j \in \Dgood} d(j, i^\star)\,.\]
Finally, by Lemma~\ref{lem:r-star-upper-bound}, the above also implies that
\[\sum_{j \in \Dgood} \Paren{\sum_{i=s(j)}^{m-1} \beta(Z_{j,i}, R_{j,i+1}) + R_{j, m}}  \leq \hat{f} + (\rho + 6 \delta) \cdot \sum_{j \in \Dgood} d(j, i^\star)\,.\]
\end{proof}

We are ready to prove the main theorem of this section.

\begin{proof}[Proof of Theorem~\ref{theorem:alg-lmp-robust}]
    Consider an optimal solution of cost $\OPT$ to the original problem with facility cost $f = \frac{\hat{f}}{\rho + 6 \delta}$, and let $k^\star$ denote the number of facilities opened in this solution.
    Let $\{i_1^\star, \ldots, i_{k^\star}^\star\}$ be the facilities opened by the optimal solution, and let $\{D_1^\star, \ldots, D_{k^\star}^\star\}$ such that $D_\ell^\star$ is the set of clients connected to facility $i_\ell^\star$ in the optimal solution.
    
    Let $C_{\mathrm{SOL}}$ and $F_{\mathrm{SOL}}$ be the connection cost and opening cost (with respect to $f$) of our algorithm, and let $C_{\OPT}$ and $F_{\OPT}$ be defined analogously for $\OPT$.
    By Lemma~\ref{lemma:alg-cost-robust}, which bounds the cost with respect to $\hat{f}$, and using that $f = \frac{\hat{f}}{\rho+6\delta}$, we know that 
    \[C_{\mathrm{SOL}} + (\rho + 6 \delta) \cdot F_{\mathrm{SOL}} \leq \sum_{j=1}^n \Paren{\sum_{i=s(j)}^{m-1} \beta(Z_{j,i}, R_{j,i+1}) + R_{j, m}} + \delta \cdot \sum_{j=1}^n d(j,\SS^\star) + |\Sr^\star| n \eta\,.\]
    Note that $\sum_{j=1}^n d(j, \SS^\star) = C_{\mathrm{SOL}}$, so we can rewrite
    \[(1-\delta) \cdot C_{\mathrm{SOL}} + (\rho + 6 \delta) \cdot F_{\mathrm{SOL}} \leq \sum_{j=1}^n \Paren{\sum_{i=s(j)}^{m-1} \beta(Z_{j,i}, R_{j,i+1}) + R_{j, m}} + |\Sr^\star| n \eta\,.\]

    Fix now some $\ell \in \{1, \ldots, k^\star\}$ and set $i^\star = i^\star_\ell$ and $D^\star = D^\star_{\ell}$.
    Define the subset of $D^\star$
    \[\Dgood = \Set{j \in D^\star \Bigg| \alpha^\star_j \geq \rho \cdot d(j, i^\star)}\,.\]
    For the clients in $\Dgood$, we have by Lemma~\ref{lem:dgood-cost-robust} that
    \[\sum_{j \in \Dgood} \Paren{\sum_{i=s(j)}^{m-1} \beta(Z_{j,i}, R_{j,i+1}) + R_{j, m}} \leq \hat{f} + (\rho+6\delta) \sum_{j \in \Dgood} d(j, i^\star)\,.\]
    Furthermore, for the clients $j \in D^\star \setminus \Dgood$, because $R_{j,s(j)} = \alpha_j^\star$, $R_{j,m} = \beta(R_{j,m}, 0)$, $Z_{j,i} \leq R_{j,i}$ for all valid $j,i$, and $R_{j,i} \geq R_{j,i+1}$ for all valid $j,i$, we have by the superadditivity of $\beta$ that
    \[\sum_{j \in D^\star \setminus \Dgood} \Paren{\sum_{i=s(j)}^{m-1} \beta(Z_{j,i}, R_{j,i+1}) + R_{j, m}} \leq \sum_{j \in D^\star \setminus \Dgood} \alpha^\star_j \leq \rho \sum_{j \in D^\star \setminus \Dgood} d(j, i^\star)\,.\]
    Overall, then,
    \[\sum_{j \in D^\star} \Paren{\sum_{i=s(j)}^{m-1} \beta(Z_{j,i}, R_{j,i+1}) + R_{j, m}} \leq \hat{f} + (\rho + 6 \delta) \sum_{j \in D^\star} d(j, i^\star)\,.\]
    
    Note that $D_1^\star, \ldots, D^\star_{k^\star}$ partition $[n]$. Then, summing the above over all $(i_{\ell}^\star, D^\star_{\ell})$,
    \begin{align*}
    \sum_{j=1}^n \Paren{\sum_{i=s(j)}^{m-1} \beta(Z_{j,i}, R_{j,i+1}) + R_{j, m}}
    &\leq k^\star \cdot \hat{f} + (\rho + 6 \delta) \cdot C_{\OPT}\\
    &= (\rho + 6 \delta) \cdot F_{\OPT} + (\rho+6\delta) \cdot C_{\OPT}\\
    &= (\rho + 6 \delta) \cdot \OPT\,.
    \end{align*}
    Hence, 
    \[(1-\delta) \cdot C_{\mathrm{SOL}} + (\rho + 6 \delta) \cdot F_{\mathrm{SOL}} \leq (\rho + 6 \delta) \cdot \OPT + |\Sr^\star| n \eta\,,\]
    so 
    \[C_{\mathrm{SOL}} \leq \frac{\rho + 6 \delta}{1-\delta} \Paren{\OPT - F_{\mathrm{SOL}}} + 2 |\Sr^\star| n \eta \leq (1+8\delta) \rho \Paren{\OPT - F_{\mathrm{SOL}}} + 2|\Sr^\star| n \eta\]
    for $\delta > 0$ smaller than $0.125$ and $\rho \geq 1$.
    Then the conclusion follows because $F_{\mathrm{SOL}} = |\SS^\star_{\mathrm{reg}}| \cdot f$.
\end{proof}

\end{document}